%% file: mutex.tex
\documentclass{lmcs}
\usepackage{amssymb}
\usepackage{wrapfig}

\definecolor[named]{Green}{cmyk}{1,0,0.7,0.3}
\definecolor[named]{DarkBlue}{rgb}{0,0,.7}
\newfont{\bbb}{bbm10 scaled 1100}                       
\newcommand{\IN}{\mbox{\bbb N}}                         
\newcommand{\IP}{\mbox{\bbb P}}                         
\newcommand{\IT}{\mbox{\bbb T}}                         
\DeclareSymbolFont{frenchscript}{OMS}{ztmcm}{m}{n}      
\DeclareMathSymbol{\Pow}{\mathord}{frenchscript}{80}    
\DeclareMathSymbol{\A}{\mathord}{frenchscript}{65}      
\DeclareMathSymbol{\Ce}{\mathord}{frenchscript}{67}     
\DeclareMathSymbol{\D}{\mathord}{frenchscript}{68}      
\DeclareMathSymbol{\E}{\mathord}{frenchscript}{69}      
\DeclareMathSymbol{\F}{\mathord}{frenchscript}{70}      
\DeclareMathSymbol{\K}{\mathord}{frenchscript}{75}      
\DeclareMathSymbol{\Q}{\mathord}{frenchscript}{81}      
\DeclareMathSymbol{\Tsk}{\mathord}{frenchscript}{84}    
\DeclareMathSymbol{\Z}{\mathord}{frenchscript}{90}      
\renewcommand{\phi}{\varphi}
\newcommand{\dcup}{\mathbin{\mbox{\href{https://www.physicsread.com/latex-union-symbol/\#latex-disjoint-union-symbol}%
                  {$\stackrel{\mbox{\huge .}}{\cup}$}}}}     
\newtheorem{exam}[thm]{Example}
\newtheorem{obse}[thm]{Observation}

\newenvironment{definition}[1]{\begin{defi} \rm \label{df:#1} }{\end{defi}}
\newenvironment{definitionq}[2]{\begin{defi}[#2] \rm \label{df:#1} }{\end{defi}}
\newenvironment{proposition}[1]{\begin{prop} \rm \label{pr:#1} }{\end{prop}}
\newenvironment{example}[1]{\begin{exam} \rm \label{ex:#1} }{\end{exam}}
\newenvironment{observation}[1]{\begin{obse} \rm \label{obs:#1} }{\end{obse}}

\newcommand{\df}[1]{Definition~\ref{df:#1}}
\newcommand{\ex}[1]{Example~\ref{ex:#1}}
\newcommand{\Sec}[1]{Section~\ref{sec:#1}}
\makeatletter
\def\comesfrom{\@transition\leftarrowfill}
\def\goesto{\@transition\rightarrowfill}
\def\ngoesto{\@transition\nrightarrowfill}
\def\Goesto{\@transition\Rightarrowfill}
\def\nGoesto{\@transition\nRightarrowfill}
\def\xmapsto{\@transition\mapstofill}
\def\nxmapsto{\@transition\nmapstofill}
\def\@transition#1{\@@transition{#1}}
\newbox\@transbox
\newbox\@arrowbox
\def\@@transition#1#2%
   {\setbox\@transbox\hbox
      {\vrule height 1.5ex depth .8ex width 0ex\hskip0.25em$\scriptstyle#2$\hskip0.25em}
   \ifdim\wd\@transbox<1.5em
      \setbox\@transbox\hbox to 1.5em{\hfil\box\@transbox\hfil}\fi
   \setbox\@arrowbox\hbox to \wd\@transbox{#1}
   \ht\@arrowbox\z@\dp\@arrowbox\z@
   \setbox\@transbox\hbox{$\mathop{\box\@arrowbox}\limits^{\box\@transbox}$}
   \dp\@transbox\z@\ht\@transbox 10pt
   \mathrel{\box\@transbox}}
\def\nrightarrowfill{$\m@th\mathord-\mkern-6mu%
  \cleaders\hbox{$\mkern-2mu\mathord-\mkern-2mu$}\hfill
  \mkern-6mu\mathord\not\mkern-2mu\mathord\rightarrow$}
\def\Rightarrowfill{$\m@th\mathord=\mkern-6mu%
  \cleaders\hbox{$\mkern-2mu\mathord=\mkern-2mu$}\hfill
  \mkern-6mu\mathord\Rightarrow$}
\def\nRightarrowfill{$\m@th\mathord=\mkern-6mu%
  \cleaders\hbox{$\mkern-2mu\mathord=\mkern-2mu$}\hfill
  \mkern-6mu\mathord\not\mathord\Rightarrow$}
\def\mapstofill{$\m@th\mathord\mapstochar\mathord-\mkern-6mu%
  \cleaders\hbox{$\mkern-2mu\mathord-\mkern-2mu$}\hfill
  \mkern-6mu\mathord\rightarrow$}
\def\nmapstofill{$\m@th\mathord\mapstochar\mathord-\mkern-6mu%
  \cleaders\hbox{$\mkern-2mu\mathord-\mkern-2mu$}\hfill
  \mkern-6mu\mathord\not\mkern-2mu\mathord\rightarrow$}
\makeatother 
\newcommand{\goto}[2][]{\mathrel{\goesto{~#2{\color{black}\;,\;#1}~}}}        
\newcommand{\weg}[1]{}                                  
\newcommand{\plat}[1]{\raisebox{0pt}[0pt][0pt]{#1}}     
\newcommand{\Tr}{\textit{Tr}}                           
\newcommand{\source}{\textit{src\/}}                    
\newcommand{\target}{\textit{trg\/}}                    
\newcommand{\aconc}{\mathrel{\mbox{$\smile\hspace{-.95ex}\raisebox{3.1pt}{$\scriptscriptstyle\bullet$}$}}}
\newcommand{\conc}{\smile}                              
\newcommand{\nconc}{\,\not\!\smile}                     
\newcommand{\naconc}{\mathrel{\mbox{$\,\not\!\smile\hspace{-.95ex}\raisebox{3.1pt}{$\scriptscriptstyle\bullet$}$}}}
\def\powermultiset#1{\IN^{#1}}
\newcommand{\monus}{\mathrel{\raisebox{-0pt}[0pt][0pt]{$
                      \stackrel{\raisebox{-5pt}[0pt][0pt]{\huge$\cdot$}}
                               {\raisebox{0pt}[0pt][0pt]{$-$}}$}}}
\def\precond#1{{\vphantom{#1}}^\bullet #1}               
\def\postcond#1{{#1}^\bullet}                            
\newcommand{\cT}{{\rm T}}                               
\newcommand{\Left}{\textsc{l}}                          
\newcommand{\R}{\textsc{r}}                             

\newcommand{\lni}[1][i]{\mbox{\color{blue}\it ln$_{#1}$}}
\newcommand{\en}[1][i]{\mbox{\color{blue}\it en$_{#1}$}}
\newcommand{\lc}[1][i]{\mbox{\color{red}\it lc$_{#1}$}}
\newcommand{\ec}[1][i]{\mbox{\color{red}\it ec$_{#1}$}}

\newcommand{\lnB}{\mbox{\color{magenta}\it ln$_B$}}
\newcommand{\enB}{\mbox{\color{magenta}\it en$_B$}}
\newcommand{\lnA}{\mbox{\color{blue}\it ln$_A$}}
\newcommand{\enA}{\mbox{\color{blue}\it en$_A$}}
\newcommand{\lcB}{\lc[B]}
\newcommand{\ecB}{\ec[B]}
\newcommand{\lcA}{\lc[A]}
\newcommand{\ecA}{\ec[A]}
\newcommand{\MEA}{{\color{Green}ORD}}
\newcommand{\MEB}{{\color{Green}ME}}
\newcommand{\MEC}[1][CC]{{\color{Green}EC$^{\it #1}$}}
\newcommand{\MED}[1][CC]{{\color{Green}LC$^{\it #1}$}}
\newcommand{\MEE}[1][CC]{{\color{Green}EN$^{\it #1}$}}
\newcommand{\MEF}[1][CC]{{\color{Green}LN$^{\it #1}$}}

\newcommand{\rA}{{\it readyA}}
\newcommand{\rB}{{\it readyB}}
\newcommand{\tu}{{\it turn}}
\newcommand{\tr}{{\it true}}
\newcommand{\fa}{{\it false}}

\newcommand{\procA}{{\rm A}}
\newcommand{\procB}{{\rm B}}

\newcommand{\lm}[3]{$\begin{array}{c}#3\\[-3pt]\ell_{#1} m_{#2}\end{array}$}

\newcommand{\Tu}[1]{{\it Turn}^{#1}}
\newcommand{\RA}[1]{{\it ReadyA}^{\,#1}}
\newcommand{\RB}[1]{{\it ReadyB}^{\,#1}}
\newcommand{\ass}[2]{{\it asgn}_{#1}^{\,#2}}
\newcommand{\noti}[2]{{\it n}_{#1}^{\,#2}}

\newcommand{\rt}{{\rm t}}                               

\begin{document}

\title[Modelling Mutual Exclusion in a Process Algebra with Time-outs]
    {Modelling Mutual Exclusion\texorpdfstring{\\}{ }in a Process Algebra with Time-outs}
\author[R.J. van Glabbeek]{Rob van Glabbeek}
\address{Data61, CSIRO, Sydney, Australia\hfill School of Informatics, University of Edinburgh, UK}
\thanks{Supported by Royal Society Wolfson Fellowship RSWF\textbackslash R1\textbackslash 221008}
\address{School of Computer Science and Engineering, University of New South Wales, Sydney, Australia}
\email{rvg@cs.stanford.edu}

\begin{abstract}
I show that in a standard process algebra extended with time-outs one can correctly model mutual
exclusion in such a way that starvation-freedom holds without assuming fairness or justness, even
when one makes the problem more challenging by assuming memory accesses to be atomic.
This can be achieved only when dropping the requirement of speed independence.
\end{abstract}

\maketitle


\section{Introduction}

A \emph{mutual exclusion protocol} mediates between competing processes to make sure that at any time at most
one of them visits a so-called \emph{critical section} in its code.  Such a protocol is
\emph{starvation-free} when each process that intends to enter its critical section will eventually
be allowed to do so.

As shown in \cite{KW97,Vogler02,GH15b}, it is fundamentally impossible to correctly model a mutual
exclusion protocol as a Petri net or in standard process algebras, such as CCS \cite{Mi90ccs}, CSP
\cite{BHR84,Ho85} or ACP \cite{BW90,Fok00}, unless starvation-freedom hinges on a fairness
assumption. The latter, in the view of \cite{GH15b}, does not provide an adequate solution, as
fairness assumptions are in many situations unwarranted and lead to false conclusions.

In \cite{EPTCS255.2} a correct process-algebraic rendering of mutual exclusion is
given, but only after making two important modifications to standard process algebra.  The first
involves making a justness assumption. Here \emph{justness} \cite{GH15a,GH19} is an alternative to
fairness, in some sense a much weaker form of fairness---meaning weaker than weak fairness.%
\footnote{Justness is the assumption that when a certain activity \emph{can} occur in a distributed system,
  eventually either it \emph{will} occur, or one of the resources needed to perform this activity is
  used for some other purpose. This is illustrated by the forthcoming Examples
  \ref{ex:beer}--\ref{ex:Bart separated}; justness is a strong enough assumption to conclude that
  Bart gets a beer in \ex{Bart separated}, but to ensure this even for \ex{beer} one needs the
  stronger assumption of weak fairness.}
Unlike (strong or weak) fairness, its use typically \emph{is} warranted and does not lead to false conclusions.
The second modification is the addition of a nonstandard construct---\emph{signals}---to CCS, or any other
standard process algebra.  Interestingly, both modifications are necessary; merely assuming justness, or
merely adding signals, is insufficient.

A similar process-algebraic rendering of mutual exclusion was given earlier in \cite{CDV09},
using a fairness assumption proposed in \cite{CDV06a} under the name \emph{fairness of actions}. In
\cite{GH19} fairness of actions (there called \emph{fairness of events}) was seen to coincide with
justness.

Bouwman \cite{Bou18,BLW20} points out that it is possible to correctly model mutual exclusion
without adding signals to the language at all, instead reformulating the justness requirement in
such a way that it effectively turns some actions into signals. Since the justness assumption was
fairly new, and needed to be carefully defined to describe its interaction with signals anyway,
redefining it to better capture read actions in mutual exclusion protocols is a plausible solution.

Yet justness is essential in all the above approaches.
This may be seen as problematic, because large parts of the foundations of process algebra are
incompatible with justness, and hence need to be thoroughly reformulated in a justness-friendly
way. This is pointed out in~\cite{vG19c}.\footnote{This problem has however been mitigated in
  \cite{BLW20}, where a standard process algebra is used in a way that is compatible with justness.
  This led to the successful verification of Peterson's mutual exclusion protocol under the
  assumption of justness, using the mCRL2 toolset \cite{BGKLNVWWW19}. The price to be paid for this
  is that more information needs to be encoded in the actions that are used as transition labels,
  and that many transitions that in classical models would be labelled with the hidden action
  $\tau$, need now have a visible label. The latter inhibits state-space reduction techniques that
  abstract from such actions.}

In \cite{GH15b}, the inability to correctly capture mutual exclusion in CCS and related process
algebras was seen as a sign that these process algebras lack some degree of universal expressiveness,
rather than as a statement about the impossibility of mutual exclusion. The repairs in
\cite{CDV09,EPTCS255.2,Bou18,BLW20} seek to rectify this lack of expressiveness, either by considering
language extensions, or by changing the definition of justness for the language.
My presentation \cite{vG18c} took a different perspective, and claimed that the impossibility results
of \cite{KW97,Vogler02,GH15b} can be seen as saying something about the real world, rather than
about formalisms we use to model it. The argument rests on two crucial features of mutual exclusion
protocols that I call \emph{atomicity} and \emph{speed independence}. Instead of protocol features
they can also be seen as assumptions on the hardware on which the mutual exclusion protocol will be running.

Atomicity is the assumption that \emph{memory accesses such as reads and writes take a positive amount of time,
yet two such accesses to the same store or register cannot overlap in time, so that a second memory
access can take place only after a first access is completed}.\footnote{This appears to be a
  consequence of seeing reads and writes as \emph{atomic actions}. In this paper ``atomicity''
  refers to the above assumption; it will not include the case where one memory access may abort
  or interrupt another, even when this entails that memory accesses cannot overlap in time.
  Subtly different notions of atomicity are that of an \emph{atomic register}, which merely behaves
  \emph{as if} its reads and writes are atomic, and that of an \emph{atomic transaction} in database
  systems, which needs to either complete, or be rolled back completely.}
Speed independence says that nothing may be assumed about the relative speed of processes competing
for access to the critical section, or for read/write access to some register. In particular, if two
processes are engaged in a race, and one of them has nothing else to do but performing the action
that wins the race, whilst the other has a long list of tasks that must be done first, it may still
happen that the other process wins.

When rejecting solutions to the mutual exclusion problem that are merely probabilistically correct,
or where starvation-freedom hinges on a fairness assumption, \cite{vG18c} claims, although without
written evidence, that when assuming atomicity as well as speed independence, mutual exclusion is impossible.
\Sec{impossible} of the present paper illustrates and substantiates this claim for Peterson's mutual exclusion protocol.

In \cite{GH19} the notion of justness from \cite{GH15a} was reformulated in terms of a concurrency
relation $\aconc$ between the transitions in a labelled transition system. This relation may be
inherited from a similar relation between the transitions of a Petri net or the instructions in the
pseudocode of protocol descriptions. Here $t \naconc u$ means that transition or instruction $u$
uses (takes away) a resource that is needed to perform transition or instruction $t$,
so that if $u$ occurs prior to, or instead of, $t$, it is not possible for $t$ to commence before $u$
is finished.\footnote{In standard Petri nets, standard process algebras, and many other models of
  concurrency, the concurrency relation is symmetric, in the sense that  $t \naconc u \Rightarrow u \naconc t$.
  Exceptions to symmetry are rare, but they will play a vital r\^ole in parts of this paper.
  They can occur when a transition needs a resource (like sunshine) without blocking it for others, or when a
  transition uses a resource when available, without actually needing it.}
The definitions of justness from  \cite{GH15a} and \cite{GH19} were shown equivalent in \cite{vG19}.

The assumption of atomicity can be formulated directly in terms of the concurrency relation $\aconc$.
It says about read or write instructions or transitions $\ell$ and $m$,
\begin{center}
  if $\ell$ and $m$ access the same register then $\ell \naconc m$ and $m \naconc \ell$.
\end{center}
In other words, in case $\ell$ and $m$ try to access the same register in parallel, and
$m$ wins the race for access to this register, $\ell$ cannot take place until $m$ is completed.
The case that is relevant for the mutual exclusion problem is where $\ell$ writes and $m$ reads.

I see only two alternatives to $\ell \naconc m \wedge m \naconc \ell$. One is that the memory accesses $\ell$ and $m$
overlap in time. This possibility has been investigated by Lamport~\cite{bakery}, who assumes that a
read action that overlaps with a write on the same register can return any possible value of that register.
Since the return of an unexpected value increases the set of possible behaviours of a mutual exclusion protocol,
Lamport implicitly takes the position that assuming overlap of actions makes the mutual exclusion
problem more challenging than assuming atomicity. Yet, he shows that his bakery algorithm~\cite{bakery}
constitutes an entirely correct solution. It moreover trivially is speed independent.
However, according to \cite{vG18c} atomicity is the more challenging assumption, as when assuming
overlap a correct speed-independent solution exists, and when assuming atomicity it does not.

The second alternative to $\ell \naconc m \wedge m \naconc \ell$ retains the assumption that memory
accesses to the same register cannot overlap in time, but assumes write actions to have priority
over reads. A write simply aborts a read that happens to be in progress, which can restart after the write is over.
Following \cite{CDV09}, I refer to this assumption as \emph{non-blocking reading}.
When $\ell$ is a write action and $m$ a read of the same register, it stipulates that $\ell \aconc m$,
yet $m \naconc \ell$. This yields an asymmetric concurrency relation, which was not foreseen in 
classical treatments of concurrency \cite{NPW81,GM84,BC87,Wi87a,GV87,DDM87,Ol87}.

The assumption of speed independence is built in in CCS and Petri nets, in the sense
that any correct mutual exclusion protocol formalised therein is automatically speed independent.
This is because these models lack the expressiveness to make anything dependent on speed.
In \Sec{LTSC}, following \cite{vG19c}, I define a (symmetric) concurrency relation between Petri net
transitions and between CCS transitions that is consistent with the work in \cite{NPW81,GM84,BC87,Wi87a,GV87,DDM87,Ol87}.
It always yields $\ell \naconc m$ when $\ell$ and $m$ both access the same register.
When taking this concurrency relation as an integral part of semantics of CCS or Petri nets,
it follows that also the assumption of atomicity is built in in these frameworks.
This makes the impossibility results of \cite{KW97,Vogler02,GH15b} special cases of the
impossibility claim from \cite{vG18c}. The latter can be seen as a generalisation of the former that
is not dependent on a particular modelling framework.

The process algebras of \cite{CDV09} and \cite{EPTCS255.2} model the possibility of\pagebreak[3]
non-blocking reading. This enables a correct rendering of speed independent mutual exclusion
without resorting to a fairness assumption. The first such correct model of exclusion occurs in
Vogler \cite{Vogler02} in terms of Petri nets extended with read arcs; the latter enable the modelling of
non-blocking reading. 
The correct modelling of mutual exclusion within CCS as
proposed by Bouwman \cite{Bou18} also exploits non-blocking reading. The justness assumption as
formulated by Bouwman can in retrospect be seen as an instance of justness as defined in \cite{GH19},
but based on a concurrency relation $\aconc$ between CCS transitions that essentially differs from
the one in \Sec{LTSC}, and that is not consistent with the work in
\cite{NPW81,GM84,BC87,Wi87a,GV87,DDM87,Ol87}, although it is entirely consistent with the
interleaving semantics of CCS given by Milner \cite{Mi90ccs}.
The claim in \cite{GH15b,EPTCS255.2} that mutual exclusion cannot be rendered satisfactory in CCS
holds only when seeing the concurrency relation of \Sec{LTSC} (or the resulting notion of justness)
as an integral part of this language, and hence is not in contradiction with the findings of Bouwman \cite{Bou18}.

In \cite{vG21} I extended standard process algebra with a time-out operator, thereby increasing its
absolute expressiveness, while remaining within the realm of untimed process algebra, in the sense
that the progress of time is not quantified. The present paper shows that the addition of time-outs
to standard process algebra makes it possible to correctly model mutual exclusion under the
assumption of atomicity, such that starvation-freedom holds without assuming fairness.
My witness for this claim will be a model of Peterson's mutual exclusion protocol.
In view of the above, this model will not be speed independent.

Moreover, starvation-freedom can be shown to hold, not only without assuming fairness, but even
without assuming justness.  Instead, one should make the assumption called \emph{progress} in
\cite{GH19}, which is weaker than justness, uncontroversial, unproblematic, and made (explicitly or
implicitly) in virtually all papers dealing with issues like mutual exclusion.
In contrast, \cite{vG18c} claims that even when dropping atomicity it is not possible to correctly
model mutual exclusion in a speed-independent way without at least assuming justness to obtain
starvation-freedom. \Sec{Peterson}/\ref{sec:EG} of the present paper illustrates and substantiates
also that claim for Peterson's mutual exclusion protocol.
\bigskip

\paragraph{Reading Guide}
Part~\ref{timeouts} of this paper shows how Peterson's mutual exclusion protocol can be modelled in
an extension of CCS with time-outs. This process algebra assumes atomicity, as one has $\ell\naconc
m$ whenever $m$ and $\ell$ are read and write transitions on the same register.  The model satisfies
all basic requirements for mutual exclusion protocols, and in addition achieves starvation-freedom
without assuming more than progress.


Part~\ref{impossibility results} recalls all impossibility claims discussed above, and illustrates
or substantiates them for Peterson's mutual exclusion protocol.
To make the impossibility claims precise, I have to define unambiguously what does and what does not
constitute a correct mutual exclusion protocol. This happens in Part~\ref{requirements}.
That part also covers \emph{fair schedulers} \cite{GH15b}, which are akin to mutual exclusion
protocols, and were used in \cite{GH15b} as a stepping stone to prove the impossibility result for
mutual exclusion in CCS\@.
I formalise four requirements on fair schedulers and six on mutual exclusion protocols that in
combination determine their correctness. Some of these requirements, including
starvation-freedom for mutual exclusion protocols, are parametrised with an assumption such as
progress, justness or fairness, that needs to be made to fulfil this requirement.
This leads to a hierarchy of quality criteria for fair schedulers and mutual exclusion protocols,
where the quality of such a protocol is higher when it depends on a weaker assumption.
I also propose two related mutual exclusion protocols, the \emph{gatekeeper} and the \emph{encapsulated gatekeeper},
that meet all correctness criteria when allowing (weak) fairness as parameter in some of the requirements.
As I expect most researchers in the area of mutual exclusion to agree with me that the (encapsulated)
gatekeeper is not an acceptable protocol, this underpins the verdict of \cite{GH15b} that assuming
(weak) fairness does not yield an acceptable solution.


The requirements on fair schedulers and  mutual exclusion protocols in Part~\ref{requirements}
are formulated in the language of \emph{linear-time temporal logic} \cite{Pnueli77,HuthRyan04}.
However, standard treatments of temporal logic turned out to be inadequate to formalise these requirements.
For this reason, Part~\ref{RTL} presents a form of temporal logic that is adapted for the study of
reactive systems, interacting with their environments through synchronisation of actions.
(Reactive) temporal logic primarily applies to distributed systems formalised as states in a Kripke structure.
However, it smoothly lifts to distributed systems formalised, for instance, as states in labelled
transition systems, as Petri nets, or as expressions in a process algebra like CCS\@. As explained
in \Sec{Models}, this is achieved through canonical translations from (states in) labelled transition
systems to (states in) Kripke structures, and from Petri nets or process algebra expressions to states in 
labelled transition systems. Assumptions such as progress, justness and fairness are gathered under
the heading \emph{completeness criteria}, as in essence they say which execution paths are regarded
as complete runs of a represented system. These criteria are incorporated in reactive temporal logic.
To capture justness, \Sec{LTSC} defines a concurrency relation $\naconc$ on the labelled transition systems
that occur in the translation steps from CCS or Petri nets to Kripke structures.



As a reading guide, I offer a table of contents and the diagram of Figure~\ref{dependence}.
\begin{figure}[ht]
\newcommand{\ds}[1]{\raisebox{-1.2pt}{\ref{sec:#1}}}
\input{dependence}
\centerline{\raisebox{1ex}{\box\graph}}
\color{black}
\caption{\it Dependence relations between the sections of this paper}\label{dependence}
\end{figure}
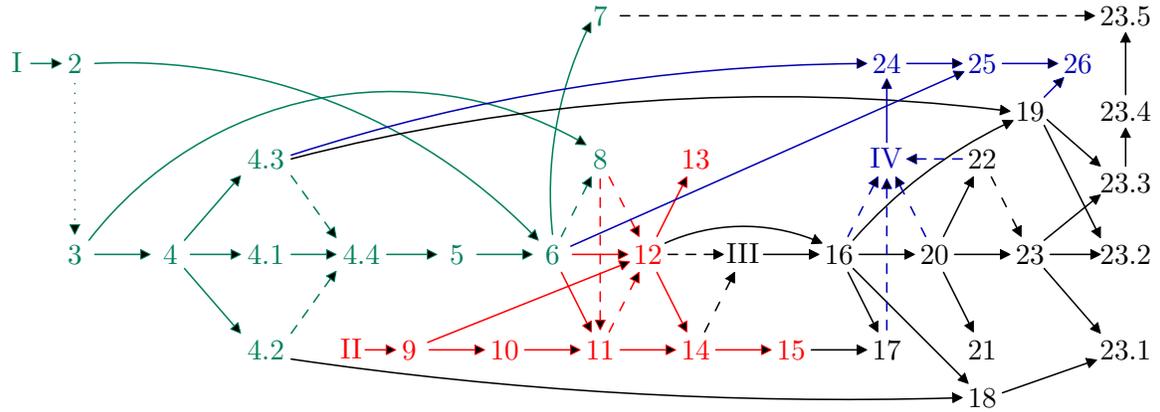
Here a dashed arrow indicates that although a concept introduced in the source section returns in
the target, in spite of this the target section can be read independently of the source.
The different colours mark the four parts of which this paper consists.
Sections \ref{sec:motivation}--\ref{sec:blocking} and \ref{sec:history}--\ref{sec:formalising mutex}
are taken from \cite{EPTCS322.6}; the only added novelty is the treatment of the next-state operator
{\bf X} in reactive linear-time temporal logic, and the corresponding mild simplification of the
requirements on fair schedulers and mutual exclusion protocols in Sections~\ref{sec:formalising FS}
and~\ref{sec:formalising mutex}.
\Sec{translating}, on translating reactive temporal logic into standard temporal logic, is also new here,
as well as \Sec{safety}, characterising a fragment of linear-time temporal logic that denotes safety formulas.
On that fragment there is no difference between reactive and standard temporal logic.
\advance\textheight 27pt

\newpage

\tableofcontents
\advance\textheight -27pt

\addtocontents{toc}{\protect\vspace{-4pt}}
\part{Reactive Temporal Logic}\label{RTL}

Whereas standard treatments of temporal logic are adequate for \emph{closed systems}, having no
run-time interactions with their environment, they fall short for \emph{reactive systems},
interacting with their environments through synchronisation of actions.
Here I present \emph{reactive temporal logic} \cite{EPTCS322.6}, a form of temporal logic adapted
for the study of reactive systems.

\section{Motivation}\label{sec:motivation}

\emph{Labelled transition systems} are a common model of distributed systems.  They consist of sets
of states, also called \emph{processes}, and transitions---each transition going from a source state
to a target state. A given distributed system $\D$ corresponds to a state $P$ in a transition system
$\IT$---the initial state of $\D$.  The other states of $\D$ are the processes in $\IT$
that are reachable from $P$ by following the transitions. The transitions are labelled by
\emph{actions}, either visible ones or the invisible action $\tau$. Whereas a $\tau$-labelled
transition represents a state-change that can be made spontaneously by the represented system,
$a$-labelled transitions, for $a\neq \tau$, merely represent potential activities of $\D$,
for they require cooperation from the \emph{environment} in which $\D$ will be running, sometimes
identified with the \emph{user} of system $\D$. A typical example is the acceptance of a
coin by a vending machine. For this transition to occur, the vending machine should be in a state
where it is enabled, i.e., the opening for inserting coins should not be closed off,
but also the user of the system should partake by inserting the coin.

{\makeatletter
\let\par\@@par
\par\parshape0
\everypar{}\begin{wrapfigure}[4]{r}{0.2\textwidth}
 \vspace{-2ex}
 \input{pretzel}
  \centerline{\raisebox{1ex}{\box\graph}}
 \end{wrapfigure}
Consider a vending machine that alternatingly accepts a coin ($c$) and dispenses a pretzel ($p$).
Its labelled transition system is depicted on the right. In standard temporal logic one can express
that each action $c$ is followed by $p$: whenever a coin is inserted, a pretzel will be dispensed.
Aligned with intuition, this formula is valid for the depicted system.
However, by symmetry one obtains the validity of a formula saying that each $p$ is followed by a
$c$: whenever a pretzel is dispensed, eventually a new coin will be inserted. But that clashes with intuition.
\par}
In \cite{EPTCS322.6} I enriched temporal logic judgements $P \models \phi$, saying that system $P$ satisfies
formula $\phi$, with a third argument $B$, telling which actions can be blocked by the environment
(by failing to act as a synchronisation partner) and which cannot. When stipulating that the coin
needs cooperation from a user, but producing the pretzel does not, the two temporal judgements can be
distinguished, and only one of them holds. I also introduced a fourth argument $CC$---a completeness
criterion---that incorporates progress, justness and fairness assumptions employed when making a
temporal judgement. This yields statements of the form $P \models^{CC}_B \phi$.\footnote{The
  technical development introduces ternary judgements $P \models^{CC} \phi$ as a primitive,
  and obtains the quaternary judgements $P \models^{CC}_B \phi$ by employing a completeness
  criterion $CC(B)$ that itself is parametrised by a set $B$ of blocking actions.}

The work in \cite{EPTCS322.6} builds on an earlier approach from \cite{GH15a}, where
judgements $P \models^{CC}_B \phi$ were effectively introduced. However, there they were written
$P \models \phi$, based on the assumption that for a given application a completeness criterion $CC$
and a set of blocking actions $B$ would be fixed. The idea was that at the beginning of a paper
employing temporal logic, a given $CC$ and $B$ would be declared, after which all forthcoming judgements
$P \models \phi$ would be interpreted as $P \models^{CC}_B \phi$. The novelty of the approach in
\cite{EPTCS322.6} is to make $CC$ and $B$ as variable as $P$ and $\phi$, so that in the description
of a single system, temporal judgements with different values of $CC$ and $B$ can be combined.

Suppose that $P$ is the initial state of the example above,
and ${\bf G}(a \Rightarrow {\bf F}b)$ is a formula that says that each action $a$ is followed by a $b$.
Abstracting from the completeness criterion for the moment, one has
\[ P \models_{\{c\}} {\bf G}(c \Rightarrow {\bf F}p) \qquad
   P \not\models_{\{c\}} {\bf G}(p \Rightarrow {\bf F}c) \qquad
   P \models_\emptyset {\bf G}(p \Rightarrow {\bf F}c)\;. \]
The first judgement says that whenever a coin is inserted, a pretzel will be dispensed,
even if we operate in an environment that may never insert a coin. By taking $B=\{c\} \not\ni p$,
the judgement also assumes that the environment will never block the production of a pretzel.

The second judgement says that in the same environment there is no guarantee that each production
of a pretzel is followed by the insertion of another coin.

The third judgement says that if we happen to run our vending machine in an environment where
the user is perpetually eager to insert a new coin, after each pretzel, acceptance of the next coin is guaranteed.
This is an important correctness property of the vending machine. Without such a property the
machine is rather unsatisfactory. Hence a specification of the kind of vending machine one would
like to have could be a combination of the first and third judgement above. This kind of
specification was not foreseen in \cite{GH15a}.

\section{Kripke Structures and Linear-time Temporal Logic}\label{sec:Kripke}

\begin{definition}{Kripke}
Let $AP$ be a set of \emph{atomic predicates}.
A \emph{Kripke structure} over $AP$ is tuple $(S, \rightarrow, \models)$ with $S$ a set (of \emph{states}),
${\rightarrow} \subseteq S \times S$, the \emph{transition relation}, and ${\models} \subseteq S \times AP$.
$s \models p$ says that predicate $p\in AP$ \emph{holds} in state $s \in S$.
\end{definition}

\noindent
Here I generalise the standard definition (see for instance \cite{HuthRyan04}) by dropping the condition of \emph{totality},
requiring that for each state $s\in S$ there is a transition $(s,s')\in {\rightarrow}$.
A \emph{path} in a Kripke structure is a nonempty finite or infinite sequence $s_0,s_1,\dots$ of states, such
that $(s_i,s_{i+1}) \in {\rightarrow}$ for each adjacent pair of states $s_i,s_{i+1}$ in that sequence.
Write $\rho \leq \pi$ if path $\rho$ is a prefix of path $\pi$.
If $\rho \leq \pi$ and $\rho$ is finite, then $\pi{\upharpoonright}\rho$ denotes the suffix of $\pi$ that remains after removing
the prefix $\rho$, but not the last state of $\rho$.
The \emph{length} of a path $\pi$, denoted $|\pi|\in\IN\cup\{\infty\}$, is the number of
transitions in $\pi$; for instance $l(s_0 s_1 s_2 s_3)=3$.

A distributed system $\D$ can be modelled as a state $s$ in a Kripke structure $K$.
A run of $\D$ then corresponds with a path in $K$ starting in $s$.
Whereas each finite path in $K$ starting from $s$ models a \emph{partial run} of $\D$,
i.e., an initial segment of a (complete) run, typically not each path models a run.
Therefore a Kripke structure constitutes a good model of distributed systems
only in combination with a \emph{completeness criterion} \cite{vG19}: a selection of a
set of paths as \emph{complete paths}, modelling runs of the represented system.

The default completeness criterion, implicitly used in almost all work on temporal logic, classifies
a path as complete iff it is infinite. In other words, only the infinite paths, and all of them,
model (complete) runs of the represented system. This applies when adopting the condition of
totality, so that each finite path is a prefix of an infinite path.  Naturally, in this setting
there is no reason to use the word ``complete'', as ``infinite'' will do.  As I plan to discuss
alternative completeness criteria in \Sec{completeness criteria}, I here already refer to paths
satisfying a completeness criterion as ``complete'' rather than ``infinite''.
Moreover, when dropping totality, the default completeness criterion is adapted to declare a path
complete iff it either is infinite or ends in a state without outgoing transitions \cite{DV95}.

\emph{Linear-time temporal logic} (LTL) \cite{Pnueli77,HuthRyan04} is a formalism explicitly designed to formulate
properties such as the safety and liveness requirements of mutual exclusion protocols. Its syntax is
\[\phi,\psi ::= p \mid \neg \phi \mid \phi \wedge \psi  \mid {\bf X}\phi \mid  {\bf F}\phi \mid
           {\bf G}\phi \mid  \psi {\bf U} \phi \]
with $p \in AP$ an atomic predicate. The propositional connectives $\Rightarrow$ and $\vee$ can be
added as syntactic sugar. It is interpreted on the paths in a Kripke structure.
The relation $\models$ between paths and LTL formulae, with $\pi\models \phi$ saying that the path
$\pi$ \emph{satisfies} the formula $\phi$, or that $\phi$ is \emph{valid} on $\pi$, is inductively
defined by
\begin{itemize}
\item $\pi \models p$, with $p\in AP$, iff $s\models p$, where $s$ is the first state of $\pi$,
\item $\pi \models \neg\phi$ iff $\pi \not\models \phi$, 
\item $\pi \models \phi \wedge \psi$ iff $\pi \models \phi$ and $\pi \models \psi$, 
\item $\pi \models {\bf X} \phi$ iff $|\pi|\mathbin>0$ and
  $\pi_{+1}\models\phi$, where $\pi_{+1}$ is obtained from $\pi$ by omitting its first state,
\item $\pi \models {\bf F} \phi$ iff $\pi{\upharpoonright}\rho\models\phi$ for some finite prefix $\rho$ of $\pi$,
\item $\pi \models {\bf G} \phi$ iff $\pi{\upharpoonright}\rho\models\phi$ for each finite prefix $\rho$ of $\pi$, and
\item $\pi \models \psi {\bf U} \phi$ iff $\pi{\upharpoonright}\rho\models\phi$ for some finite prefix $\rho$ of $\pi$,
  and $\pi{\upharpoonright}\zeta \models \psi$ for each $\zeta<\rho$.
\end{itemize}
In the standard treatment of LTL \cite{Pnueli77,HuthRyan04}, judgements $\pi\models \phi$ are
pronounced only for infinite paths $\pi$. Here I apply the same definitions verbatim to
finite paths as well. Extra care is needed only in the definition of the \emph{next-state}
operator ${\bf X} \phi$; here the condition $|\pi|\mathbin>0$ is redundant when $\pi$ is infinite.
One can define a \emph{weak next-state} operator ${\bf Y} \phi$ by
\begin{itemize}
\item $\pi \models {\bf Y} \phi$ iff $|\pi|=0$ or
  $\pi_{+1}\models\phi$, where $\pi_{+1}$ is obtained from $\pi$ by omitting the first state.
\end{itemize}
Now ${\bf Y}$ is the \emph{dual} of ${\bf X}$, in the sense ${\bf Y} \phi \equiv \neg {\bf X} \neg \phi$
and ${\bf X} \phi \equiv \neg {\bf Y} \neg \phi$, just like ${\bf F}$ is the dual of ${\bf G}$.
Here $\phi \equiv \psi$ means that $(\pi \models \phi) \Leftrightarrow (\pi \models \psi)$ for all
paths $\pi$ in all Kripke structures.
The distinction between strong and weak next-state operators stems from \cite{LPZ85},
where ${\bf F}$, ${\bf G}$, ${\bf X}$ and ${\bf Y}$ are written {\Large $\diamond$}, $\square$,
\raisebox{1pt}{\scriptsize $\bigcirc$} and {\large $\odot$}.
When only infinite paths are considered, there is no difference between ${\bf X}$ and ${\bf Y}$.

Having given meaning to judgements $\pi \models \phi$, as a derived concept one defines when an
{LTL} formula $\phi$ holds for a state $s$ in a Kripke structure, modelling a distributed system
$\D$, notation $s \models \phi$ or $\D \models \phi$. This is the case iff $\phi$
holds for all runs of $\D$.
\begin{definition}{validity}
$s \models \phi$ iff $\pi \models \phi$ for all complete paths $\pi$ starting in state $s$.
\end{definition}

\noindent
This definition depends on the underlying completeness criterion, telling which paths
model actual system runs. In situations where I consider different completeness criteria, I make this
explicit by writing $s \models^{CC} \phi$, with $CC$ the name of the completeness criterion used.
When leaving out the superscript $CC$ I refer to the default completeness criterion, defined above.

\begin{example}{beer}
Alice, Bart and Cameron stand behind a bar, continuously ordering and drinking beer.
Assume they do not know each other and order individually.
As there is only one bartender, they are served sequentially.
Also assume that none of them is served twice in a row, but as it takes no longer to drink a beer
than to pour it, each of them is ready for the next beer as soon as another person is served.
{\makeatletter
\let\par\@@par
\par\parshape0
\everypar{}\begin{wrapfigure}[6]{r}{0.25\textwidth}
 \vspace{-1ex}
 \input{Bart}
  \centerline{\raisebox{1ex}{\box\graph}}
 \end{wrapfigure}

A Kripke structure of this distributed system $\D$ is drawn on the right.
The initial state of $\D$ is indicated by a short arrow. The other three states are
labelled with the atomic predicates $A$, $B$ and $C$, indicating that Alice, Bart or Cameron,
respectively, has just acquired a beer. When assuming the default completeness criterion, valid
{LTL} formulae are ${\bf F}(A \vee C)$, saying that eventually either Alice or Cameron will get a
beer, or ${\bf G}(A \Rightarrow {\bf F}\neg A)$, saying that each time Alice got a beer
is followed eventually by someone else getting one. However, it is not guaranteed that Bart will
ever get a beer: $\D \not\models {\bf F}B$. A counterexample for this formula is the infinite run in
which Alice and Cameron get a beer alternatingly.
\par}
\end{example}

\begin{example}{Bart alone}
Bart is the only customer in a bar in London, with a single bartender.
He only wants one beer.
A Kripke structure of this system $\E$ is drawn on\hspace{10pt}\hspace{0.25\textwidth}\mbox{}\linebreak\vspace{-13pt}
{\makeatletter
\let\par\@@par
\par\parshape0
\everypar{}\begin{wrapfigure}[6]{r}{0.25\textwidth}
 \vspace{-3ex}
 \input{Bart2}
  \centerline{\raisebox{1ex}{\box\graph}}
 \end{wrapfigure}
\noindent
 the right. When assuming the default completeness criterion, this time Bart gets his beer:
$\E \models {\bf F}B$.
\par}
\end{example}

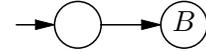
\begin{example}{Bart separated}
Bart is the only customer in a bar in London, with a single bartender.
He only wants one beer.
At the same time, Alice and Cameron are in a bar in Tokyo.
They drink a lot of beer. Bart is not in contact with Alice and Cameron, nor%
\hspace{10pt}\hspace{0.25\textwidth}\mbox{}\linebreak\vspace{-13pt}
{\makeatletter
\let\par\@@par
\par\parshape0
\everypar{}\begin{wrapfigure}[4]{r}{0.25\textwidth}
 \vspace{-2ex}
 \input{Bart3}
  \centerline{\raisebox{1ex}{\box\graph}}
 \end{wrapfigure}
\noindent
is there any connection between the two bars.
Yet, one may choose to model the drinking in these two bars as a single distributed system.
A Kripke structure of this system $\F$ is drawn on the right, collapsing the orders of
Alice and Cameron, which can occur before or after Bart gets a beer, into self-loops.
When assuming the default completeness criterion, Bart cannot count on a beer:
$\F\not\models {\bf F}B$.
\par}
\end{example}

\section{Labelled Transition Systems, Process Algebra and Petri Nets}\label{sec:Models}

The most common formalisms in which to present reactive distributed systems are pseudocode,
process algebra and Petri nets. The semantics of these formalisms is often given through translations into
labelled transition systems (LTSs), and these in turn can be translated into Kripke structures, on
which temporal formulae from languages such as LTL are interpreted. These translations make the
validity relation $\models$ for temporal formulae applicable to all these formalisms. A state in
an LTS, for example, is defined to satisfy an {LTL} formula $\phi$ iff its translation into a
state in a Kripke structure satisfies this formula.

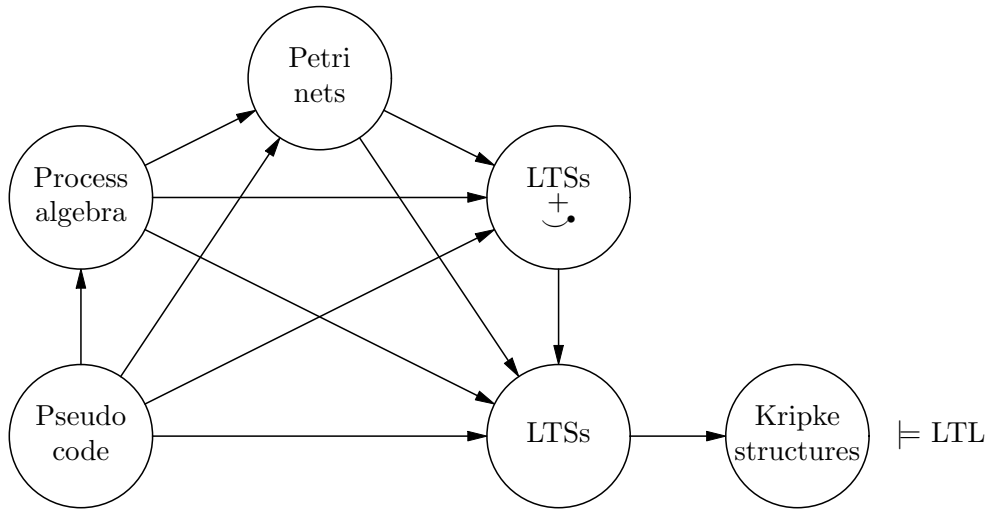
\begin{figure}[ht]
\input{models}
\centerline{\raisebox{1ex}{\box\graph}}
\caption{\it Formalisms for modelling mutual exclusion protocols}\label{models}
\end{figure}

\noindent
Figure~\ref{models} shows a commutative diagram of semantic translations found in the literature, from
pseudocode, process algebra and Petri nets via LTSs to Kripke structures. Each step in the
translation abstracts from certain features of the formalism at its source.
Some useful requirements on distributed systems can be adequately formalised in
process algebra or Petri nets, and informally described for pseudocode, whereas LTSs and Kripke
structures have already abstracted from the relevant information. An example will be FS\hspace{1pt}1 on page \pageref{FS1}.
I also consider LTSs upgraded with a concurrency relation $\aconc$ between transitions; these will be
expressive enough to formalise some of these requirements.

\subsection{Labelled Transition Systems}\label{sec:LTS}

\begin{definition}{LTS}
Let $A$ be a set of \emph{observable actions}, and let $Act := A \dcup\{\tau\}$, with $\tau\notin A$
the \emph{hidden action}.
A \emph{labelled transition system} (LTS) over $Act$ is a tuple $(\IP,  \Tr, \source,\target,\ell)$
with $\IP$ a set (of \emph{states} or \emph{processes}),
$\Tr$ a set (of \emph{transitions}), $\source,\target:\Tr\rightarrow \IP$ and $\ell:\Tr\rightarrow Act$.
\end{definition}
\noindent
Write $s \goesto{\alpha} s'$ if there is a $t\in\Tr$ with $\source(t)=s \in\IP$,
$\ell(t)=\alpha\in Act$ and $\target(t)=s'\in\IP$. In this case $t$ \emph{goes} from $s$ to $s'$,
and is an \emph{outgoing transition} of $s$. States $s$ and $s'$ are the \emph{source} and \emph{target} of $t$.
A \emph{path} in an LTS is a finite or infinite alternating sequence of states and transitions,
starting with a state, such that each transition goes from the state before it to the state after it
(if any). A \emph{completeness criterion} on an LTS is a set of its paths.

As for Kripke structures, a distributed system $\D$ can be modelled as a state $s$ in an
LTS upgraded with a completeness criterion. A (complete) run of $\D$ is then modelled 
by a complete path starting in $s$. As for Kripke structures, the default completeness criterion
deems a path complete iff it either is infinite or ends in a \emph{deadlock}, a state without
outgoing transitions. An alternative completeness criterion could declare some infinite paths
incomplete, saying that they do not model runs that can actually occur, and/or declare some
finite paths that do not end in deadlock complete. A complete path $\pi$ ending in a state
models a run of the represented system that follows the path until its last state, and then stays
in that state forever, without taking any of its outgoing transitions. A complete path that ends in
a transition models a run in which the action represented by this last transition starts occurring but
never finishes. It is often assumed that transitions are instantaneous, or at least of finite duration.
This assumption is formalised through the adoption of a completeness criterion that holds all paths
ending in a transition to be incomplete.

The most prominent translation from LTSs to Kripke structures stems from De Nicola \& Vaandrager \cite{DV95}.
Its purpose is merely to efficiently lift the validity relation $\models$ from Kripke structures to LTSs.
It simply creates a new state halfway along any transition labelled by a visible action, and moves the
transition label to that state.
\begin{definition}{DV translation}
Let $(\IP,  \Tr, \source,\target,\ell)$ be an LTS over $Act = A \cup\{\tau\}$. The associated Kripke structure
$(S,\rightarrow,\models)$ over $A$ is given by\vspace{-1pt}
\begin{itemize}
\item $S := \IP \dcup \{t \in\Tr\mid \ell(t)\neq\tau\}$,
\item ${\rightarrow} := \{(\source(t),t),(t,\target(t)) \mid t \mathbin\in\Tr\wedge \ell(t)\mathbin{\neq}\tau\} \cup 
      \{(\source(t),\target(t))\mid t\mathbin\in\Tr \wedge \ell(t)=\tau\}$
\item and ${\models} := \{(t,\ell(t)) \mid t \in \Tr \wedge \ell(t)\mathbin{\neq}\tau\}$.
\end{itemize}
\end{definition}
\noindent
Ignoring paths ending within a $\tau$-transition, which are never deemed complete anyway,
this translation yields a bijective correspondence between the paths in an LTS and those in its
associated Kripke structure. Consequently, any completeness criterion on the LTS induces a
completeness criterion on the Kripke structure.
Hence it is now well-defined when $s \models^{CC} \phi$, with $s$ a state in an LTS, $CC$ a completeness
criterion on this LTS and $\phi$ an {LTL} formula.

\subsection{Petri Nets}\label{sec:nets}

\begin{definition}{net}
  A \emph{(labelled) Petri net} over $Act$ is a tuple
  $N = (S, T, F, M_0, \ell)$ where
  \begin{itemize}
    \item $S$ and $T$ are disjoint sets (of \emph{places} and \emph{transitions}),
    \item $F: (S \mathord\times T \mathrel\cup T \mathord\times S) \rightarrow \IN$
      (the \emph{flow relation}) 
      such that  $\forall t \mathbin\in T.\,  \exists s\mathbin\in S.~ F(s, t) > 0$,
    \item $M_0 : S \rightarrow \IN$ (the \emph{initial marking}), and
    \item $\ell: T \rightarrow Act$ (the \emph{labelling function}).
  \end{itemize}
\end{definition}

\noindent
Petri nets are usually depicted by drawing the places as circles and the transitions as boxes, containing
their label. For $x,y \mathbin\in S\cup T$ there are $F(x,y)$ arrows (\emph{arcs}) from $x$ to $y$.
When a Petri net represents a distributed system, a global state of this system is given as a
\emph{marking}, a multiset of places, depicted by placing $M(s)$ dots (\emph{tokens}) in each place
$s$.  The initial state is $M_0$.  The behaviour of a Petri net is defined by the possible moves between
markings $M$ and $M'$, which take place when a transition \emph{fires}.
In that case, the transition $t$ consumes $F(s,t)$ tokens from each place $s$.
Naturally, this can happen only if $M$ makes these tokens available in the first place. Next,
the transition produces $F(t,s)$ tokens in each place $s$. \df{firing} formalises this notion of behaviour.

A {\em multiset} over a set $X$ is a function $A\!:X \rightarrow \IN$, i.e.\ $A\in \powermultiset{X}\!\!$.
Object $x \in X$ is an \emph{element of} $A$ iff $A(x) > 0$.
The \emph{empty} multiset, without elements, is denoted $\emptyset$.
For multisets $A$ and $B$ over $X$ I write $A \leq B$ iff \mbox{$A(x) \leq B(x)$} for all $x \mathbin\in X$;
$A + B$ denotes the multiset over $X$ with $(A + B)(x):=A(x)+B(x)$,
$A \cap B$ is given by $(A \cap B)(x):= \min(A(x),B(x))$ and
$A - B$ is given by $(A - B)(x):=A(x)\monus B(x)=\mbox{max}(A(x)-B(x),0)$.

\begin{definition}{preset}
  Let $N\!=\!(S, T, F, M_0, \ell)$ be a Petri net and $t\in T$.
The multisets $\precond{t},~\postcond{t}: S \rightarrow \IN$ are given by $\precond{t}(s)=F(s,t)$ and
$\postcond{t}(s)=F(t,s)$ for all $s \in S$.
The elements of $\precond{t}$ and $\postcond{t}$ are
called \emph{pre-} and \emph{postplaces} of $t$, respectively.
\end{definition}

\begin{definition}{firing}
  Let $N \mathbin= (S, T, F, M_0,\ell)$ be a Petri net,
  $t \in T$ and $M, M' \in \IN^S\!$.
  Transition $t$ is \emph{enabled} under $M$ iff $\precond{t}\leq M$.
  In that case $M\goesto{t}_N M'$, where $M' = (M - \mbox{$^\bullet t$}) + t^\bullet$. 
\end{definition}

\noindent
A marking $M\in \IN^S$ is \emph{reachable} in a Petri net $(S, T, F, M_0,\ell)$
iff there are transitions $t_i\in T$ and markings $M_i\in \powermultiset S$
for $i\mathbin=1,\dots,k$, such that $M_k=M$ and $M_{i{-\!}1} \!\goesto{t_{i}}_N M_{i}$ for $i\mathbin=1,\dots,k$.

\begin{definitionq}{structural conflict}{\cite{GGS11}}
  A Petri net $N = (S, T, F, M_0, \ell)$ is a \emph{structural conflict net} if
  for all $t, u\in T$ with $\precond{t} \cap \precond{u} \neq \emptyset$
  and for all reachable markings $M$, one has $\precond{t} + \precond{u} \not\leq M$.
\end{definitionq}
\noindent
Here I restrict myself to structural conflict nets, henceforth simply called \emph{nets}, a class of
Petri nets containing the \emph{safe} Petri nets that are normally used to give semantics to process algebras.
\begin{definition}{net2lts}
Given a net $N = (S, T, F, M_0, \ell)$, its associated LTS $(\IP,  \Tr, \source,\target,\ell)$
is given by $\IP := \powermultiset{S}$, $\Tr:=\{(M,t)\in\powermultiset{S}\times T \mid \precond{t}\leq M\}$,
$\source(M,t) := M$, $\target(M,t) := (M - \precond{t}) + \postcond{t}$ and $\ell(M,t) := \ell(t)$.
The net $N$ maps to the state $M_0$ in this LTS\@.
\end{definition}
\noindent
A \emph{completeness criterion} on a net is a completeness criterion on its associated LTS\@.
Now $N \models^{CC} \phi$ is defined to hold iff $M_0 \models^{CC} \phi$ in the associated LTS\@.

\subsection{CCS}\label{sec:CCS}
\mbox{}\\[2ex]
\newcommand{\ca}{a}
\noindent
The \emph{Calculus of Communicating Systems} (CCS) \cite{Mi90ccs} is parametrised with sets ${\K}$ of \emph{agent identifiers} and $\A$ of \emph{names};
each $X\in\K$ comes with a defining equation \plat{$X \stackrel{{\it def}}{=} P$} with $P$ being a CCS expression as defined below.
\href{https://en.wikipedia.org/wiki/List_of_mathematical_symbols_by_subject#Set_operations}{\raisebox{0pt}[10pt]{$Act := \A\dcup\bar\A\dcup \{\tau\}$}}
is the set of {\em actions}, where $\tau$ is a special \emph{internal action}
and $\bar{\A} := \{ \bar{\ca} \mid \ca \in \A\}$ is the set of \emph{co-names}.
Complementation is extended to $\bar\A$ by setting $\bar{\bar{\mbox{$\ca$}}}=\ca$.
Below, $\ca$ ranges over $\A\cup\bar\A$, $\alpha$ over $Act$, and $X,Y$ over $\K$.
A \emph{relabelling} is a function $f\!:\A\mathbin\rightarrow \A$; it extends to $Act$ by
$f(\bar{\ca})\mathbin=\overline{f(\ca)}$ and $f(\tau):=\tau$.
The set $\cT_{\rm CCS}$ of CCS expressions or \emph{processes} is the smallest set including:
\begin{center}
\begin{tabular}{@{}lll@{}}
$\sum_{i\in I}\alpha_i.P_i$ & for $I$ an index set, $\alpha_i\mathbin\in Act$ and $P_i\mathbin\in\cT_{\rm CCS}$ & \emph{guarded choice} \\
$P|Q$ & for $P,Q\mathbin\in\cT_{\rm CCS}$ & \emph{parallel composition}\\
$P\backslash L$ & for $L\subseteq\A$ and $P\mathbin\in\cT_{\rm CCS}$ & \emph{restriction} \\
$P[f]$ & for $f$ a relabelling and $P\mathbin\in\cT_{\rm CCS}$ & \emph{relabelling} \\
$X$ & for $X\in\K$ & \emph{agent identifier}\\
\end{tabular}
\end{center}
The process $\sum_{i\in \{1,2\}}\alpha_i.P_i$ is often written as $\alpha_1.P_1 + \alpha_2.P_2$,
$\sum_{i\in \{1\}}\alpha_i.P_i$ as $\alpha_1.P_1$ and $\sum_{i\in \emptyset}\alpha_i.P_i$ as ${\bf 0}$.
The semantics of CCS is given by the transition relation
$\mathord\rightarrow \subseteq \cT_{\rm CCS}\times Act \color{black} \times \Pow(\Ce) \color{black} \times\cT_{\rm CCS}$,
where transitions \plat{$P\goto[C]{\alpha}Q$} are derived from the rules of \autoref{tab:CCS}.
\begin{table*}[t]
\caption{Structural operational semantics of CCS}
\label{tab:CCS}
\normalsize
\begin{center}
  \newcommand\mylabel\weg
  \renewcommand\eta\alpha
  \renewcommand\ell\alpha
  \renewcommand\strut{\rule{0pt}{12pt}}
\framebox{$\begin{array}{c@{\quad}c@{\qquad}c}
&\sum_{i\in I}\alpha_i.P_i \goto[\{\varepsilon\}]{\alpha_j} P_j\makebox[20pt][l]{~~~~($j\in I$)} \\[3ex]
\displaystyle\frac{P\goto[C]{\eta} P'}{P|Q \strut\goto[\Left\cdot C]{\eta} P'|Q} \mylabel{Par-l}&
\displaystyle\frac{P\goto[C]{\ca} P' ,~ Q \goto[D]{\bar{\ca}} Q'}{P|Q \strut\goto[\Left\cdot C \;\cup\; \R\cdot D]{\tau} P'| Q'} \mylabel{Comm}&
\displaystyle\frac{Q\goto[D]{\eta} Q'}{P|Q \strut\goto[\R\cdot D]{\eta} P|Q'} \mylabel{Par-r}\\[4ex]
\displaystyle\frac{P \goto[C]{\ell} P'}{P\backslash L \strut\goto[C]{\ell}P'\backslash L}~~(\ell,\bar{\ell}\not\in L) \mylabel{Res}&
\displaystyle\frac{P \goto[C]{\ell} P'}{P[f] \strut\goto[C]{f(\ell)} P'[f]} \mylabel{Rel} &
\displaystyle\frac{P \goto[C]{\alpha} P'}{X\strut\goto[C]{\alpha}P'}~~(X \stackrel{{\it def}}{=} P) \mylabel{Rec}
\end{array}$}
\end{center}
\end{table*}
Ignoring the labels {\color{black}$C\in \Pow(\Ce)$} for now, such a transition indicates that process
$P$ can perform the action $\alpha\in Act$ and transform into process $Q$.
The process $\sum_{i\in I}\alpha_i.P_i$ performs one of the actions $\alpha_j$ for $j\in I$ and subsequently acts as $P_j$.
The parallel composition $P|Q$ executes an action from $P$, an action from $Q$, or a synchronisation
between complementary actions $c$ and $\bar{c}$ performed by $P$ and $Q$, resulting in an internal action $\tau$. 
The restriction operator $P \backslash L$ inhibits execution of the actions from $L$ and their complements. 
The relabelling $P[f]$ acts like process $P$ with all labels $\alpha$ replaced by $f(\alpha)$.
Finally, the rule for agent identifiers says that an agent $X$ has the same transitions as the body $P$ of its defining equation.
The standard version of CCS \cite{Mi90ccs} features a \emph{choice} operator $\sum_{i\in I}P_i$; here
I use the fragment of CCS that merely features guarded choice.

The second label of a transition indicates the set of (parallel) \emph{components} involved in
executing this transition. The set $\Ce$ of components is defined as $\{\Left,\R\}^*$, that is, the set
of strings over the indicators $\Left$eft and $\R$ight, with $\varepsilon\mathbin\in\Ce$ denoting the empty string
and $\textsc{d}{\cdot} C := \{\textsc{d}\sigma \mid \sigma\mathbin\in C\}$ for $\textsc{d}\mathbin\in\{\Left,\R\}$ and
$C\mathbin\subseteq \Ce\!$.

\begin{example}{CCS transitions}
The process $P\mathbin{:=}(X|\bar a.{\bf 0})|\bar a.b.{\bf 0}$ with $X\mathbin{\stackrel{{\it def}}{=}} a.X$ has as outgoing transitions
$P \mathbin{\goto[\{\Left\Left\}]a} P$,
$~P \mathbin{\goto[\{\Left\Left,\Left\R\}]\tau} (X|{\bf 0})|\bar a.b.{\bf 0}$,
$~P \mathbin{\goto[\{\Left\R\}]{\bar a}} (X|{\bf 0})|\bar a.b.{\bf 0}$,
$~P \mathbin{\goto[\{\Left\Left,\R\}]\tau} (X|\bar a.{\bf 0})|b.{\bf 0}$ and
$P \goto[\{\R\}]{\bar a} (X|\bar a.{\bf 0})|b.{\bf 0}$.
\end{example}
\noindent
These components stem from Victor Dyseryn [personal communication, 2017] and were
introduced in \cite{vG19c}. They were not part of the standard semantics of CCS \cite{Mi90ccs},
which can be retrieved by ignoring them.

\begin{definition}{CCS2LTS}
The LTS of CCS is $(\cT_{\mbox{\tiny CCS}}, \Tr, \source,\target,\ell)$, with
$\Tr$ the set of derivable transitions $P\mathbin{\goto[C\!\!]{\!\!\alpha}}Q$,
$\ell(P\mathbin{\goto[C\!\!]{\!\!\alpha\!}} Q) \mathbin=\alpha$,
$\source(P\mathbin{\goto[C\!\!]{\!\!\alpha\!}} Q) \mathbin=P$ and $\target(P\mathbin{\goto[C\!\!]{\!\!\alpha\!}} Q)\mathbin=Q$.
Employing\linebreak this interpretation of CCS, one can pronounce judgements $P\models^{CC}\phi$ for CCS processes $P$.
\end{definition}

\subsection{Labelled Transition Systems with Concurrency}\label{sec:LTSC}

\begin{definition}{LTSC}
  An \emph{LTS with concurrency} (LTSC) is a tuple $(\IP, \Tr, \source,\target,\ell,\aconc)$
  consisting of a LTS $(\IP, \Tr, \source,\target,\ell)$ and a \emph{concurrency relation}
  ${\aconc} \subseteq \Tr \times \Tr$, such that:
  \begin{equation}\label{irreflexivity}
  \mbox{$t \naconc t$ for all $t \in\Tr$,}
  \end{equation}
  \begin{equation}\label{closure}\begin{minipage}{5.4in}{
  if $t\in\Tr$ and $\pi$ is a path from $\source(t)$ to $s\in \IP$ such that $t \aconc v$ for
  all transitions $v$ occurring in $\pi$, then there is a $u\in\Tr$ such that $\source(u)=s$,
  $\ell(u)=\ell(t)$ and $t \naconc u$.}
  \end{minipage}\end{equation}
\end{definition}
\noindent
Informally, $t\aconc v$ means that the transition $v$ does not interfere with $t$, in the sense that
it does not affect any resources that are needed by $t$, so that in a state where $t$ and $v$ are
both possible, after doing $v$ one can still do a future variant $u$ of $t$.
Write $t\conc v$ for $t\aconc v \wedge v \aconc t$.

LTSCs were introduced in \cite{vG19}, although there the model is more general
on various counts; I do not need this generality here.

The LTS associated with CCS can be turned into an LTSC by defining $(P\goto[C]\alpha P') \aconc (Q\goto[D]\beta Q')$
iff $C\cap D = \emptyset$, that is, two transitions are concurrent iff they stem from disjoint sets
of components \cite{GH19,vG19c}.
In this LTSC, and many others, including the ones associated to nets below,
$\aconc$ is symmetric, and thus the same as $\conc$.

\begin{example}{CCS transitions concurrency}
Let the 5 transitions from \ex{CCS transitions} be $t$, $u$, $v$, $w$ and $x$, respectively.
Then $t \nconc w$ because these transitions share the component $\Left\Left$. Yet $v \conc w$.
\end{example}

\noindent
The LTS associated with a net can be turned into an LTSC by defining $(M,t) \aconc (M',u)$ iff
$\precond{t} \cap \precond{u} = \emptyset$, i.e., the two LTS-transitions stem from net-transitions
that have no preplaces in common.
Naturally, any LTSC can be turned into a LTS, and further into a Kripke structure, by forgetting $\aconc$.

\section{Progress, Justness and Fairness}\label{sec:completeness criteria}

With the above definitions one can pronounce judgements $\D\models^{CC}\phi$ for distributed systems
$\D$ given as a net or a CCS expression, for instance. Through the translations of
Definitions~\ref{df:net2lts} or~\ref{df:CCS2LTS} one renders $\D$ as a state $P$ in an LTS\@.
The completeness criterion $CC$ is given as a set of paths on that LTS\@.
Then, using \df{DV translation}, $P$ is seen as a state in a Kripke structure, and $CC$ as a set of
paths on that Kripke structure. Here it is well-defined when $P\models^{CC}\phi$ holds, and this verdict
applies to the judgement $\D\models^{CC}\phi$.

The one thing left to explain is where the completeness criterion $CC$ comes from.
In this section I define completeness criteria
$CC \in \{{\it SF}(\Tsk),{\it WF}(\Tsk),J, {\it Pr},\top \mid \Tsk \mathbin\in \Pow(\Pow(\Tr))\}$
on LTSs $(\IP, \Tr, \source,\target,\ell)$, to be used in
judgements $P \models^{CC} \phi$, for $P \mathbin\in \IP$ and $\phi$ an {LTL} formula. These
criteria are called \emph{strong fairness} (\emph{SF}), \emph{weak fairness} (\emph{WF}), both
parametrised with a set $\Tsk\subseteq \Pow(\Tr)$ of \emph{tasks}, \emph{justness} ($J$),
\emph{progress} (\emph{Pr}) and the \emph{trivial} completeness criterion ($\top$).
Justness is merely defined on LTSCs.
I confine myself to criteria that hold finite paths ending within a transition to be incomplete.

Reading \ex{beer}, one could find it unfair that Bart might never get a beer.
Strong and weak \emph{fairness} are completeness criteria that postulate that Bart will get a beer,
namely by ruling out as incomplete the infinite paths in which he does not.
They are formalised by introducing a set $\Tsk$ of \emph{tasks},\pagebreak[3] each being a set of
transitions (in an LTS or Kripke structure). 
\begin{definitionq}{fairness}{\cite{GH19}}
A task $T \in \Tsk$ is \emph{enabled} in a state $s$ iff $s$ has an outgoing transition from $T$.
It is \emph{perpetually enabled} on a path $\pi$ iff it is enabled in every state of $\pi$.
It is \emph{relentlessly enabled} on $\pi$, if each suffix of $\pi$ contains a state
in which it is enabled.\footnote{This is the case if the task is enabled in infinitely many states
of $\pi$, in a state that occurs infinitely often in $\pi$, or in the last state of a finite $\pi$.}
It \emph{occurs} in $\pi$ if $\pi$ contains a transition $t\in T$.

A path $\pi$ is \emph{weakly fair} if, for every suffix $\pi'$ of $\pi$,
each task that is perpetually enabled on $\pi'$, occurs in $\pi'$.
It is \emph{strongly fair} if, for every suffix $\pi'$ of $\pi$,
each task that is relentlessly  enabled on $\pi'$, occurs~in~$\pi'$.
\end{definitionq}
\noindent
As completeness criteria, these notions take only the fair paths to be complete.
In \ex{beer} it suffices to have a task ``Bart gets a beer'', consisting of the three transitions
leading to the $B$ state. Now in any path in which Bart never gets a beer this task is perpetually
enabled, yet never taken. Hence weak fairness suffices to rule out such paths.
One has $\D \models^{{\it WF}(\Tsk)} {\bf F}B$.

\emph{Local fairness} \cite{GH19} allows the tasks $\Tsk$ to be declared on an ad hoc basis for the
application at hand. On this basis one can call it unfair if Bart doesn't get a beer,
without requiring that Cameron should get a beer as well. \emph{Global fairness}, on the other hand,
distils the tasks of an LTS in a systematic way out of the structure of a formalism, such as
pseudocode, process algebra or Petri nets, that gave rise to the LTS\@.
A classification of many ways to do this, and thus of many notions of strong and weak fairness,
appears in \cite{GH19}. In \emph{fairness of directions} \cite{Fr86}, for instance, each
transition in an LTS is assumed to stem from a particular \emph{direction}, or \emph{instruction},
in the pseudocode that generated the LTS; now each direction represents a task, consisting of all
transitions derived from that direction.

In \cite{GH19} the assumption that a system will never stop when there are
transitions to proceed is called \emph{progress}. In \ex{Bart alone} it takes a progress
assumption to conclude that Bart will get his beer. Progress fits the default completeness criterion
introduced before, i.e., $\models^{\it Pr}$ is the same as $\models$. Not (even) assuming progress
can be formalised by the trivial completeness criterion $\top$ that declares all paths to be
complete. Naturally, $\E \not\models^\top {\bf F}B$.

Completeness criterion $D$ is called \emph{stronger} than criterion $C$ if it rules out
more paths as incomplete. So $\top$ is the weakest of all criteria, and, for any given collection
$\Tsk$, strong fairness is stronger than weak fairness. When assuming that each transition occurs in
at least one task---which can be ensured by incorporating a default task consisting of all
transitions---progress is weaker than weak fairness.

\emph{Justness} \cite{GH19} is a strong form of progress, defined on LTSCs.
\begin{definition}{justness}
A path $\pi$ is \emph{just} if for each transition $t$ whose source state $s := \source(t)$ occurs in $\pi$,
the (or any) suffix of $\pi$ starting at $s$ contains a transition $u$ with $t \naconc u$.
\end{definition}
\begin{example}{CCS justness}
The infinite path $\pi$ that only ever takes transition $t$ in \ex{CCS transitions}/\ref{ex:CCS transitions concurrency} is unjust.
Namely with transition $v$ in the r\^ole of the $t$ from \df{justness}, $\pi$ contains no transition $y$ with $v \nconc y$.
\end{example}
\noindent
Informally, the only reason for an enabled transition not to occur, is that one of its resources is
eventually used for some other transition. In \ex{Bart separated} for instance, the orders of Alice
and Cameron are clearly concurrent with the one of Bart, in the sense that they do not compete for
shared resources. Taking $t$ to be the transition in which Bart gets his beer, any path in which
$t$ does not occur is unjust. Thus $\F \models^J {\bf F}B$.

For most choices of $\Tsk$ found in the literature, weak fairness is a strictly stronger
completeness criterion than justness. In \ex{beer}, for instance, the path in which Bart does not
get a beer is just. Namely, any transition $u$ giving Alice or Cameron a beer competes for the same
resource as the transition $t$ giving Bart a beer, namely the attention of the bartender.
Thus $t \naconc u$, and consequently $\D \not\models^J {\bf F}B$.

\section{Blocking Actions}\label{sec:blocking}

I now present \emph{reactive} temporal logic by extending the ternary judgements $P \models^{CC} \phi$
defined above to quaternary judgements $P \models^{CC}_B \phi$, with $B\subseteq A$ a set of
\emph{blocking} actions.
Here $A$ is the set of all observable actions of the LTS on which {LTL} is interpreted.
The intuition is that actions $b\in B$ may be blocked by the environment, but actions
$a \in A {\setminus} B$ may not. The relation $\models_B$ can be used to formalise
the assumption that the actions in $A {\setminus} B$ are not under the control of the user of the
modelled system, or that there is an agreement with the user not to block them. Either way, it is a
disclaimer on the wrapping of our temporal judgement, that it is valid merely when applying the
involved distributed system in an environment that may block actions from $B$ only. The hidden
action $\tau$ may never be blocked.

I will present the relations $\models^{CC}_B$ for each choice of $CC \neq \top$ discussed in the
previous section, and each $B\subseteq A$.  When writing $P \models^{CC}_B\phi$ the modifier $B$
adapts the default completeness criterion by declaring certain finite paths complete, and the
modifier $CC\neq \top,{\it Pr}$ adapts it by declaring some infinite paths incomplete.

Starting with $CC={\it Pr}$, \hypertarget{Bprogressing}{I call a path \emph{$B$-progressing}
iff it is either infinite or ends in a state of which all outgoing transitions have a label from $B$},
and write $s \models^{\it Pr}_B \phi$, or  $s \models_B \phi$ for short, if $\pi \models\phi$
holds for all $B$-progressing paths $\pi$ starting in $s$.
The completeness criterion \emph{$B$-progress}, which takes the $B$-progressing to be the complete ones,
says that a system may stop in a state with outgoing transitions only when they are all blocked by
the environment. Note that the standard {LTL} interpretation $\models$ is simply
$\models_\emptyset$, obtained by taking the empty set of blocking actions.

In the presence of the modifier $B$, \df{justness} is adapted as follows:
\begin{definition}{Bjustness}
\hypertarget{Bjust}{A path $\pi$ is \emph{$B$-just} if for each $t\in\Tr$ with
$\ell(t)\notin B$ and whose source state $s := \source(t)$ occurs in $\pi$,
any suffix of $\pi$ starting at $s$ contains a transition $u$ with $t \naconc u$.}
\end{definition}
\noindent
It doesn't matter whether $\ell(u)\in B$.
The completeness criterion \emph{$B$-justness} takes the $B$-just paths to be the complete ones.
Write $s \models^{\it J}_B \phi$ if $\pi \models\phi$ for all $B$-just paths $\pi$ starting in $s$.

For the remaining cases $CC={\it SF}(\Tsk)$ and $CC={\it WF}(\Tsk)$, adapt the first sentence
of \df{fairness} as follows.

\begin{definition}{Bfairness}
A task $T \in \Tsk$ is \emph{$B$-enabled} in a state $s$ iff $s$ has an outgoing transition $t\in T$
with $\ell(t)\notin B$.
\end{definition}
\noindent
Strong and weak $B$-fairness of paths is then defined as in \df{fairness}, but replacing ``enabled''
by ``$B$-enabled''.

The above completes the formal definition of the validity of temporal judgements $P \models^{CC}_B \phi$
with $\phi$ an {LTL} formula, $B \subseteq A$, and either
\begin{itemize}
\item $CC={\it Pr}$ and $P$ a state in an LTS, a Petri net or a CCS expression,
\item $CC={\it J}$ and $P$ a state in an LTSC, a Petri net or a CCS expression,
\item $CC={\it WF}(\Tsk)$ or ${\it SF}(\Tsk)$ and $P$ a state in an LTS
   $(\IP, \Tr, \source,\target,\ell)$ with $\Tsk\in\Pow(\Pow(\Tr))$,
   or $P$ a net or CCS expression with associated LTS
   $(\IP, \Tr, \source,\target,\ell)$ and $\Tsk\mathbin\in\Pow(\Pow(\Tr))$.
\end{itemize}
Namely, in case $P$ is a state in an LTS, it is also a state in the associated Kripke structure $K$.
Moreover, $B$ and $CC$ combine into a single completeness criterion ${\it CC}(B)$ on that LTS, which
translates as a completeness criterion ${\it CC}(B)$ on $K$. Now \df{validity} tells whether
$P\models^{{\it CC}(B)} \phi$ holds.

In case $CC=J$ and $P$ a state in an LTSC,
$B$ and $J$ combine into a single completeness criterion $J(B)$ on that LTSC, which is also a
completeness criterion on the associated LTS; now proceed as above.

In case $P$ is a Petri net or CCS expression, first translate it into a state in an LTS or LTSC,
using Definitions~\ref{df:net2lts} or~\ref{df:CCS2LTS}, respectively, and proceed as above.

Temporal judgements $P \models^{CC}_B \phi$, as introduced above, are not limited to the case that
$\phi$ is an LTL formula. In \cite{EPTCS322.6} I show that allowing $\phi$ to be a CTL formula instead
poses no additional complications, and I expect the same to hold for other temporal logics.

Judgements $P \models^{CC}_B \phi$ get stronger (= less likely true) when the completeness criterion
$CC$ is weaker, and the set $B$ of blocking actions larger.

\section{Translating Reactive LTL into Standard LTL}\label{sec:translating}

Here I translate judgements $s \models^{CC}_B \phi$ in reactive LTL\footnote{\emph{Reactive LTL}
refers to judgements of the form $s \models^{CC}_B \phi$ or $\D \models^{CC}_B \phi$ where $\phi$ is
an LTL formula.} into equivalent judgements
$\hat s \models \psi$ in traditional LTL, albeit with infinite conjunctions.  The price to be paid for this
is an extra dose of atomic propositions. I start with judgements $s \models^{CC}_B \phi$
interpreted in an LTS $(\IP,\Tr,\source,\target,\ell)$, as this is where the completeness criterion $CC(B)$ takes shape.
I use a slightly different translation from LTSs to Kripke structures than the one of \df{DV translation};
it inserts a state halfway along \emph{any} transition, even if it is labelled $\tau$. However,
$\tau$ will not be an atomic proposition of the resulting Kripke structure $K$, and the new halfway
states do not inherit a transition label. This change affects the bookkeeping for next-state
operators, but not in a bad way.

To translate $\models^{CC}_B$ into $\models$, I have to make provisions for finite $CC(B)$-complete
paths that are ${\it Pr}(\emptyset)$-incomplete, and for infinite $CC(B)$-incomplete
paths that are ${\it Pr}(\emptyset)$-complete. In order not to tackle these
opposite forces in the same step, I first present a translation from reactive LTL into LTL with the
trivial completeness criterion. That is, for each choice of $CC$ from \Sec{completeness criteria},
each set $B \subseteq A$ of blocking actions, and each LTL formula $\phi$, I present an LTL
formula $\phi_B^{CC}$ such that $s \models^{CC}_B \phi$ iff $s \models^\top \phi^{CC}_B$
for each $s \in \IP$.

Given a collection $\Tsk$ of tasks and a set $B$ of blocking actions, introduce for each task
$T\in\Tsk$ two atomic propositions ${\it en}_B^T$ and ${\it oc}^T$. Proposition ${\it en}_B^T$
holds in those states of $K$ that stem from states of $s \in \IP$ in which the task $T$ is $B$-enabled,
i.e., that have an outgoing transition $t\in T$ with $\ell(t)\notin B$. Additionally, it holds in
those states of $K$ that stem from transitions $t \in \Tr$ such that $T$ is $B$-enabled in both
$\source(t)$ and $\target(t)$. Proposition ${\it oc}^T$ holds in those states of $K$ that stem from
a transition $t \in T$; this is where the task \emph{occurs}. Now the formula
$${\it WF}(\Tsk)_B := \bigwedge_{T \in \Tsk} {\bf G}({\bf G}{\it en}_B^T \Rightarrow {\bf F}{\it oc}^T)$$
holds for a path $\pi$ of $K$ exactly when $\pi$ is weakly $B$-fair.
Hence the formula \plat{${\it WF}(\Tsk)_B \Rightarrow \phi$} says that $\phi$ holds on 
weakly $B$-fair paths, and one has $s \models^{\plat{$\scriptstyle{\it WF}(\Tsk)$}}_B \phi$ iff
$s \models^\top {\it  WF}(\Tsk)_B \Rightarrow \phi$.\pagebreak[3]
Likewise,\vspace{-4pt}
$${\it SF}(\Tsk)_B := \bigwedge_{T \in \Tsk} {\bf G}({\bf GF}{\it en}_B^T \Rightarrow {\bf F}{\it oc}^T)\vspace{-2pt}$$
holds for a path of $K$ iff it is strongly $B$-fair, so that
$s \models^{\plat{$\scriptstyle{\it SF}(\Tsk)$}}_B \phi$ iff $s \models^\top {\it SF}(\Tsk)_B \Rightarrow \phi$.
The formulas ${\bf G}({\bf G}{\it en} \Rightarrow {\bf F}{\it oc})$ and
${\bf G}({\bf GF}{\it en} \Rightarrow {\bf F}{\it oc})$ stem from \cite{GPSS80}.
In the literature one sometimes find the equivalent forms ${\bf FG}{\it en} \Rightarrow {\bf GF}{\it oc}$ and
${\bf GF}{\it en} \Rightarrow {\bf GF}{\it oc}$.

Progress, i.e., the case $CC={\it Pr}$, can be dealt with in the same way, by recognising it as weak
or strong fairness involving a single task, spanning all transitions.

To deal with $B$-justness, introduce atomic propositions ${\it en}^t$ and $\sharp t$ for each transition
$t\in\Tr$ with $\ell(t)\notin B$. Proposition ${\it en}^t$ holds in the state $\source(t)$ only,
whereas $\sharp t$ holds in all states of $K$ that stem from a transition $u \in \Tr$ with $t \naconc u$.
Now $$J_B := \bigwedge_{t \in \Tr,~\ell(t)\notin B}  {\bf G}({\it en}^t \Rightarrow {\bf F} \sharp t)\vspace{-4pt}$$
holds for a path of $K$ iff it is $B$-just. Consequently, 
$s \models^J_B \phi$ iff $s \models^\top J_B \Rightarrow \phi$.

It remains to translate $\models^\top$ into $\models$.
To this end, I transform $K$ into \plat{$\widehat K$} by adding a self-loop at each of its states.
I also introduce a fresh atomic proposition ${\it tr}$, for \emph{transition},
and use it to label all states in $\widehat K$ that stem from a transition from $\Tr$.
For each finite path $\pi$ in $K$ let $\pi^\infty$ be the infinite path in $\widehat K$ obtained by
repeating the last state of $\pi$ infinitely often.
In case $\pi$ is infinite, let $\pi^\infty:=\pi$.
Let $Z$ be the completeness criterion on $\widehat K$ that declares each path of the form
$\pi^\infty$ complete. These are exactly the infinite paths without a subsequence $s s t$ with $s\neq t$.
So $\_^\infty$ is a bijection between the paths of $K$ and the $Z$-complete paths of $\widehat K$.
Let $\Q$ be the transformation on LTL formula, defined by
\[\begin{array}{r@{~:=~}l@{\qquad}r@{~:=~}l@{\qquad}r@{~:=~}l}
  \Q(p) & p &
  \Q(\neg\phi) & \neg \Q(\phi) &
  \Q(\phi \wedge \psi) & \Q(\phi) \wedge \Q(\psi) \\
  \Q({\bf F}\phi) & {\bf F} \Q(\phi) &
  \Q({\bf G}\phi) & {\bf G} \Q(\phi) &
  \Q(\phi{\bf U}\psi) & \Q(\phi) {\bf U} \Q(\psi)\\
  \Q({\bf X}\phi) & \multicolumn{5}{@{}l@{}}{\big({\it tr} \Rightarrow {\bf X}(\neg{\it tr} \wedge \Q(\phi))\big) \wedge
                                             \big(\neg{\it tr} \Rightarrow {\bf X}({\it tr} \wedge \Q(\phi))\big).}
\end{array}\]
A trivial induction on $\phi$ shows that $\pi \models \phi$ holds in $K$ iff $\pi^\infty \models \Q(\phi)$ holds in $\widehat K$.
This implies that $s \models^\top \phi$ holds in $K$ iff $s \models^Z \phi$ holds in $\widehat K$; I
will write the latter as $\hat s \models^Z \phi$.
$Z$-completeness can be stated as\vspace{-1pt}
\[\Z := {\bf G}\big({\it tr} \Rightarrow ({\bf G}{\it tr} \vee {\bf X}(\neg{\it tr}))\big) \wedge
        {\bf G}\big(\neg{\it tr} \Rightarrow ({\bf G}(\neg{\it tr}) \vee {\bf X}{\it tr})\big).\vspace{-1pt}\]
Hence $s \models^\top \phi$ iff $\hat s \models \Z \Rightarrow \phi$.

The above translation from reactive LTL into standard LTL may give the impression that reactive LTL
is not more expressive than standard LTL\@. This conclusion is not really valid, due to the addition
to the formalism of fresh atomic propositions. It is for instance widely accepted that LTL is not
more expressive than CTL\@. Yet, if one introduces an atomic proposition $p_\phi$ for each CTL
formula $\phi$, one that is declared to hold for all states that satisfy $\phi$,
one trivially obtains $s \models \phi$ iff $s \models p_\phi$. This would suggest that CTL can be
faithfully translated into LTL, even without using any of the modal operators of LTL\@.

\section{Safety Properties}\label{sec:safety}

A \emph{safety property} is a temporal formula $\phi$ that holds for a path $\pi$ iff it holds for
all finite prefixes of $\pi$. In that case $\D \models \phi$ iff $\pi \models \phi$ for all finite paths
$\pi$ starting in the initial state of $\D$.
Such a property can be thought to say that nothing bad will ever happen \cite{Lam77}.
The intuition is that a bad thing must be observable in a finite prefix of a run,
so that $\D \models \neg\phi$ iff $\pi \models \neg\phi$ for some finite path
$\pi$ starting in the initial state of $\D$.

\begin{proposition}{safety fragment}
The fragment of LTL given by the grammar
\[\phi,\psi ::= p \mid \neg p \mid \phi \wedge \psi  \mid \phi \vee \psi \mid {\bf Y}\phi \mid  {\bf G}\phi \mid
           \psi {\bf W} \phi \]
describes only safety properties. Here {\bf Y} is the dual of {\bf X}, and $\psi {\bf W} \phi$
abbreviates $(\psi {\bf U} \phi) \vee {\bf G}\psi$.
\end{proposition}
\begin{proof}
Let $\phi$ and $\psi$ be safety properties.
I show that also $\phi \vee \psi$ is a safety property.

Let $\pi\models \phi \vee \psi$. Then either $\pi \models \phi$ or $\pi\models \psi$; for reasons of
symmetry I may assume the former. Let $\pi'$ be a finite prefix of $\pi$. Since $\phi$ is a safety
property, $\pi' \models \phi$. Hence $\pi' \models \phi \vee \psi$, which needed to be shown.

Now let $\pi$ be a path such that $\pi' \models \phi \vee \psi$ for all finite prefixes $\pi'$ of $\pi$.
In case $\pi' \models \phi$ for all finite prefixes $\pi'$ of $\pi$, it follows that $\pi \models \phi$,
since $\phi$ is a safety property, and hence $\pi\models \phi \vee \psi$.
So assume $\pi$ has a finite prefix $\pi''$ for which $\pi''\not\models \phi$.
Since $\phi$ is a safety property, it follows that $\pi'\not\models \phi$ for each finite path $\pi'$
with $\pi'' \leq \pi' \leq \pi$. So $\pi'\models \psi$ for all such $\pi'$.
As $\pi''\models \psi$, and $\psi$ is a safety property, one also has $\pi'\models \psi$ for all
prefixes $\pi'$ of $\pi''$. Thus $\pi'\models \psi$ for \emph{all} prefixes $\pi'$ of $\pi$.
As $\psi$ is a safety property, it follows that $\pi\models\psi$, so $\pi\models \phi \vee \psi$.

That $p$ and $\neg p$ are safety properties, for $p \in AP$, and that the safety properties are
closed under conjunction, is trivial.

Let $\phi$ be a safety property. I show that ${\bf Y}\phi$ and ${\bf G}\phi$ are safety properties.

Let $\pi \models {\bf Y}\phi$. Then either $|\pi|=0$ or $\pi_{+1} \models \phi$.
Let $\pi'$ be a finite prefix of $\pi$.
Then either $|\pi'|=0$, so that trivially $\pi' \models {\bf Y}\phi$, or $\pi'_{+1}$ is a finite
prefix of $\pi_{+1}$. Since $\phi$ is a safety property, $\pi'_{+1} \models \phi$, and thus $\pi' \models {\bf Y}\phi$. 

Now let $\pi$ be a path such that $\pi' \models {\bf Y}\phi$ for all finite prefixes $\pi'$ of $\pi$.
In case $|\pi|=0$, trivially $\pi \models {\bf Y}\phi$. So assume $|\pi|>0$.
Then $\pi'_{+1} \models \phi$ for all finite prefixes $\pi'$ of $\pi$ with $|\pi'|>0$, that is,
$\rho \models \phi$ for all finite prefixes $\rho$ of $\pi_{+1}$.
Since $\phi$ is a safety property it follows that $\pi_{+1}\models \phi$, and hence $\pi \models {\bf Y}\phi$.

Let $\pi \models {\bf G}\phi$. Then $\pi{\upharpoonright}\rho \models \phi$ for each finite prefix $\rho$ of $\pi$.
As $\phi$ is a safety property, $\rho' \models \phi$ for each finite prefix $\rho'$ of $\pi{\upharpoonright}\rho$,
for each finite prefix $\rho$ of $\pi$, that is,
$\pi'{\upharpoonright}\rho \models \phi$ for each pair of finite prefixes $(\rho,\pi')$ with $\rho \leq \pi' \leq \pi$. 
Thus $\pi' \models  {\bf G}\phi$ for each finite prefix $\pi'$ of $\pi$.

Now let $\pi$ be a path such that $\pi' \models {\bf G}\phi$ for all finite prefixes $\pi'$ of $\pi$.
Then $\pi'{\upharpoonright}\rho \models \phi$ for each pair of finite prefixes $(\rho,\pi')$ with $\rho \leq \pi' \leq \pi$,
that is, $\rho' \models \phi$ for each finite prefix $\rho'$ of $\pi{\upharpoonright}\rho$,
for each finite prefix $\rho$ of $\pi$.
As $\phi$ is a safety property, $\pi{\upharpoonright}\rho \models \phi$ for each finite prefix $\rho$ of $\pi$.
Thus $\pi \models {\bf G}\phi$.

Finally, let $\phi$ and $\psi$ be safety properties.
I show that $\psi {\bf W} \phi$ is a safety property.

Let $\pi \models \psi {\bf W} \phi = (\psi {\bf U} \phi) \vee {\bf G}\psi$. One possibility is that
$\pi \models {\bf G}\psi$. Then, as shown above, for each finite prefix $\pi'$ of $\pi$ one has 
$\pi' \models  {\bf G}\psi$, and thus $\pi' \models \psi {\bf W} \phi$.
The other possibility is that $\pi \models\psi {\bf U} \phi$. Then there is a finite prefix $\rho$ of $\pi$
such that $\pi{\upharpoonright}\rho\models \phi$, and for each $\zeta < \rho$ one has $\pi{\upharpoonright}\zeta\models \psi$.
Now let $\pi'$ be a finite prefix of $\pi$.

First assume that $\pi' < \rho$.
Then for each $\zeta\leq \pi'$ the path $\pi'{\upharpoonright}\zeta$ is a finite prefix of $\pi{\upharpoonright}\zeta$.
Since $\psi$ is a safety property, this implies $\pi'{\upharpoonright}\zeta\models \psi$.
Thus $\pi'\models {\bf G}\psi$ and hence $\pi' \models \psi {\bf W} \phi$.

Next assume that $\rho \leq \pi'$.
Then  $\pi'{\upharpoonright}\rho$ is a finite prefix of $\pi{\upharpoonright}\rho$.
Since $\phi$ is a safety property, $\pi'{\upharpoonright}\rho\models \phi$.
For each $\zeta < \rho$, $\pi'{\upharpoonright}\zeta$ is a finite prefix of $\pi{\upharpoonright}\zeta$.
Thus $\pi'{\upharpoonright}\zeta\models\psi$, as $\psi$ is a safety property.
It follows that $\pi'\models \psi {\bf U} \phi$ and hence $\pi' \models \psi {\bf W} \phi$.

Now let $\pi \not \models \psi {\bf W} \phi$. Then $\pi \not \models{\bf G}\psi$,
so there is a finite prefix $\rho$ of $\pi$ with $\pi{\upharpoonright}\rho \not \models\psi$,
and, choosing $\rho$ as short as possible, $\pi{\upharpoonright}\zeta \models\psi$ for all $\zeta<\rho$.
Since $\pi \not \models \psi {\bf U} \phi$, one has $\pi{\upharpoonright}\zeta \not\models\phi$ for each $\zeta\leq\rho$.%
\footnote{This argument shows that {\bf U} and {\bf W} are almost duals:
$\neg(\psi {\bf W} \phi) \equiv (\neg\phi){\bf U} (\neg \psi \wedge \neg\phi)$, and thus,
using that $\psi{\bf U}\phi \equiv (\psi{\bf W}\phi) \wedge {\bf F} \phi$,
also $\neg(\psi {\bf U} \phi) \equiv (\neg\phi){\bf W} (\neg \psi \wedge \neg\phi)$.
}
As $\phi$ is a safety property, for each $\zeta \leq \rho$ there exists
a finite prefix $\zeta'$ of $\pi{\upharpoonright}\zeta$ with $\zeta'\not\models\phi$;
or in other words, for each $\zeta \leq \rho$ there is a finite
$\zeta \leq \pi_\zeta \leq \pi$ with $\pi_\zeta{\upharpoonright}\zeta\not\models\phi$.
As $\psi$ is a safety property, there exists
a finite prefix $\zeta'$ of $\pi{\upharpoonright}\rho$ with $\zeta'\not\models\psi$;
or in other words, a finite $\rho \leq \bar\pi \leq \pi$ with $\bar\pi{\upharpoonright}\rho\not\models\psi$.
Now let $\pi'\leq \pi$ be the longest of the finite paths $\bar\pi$ and $\pi_\zeta$ for each $\zeta \leq \rho$.
Since $\psi$ and $\phi$ and safety properties, $\pi'{\upharpoonright}\rho\not\models\psi$, and
$\pi'{\upharpoonright}\zeta\not\models\psi$ for each $\zeta \leq \rho$.
It follows that $\pi' \not \models \psi {\bf W} \phi$.
\end{proof}

\noindent
For safety properties $\phi$, the reactive part of reactive temporal logic is irrelevant,
as one has $\D \models^{CC}_B \phi$ iff $\D \models \phi$, for all completeness criteria $CC$ and
all $B \subseteq A$. Namely, both hold iff $\pi \models \phi$ for all finite paths
$\pi$ starting in the initial state of $\D$. This hinges on the requirement of \emph{feasibility}
imposed on completeness criteria \cite{AFK88,GH19}: any finite partial run can be extended to a
complete run; or any finite path must be a prefix of a complete path.

\addtocontents{toc}{\protect\vspace{-4pt}}
\part{\texorpdfstring{\!}{}Formalising Mutual Exclusion and Fair Scheduling in Reactive LTL\texorpdfstring{\!\!\!}{}}
\label{requirements}

Here I recall the mutual exclusion problem as posed by Dijkstra \cite{Dijk65}, and the related
notion of a fair scheduler \cite{GH15a}. Employing reactive LTL, I formulate requirements that tell
exactly what does and does not count as a mutual exclusion protocol, and as a fair scheduler.
Since my requirements are parametrised by completeness criteria, which are progress, justness or
fairness assumptions, I obtain a hierarchy of quality criteria for mutual exclusion protocols and
fair schedulers, where a weaker completeness criterion characterises a higher quality protocol.
When allowing (strong or) weak fairness as parameter in my requirements, an intuitively unsatisfactory
mutual exclusion protocol or fair scheduler, which I call the \emph{gatekeeper}, meets all requirements.
This indicates that weak fairness is too strong an assumption to be used in these parameters.

\section{The Mutual Exclusion Problem and its History}\label{sec:history}

The mutual exclusion problem was presented by Dijkstra in \cite{Dijk65} and formulated as follows:
{\addtolength\leftmargini{-0.1in}
\begin{quote}
   ``To begin, consider $N$ computers, each engaged in a
process which, for our aims, can be regarded as cyclic. In
each of the cycles a so-called ``critical section'' occurs and
the computers have to be programmed in such a way that
at any moment only one of these $N$ cyclic processes is in
its critical section. In order to effectuate this mutual
exclusion of critical-section execution the computers can
communicate with each other via a common store. Writing
a word into or nondestructively reading a word from this
store are undividable operations; i.e., when two or more
computers try to communicate (either for reading or for
writing) simultaneously with the same common location,
these communications will take place one after the other,
but in an unknown order.''
\end{quote}}
\noindent
Dijkstra proceeds to formulate a number of requirements
that a solution to this problem must satisfy, and then
presents a solution that satisfies those requirements. 
The most central of these are:
\begin{itemize}
\item (\emph{Mutex}) ``no two computers can be in their critical section simultaneously'', and
\item (\emph{Deadlock-freedom}) if at least one computer intends to enter its critical section,
then at least one ``will be allowed to enter its critical section in due time''.
\pagebreak[3]
\end{itemize}
Two other important requirements formulated by Dijkstra are
\begin{itemize}
\item (\emph{Speed independence}) ``Nothing may be assumed about the relative speeds of the $N$
  computers'',
\item and (\emph{Optionality}) ``If any of the computers is stopped well outside its
critical section, this is not allowed to lead to potential blocking of the others.''
\end{itemize}
A crucial assumption is that each computer, in each cycle, spends only a finite amount of time in its
critical section. This is necessary for the correctness of any mutual exclusion protocol.

For the purpose of the last requirement one can partition each cycle into a \emph{critical section}, a
\emph{noncritical section} (in which the process starts), an \emph{entry protocol} between the
noncritical and the critical section, during which a process prepares for entry in negotiation with the
competing processes, and an \emph{exit protocol}, that comes right after the critical section and
before return to the noncritical section. Now ``well outside its critical section'' means in the
noncritical section. 
Optionality can equivalently be stated as admitting the possibility that a process chooses to
remain forever in its noncritical section, without applying for entry in the critical section ever again.

Knuth \cite{Knuth66} proposes a strengthening of the deadlock-freedom requirement, namely
\begin{itemize}
\item (\emph{Starvation-freedom}) If a computer intends to enter its critical section,
then it will be allowed to enter in due time.
\end{itemize}
He also presents a solution that is shown to satisfy this requirement, as well as Dijkstra's requirements.%
\footnote{However, Knuth's solution satisfies starvation-freedom, and even deadlock-freedom, only
  when making a fairness assumption. In fact, all mutual exclusion protocols, including the ones of
  \cite{EWD35,bakery,Pet81,Szy88} discussed below, need a fairness assumption to solve
  the problem as stated by Dijkstra above in a starvation-free way.
  This will be discussed in Part~\ref{impossibility results} of this paper.}
Henceforth I define a correct solution of the mutual exclusion problem as one that satisfies both
mutex and starvation-freedom, as formulated above, as well as optionality. I speak of
``speed-independent mutual exclusion'' when also insisting on the requirement of speed independence.

The special case of the mutual exclusion problem for two processes ($N=2$) was presented by Dijkstra in
\cite{EWD35}, two years prior to \cite{Dijk65}. There Dijkstra presented a solution found by
T.J. Dekker in 1959, and shows that it satisfies all requirements of \cite{Dijk65}.
Although not explicitly stated in \cite{EWD35}, the arguments given therein imply straightforwardly
that Dekker's solution also satisfies Knuth's starvation-freedom requirement above.

Peterson \cite{Pet81} presented a considerable simplification of Dekker's algorithm that satisfies
the same correctness requirements.
Many other mutual exclusion protocols appear in the literature, the most prominent being
Lamport's bakery algorithm \cite{bakery} and Szyma\'nski's mutual exclusion algorithm \cite{Szy88}.
These guarantee some additional correctness criteria besides the ones discussed above.

\section{Fair Schedulers}\label{sec:FS}

In \cite{GH15b} a \emph{fair scheduler} is defined as
{\addtolength\leftmargini{-0.1in}
\begin{quote}
``a reactive system with two input channels: one on which it can receive requests $r_1$ from its
  environment and one on which it can receive requests $r_2$. We allow the scheduler to be too busy
  shortly after receiving a request $r_i$ to accept another request $r_i$ on the same channel.
  However, the system will always return to a state where it remains ready to accept the next
  request $r_i$ until $r_i$ arrives. In case no request arrives it remains ready forever. The
  environment is under no obligation to issue requests, or to ever stop issuing requests.  Hence for
  any numbers $n_1$ and $n_2\in\IN\cup\{\infty\}$ there is at least one run of the system in which
  exactly that many requests of type $r_1$ and $r_2$ are received.

  Every request $r_i$ asks for a task $t_i$ to be executed.  The crucial property of the fair
  scheduler is that it will eventually grant any such request. Thus, we require that in any run of
  the system each occurrence of $r_i$ will be followed by an occurrence of $t_i$.''

``We require that in any partial run of the scheduler there may not be more occurrences of $t_i$
  than of $r_i$, for $i=1,2$.

  The last requirement is that between each two occurrences of $t_i$ and $t_j$ for $i,j\in\{1,2\}$
  an intermittent activity $e$ is scheduled.''
\end{quote}}
\noindent
\makebox[0pt][l]{\raisebox{36ex}[0pt][0pt]{\small FS\hspace{1pt}1}}%
\makebox[0pt][l]{\raisebox{18ex}[0pt][0pt]{\small FS\hspace{1pt}2}}%
\makebox[0pt][l]{\raisebox{10ex}[0pt][0pt]{\small FS\hspace{1pt}3}}%
\makebox[0pt][l]{\raisebox{4ex}[0pt][0pt]{\small FS\hspace{1pt}4}}%
This fair scheduler serves two clients, but the concept generalises smoothly to $N$ clients.

The intended applications of fair schedulers are for instance in operating systems, where multiple
application processes compete for processing on a single core, or radio broadcasting stations, where 
the station manager needs to schedule multiple parties competing for airtime. In such cases
each applicant must get a turn eventually. The event $e$ signals the end of the time slot allocated
to an application process on the single core, or to a broadcast on the radio station.

Fair schedulers occur (in suitable variations) in many distributed systems.  Examples are
\emph{First in First out}\footnote{Also known as First Come First Served (FCFS)},
\emph{Round Robin}, and
\emph{Fair Queueing} 
scheduling algorithms\footnote{\url{http://en.wikipedia.org/wiki/Scheduling_(computing)}}
as used in network routers~\cite{rfc970,Nagle87} and operating systems~\cite{Kleinrock64}, or the
\emph{Completely Fair Scheduler},\footnote{\url{http://en.wikipedia.org/wiki/Completely_Fair_Scheduler}}
which is the default scheduler of the Linux kernel since version 2.6.23.

Each action $r_i$, $t_i$ and $e$ can be seen as a communication
between the fair scheduler and one of its clients. In a reactive system such communications will take place
only if both the fair scheduler and its client are ready for it. Requirement FS\hspace{1pt}1 of a fair
scheduler quoted above effectively shifts the responsibility for executing $r_i$ to the client.
The actions $t_i$ and $e$, on the other hand, are seen as the responsibility of the fair scheduler.
We do not consider the possibility that the fair scheduler fails to execute $t_i$ merely because the
client does not collaborate. Hence \cite{GH15b} assumes that the client cannot prevent the actions
$t_i$ and $e$ from occurring. It is furthermore assumed that executing the actions $r_i$, $t_i$
and $e$ takes a finite amount of time only.

A fair scheduler closely resembles a mutual exclusion protocol.
However, its goal is not to achieve mutual exclusion. In most applications, mutual exclusion can be
taken for granted, as it is physically impossible to allocate the single core to multiple applications
at the same time, or the (single frequency) radio sender to multiple simultaneous broadcasts.
Instead, its goal is to ensure that no applicant is passed over forever.

It is not hard to obtain a fair scheduler from a mutual exclusion protocol.
For suppose we have a mutual exclusion protocol $M$, serving two processes $P_i$ ($i=1,2$).
I instantiate the noncritical section of Process $P_i$ as patiently awaiting the request $r_i$.
As soon as this request arrives, $P_i$ leaves the noncritical section and starts the entry protocol
to get access to the critical section. Starvation-freedom guarantees that
$P_i$ will reach its critical section. Now the critical section consists of scheduling task $t_i$,
followed by the intermittent activity $e$. Trivially, the composition of the two processes $P_i$, in
combination with protocol $M$, constitutes a fair scheduler, in that it meets the above four requirements.

One cannot quite construct a mutual exclusion protocol from a fair scheduler, due to the fact that
in a mutual exclusion protocol leaving the critical section is controlled by the client process.
For this purpose one would need to adapt the assumption that the client of a fair scheduler cannot block
the intermittent activity $e$ into the assumption that the client can postpone this action, but for
a finite amount of time only. In this setting one can build a mutual exclusion protocol, serving two
processes $P_i$ ($i=1,2$), from a fair scheduler $F$. Process $i$ simply issues request $r_i$ at
$F$ as soon as it has left the noncritical section, and when $F$ communicates the action $t_i$,
Process $i$ enters its critical section. Upon leaving its critical section, which is assumed to
happen after a finite amount of time, it participates in the synchronisation $e$ with
$F$. Trivially, this yields a correct mutual exclusion protocol.

\section{Formalising the Requirements for Fair Schedulers in Reactive LTL}\label{sec:formalising FS}

The main reason fair schedulers were defined in \cite{GH15b} was to serve as an example
of a realistic class of systems of which no representative can be correctly specified in CCS, or
similar process algebras, or in Petri nets. Proving this impossibility result necessitated a precise
formalisation of the four requirements quoted in \Sec{FS}. Through the provided translations of CCS
and Petri nets into LTSs, a fair scheduler rendered in CCS or Petri nets can be seen as a state $F$
in an LTS over the set $\{r_i,t_i,e \mid i=1,2\}$ of visible actions; all other actions can be
considered internal and renamed into $\tau$.

Let a \emph{partial trace} of a state $s$ in an LTS be the sequence of visible actions encountered on
a path starting in $s$ \cite{vG93}. Now the last two requirements (FS\hspace{1pt}3 and FS\hspace{1pt}4) of a fair
scheduler are simple properties that should be satisfied by all partial traces $\sigma$ of state $F$:
\begin{trivlist}
\item[\hspace{\labelsep}\small (FS\hspace{1pt}3)] ~$\sigma$ contains no more occurrences of $t_i$ than of $r_i$, for $i=1,2$,
\item[\hspace{\labelsep}\small (FS\hspace{1pt}4)] ~$\sigma$ contains an occurrence of $e$ between each two occurrences of $t_i$ and $t_j$ for $i,j\mathbin\in\{1,2\}$.
\end{trivlist}
\noindent
FS\hspace{1pt}4 can be conveniently rendered in {LTL}:
\begin{trivlist}\vspace{4pt}
\item[\hspace{\labelsep}(FS\hspace{1pt}4)]
~\(F \models {\bf G}\left(t_i \Rightarrow {\bf Y} 
 \big((\neg t_1 \wedge \neg t_2) {\bf W} e\big)\right)\)
\vspace{3pt}
\end{trivlist}
\noindent
for $i\in\{1,2\}$.
Since FS\hspace{1pt}4 is a safety property, it makes no difference whether and how $\models$ is
annotated with $B$ and $CC$. In \cite{Lam83}, Lamport argues against the use of the next-state
operator {\bf X}, as it is incompatible with abstraction from irrelevant details in system descriptions.
When following this advice, the weak next-state operator {\bf Y} in FS\hspace{1pt}4
can be replaced by $t_i {\bf W}$; on Kripke structures distilled from LTSs the meaning is the same.

Unfortunately, FS\hspace{1pt}3 cannot be formulated in {LTL}, due to the need to keep count of the
difference in the number of $r_i$ and $t_i$ actions encountered on a path. However, one could strengthen FS\hspace{1pt}3 into
\begin{trivlist}
\item[\hspace{\labelsep}\small (FS\hspace{1pt}3$'$)]
\begin{tabular}{l}$\sigma$ contains an occurrence of $r_i$ between each two occurrences of $t_i$,\\ and
  prior to the first occurrence of $t_i$, for $i\in\{1,2\}$.
\end{tabular}
\end{trivlist}
\noindent
This would restrict the class of acceptable fair schedulers, but keep the most interesting examples.
Consequently, the impossibility result from \cite{GH15b} applies to this modified class as well.
FS\hspace{1pt}3$'$ can be rendered in {LTL} in the same style as FS\hspace{1pt}4:
\begin{trivlist}\vspace{4pt}
\item[\hspace{\labelsep}(FS\hspace{1pt}3$'$)]
 ~\(F \models  {\big(}(\neg t_i) {\bf W} r_i {\big)}
  \wedge {\bf G}\left(t_i \Rightarrow
   {\bf Y} \big((\neg t_i) {\bf W} r_i\big) \right)\)
\vspace{3pt}
\end{trivlist}
\noindent
for $i\in\{1,2\}$.

Requirement FS\hspace{1pt}2 involves a quantification over all complete runs of the system, and thus depends
on the completeness criterion $CC$ employed. It can be formalised as
\begin{trivlist}\vspace{4pt}
\item[\hspace{\labelsep}(FS\hspace{1pt}2)] ~$F \models_B^{CC}{\bf G}(r_i \Rightarrow {\bf F}t_i)$
\vspace{3pt}
\end{trivlist}
\noindent
for $i\in\{1,2\}$, where $B=\{r_1,r_2\}$.
The set $B$ should contain $r_1$ and $r_2$, as these actions are supposed to be under the control
of the users of a fair scheduler. However, actions $t_1$, $t_2$ and $e$ should not be in $B$, as they
are under the control of the scheduler itself.
In \cite{GH15b}, the completeness criterion employed is justness, so the above formula with $CC:=J$
captures the requirement on the fair schedulers that are shown in \cite{GH15b} not to exist in CCS
or Petri nets. However, keeping $CC$ a variable allows one to pose the question under which
completeness criterion a fair scheduler \emph{can} be rendered in CCS\@. Naturally, it needs to be
a stronger criterion than justness. In \cite{GH15b} it is shown that weak fairness suffices.

FS\hspace{1pt}2 is a good example of a requirement that can \emph{not} be rendered correctly in standard LTL\@.
Writing $F \models{\bf G}(r_i \Rightarrow {\bf F}t_i)$ would rule out the complete runs of $F$
that end because the user of $F$ never supplies the input $r_j \in B$.  The CCS process\vspace{-3pt}
$$F \stackrel{{\it def}}{=} r_1.r_2.t_1.e.t_2.e.F$$ for instance
satisfies this formula, due to its unique infinite path, as well as FS\hspace{1pt}3 and 4; yet it
does not satisfy Requirement FS\hspace{1pt}2.
Namely, the path consisting of the $r_1$-transition only is complete, since it ends in a state
of which the only outgoing transition has the label $r_2\in B$. It models a (complete) run that
can occur when the environment never issues a request $r_2$, as allowed by FS\hspace{1pt}1.
Yet on this path $r_1$ is not followed by $t_1$.

Requirement FS\hspace{1pt}1 is by far the hardest to formalise. In \cite{GH15b} two formalisations are shown to
be equivalent: one involving a coinductive definition of $B$-just paths that exploits the syntax of
CCS, and the other requiring that Requirements FS\hspace{1pt}2--4 are preserved under putting an input
interface around Process~$F$. The latter demands that also
$$\widehat F:= (I_1\,|\,F[f]\,|\,I_2)\backslash \{c_1,c_2\}$$ should satisfy FS\hspace{1pt}2--4;\vspace{2pt}
here $f$ is a relabelling with $f(r_i)=c_i$, $f(t_{i})=t_{i}$ and $f(e)=e$ for $i=1,2$,
and \plat{$I_{i}\stackrel{\it def}{=} r_i.\bar{c_i}.I_{i}$} for $i\in\{1,2\}$.
A similar interface for Petri nets occurs in \cite{KW97}.

A formalisation of FS\hspace{1pt}1 on Petri nets also appears in \cite{GH15b}: each complete path $\pi$ with
only finitely many occurrences of $r_i$ should contain a state (= marking) $M$, such that there is a
transition $v$ with $\ell(v)=r_i$ and $\precond{v}\leq M$, and for each transition $u$ that occurs
in $\pi$ past $M$ one has $\precond{v} \cap \precond{u} = \emptyset$.

When discussing proposals for fair schedulers by others, FS\hspace{1pt}1 is
the requirement that is most often violated, and explaining why is not always easy.

In reactive {LTL}, this requirement is formalised as
\begin{trivlist}\vspace{4pt}
\item[\hspace{\labelsep}(FS\hspace{1pt}1)] ~$F \models_{B\setminus\{r_i\}}^J {\bf GF} r_i$
\vspace{3pt}
\end{trivlist}
\phantomsection\label{FS1}
\noindent
for $i\mathbin\in\{1,2\}$, or $F \models_{B\setminus\{r_i\}}^{\it CC} {\bf GF} r_i$ if one wants to discuss the
completeness criterion $CC$ as a parameter. The surprising element in this temporal judgement is the subscript
${B{\setminus}\{r_i\}} = \{r_{3-i}\}$, which contrasts with the assumption that requests are under
the control of the environment.
FS\hspace{1pt}1 says that, although we know that there is no guarantee that user $i$ of $F$ will ever issue
request $r_i$, under the assumption that the user \emph{does} want to make such a request, making
the request should certainly succeed. This means that the protocol itself does not sit in the way of making this request.

The combination of Requirements FS\hspace{1pt}1 and 2, which use different sets of blocking actions as a
parameter, is enabled by reactive {LTL} as presented here.

The following examples, taken from \cite{GH15b}, show that all the above requirements are necessary for the result from
\cite{GH15b} that fair schedulers cannot be rendered in CCS\@.
\begin{itemize}
\item The CCS process $F_1|F_2$ with {$F_i \stackrel{{\it def}}{=} r_i.t_i.e.F_i$} satisfies 
  FS\hspace{1pt}1, FS\hspace{1pt}2 and FS\hspace{1pt}3$'$. In FS\hspace{1pt}1 and 2 one needs to take ${\it CC} := J$, as progress is not a strong
  enough assumption here.
\item The process $E_1|G|E_2$ with {$E_i \mathbin{\stackrel{{\it def}}{=}} r_i.E_i$} and
  \plat{$G \stackrel{{\it def}}{=} t_1.e.t_2.e.G$} satisfies FS\hspace{1pt}1, 2 and 4,
  again with ${\it CC} \mathbin{:=} J$.\vspace{2pt}
\item The process $E_1|E_2$ satisfies FS\hspace{1pt}1, 3$'$ and 4, again with ${\it CC} \mathbin{:=} J$ in FS\hspace{1pt}1.
\item The process $F_0$ with {$F_0 \stackrel{{\it def}}{=} r_1.t_1.e.F_0 + r_2.t_2.e.F_0$} satisfies FS\hspace{1pt}2--4.
  Here FS\hspace{1pt}2 merely needs ${\it CC} \mathbin{:=} {\it Pr}$, that is, the assumption of progress.
  Furthermore, it satisfies FS\hspace{1pt}1 with ${\it CC} := {\it SF}(\Tsk)$, as long as
  $\textsc{r}_1,\textsc{r}_2 \in \Tsk$. Here $\textsc{r}_i$ is the set of transitions with label $r_i$.
\end{itemize}

\noindent
  The process $X$ given by\hfill $X \stackrel{{\it def}}{=} r_1.Y + r_2.Z$,\hfill
  \plat{$Y \stackrel{{\it def}}{=} r_2.t_1.e.Z + t_1.(r_2.e.Z + e.X)$}\hfill and\vspace{3pt}\linebreak
  \plat{$Z \stackrel{{\it def}}{=} r_1.t_2.e.Y + t_2.(r_1.e.Y + e.X)$}, the \emph{gatekeeper},
  is depicted
{\makeatletter
\let\par\@@par
\par\parshape0
\everypar{}\begin{wrapfigure}[14]{r}{0.297\textwidth}
\vspace{-3ex}
 \input{Gatekeeper1}
 \centerline{\raisebox{1ex}{\box\graph}}
\end{wrapfigure}
\noindent
 on the right.
The grey shadows represent copies of the states at the opposite end of the diagram, so the
transitions on the far right and bottom loop around.
This process satisfies FS\hspace{1pt}3$'$ and 4, FS\hspace{1pt}2 with ${\it CC} \mathbin{:=} {\it Pr}$,
  and FS\hspace{1pt}1 with ${\it CC} := {\it WF}(\Tsk)$, thereby improving Process $F_0$, and
  constituting the best CCS approximation of a fair scheduler seen so far.
  Yet, intuitively FS\hspace{1pt}1 is not ensured at all, meaning that weak
  fairness is too strong an assumption. Nothing really prevents all the choices between $r_2$ and
  any other action $a$ to be made in favour of $a$.
\par}

\section{Formalising Requirements for Mutual Exclusion in Reactive LTL}\label{sec:formalising mutex}

Define a process $i$ participating in a mutual exclusion protocol to cycle through the stages
\emph{noncritical section}, \emph{entry protocol}, \emph{critical section}, and \emph{exit protocol},
in that order, as explained in \Sec{history}. Modelled as an LTS, its visible actions will be
$\en$, $\lni$, $\ec$ and $\lc$, of entering and leaving its (non)critical section.
Put $\lni$ in $B$ to make leaving the noncritical section a blocking action.
The environment blocking it is my way of allowing the client process to stay in its noncritical
section forever. This is the manner in which the requirement \emph{optionality} is captured in reactive temporal logic.
On the other hand, $\ec$ should not be in $B$, for one does not consider the starvation-freedom property of a
mutual exclusion protocol to be violated simply because the client process refuses to enter the
critical section when allowed by the protocol. Likewise, $\en$ is not in $B$. Although exiting the
critical section is in fact under control of the client process, it is assumed that it
will not stay in the critical section forever. In the models of this paper this can be simply
achieved by leaving $\lc$ outside $B$. Hence $B := \{\lni \mid i=1,\dots,N\}$.
\newcommand{\act}{\,{\it act}_i}

\hypertarget{ME1}{%
My first requirement on mutual exclusion protocols $P$ simply says that the actions $\en$, $\lni$,
$\ec$ and $\lc$ have to occur in the right order:
\begin{trivlist}
\vspace{4pt}
\item[\hspace{1pt} (\MEA)] ~$P \models {\big(}(\neg \act) {\bf W} \lni {\big)}
  \begin{array}[t]{@{}l@{}}
  \mbox{} \wedge {\bf G}\left(\lni \Rightarrow
   {\bf Y} \big((\neg \act) {\bf W} \ec\big)\right)
  \wedge {\bf G}\left(\ec \Rightarrow
   {\bf Y} \big((\neg \act) {\bf W} \lc\big)\right)
  \\ \mbox{}
  \wedge {\bf G}\left(\lc \Rightarrow
   {\bf Y} \big((\neg \act) {\bf W} \en\big)\right)
  \wedge {\bf G}\left(\en \Rightarrow
   {\bf Y} \big((\neg \act) {\bf W} \lni\big)\right)
\end{array}$
\end{trivlist}
\noindent
for $i=1,\dots,N$. Here $\act := (\lni \vee \ec \vee \lc \vee \en)$.}

The second is a formalisation of \textit{mutex}, saying that only one process can be in its critical
section at the same time:
\hypertarget{ME}{%
\begin{trivlist}
\vspace{3pt}
\item[\hspace{5pt}(\MEB)]
~$P \models {\bf G}\left(\ec \Rightarrow \big((\neg \textcolor{red}{\textit{ec}_j}) {\bf W} \lc\big)\right)$
\vspace{2pt}
\end{trivlist}
\noindent
for all $i,j=1,\dots,N$ with $i\neq j$.
Both {\MEA} and {\MEB} are safety properties, and thus unaffected by changing $\models$ into $\models^{CC}$ or $\models_B^{CC}$.
}

The starvation-freedom requirement of \Sec{history} can be formalised as
\begin{trivlist}
\vspace{4pt}
\item[\hspace{5pt}(\MEC)] ~$P \models_B^{CC} {\bf G}(\lni \Rightarrow {\bf F}\ec)$
\vspace{3pt}
\end{trivlist}
\noindent
for $i=1,\dots,N$.
Here the choice of a completeness criterion is important.
Finally, the following requirements are similar to starvation-freedom, and state that from each section in
the cycle of a process $i$, the next section will in fact be reached.
In regards to reaching the end of the noncritical section, this should be guaranteed only when
assuming that the process wants to leave its noncritical section; hence $\lni$ is excepted from $B$.
\begin{trivlist}
\vspace{4pt}
\item[\hspace{5pt}(\MED)] ~$P \models_B^{CC} {\bf G}(\ec \Rightarrow {\bf F}\lc)$
\item[\hspace{5pt}(\MEE)] ~$P \models_B^{CC} {\bf G}(\lc \Rightarrow {\bf F}\en)$
\item[\hspace{5pt}(\MEF)] ~$P \models_{B\setminus\mbox{\scriptsize\{\lni\}}}^{CC} {\bf F}\lni
   \wedge {\bf G}(\en \Rightarrow {\bf F}\lni)$
\end{trivlist}
\noindent
for $i=1,\dots,N$.

The requirement \emph{speed independence} is automatically satisfied for models of mutual
exclusion protocols rendered in any of the formalisms discussed so far, as these formalisms
lack the expressiveness to make anything dependent on speed.

The following examples show that none of the above requirements are redundant.
\begin{itemize}
\item The CCS process $F_1|F_2| \cdots|F_N$ with \plat{$F_i \stackrel{{\it def}}{=}
  \lni.\ec.\lc.\en.F_i$} satisfies all requirements, with ${\it CC} := J$, except for \hyperlink{ME}{\MEB}.\vspace{2pt}
\item The process $R_1|R_2| \cdots|R_N$ with \plat{$R_i \stackrel{{\it def}}{=} \lni.{\bf 0}$}
  satisfies all requirements except for \hyperlink{ME}{\MEC[]}.\vspace{2pt}
\item For $N\mathbin=2$, process \plat{$H \stackrel{{\it def}}{=}$}
   {\color{blue}\it ln$_1$}.{\color{red}\it ec$_1$}.{\color{blue}\it ln$_2$}.{\color{red}\it lc$_1$}.%
   {\color{blue}\it en$_1$}.{\color{red}\it ec$_2$}.{\color{red}\it lc$_2$}.{\color{blue}\it en$_2$}.$H$
   \!+\! {\color{blue}\it ln$_2$}.{\color{red}\it ec$_2$}.{\color{blue}\it ln$_1$}.{\color{red}\it lc$_2$}.%
   {\color{blue}\it en$_2$}.{\color{red}\it ec$_1$}.{\color{red}\it lc$_1$}.{\color{blue}\it en$_1$}.$H$
   satisfies all requirements but \hyperlink{ME}{\MED[]}.
   The case $N>2$ is only notationally more cumbersome.
   Similarly, one finds examples failing only on \hyperlink{ME}{\MEE[]}, or on the
   second conjunct of \hyperlink{ME}{\MEF[]}.\vspace{2pt}
\item The process ${\bf 0}$ satisfies all requirements except for the first conjunct of \hyperlink{ME}{\MEF[]}.\vspace{2pt}
\item In case $N\mathbin=1$, the process
   \plat{$W \stackrel{{\it def}}{=} \mbox{\color{red}\it lc$_1$}.
   \mbox{\color{red}\it ec$_1$}.\mbox{\color{red}\it lc$_1$}.\mbox{\color{blue}\it en$_1$}.
   \mbox{\color{blue}\it ln$_1$}.W$} satisfies all requirements but \hyperlink{ME1}{\MEA}.
\end{itemize}
\noindent
  The process $X$, a gatekeeper variant, given by 
  \plat{$X \stackrel{{\it def}}{=} \lni[1].Y + \lni[2].Z$},\\[1pt]
  \plat{$Y \stackrel{{\it def}}{=} \lni[2].\ec[1].\lc[1].\en[1].Z + \ec[1].(\lni[2].\lc[1].\en[1].Z +
  \lc[1].(\lni[2].\en[1].Z + \en[1].X))$}\\[2pt]
  \plat{$Z\hspace{-.35pt} \stackrel{{\it def}}{=} \lni[1].\ec[2].\lc[2].\en[2].Y + \ec[2].(\lni[1].\lc[2].\en[2].Y +
  \lc[2].(\lni[1].\en[2].Y + \en[2].X))$}\vspace{2pt}
{\makeatletter
\let\par\@@par
\par\parshape0
\everypar{}\begin{wrapfigure}[12]{r}{0.383\textwidth}
 \vspace{-2ex}
 \input{gatekeeper2}
  \centerline{\raisebox{1ex}{\box\graph}}
 \end{wrapfigure}
\noindent
is depicted on the right.
It satisfies \hyperlink{ME1}{\MEA}, \hyperlink{ME}{\MEB}, \hyperlink{ME}{\MEC},
\hyperlink{ME}{\MED} and  \hyperlink{ME}{\MEE} with ${\it CC} \mathbin{:=} {\it Pr}$ and
\hyperlink{ME}{\MEF[]} with ${\it CC} \mathbin{:=} {\it WF}(\Tsk)$,
where $\textsc{ln}_1,\textsc{ln}_2 \in \Tsk$.
It can be seen as a mediator that synchronises, on the actions $\lni$, $\ec$, $\lc$ and $\en$,
with the actual processes that need to exclusively enter their critical sections.
Yet, it would not be commonly accepted as a valid mutual exclusion protocol, since nothing prevents
it to never choose $\lni[2]$.\linebreak[4] This means that merely requiring weak
fairness in \hyperlink{ME}{\MEF[]} makes this requirement unacceptably weak.
The problem with this protocol is that it ensures starvation-freedom by making it hard for processes
to leave their noncritical sections.
\par}

\section{State-oriented Requirements for Mutual Exclusion}\label{sec:state-oriented}

Some readers may prefer a state-oriented view of mutual exclusion over the action-oriented view of
\Sec{formalising mutex}. In such a view, the mutex requirement (\hyperlink{ME}{\MEB} in \Sec{formalising mutex})
simply says that different processes $i$ and $j$ cannot \emph{be} in the critical section at the
same time, rather than encoding this in terms of the actions of entering and leaving the critical section.
This section translates the requirements on mutual exclusion from the action-oriented view of
\Sec{formalising mutex} to a state-oriented view.

Let's model the protocol with a Kripke structure that features the atomic predicates $C_i$ and $I_i$, for $i=1,2$.
Predicate $C_i$ holds when Process $i$ is in its critical section, and $I_i$ when Process $i$ intends to
enter its critical section, but isn't there yet. Predicate $I_i$ should thus hold when Process $i$
has left the noncritical section and is executing its entry protocol. In this presentation one can
lump together the exit protocol and the noncritical section of process $i$; these comprise the
states satisfying $\neg(I_i \vee C_i)$.
Now Requirements \hyperlink{ME}{\MEA--\MEF[]} can be reformulated as follows, for $i,j=1,\dots,N$.
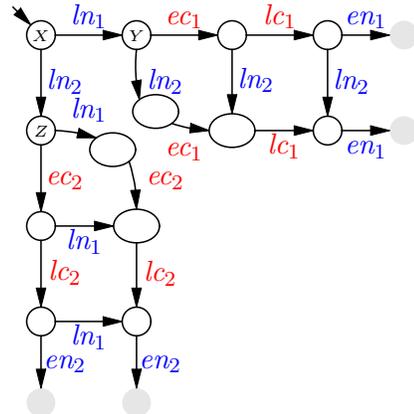
\begin{trivlist}
\vspace{4pt}
\item[\hspace{1pt} (\hyperlink{ME1}{\MEA})] \(\quad\!\! P \models 
  \begin{array}[t]{@{}l@{}}
  {\bf G}\neg(C_i \wedge I_i)
  \wedge \neg(I_i \vee C_i)
  \wedge {\bf G}\left(I_i \Rightarrow (I_i {\bf W} C_i\right))
  \wedge {\bf G}\left(C_i \Rightarrow (C_i {\bf W} \neg(I_i \vee C_i))\right)
  \\ \mbox{}
  \wedge {\bf G}\left(\neg(I_i \vee C_i) \Rightarrow (\neg(I_i \vee C_i) {\bf W} I_i)\right)
\end{array}\)
\item[\hspace{5pt}(\hyperlink{ME}{\MEB})]
~~\quad $P \models {\bf G}\left(\neg(C_i \wedge C_j)\right)$ \qquad\qquad\qquad\; for all $j\neq i$
\item[\hspace{5pt}(\hyperlink{ME}{\MEC})] ~$P \models_B^{CC} {\bf G}(I_i \Rightarrow {\bf F}C_i)$
\item[\hspace{5pt}(\hyperlink{ME}{\MED})] ~$P \models_B^{CC} {\bf G}(C_i \Rightarrow {\bf F}\neg(I_i \vee C_i))$
\item[\hspace{5pt}(\hyperlink{ME}{\MEF})] ~$P \models_{B\setminus\{I_i\}}^{CC} {\bf G}(\neg(I_i \vee C_i) \Rightarrow {\bf F}I_i)$
\vspace{3pt}
\end{trivlist}

{\makeatletter
\let\par\@@par
\par\parshape0
\everypar{}\begin{wrapfigure}[19]{r}{0.35\textwidth}
 \vspace{-13ex}
 \tiny
 \input{gatekeeper}
 \centerline{\raisebox{1ex}{\box\graph}}
\end{wrapfigure}
\noindent
Here $B$ can be rendered as a set of atomic predicates $p$, and refers to all those ``blocking transitions''
that go from a state where $p$ does not hold to one where $p$ holds, for some $p\in B$.
So here $B=\{I_i \mid i=1,\dots,N\}$.

In this setting, the gatekeeper is depicted on the right.
It satisfies \hyperlink{ME1}{\MEA}, \hyperlink{ME}{\MEB}, \hyperlink{ME}{\MEC}
and \hyperlink{ME}{\MED} above, with $CC=Pr$,
but {\MEF} only with $CC={\it WF}$.
\par}



\section{A Hierarchy of Quality Criteria for Mutual Exclusion Protocols}\label{sec:hierarchy}

Formalising the quality criteria for fair schedulers as FS\hspace{1pt}1--4, one sees that, unlike 
FS\hspace{1pt}3 and 4, Requirements FS\hspace{1pt}1 and 2 are parametrised by the choice
of a completeness criterion $CC$. In each of FS\hspace{1pt}1 and FS\hspace{1pt}2, CC can be
instantiated with either $\top$, ${\it Pr}$, $J$, ${\it WF}(\Tsk)$ or ${\it SF}(\Tsk)$
for a suitable collection of tasks $\Tsk$. When seeing ${\it WF}(\Tsk)$ or ${\it SF}(\Tsk)$ as
single choices, allowing them to utilise the most appropriate choice of $\Tsk$,
this yields a hierarchy of $5\times 5 =25$ different quality criteria for fair schedulers, partially
depicted in Figure~\ref{hierarchy}.
\begin{figure}[hbt]
\input{hierarchy}
\centerline{\raisebox{1ex}{\box\graph}}
\vspace{-3pt}
\caption{\it A hierarchy of quality criteria for fair schedulers and mutual exclusion protocols}
\label{hierarchy}
\end{figure}
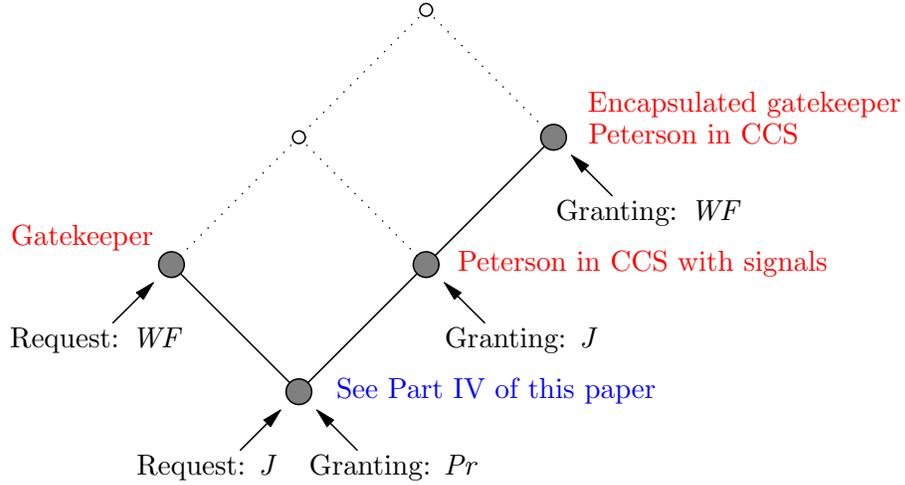
Here ``Request'' indicates the completeness criterion used in Requirement FS\hspace{1pt}1
and ``Granting'' the one taken in FS\hspace{1pt}2. Note that quality criteria encountered further up
in the figure employ stronger fairness assumptions, and thus yield weaker, or less impressive, fair schedulers.

I have not rendered all 25 quality criteria in Figure~\ref{hierarchy}, as many are irrelevant. Since
no meaningful liveness property holds when merely assuming the trivial completeness criterion
$\top$, one can safely discard it from consideration; there can exist no fair scheduler satisfying
FS\hspace{1pt}1 or 2 with $CC=\top$.  Likewise, one can forget about the possibility ``Request: {\it Pr}''.
As an infinite run such as $(r_2 t_2 e)^\infty$ in which a request $r_1$ is never received,
and consequently a request-granting action $t_1$ never occurs, should be a complete run of the
system, progress is not strong enough an assumption to ensure that when user $1$ wants to issue
request $r_1$ it will actually succeed. The least one should assume here is justness.

At the other end of the hierarchy I have dropped the choice {\it SF}. The reason is that there turn
out to be completely satisfactory solutions merely assuming weak fairness in either dimension; so
the choice of strong fairness makes the fair scheduler unnecessary weak.

The same hierarchy of quality criteria applies to mutual exclusion protocols.
Now ``Request'' indicates the completeness criterion used in Requirement \hyperlink{ME}{\MEF[]}
and ``Granting'' the one taken in \hyperlink{ME}{\MEC[]}. 
The latter concerns the starvation-freedom property of mutual exclusion; it indicates how hard it is to
reach the critical section after a process' interest in doing this has been expressed.
The former indicates how hard it is to express such an interest in the first place.
Again, the choice ``Request: {\it Pr}'' can be discarded, as no mutual exclusion protocol can meet
this requirement. This is due to the infinite run $(\lni[2] \ec[2] \lc[2] \en[2])^\infty$,
in which process $1$ never requests access to the critical section. When merely assuming progress,
one cannot tell whether this is because process 1 does not want to leave its noncritical section,
or because it wants to but doesn't succeed, as always another action is chosen.

In principle, there are two more dimensions in classifying the quality criteria for mutual exclusion
protocols, namely the choice of a completeness criterion for Requirements \hyperlink{ME}{\MED[] and \MEE[]}.
These indicate how hard it is to leave the critical section after entering, and to enter the
noncritical section after leaving the critical one, respectively. As both tasks are really easy to
accomplish, these two dimensions are not indicated in Figure~\ref{hierarchy}.
Nevertheless, they should not be forgotten in the forthcoming analysis.
There are no further dimensions for \hyperlink{ME1}{{\MEA} and \MEB}, as these are safety properties.

When allowing weak fairness in the ``request'' dimension, the gatekeeper, described for fair schedulers
in \Sec{formalising FS} and for mutual exclusion in \Sec{formalising mutex}, is a good solution.
It merely requires progress in the ``granting'' dimension, and for mutual exclusion also in
\hyperlink{ME}{\MED[] and \MEE[]}.
As most researchers in the area of mutual exclusion would agree that nevertheless the gatekeeper is
not an acceptable protocol, we have evidence that weak fairness in the ``request'' dimension is too
strong an assumption.

\section{An Input Interface for Implementing LN}\label{sec:interface}

In \cite[Section~13]{GH15b} an input interface is proposed that can be put around any potential fair
scheduler expressed in CCS---it was recalled in \Sec{formalising FS}. It was shown that a process $F$ satisfies
Requirements FS\hspace{1pt}1--4 iff the \emph{encapsulated} process $\widehat F$---the result of
putting $F$ in this interface---satisfies FS\hspace{1pt}2--4. Here I propose a similar
interface for mutual exclusion protocols, restricting attention to the case of $N=2$ parties.

\begin{definition}{input interface}
For any expression $P$, let $\widehat P := (I_1 \mid P[f] \mid I_2){\setminus} \{c_1,c_2,d_1,d_2\}$
where \plat{$I_i \stackrel{{\it def}}{=} \lni.\bar{c_i}.\bar{d_i}.\en. I_i$} for $i\in\{1,2\}$ and $f$
is a relabelling with $f(\lni)=c_i$ and $f(\en)=d_i$ for $i\in\{1,2\}$,
such that $f(a)=a$ for all other actions $a$ occurring in $P$.
\end{definition}

\begin{observation}{input interface}
Suppose that $P$ satisfies \hyperlink{ME1}{\MEA}, and criterion $CC$ is at least as strong as justness.
Then $\widehat P$ satisfies \hyperlink{ME1}{\MEA} as well as \hyperlink{ME}{\MEF[J]}, and
\begin{itemize}
\item $\widehat P$ satisfies \hyperlink{ME}{\MEB} iff $P$ satisfies \hyperlink{ME}{\MEB}.
\item $\widehat P$ satisfies \hyperlink{ME}{\MEC} iff $P$ satisfies \hyperlink{ME}{{\MEF} and \MEC},
\item $\widehat P$ satisfies \hyperlink{ME}{\MED} iff $P$ satisfies \hyperlink{ME}{\MED}.
\item $\widehat P$ satisfies \hyperlink{ME}{\MEE} iff $P$ satisfies \hyperlink{ME}{\MEE}.
\end{itemize}
\end{observation}
\noindent
Let the \emph{encapsulated gatekeeper} $H$ be the result of putting this input interface
around the gatekeeper for mutual exclusion. It can be described as follows, and its labelled
transition system is depicted in Figure~\ref{encapsulated}.
Here I added subscripts to $\tau$-action to indicate their origins.
\[\begin{array}{l}
  H := (I_1\,|\,X\,|\,I_2)\backslash \{c_1,c_2,d_1,d_2\}\\
  I_i \stackrel{{\it def}}{=} \lni.\bar{c_i}.d_i.\en. I_i \qquad\qquad \mbox{for}~ i\in\{1,2\},\\
  X \stackrel{{\it def}}{=} c_1.Y + c_2.Z,\\
  Y \stackrel{{\it def}}{=} c_2.\ec[1].\lc[1].\bar{d_1}.Z + \ec[1].(c_2.\lc[1].\bar{d_1}.Z +
  \lc[1].(c_2.\bar{d_1}.Z + \bar{d_1}.X))\\
  Z\hspace{-.35pt} \stackrel{{\it def}}{=} c_1.\ec[2].\lc[2].\bar{d_2}.Y + \ec[2].(c_1.\lc[2].\bar{d_2}.Y +
  \lc[2].(c_1.\bar{d_2}.Y + \bar{d_2}.X))
\end{array}\]

{\makeatletter
\let\par\@@par
\par\parshape0
\everypar{}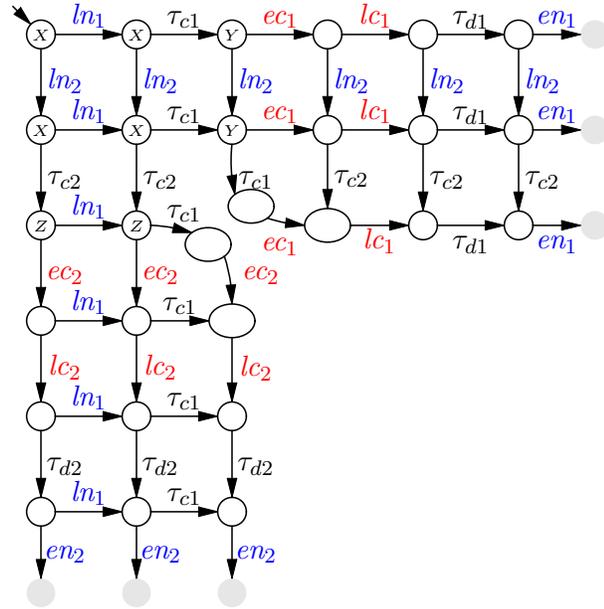
\begin{wrapfigure}[20]{r}{0.52\textwidth}
 \vspace{-1.5ex}
 \input{Encapsulated}
  \centerline{\raisebox{1ex}{\box\graph}}
 \caption{\it Encapsulated gatekeeper}
 \label{encapsulated}
 \end{wrapfigure}
\noindent
This mutual exclusion protocol satisfies \hyperlink{ME}{\MEF[J]},
for as soon as $\en$ has occurred (and in the initial state)
Process $I_i$ is in its initial state $\lni.\bar{c_i}.d_i.\en. I_i$, and nothing stands in the way of the
action $\lni$. In other words, justness is a strong enough assumption for $\lni$ to occur.
Clearly, the protocol also satisfies \hyperlink{ME1}{\MEA} and \hyperlink{ME1}{\MEB}, as well as
\hyperlink{ME}{\MED[Pr]} and \hyperlink{ME}{\MED[Pr]}.
The only downside is that it takes weak fairness to achieve \mbox{\hyperlink{ME}{\MEC[]}},
starvation-freedom. This assumption is needed to assure that
the synchronisation between actions $\bar{c_i}$ and $c_i$ will actually occur.\\
\mbox{}\hspace{5pt} Intuitively, the encapsulated gatekeeper is an equally unacceptable mutual exclusion protocol
as the gatekeeper, for the input interface ought to make no difference.
This shows that weak fairness in any dimension of Figure~\ref{hierarchy} is too strong an assumption.
However, due to the impossibility result of \cite{GH15b}, the two remaining entries of Figure~\ref{hierarchy}
cannot be realised in CCS\@. Theoretically, that result leaves open the possibility of achieving
justness in both the dimensions ``request'' and ``granting'', at the expense of assuming weak
fairness for \hyperlink{ME}{\MED[] or \MEE[]}. I do not think this is actually possible, and even if it were,
a solution that requires weak fairness to escape the critical section, or to enter the noncritical
one, appears equally unacceptable as the (encapsulated) gatekeeper.
\par}

\addtocontents{toc}{\protect\vspace{-4pt}}
\part{Impossibility Results for Peterson's Mutual Exclusion Algorithm}\label{impossibility results}

Here I recall three impossibility results for mutual exclusion protocols that have been shown or
claimed earlier, and illustrate or substantiate them for Peterson's mutual exclusion protocol.
I could have equally well done this for another mutual exclusion protocol, such as
Lamport's bakery algorithm \cite{bakery}, Aravind's mutual exclusion algorithm \cite{Ar11}
or the \emph{round-robin} scheduler.\footnote{As a mutual exclusion protocol, the round-robin is a 
  central scheduler that grants access to the critical section to $N$ processes numbered $1-N$
  by cycling through all competing processes in the order $1-N$. Each time it is the turn of
  Process $i$, the round-robin scheduler checks whether Process $i$ wants to enter the critical
  section, and if so, grants access. When Process $i$ leaves the critical section, or if it didn't
  want to enter, it will be the turn of Process $i{+}1$ mod $N$.

  When confronted with the claim that under some natural assumptions no correct mutual exclusion
  protocols exist, some people reply that in that case one could always use a round-robin scheduler,
  as if this somehow constitutes an exception.}
My reason for choosing Peterson's protocol in this paper is that it is (one of) the simplest of all
mutual exclusion protocols.

The first impossibility result stems from \cite{Vogler02,KW97,GH15b}, and says that in Petri nets,
and in CCS and similar process algebras, it is not possible to model a mutual exclusion protocol
in such a way that it is correct without making an assumption as strong as weak fairness.
  In \cite{Vogler02} this is shown for finite Petri nets, and in \cite{KW97} for a class of
  Petri nets that interact with their environment through an interface of a particular shape,
  similar to the one of \Sec{interface}. In \cite{GH15b} the same is shown for all structural
  conflict nets (and thus for all safe nets), as well as for the process algebra CCS, with strong
  hints on how the result extends to many similar process algebras. For the latter result, either
  the concurrency relation between CCS transitions defined in \Sec{LTSC}, or directly the resulting
  concept of a just path, needs to be seen as an integral part of CCS\@.
In Sections~\ref{sec:Peterson Petri} and~\ref{sec:Peterson CCS} I will illustrate this impossibility
result for a rendering of Peterson's protocol as a Petri net and as a CCS expression, respectively. 
In \cite{GH15b,GH19} moreover the point of view is defended that assuming (strong or weak) fairness is typically
unwarranted, in the sense that there is no reason to assume that reality will behave in a fair way.
From this point of view, a model of a mutual exclusion that hinges on a fairness assumption
can be seen as incorrect or unsatisfactory. This makes the above into a real impossibility result.

The second impossibility result was claimed in \cite{vG18c}, but unaccompanied by written evidence.
It blames the first impossibility result above on the combination of two assumptions or protocol
features, which I here call \emph{atomicity} and \emph{speed independence}.
Atomicity, or rather the special case of atomicity that is relevant for the second impossibility result,
will be formally defined as (\ref{atomicity}) in \Sec{atomicity}. It can be seen as an assumption on the behaviour
of the hardware on which an implementation of a mutual exclusion protocol will be running.
Atomicity is explicitly assumed in the original paper of Dijkstra where the mutual exclusion problem
was presented \cite{Dijk65}, and implicitly in many other papers on mutual exclusion,
but not in the work of Lamport \cite{bakery}.
Speed independence can either be seen as an assumption on the underlying hardware, or as a feature
of a mutual exclusion protocol. The assumption stems from Dijkstra \cite{Dijk65} and was quoted in \Sec{history}.
The claims of \cite{vG18c} employ a rather strict interpretation of speed independence, illustrated
in \ex{speed independence}.

In a setting where solutions based on the assumption of weak fairness are rejected, as well as
solutions that are merely probabilistically correct, \cite{vG18c} claims that when assuming atomicity,
speed-independent mutual exclusion is impossible. This means that when assuming atomicity and speed
independence, there is no mutual exclusion protocol satisfying \hyperlink{ME1}{\MEA}--\hyperlink{ME}{\MEF[]}
with $CC=J$. The assumption of speed independence is built in in CCS and Petri nets, in the sense
that any correct mutual exclusion protocol formalised therein is automatically speed independent.
This is because these models lack the expressiveness to make anything dependent on speed.
When taking the concurrency relation between CCS or net transitions defined in \Sec{LTSC} as an
integral part of semantics of CCS or Petri nets, also the assumption of atomicity is built in in
these frameworks. This makes the first impossibility result a special case of the second.
The latter can be seen as a generalisation of the former that is not dependent on a particular
modelling framework. In \Sec{impossible} I will substantiate the above claim of \cite{vG18c} for the special
case of Peterson's protocol.

The third impossibility result, also claimed in \cite{vG18c}, says that when dropping the assumption
of atomicity, but keeping speed independence, there still exists no mutual exclusion protocol
satisfying \hyperlink{ME}{\MEC[Pr]}. That is, the assumption of progress is not
strong enough to obtain starvation-freedom of any speed-independent mutual exclusion protocol.
I will substantiate this claim for the special case of Peterson's protocol in \Sec{Peterson}.
In Part~\ref{timeouts} I aim for a formalisation of Peterson's protocol that satisfies
\hyperlink{ME}{\MEC[Pr]}. It follows that there I will have to drop speed independence.

\section{Peterson's Mutual Exclusion Protocol}\label{sec:Peterson}

A pseudocode rendering of Peterson's protocol is depicted in \autoref{fig:peterson}.
The two processes, here called {\procA} and {\procB}, use three shared variables: $\rA$, $\rB$ and $\tu$. 
The Boolean variable $\rA$ can be
written by Process $\procA$ and read by Process $\procB$, whereas $\rB$ can be
written by $\procB$ and read by $\procA$. By setting $\rA$ to $\tr$, Process
$\procA$ signals to Process $\procB$ that it wants to enter the critical
section. The variable $\tu$ can be written and read by both processes. 
Its carefully designed  functionality guarantees  mutual exclusion as well as deadlock-freedom. 
Both $\rA$ and $\rB$ are initialised with\/ $\fa$ and $\tu$ with $\it A$.
\begin{figure}[t]
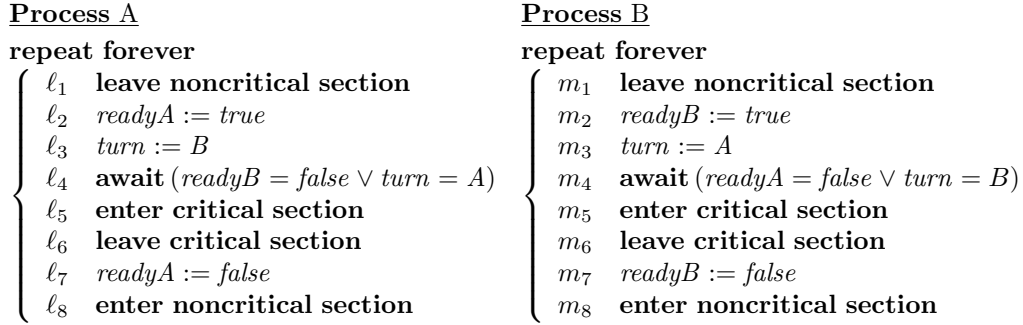

\centering
\small$
\begin{array}{@{}l@{}}
\underline{\bf Process~\procA}\\[.5ex]
{\bf repeat~forever}\\
\left\{\begin{array}{ll}
\ell_1 & {\bf leave~noncritical~section}\\
\ell_2 & \it \rA := \tr	\\
\ell_3 & \it \tu := B \\
\ell_4 & {\bf await}\,(\it \rB = \fa \vee \tu = A) \\
\ell_5 & {\bf enter~critical~section}\\
\ell_6 & {\bf leave~critical~section}\\
\ell_7 & \it \rA := \fa\\
\ell_8 & {\bf enter~noncritical~section}\\
\end{array}\right.
\end{array}
~~~~~~
\begin{array}{@{}l@{}}
\underline{\bf Process~\procB}\\[.5ex]
{\bf repeat~forever}\\
\left\{\begin{array}{ll}
m_1 & {\bf leave~noncritical~section}\\
m_2 & \it \rB := \tr	\\
m_3 & \it \tu := A \\
m_4 & {\bf await}\,(\rA = \fa \vee \tu = B) \\
m_5 & {\bf enter~critical~section}\\
m_6 & {\bf leave~critical~section}\\
m_7 & \it \rB := \fa \\
m_8 & {\bf enter~noncritical~section}\\
\end{array}\right.
\end{array}$
\caption{\it Peterson's algorithm (pseudocode)}
\label{fig:peterson}
\end{figure}

\begin{figure}[ht]
{\scriptsize
\input{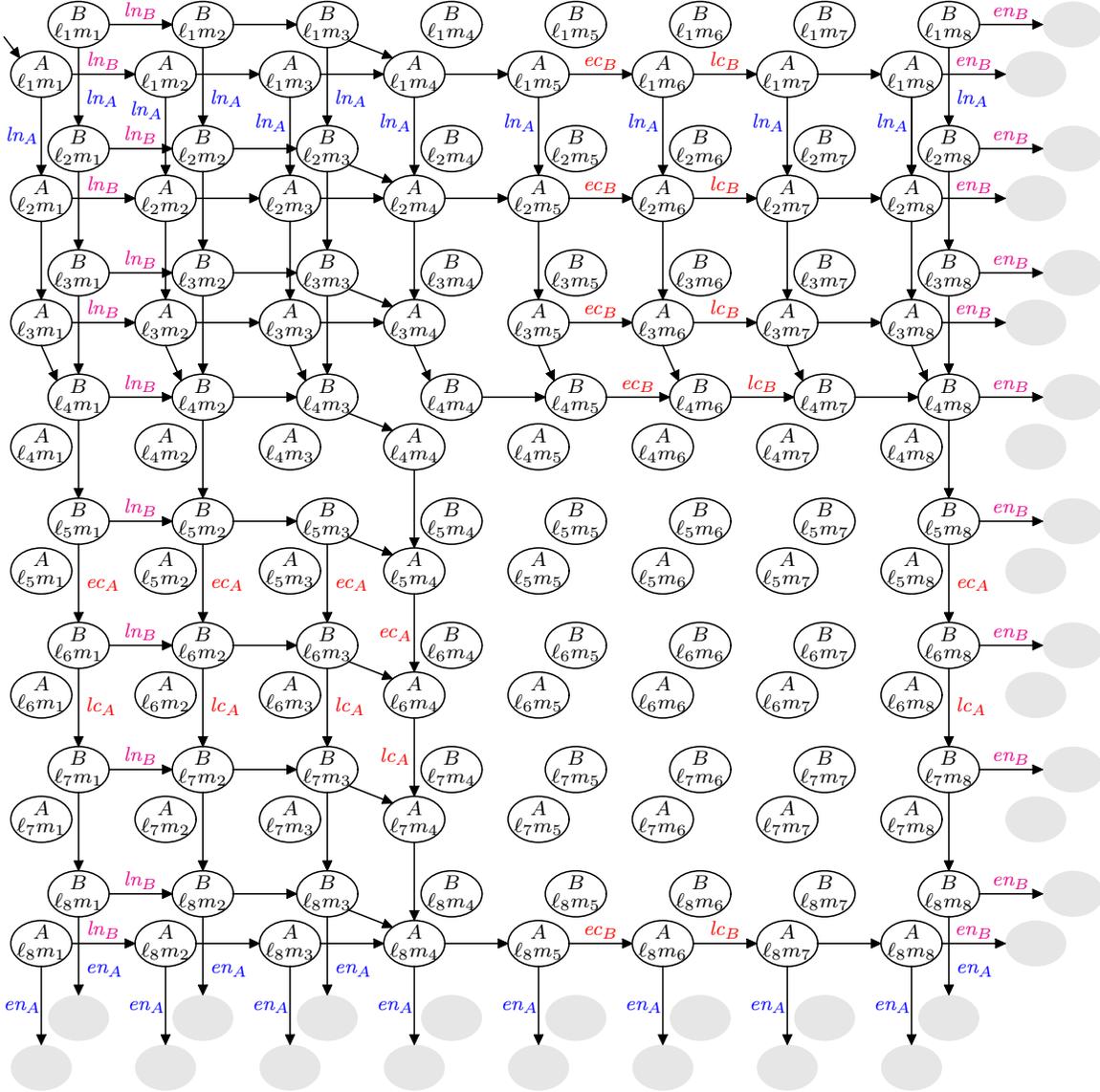}
\centerline{\box\graph}
}
\caption{\it LTS of Peterson's mutual exclusion algorithm}
\label{Peterson}
\end{figure}

Figure~\ref{Peterson} presents a labelled transition system for this protocol.
The name of a state \plat{\scriptsize\lm ijT} where $T$ is $A$ or $B$ indicates that Process $\procA$ is in
State $\ell_i$, Process $\procB$ in State $m_j$ and the variable $turn$ has value $T$,
with the convention that Instruction $\ell_i$ leads from State $\ell_i$ to State $\ell_{i+1}$.
This completely determines the values of the variables $\rA$ and $\rB$.
The actions
$\lnA$,
$\lnB$,
$\ecA$,
$\ecB$,
$\lcA$,
$\lcB$,
$\enA$ and
$\enB$
are visible; they correspond to the instructions $\ell_1$, $m_1$,  $\ell_5$, $m_5$,
$\ell_6$, $m_6$,  $\ell_8$ and $m_8$, respectively. All other transitions, unlabelled in
Figure~\ref{Peterson}, are labelled $\tau$. When assuming speed independence, none of the paths in 
Figure~\ref{Peterson} can be ruled out due to timing considerations. This makes 
Figure~\ref{Peterson} an adequate rendering of Peterson's protocol.

From the pseudocode in Figure~\ref{fig:peterson} one sees immediately that Peterson's protocol
satisfies \hyperlink{ME1}{\MEA} (all visible actions occur in the right order) and
\hyperlink{ME}{\MEF[J]} (nothing stands in the way of a process leaving its noncritical section).
By inspecting the LTS in Figure~\ref{Peterson} one sees that it moreover satisfies
\hyperlink{ME}{\MEB} (the mutual exclusion property) and \hyperlink{ME}{\MED[Pr]} (assuming progress
and willingness to do so is enough to ensure that a process will leave its critical section
after entering). One obtains \hyperlink{ME}{\MEE[J]} (assuming justness suffices to ensure that a process
always enters its noncritical section after leaving its critical section) by combining the code and the LTS\@.
The LTS shows that assuming progress suffices to ensure that $\lcB$ is always followed by $m_7$.
The code shows that assuming justness suffices to ensure that $m_7$ is always followed by $\lnB$.
Of course the same applies to Process A\@.

More problematic is Requirement \hyperlink{ME}{\MEC[]}, starvation-freedom.
Thanks to symmetry, I may restrict attention to \hyperlink{ME}{\MEC[]} for Process A\@.
Will Instruction $\ell_1 = \lnA$ always be followed by $\ell_5 = \ecA$?
The LTS shows that $\ell_2$ is always followed by $\ell_5=\ecA$, even when
merely assuming progress.  However, it is less clear whether $\ell_1=\lnA$ is always followed by $\ell_2$.
The only\footnote{when considering two paths essentially the same if they differ merely on a finite prefix}
progressing path $\pi_P$ on which $\ell_1$ is not followed by $\ell_2$ visits state
\plat{\scriptsize\lm 24A} infinitely often and always takes the transition going right.
This path witnesses that progress is not strong enough an assumption to ensure \hyperlink{ME}{\MEC[]},
whereas weak fairness is, provided all the transitions stemming from Instruction $\ell_2$ form a task.%
\footnote{Formally, $\mbox{\it Peterson} \models_B^{{\it WF}(\Tsk)} {\bf G}(\lni \Rightarrow {\bf F}\ec)$
when $\textsc{l}_2, \textsc{m}_2 \in \Tsk$, where $\textsc{l}_2$ (resp.\ $\textsc{m}_2$) contains
all transitions stemming from instruction $\ell_2$ (resp.\ $m_2$).
Namely,  $\pi_P$ fails to be weakly fair, for it has a suffix on which
task $\textsc{l}_2$ is perpetually enabled but never occurs.}
Whether justness is a strong enough assumption for \hyperlink{ME}{\MEC[]} depends solely on
the question whether $\pi_P$ is just.

To answer that question one can interpret the LTS of Figure~\ref{Peterson} as an LTSC, by 
investigating an appropriate concurrency relation $\aconc$ on the transitions. Whether two
transitions are concurrent ought to depend solely on the instructions $\ell_i$ and $m_j$ from
Figure~\ref{fig:peterson} that gave rise to these transitions, that is, whether $\ell_i \conc m_j$
for $i,j =1,\dots,8$. Assuming that the three variables $\rA$, $\rB$ and $\tu$ are stored in
independent stores or registers, the only pairs that may violate $\ell_i \conc m_j$ are
$\ell_3 \nconc m_3$, 
$\ell_2 \nconc m_4$, 
$\ell_3 \nconc m_4$, 
$\ell_7 \nconc m_4$,
$\ell_4 \nconc m_2$, 
$\ell_4 \nconc m_3$ and
$\ell_4 \nconc m_7$,
for these instruction compete for access to the same register.\footnote{The verdicts
$\ell_1 \conc m_j$,
$\ell_8 \conc m_j$,
$\ell_i \conc m_1$ and
$\ell_i \conc m_8$ were already used implicitly in the above derivation of \hyperlink{ME}{\MEF[J]}
and \hyperlink{ME}{\MEE[J]} from the pseudocode.} The only pair out of these 7 that affects the
justness of $\pi_P$ is $\ell_2 \nconc m_4$. Considering that in State \plat{\scriptsize\lm 24A} it
is possible to perform Instruction $m_4$, moving to State \plat{\scriptsize\lm 25A},
yet in State \plat{\scriptsize\lm 34A}, thus after performing $\ell_2$, it is no longer possible
to perform Instruction $m_4$, one surely has $m_4 \naconc \ell_2$, that is, Instruction $m_4$ is
affected by $\ell_2$---compare (\ref{closure}). On the other hand, one cannot derive from 
Figure~\ref{Peterson} alone whether $\ell_2 \naconc m_4$. In case one decides that $\ell_2 \aconc m_4$
then $\pi_P$ is not $B$-just by \df{Bjustness}; as nothing blocks the execution of Instruction
$\ell_2$, it must eventually occur. In this case \hyperlink{ME}{\MEC[J]} holds. However, in case one
decides that $\ell_2 \naconc m_4$, then $\pi_P$ is just and \hyperlink{ME}{\MEC[J]} does not hold.

Sections~\ref{sec:Peterson Petri}--\ref{sec:Peterson CCS}
examine whether $\ell_2 \aconc m_4$ holds in renderings of this protocol as a Petri net and in CCS\@.
In \Sec{atomicity} I will reflect on whether $\ell_2 \aconc m_4$ holds, and thus on whether
Peterson's protocol satisfies \hyperlink{ME}{\MEC[J]}, based on a classification as to
what could happen if two processes try to access the same register at the same time, one writing and one reading.

\section{Verifications of Starvation-Freedom Merely Assuming Progress}\label{sec:EG}

Above, I argued that in any formalisation $P$ of Peterson's algorithm that is consistent with the LTS of 
Figure~\ref{Peterson}---including any speed-independent formalisation---one has
$P\not\models_B^{\it Pr} {\bf G}(\lni \Rightarrow {\bf F}\ec)$, that is,
Requirement \hyperlink{ME}{\MEC[Pr]} does not hold, or, the formalisation $P$ does
not satisfy starvation-freedom when merely assuming progress. The path $\pi_P$ from \Sec{Peterson}
constitutes a counterexample. In Part~\ref{timeouts} of this paper I will bypass this verdict
by proposing a formalisation of Peterson's algorithm that is not consistent with the LTS of Figure~\ref{Peterson}.

In \cite{Walker89}, Peterson's algorithm was formalised in a way that is consistent with
Figure~\ref{Peterson}, yet starvation-freedom was proven, even by automatic means,
while assuming no more than progress.
The contradiction with the above is only apparent, because the starvation-freedom property obtained
by \cite{Walker89} can be stated as \plat{$P \not\models_B^{\it Pr} {\bf G}(\ell_2 \Rightarrow {\bf F}\ecA)$},
and symmetrically $P \not\models_B^{\it Pr} {\bf G}(m_2 \Rightarrow {\bf F}\ecB)$,
that is, once a process expresses its intention to enter its critical section, by executing
Instruction $\ell_2$ (or $m_2$), then it will surely reach its critical section.%
\footnote{In the literature \cite{SGG12,Raynal13} it is frequently claimed that once Process $\procA$ expresses
  its intention to enter its critical section, by executing Instruction $\ell_2$,
  Process $\procB$ will enter the critical section at most once before $\procA$ does (and the same
  with $\procA$ and $\procB$ reversed, of course).
  A counterexample is provided by the path that turns right as much as possible from the state
  \raisebox{1pt}[0pt][0pt]{\scriptsize\lm 35A}. Here Process $\procA$ has just executed $\ell_2$, but Process $\procB$
  enters the critical section twice before $\procA$ does. (By the definitions in \cite{Raynal13},
  in state \raisebox{1pt}[0pt][0pt]{\scriptsize\lm 35A} the two processes are \emph{competing}, and $\procA$ loses the
  competition twice.)
}
The\linebreak same can be said for the verification of starvation-freedom of Peterson's algorithm in
\cite{Valmari96}, and, in essence, for the verification of starvation-freedom of Dekker's algorithm in
\cite{EsparzaBruns96}.

It can be (and has been) debated which is the better formalisation of starvation-freedom.
Since the greatest hurdle in the protocol is Instruction $\ell_2$, it seems unfair to say that a
process is only deemed interested in entering the critical section when this hurdle is taken.%
\footnote{Suppose I promise all of you \$1000, if only you express interest in getting it, by
  filling in Form 316F. Then I implement this promise by making it impossible to fill in Form 316F.
  In that case you might argue against the claim that the Requirement
  $~~{\bf G}(\mbox{``has interest''} \Rightarrow {\bf F}\mbox{``receive \$1000''})~~$ can be verified.
  The argument would be that filling in Form 316F is an inadequate formalisation of having interest
  in getting this prize.}
Another tactic is to consider a form of Peterson's protocol in which Processes $\procA$ and
$\procB$ are merely interfaces that interact with the real clients that compete for access to the
critical section by synchronising on the actions
$\lnA$,
$\lnB$,
$\ecA$,
$\ecB$,
$\lcA$,
$\lcB$,
$\enA$ and
$\enB$.
In such a case, the real intent of Client $\procA$ to enter its critical section is expressed in a
message to Interface $\procA$ that occurs strictly before Interface $\procA$ executes Instruction $\ell_2$.

Such arguments are less necessary now that I have formalised \hyperlink{ME}{\MEF[]} as an
additional requirement. Suppose that one would redefine the action $\lnA$ that appears in the
antecedent of the starvation-freedom requirement \hyperlink{ME}{\MEC[]} as the occurrence of
instruction $\color{blue} \ell_2$ from Peterson's algorithm, thus turning \hyperlink{ME}{\MEC[]} for
Process $\procA$ into $P \models_B^{CC} {\bf G}({\color{blue}\ell_2} \Rightarrow {\bf F}\ecA)$,
then \hyperlink{ME}{\MEF[]} becomes
\(P \models_{B\setminus\{{\color{blue}\ell_2}\}}^{CC}
  {\bf F}{\color{blue}\ell_2} \wedge {\bf G}(\enA \Rightarrow {\bf F}{\color{blue}\ell_2})\).
Either way, it is required that $\ell_1$ will be followed by $\ell_2$; this
requirement is either part of \hyperlink{ME}{\MEF[]} or of \hyperlink{ME}{\MEC[]}.
In case it takes weak fairness to perform Instruction $\ell_2$, either \hyperlink{ME}{\MEC[]}
or \hyperlink{ME}{\MEF[]} will hold only for $CC={\it WF}$, and depending on one's modelling
preferences one can choose which one.

In Sections~\ref{sec:Peterson Petri} and~\ref{sec:Peterson CCS} I will show that in
formalisations of Peterson's algorithm as a Petri net or a CCS expression, the path $\pi_P$ from
\Sec{Peterson} is just, implying that it takes weak fairness to perform Instruction
$\ell_2$. This implies that those formulations can in terms of Figure~\ref{hierarchy} be situated
either at the coordinates $({\it WF},{\it Pr})$ or $({\it J},{\it WF})$, depending on one's
preferred modelling of intent to enter the critical section. Either way, such a rendering of
Peterson scores no better than the (encapsulated) gatekeeper.

\section{Modelling Peterson's Protocol as a Petri Net}\label{sec:Peterson Petri}

Figure~\ref{PetersonPN} shows a rendering of Peterson's protocol as a Petri net.
There is one place for each of the local states $\ell_1$--$\ell_8$ and $m_1$--$m_8$ of Processes
$\procA$ and $\procB$ and two for each of the Boolean variables $\rA$, $\rB$ and $\tu$.
There is one transition for each of the instructions $\ell_1$--$\ell_8$ and $m_1$--$m_8$,
except that $\ell_3$, $m_3$, $\ell_4$ and $m_4$ yield two transitions each.
For $\ell_3$ and $m_3$ this is to deal with each of the possible values of $\tu$ before the
assignment is executed; for $\ell_4$ the two transitions are for reading that  $\tu=A$, and for
reading that $\rB=\fa$.

\begin{figure}[ht]
\input{PetersonPN}
\centerline{\box\graph}
\caption{\it Petri net representation of Peterson's mutual exclusion algorithm}
\label{PetersonPN}
\end{figure}
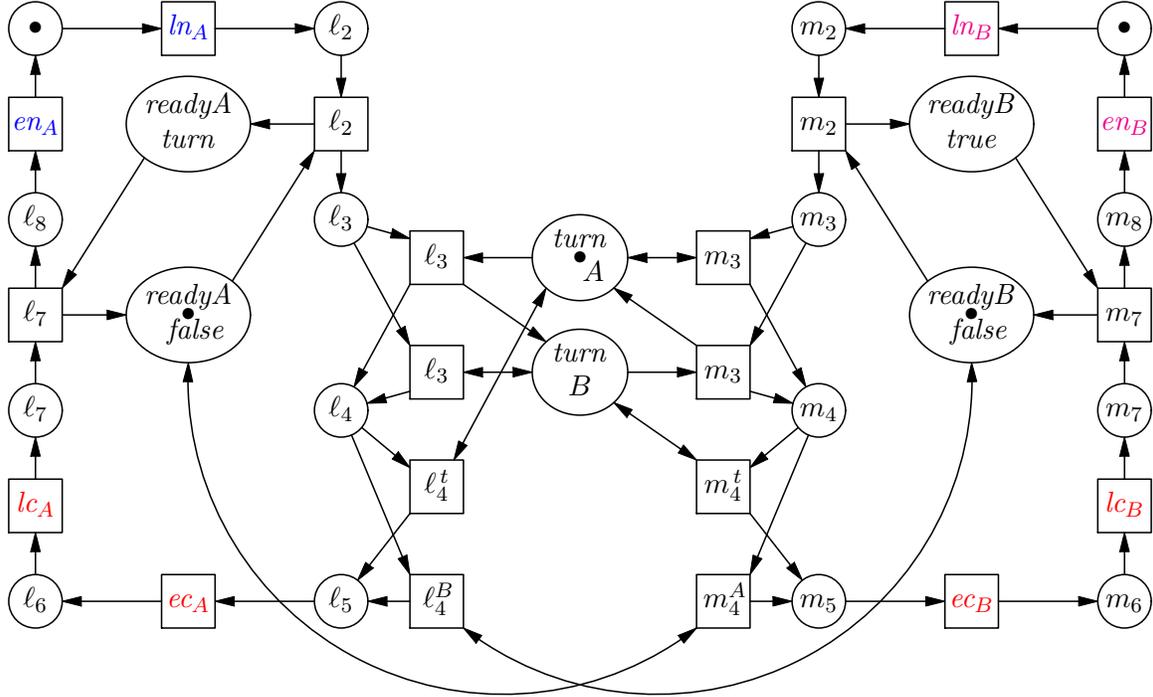

The transitions $\ell_2$ and $m_4^A$ are not concurrent, because they compete for the same token on
the place $\rA=\fa$. For this reason the run $\pi_P$, in which one token is stuck in place $\ell_2$
while the other four tokens keep moving around---with $m_4^A$ executed infinitely many times---is just.
Consequently, this Petri net does not satisfy requirement \hyperlink{ME}{\MEC[J]}.

\section{Modelling Peterson's Protocol in CCS}\label{sec:Peterson CCS}

In order to model Peterson's mutual exclusion protocol in CCS, I use the names
$\enA$,
$\enB$,
$\lnA$,
$\lnB$,
$\ecA$,
$\ecB$,
$\lcA$ and
$\lcB$
for Processes {\procA} and {\procB} entering and leaving their (non)critical section.
Following \cite{EPTCS255.2}, I describe a simple shared memory system in CCS, using the name
\plat{$\ass{x}{v}$} for the assignment of value $v$ to the variable $x$, and 
$\noti{x}{v}$ for noticing or notifying that the variable $x$ has the value $v$.
The action $\overline{\ass{x}{v}}$ communicates the assignment $x:=v$ to the shared memory,
whereas $\ass{x}{v}$ is the action of the shared memory of accepting this communication.
Likewise, $\overline{\noti{x}{v}}$ is a notification by the shared memory that $x$ equals $v$; it synchronises
with the complementary action $\noti{x}{v}$ of noticing that $x=v$.

The Processes {\procA} and {\procB} can be modelled~as
\[\begin{array}{@{}l@{~\stackrel{{\it def}}{=}~}l@{}}
\procA  & \lnA\mathbin.\overline{\ass{\rA}{\tr}}\mathbin.\overline{\ass{\tu}{B}}\mathbin.
  (\noti{\rB}{\fa}\ +\ \noti{\tu}{A})\mathbin.\ecA\mathbin.
   \lcA\mathbin.\overline{\ass{\rA}{\fa}}\mathbin.\enA\mathbin.\procA\ ,\\
\procB  & \lnB\mathbin.\overline{\ass{\rB}{\tr}}\mathbin.\overline{\ass{\tu}{A}}\mathbin.
  (\noti{\rA}{\fa}\ +\ \noti{\tu}{B})\mathbin.\ecB\mathbin.
   \lcB\mathbin.\overline{\ass{\rB}{\fa}}\mathbin.\enB\mathbin.\procB\ ,
\end{array}\]

\noindent where $(a+b).P$ is a shorthand for $a.P+b.P$.
This CCS rendering naturally captures the \textbf{await} statement, requiring Process {\procA} to wait
at Instruction $\ell_4$ until it can read that $\rB = \fa$ or $\tu = A$.
We use two agent identifiers for each Boolean variable $x$, one for each value:
\[\begin{array}{r@{\ }c@{\ }l}
  x^\tr &\stackrel{\it def}{=}& \ass{x}{\tr}\mathbin.x^\tr \ +\ \ass{x}{\fa}\mathbin.x^\fa \ +\ \overline{\noti{x}{\tr}}\mathbin.x^\tr\ ,\\
  x^\fa &\stackrel{\it def}{=}& \ass{x}{\tr}\mathbin.x^\tr \ +\  \ass{x}{\fa}\mathbin.x^\fa \ +\  \overline{\noti{x}{\fa}}\mathbin.x^\fa\ .\end{array}\]
Likewise we have, for instance,
\[\Tu{A}~\stackrel{{\it def}}{=}~\ass{\tu}{A}\mathbin.\Tu{A}\ +
\ass{\tu}{B}\mathbin.\Tu{B} \ +\ \overline{\noti{\tu}{A}}\mathbin.\Tu{A}.
\]
Peterson's mutual exclusion algorithm (PME) is the parallel composition of all these processes,
restricting all the communications
\[ (\procA \mid \procB \mid \RA{\fa} \mid \RB{\fa} \mid \Tu{A}) \backslash L\ ,\]
where $L$ is the set of all names except
$\enA$,
$\enB$,
$\lnA$,
$\lnB$,
$\ecA$,
$\ecB$,
$\lcA$ and
$\lcB$ \cite{EPTCS255.2}.

The LTS of the above CCS expression PME is exactly as displayed in Figure~\ref{Peterson}.
By the interpretation of CCS as an LTSC, defined in \Sec{LTSC}, one obtains $\ell_2 \naconc m_4$,
where I use the names $\ell_2$ and $m_4$ of the underlying instructions from
Figure~\ref{fig:peterson} to denote the two outgoing $\tau$-transitions from the state
\plat{\scriptsize\lm 24A}. In fact, this could have been concluded without studying the above CCS
rendering of Peterson's protocol, as \Sec{LTSC} remarks that in the LTSC of CCS the relation
$\aconc$ is symmetric, while \Sec{Peterson} concludes that $m_4 \naconc \ell_2$. Thus, the CCS
rendering of Peterson's algorithm does not satisfy the correctness criterion \hyperlink{ME}{\MEC[J]}.

\section{What Happens if Processes Try to Read and Write Simultaneously}\label{sec:atomicity}

A program instruction like $\ell_2$, $\ell_3$, $\ell_4$ or $\ell_7$ that reads or writes a value
$\tr$, $\fa$, $A$ or $B$ from or to a register $\rA$, $\rB$ or $\tu$ cannot be executed instantaneously,
and is thus assumed to occur during an interval of real time. Hence it may happen that Processes
$\procA$ and $\procB$ try to access the same register during overlapping periods of time.
In such a case is it common to assume that the register is \emph{safe}, meaning that
\begin{center}\begin{minipage}{5.3in}\vspace{3pt}\it
A read operation not concurrent with any write operation
returns the value written by the latest write operation,
provided the last two write operations did not overlap.\vspace{3pt}
\end{minipage}\end{center}
This assumption stems from \cite{Lam86b}, although overlapping writes were not considered there.
``No assumption is made about the value obtained by a read that overlaps a write, except that it must obtain one
  of the possible values of the register.'' \cite{Lam86b}\footnote{This is not really a restriction,
    for one can always follow a read action $r:=x$ of a variable $x$, where $r$ is a local register,
    by a default assignment $r:= v_0$ in case the read yields a value that is out of range.}\;
  In the same spirit, one may assume that two overlapping writes may put any of its possible values
  in the register, in the sense that subsequent reads will return that value.
I will assume safety in this sense of the Boolean registers $\rA$, $\rB$ and $\tu$.
In an architecture where safe registers are not available, and cannot be simulated, implementing a
correct mutual exclusion protocol appears to be hopeless.\footnote{This statement can be
  strengthened by considering its contrapositive: if one has a correct mutual exclusion protocol,
  safe registers can be simulated. Namely, when confronted with a weak memory, where it takes time
  for writes to propagate after the instruction has been encountered, one could put each read and
  write instruction in a critical section, together with an appropriate memory barrier or fence
  instruction, to ensure propagation of the write before any read occurs. This ought to yield safety.}

For the fate of Peterson's algorithm it matters what happens if one process wants to start writing
a register when another is busy reading it. There appear to be only three (or five) possibilities.
\begin{enumerate}
\item The register cannot handle a read and a write at the same time; as the read started first, the
  writing process will need to await the termination of the read action before the write can commence.
  \label{atomicity}
\item The register cannot handle a read and a write at the same time, but the write takes precedence
  and occurs when scheduled. This aborts the read action, which can restart after the write has terminated.
  \label{nonblocking read}
\item The read and write proceed as scheduled, thus overlapping in time.
  \label{overlap}
\end{enumerate}
A fourth possibility could be that reads and writes are instantaneous after all, so that overlap can
be avoided without postponing either. I deem this unrealistic and do not consider this option here.
A potential fifth possibility could be a variation of (\ref{nonblocking read}), in which the read merely is interrupted, and resumes
after the write is finished. In that case, as with option (\ref{overlap}), it seems reasonable to assume that the
read can return any value of the register.

In Dijkstra's original formulation of the mutual exclusion problem \cite{Dijk65}, possibility (\ref{atomicity})
above---\emph{atomicity}---was assumed---see the quote in \Sec{history}. Lamport, on the other hand,
assumes (\ref{overlap}) \cite{bakery}.
On his webpage \href{https://lamport.azurewebsites.net/pubs/pubs.html#bakery}{\small https://lamport.azurewebsites.net/pubs/pubs.html\#bakery}
Lamport takes the position that assuming atomicity ``cannot really be said to solve the mutual exclusion problem'',
as it assumes ``lower-level mutual exclusion''. As possibility (\ref{overlap}) adds the complication of
arbitrary register values returned by reads that overlap a write, he implicitly takes the position
that solving the mutual exclusion problem under assumption (\ref{overlap}) is more challenging; and this is
exactly what his bakery algorithm does.

Here I argue that atomicity is the more challenging assumption. The objection
that assuming reads and writes to be atomic amounts to assuming ``lower-level mutual exclusion'' is
based on the idea that securing the mutex property of a mutual exclusion protocol is the main
challenge. However, the real challenge is doing this in a starvation-free way, and this feature is
not inherited from the lower level.  By assuming atomicity one obtains $\ell_2 \naconc m_4$, that is,
transition $\ell_2$ is affected by $m_4$, and consequently Peterson's algorithm fails the
correctness requirement \hyperlink{ME}{\MEC[J]}.  In \Sec{impossible} below I will argue that this is not
merely a result of the way I choose to model things in this paper, but actual evidence of the
incorrectness of Peterson's algorithm, or any other mutual exclusion protocol for that matter,
provided one consistently assumes atomicity, as well as speed independence.

Assuming (\ref{nonblocking read}) instead yields $\ell_2 \aconc m_4$, that is, the write
$\ell_2$ is in no way affected by the read $m_4$. This means that nothing can prevent Process $\procA$
from executing $\ell_2$. This makes Peterson's algorithm correct, in the sense that it satisfies
\hyperlink{ME}{\MEC[J]}.

Assuming (\ref{overlap}) also yields $\ell_2 \aconc m_4$. Also this would make the algorithm correct,
provided that it is robust against the effects of overlapping reads and writes.
For Peterson's algorithm this is not the case; overlapping reads and writes can cause a violation
of the mutex property \hyperlink{ME}{\MEB}---see \Sec{Peterson overlap}.
However, various other mutual exclusion protocols, including the ones of \cite{bakery} and \cite{Ar11},
are robust against the effects of overlapping reads and writes and do satisfy
\hyperlink{ME}{\MEC[J]} when assuming  (\ref{overlap}).\footnote{Szyma\'nski's algorithm~\cite{Szy88}
was designed specifically for this robustness, but fails to achieve it \cite{LS21}.}

In CCS, and in Petri nets, $\aconc$ is symmetric, and one has $\ell \nconc m$ whenever $\ell$ and
$m$ are read or write instructions on the same register, at least when employing the concurrency
relation $\conc$ of \Sec{LTSC}. This amounts to assuming atomicity.

\section{Is Peterson's Protocol Resistant Against Overlapping Reads and Writes?}\label{sec:Peterson overlap}

Those mutual exclusion protocols that were designed to be robust under overlapping reads and writes,
avoid overlapping writes altogether, either by making sure that each variable can be written by only
one process (although it can be read by others) \cite{bakery,Szy88}, or by putting writes to the
same variable right before or after \cite{Ar11} the critical section, within the part of a process'
cycle that is made mutually exclusive.
Any protocol that doesn't take this precaution, including Peterson's, is regarded with suspicion by
those that make an effort to avoid ill effects due to reads overlapping with writes.
Nevertheless, until May 2021 no examples were known (to me at least) that overlapping reads and
writes actually cause any problem for Peterson's protocol. I personally believed that it was robust
under overlap, reasoning as follows [unpublished notes].
{\addtolength\leftmargini{-0.1in}
\begin{quote}\it 
Two overlapping write actions to the same register may produce any value of that register.
In Peterson's algorithm, the only register that can be written by both processes contains the
variable $\tu$. It is a Boolean register, whose values are $A$ and $B$.
The only write actions to this register are $\ell_3$ and $m_3$. When these overlap in time, any of
the register values, that is $A$ or $B$, may result.
However, $\ell_3$ tries to write the value $B$, and $m_3$ the value $A$. So if the result of a
simultaneous write is $A$, one can just as well assume that $\ell_3$ occurred before $m_3$, and if it
is $B$, that $m_3$ occurred before $\ell_3$. Thus the effects of overlapping writes are no
different than those of atomic writes, and hence harmless.

Peterson's algorithm has six cases of a read overlapping with a write, and thanks to symmetry it
suffices to study three of them. First consider the overlap of the write $\ell_2$ with the read $m_4$.
Here the overlapping read can yield any register value, that is, $\tr$ or $\fa$.
One should not ignore the possibility that Process $\procB$, while cycling around, performs
  multiple $m_4$-reads of $\rA$ that may return any sequence of $\tr$ and $\fa$ during a single
  write action $\ell_2$ of Process $\procA$.
However, any read by Process $\procB$ that returns $\rA=\tr$ does not help to pass the
{\bf await}-statement in $m_4$, and is equivalent to no read of $\rA$ being carried out.
So all reads that matter return $\rA=\fa$, and can just as well be thought of as occurring prior to $\ell_2$.
A similar argument applies to the overlap of the write $\ell_3$ with the read $m_4$, and of the
write $\ell_7$ with the read $m_4$; also here any resulting behaviour can already be generated without assuming an overlap.
\end{quote}}
\noindent
I leave it as a puzzle to the reader to find the fallacy in this argument.\footnote{On request of
  2 referees I divulge this fallacy at \url{http://theory.stanford.edu/~rvg/PMEfallacy.html}---not here
  to avoid spoilers. Alternatively, one can find out by comparing with the counterexample in~\cite{Spr21}.}
A run of Peterson's algorithm that violates the mutex property \hyperlink{ME}{\MEB}
was found in 2021 by means of the model checker mCRL2 by Myrthe Spronck \cite{Spr21}.
This involved implementing safe registers, as described  in \Sec{atomicity}, as mCRL2 processes.

\section{The Impossibility of Mutual Exclusion\texorpdfstring{\\}{} when Assuming Atomicity and Speed Independence}\label{sec:impossible}

In this section I argue that when assuming atomicity and speed independence,
Peterson's algorithm is not correct, in the sense that it fails the requirement of starvation-freedom.
The argument is that the the run corresponding with the path $\pi_P$ from \Sec{Peterson} can actually occur.
To see why, let me illustrate the form of speed independence needed for this argument by means of a
simple example in CCS\@.
\begin{example}{speed independence}
  Consider the process $(X|Y){\setminus}\{c\}$ with $X \stackrel{{\it def}}{=} a.{\bf 0} + \bar{c}.X$ and
  $Y \stackrel{{\it def}}{=} c.d.e.Y$, where none of the actions $a,d,e$ is blocked by the
  environment, that is, the environment is continuously eager to partake in these actions.
  The question is whether action $a$ is guaranteed to occur.

  Here one may argue that when Process $X$ reaches the state $a.{\bf 0} + \bar{c}.X$
  at a time when its environment, which is the Process $Y$, is not yet ready to engage in the
  synchronisation on $c$, it will proceed by executing $a$.
  If, on the other hand, both options $a$ and $\bar c$ are available, it cannot be excluded that
  $\bar c$ is
  chosen, as no priority mechanism is at work here.

  Now after execution of $\bar{c}|c$, Process $X$ again faces a choice between $a$ and $\bar c$, but
  Process $Y$ first has to execute the actions $d$ and $e$. During the time that $Y$ is busy with
  $d$ and $e$, for Process $X$ it feels like $\bar c$ is blocked, and it will do $a$.

  A strict interpretation of speed independence, which I employ here, says that actions $d$ and
  $e$ may be executed so fast that from the perspective of Process $X$ one can just as well assume
  that no actions $d$ and $e$ were scheduled at all.\vspace{1pt} Thus the answer to the above question will not
  change upon replacing Process $Y$ by \plat{$Y' \stackrel{{\it def}}{=} c.Y'$}.   However, for the
  process $(X|Y'){\setminus}\{c\}$ there is really no reason to assume that $a$ will ever occur.

  In CCS and related process algebras this form of speed independence is built in: one sees $d$ and
  $e$ as $\tau$-transitions, as for the purpose of answering the above question they can be regarded
  as unobservable, and then applies the law $a.\tau.P = a.P$ \cite{Mi90ccs}.
\end{example}

A run of Peterson's protocol, or any implementation thereof in a setting where atomicity and the above
form of speed independence may be assumed, can visit the state \plat{\scriptsize\lm 24A}
in which Process $\procA$ wants to write on register $\rA$, and that register has a choice between
being written or being read first, by Process $\procB$. One cannot exclude that the read action wins this
race, which allows Process $\procB$ to enter the critical section. During the execution of $m_4$,
reading $\rA$, Process $\procA$ has to wait. Afterwards, Process $\procB$ might execute actions
$m_5$--$m_3$ so fast that from the perspective of Process $\procA$ no time elapses at all. This
brings us again in State \plat{\scriptsize\lm 24A} where the same race between $\ell_2$ and $m_4$
occurs. Again it could be won by $m_4$. This behaviour can continue indefinitely.

The above argument stems from \cite{vG18c}, where it was made not just for Peterson's protocol,
but for all conceivable mutual exclusion algorithms. These include 
Lamport's bakery algorithm \cite{bakery}, Szyma\'nski's algorithm \cite{Szy88},
and the round-robin scheduler.
To obtain a correct mutual exclusion algorithm, one has to either employ register hardware in which
the assumption of atomicity---possibility (\ref{atomicity}) in \Sec{atomicity}---is not valid, or make the protocol
speed-dependent. The first option was already explored in \Sec{atomicity}; as mentioned there,
using hardware that works according to possibilities (\ref{nonblocking read}) or (\ref{overlap}) solves the problem.
The second option will be explored in Part~\ref{timeouts} of this paper.

\section{Variations of Petri nets and CCS with Non-blocking Reading}\label{sec:Peterson CCSS}

To escape from the failure of requirement \hyperlink{ME}{\MEC[J]} for Peterson's
protocol due to the assumption (\ref{atomicity}) of atomicity as well as speed independence, one can instead
assume possibility (\ref{nonblocking read}) from \Sec{atomicity} while keeping speed independence.
I refer to such an option as \emph{non-blocking reading} \cite{CDV09}, as a write cannot be postponed by a
read action on the same register. This yields $\ell_2 \aconc m_4$, thereby saving \hyperlink{ME}{\MEC[J]}.
Here I review how this can be modelled in variations of Petri nets and CCS\@.

\subsection{Read Arcs}\label{sec:read arcs}

A Petri net with \emph{read arcs} is a tuple $N=(S,T,F,R,M_0,\ell)$ as in \df{net}, but enriched with
an object $R\!: S\times T \rightarrow \IN$, such that $F(s,t)\mathbin>0 \Rightarrow R(s,t)\mathbin=0$.
An element $(s,t)$ of the multiset $R$ is called a \emph{read arc}.
Read arcs are drawn as lines without\linebreak[2] arrow heads.
For $t\mathbin\in T$, the multiset $\hat t\!: S \rightarrow \IN$ is given by $\hat t(s)\mathbin= R(s,t)$ for all $s\mathbin\in S$.
Transition $t$ is \emph{enabled} under $M$ iff $\precond{t}+ \hat{t} \leq M$.
In that case \plat{$M \goesto{t}_N M'$}, where $M' = (M - \precond{t}) + \postcond{t}$.

Thus, for a transition to be enabled, there need to be enough tokens at the other end of each of its read arcs.
However, these tokens are not consumed by the firing. Clearly, the transition relation $\goesto{t}_N$
between the markings of a net is unaffected when replacing each read arc $(s,t)$ by a loop between
$s$ and $t$; that is, by dropping $R$ and using the flow relation $F_R := F + R + R^{-1}$.
This does not apply to the concurrency relation between transitions.

The definition of a structural conflict net can be extended to Petri nets with read arcs by requiring,
for all $t, u\in T$ and all reachable markings $M$, that
\begin{center}
  if $(\precond{t}+\hat{t}\,) \cap \precond{u} \neq \emptyset$ then
  $\precond{t} + \hat{t} + \precond{u} \not\leq M$ or $\hat u \not\leq M$.
\end{center}
For such nets, one defines $t \aconc u$ iff $(\precond{t}+\hat{t}\,) \cap \precond{u} = \emptyset$.
This says that a transition $u$ affects a transition $t$ iff $u$ consumes a token that is needed to
enable $t$. The condition for a structural conflict net guarantees exactly that if $t \naconc u$ and
$u$ is enabled under a reachable marking $M$, then $t$ is not enabled under the marking $M-\precond{u}$.

As pointed out by Vogler in \cite{Vogler02}, the addition of read arcs makes Petri nets sufficiently
expressive to model mutual exclusion. When changing the loops
\input{smallPN}\raisebox{10pt}[0pt][0pt]{\box\graph}
and
\input{smallPN2}\raisebox{10pt}[0pt][0pt]{\box\graph}
in Figure~\ref{PetersonPN} into read arcs, one obtains $\ell_2 \aconc m_4^{\scriptscriptstyle B\!}$ and
$m_2 \aconc \ell_4^{\scriptscriptstyle A\!}$.
So the resulting net satisfies \hyperlink{ME1}{\MEA}--\hyperlink{ME}{\MEF[]} with $CC=J$, and correctly solves
the mutual exclusion problem. Two different solutions as Petri nets with read arcs are given in
\cite[Figures 8 and 9]{Vogler02}, the one in Figure~9 being a round-robin scheduler.\vspace{-1pt}

\subsection{Broadcast Communication}\label{sec:broadcast}

A process algebraic solution was presented in \cite[Section~5]{GH15a}, using an extension ABC
of CCS with broadcast communication.
The most obvious distinction between broadcast communication and CCS-style handshaking
communication is that the former allows multiple recipients of a message and the latter exactly one.
This feature of broadcast communication was not exploited in the solution of \cite{GH15a}.
A more subtle feature of broadcast communication is that the transmission of a broadcast occurs
regardless of whether anyone is listening. Thus a broadcast can be used to model a write that cannot
be blocked or postponed because the receiving register is busy being read. This yields an asymmetric
concurrency relation, in which a broadcast transition is not affected by a competing transition from
a receiver of the broadcast, whereas the competing transition \emph{is} affected by the broadcast.

Nevertheless, the semantics of ABC says that the broadcast \emph{will} be received by all processes that
are in a state with an outgoing receive transition. This allows one to make receipt of a broadcast
reliable, by giving the receiving process an outgoing receive action in each of its states. This
feature of the semantics of ABC, which is essential for modelling mutual exclusion, is somewhat
debatable, as one could argue that a process that is engaged in its own broadcast transmission
through a transition between two states that each have an outgoing receive transition, is temporary
too busy to hear an incoming message.

The solution proposed in \cite{GH15a} is not a mutual exclusion protocol, but a fair
scheduler, which can be converted into a mutual exclusion protocol in the manner of \Sec{FS}.
In fact, it is a variant of the encapsulated gatekeeper, where broadcast communication is used to
merely assume justness for all requirements. In the same spirit one can model Peterson's protocol in
ABC in such a way that \hyperlink{ME}{\MEC[J]} is satisfied. It suffers to
interpret $\overline{\ass{x}{v}}$ and $\ass{x}{v}$ as broadcast transmit and receive actions.
The LTSC semantics of ABC \cite{vG19} then yields $\ell_2 \aconc m_4$.

\subsection{Signals}\label{sec:signals}

A different process algebraic solution, arguably less debatable, was proposed in \cite{CDV09,EPTCS255.2}.
In \cite{CDV09}, processes $P$ are equipped with the possibility to perform actions that do not change
their state, and that, in synchronisation with another parallel process $Q$, describe information on
the state of $P$ that is read by $Q$ in a non-blocking way. In \cite{EPTCS255.2}, following \cite{Be88b},
such actions are called \emph{signals}.
The communication between, say, a traffic light, emitting the signal \emph{red},
and a car, coming to a halt, is binary, just like handshaking communication in CCS\@. The difference
is that the concurrency relation between transitions again becomes asymmetric, because the car is
affected by the traffic light, but the traffic light is not affected by the car. A car stopping for
a red light in no way blocks or postpones the action of the traffic light of turning green.

In \cite{Be88b,EPTCS255.2} the emission of a signal is modelled as a predicate on states, whereas
the receipt of such an emission is modelled as a transition. In \cite{CDV09,Bou18} the emission of a
signal is modelled as a transition instead. An advantage of the former approach is that it stresses
the semantic difference between signal emissions and handshaking actions, and emphasises that signal
emissions cannot possibly cause a state-change. An advantage of the latter approach is that
communication via signals can be treated in the same way as a CCS handshaking communication, thereby
simplifying the process algebra.  Technically, the two approaches are equivalent.

{\makeatletter
\let\par\@@par
\par\parshape0
\everypar{}\begin{wrapfigure}[6]{r}{0.5\textwidth}
 \vspace{-2.5ex}%
  \input{Boolean}%
  \centerline{\raisebox{1ex}{\box\graph}}%
   \vspace{-15ex}
  \end{wrapfigure}
In CCS with signals \cite{EPTCS255.2}, modelling signal emissions as transitions \cite{Bou18,vG19},
a Boolean variable $x$, such as $\rA$ in Peterson's protocol, has the exact same LTS as in the CCS
model of \Sec{Peterson CCS}:
However, this time the notifications $\overline{\noti{x}{v}}$ are signals emissions rather than
handshaking actions. The definition of the concurrency relation $\aconc$ on CCS transitions from
\Sec{LTSC} is in \cite{vG19} extended to CCS with signals in such a way that with the above way of
modelling variables, and the same processes $\procA$ and $\procB$ as in \Sec{Peterson CCS},
one obtains $\ell_2 \aconc m_4$. Here the action $m_4$ of reading the register is like a car reading a
traffic light, and does not inhibit the write action $\ell_2$ on the same register.
\par}

This shows that Peterson's algorithm can be correctly modelled in CCS with signals.
Earlier, Dekker's algorithm was correctly modelled in the process algebra PAFAS with non-blocking
reading \cite{CDV09}, and same was done for Peterson's algorithm in \cite{EPTCS54.4}.
The latter paper also points out that non-blocking reading is not strong enough an assumption to
obtain starvation-freedom, or even deadlock-freedom as defined in \Sec{history}, for Knuth's
algorithm~\cite{Knuth66}; this requires a fairness assumption.

\subsection{Modelling Non-blocking Reading in CCS}\label{sec:Bouwman}

The rendering in \cite{EPTCS255.2} of Peterson's protocol employs an extension of CCS with signals
that arguably strictly increases the expressiveness of the language. Namely in CCS with signals one obtains
an asymmetric concurrency relation $\aconc$, which turned out to be essential for the satisfactory
modelling of Peterson; restricted to proper CCS this relation is symmetric.  Bouwman
\cite{Bou18} proposes the same modelling of Peterson's protocol, but entirely within the confines of
the existing language CCS, namely by simply declaring some of the action names of CCS to be signals.
As it is essential that the emission of a signal never causes a state-change, care has to be taken
to only use CCS-expressions in which the actions that are chosen to be seen as signal emissions
occur in self-loops only. When doing this, creating a satisfactory rendering of Peterson's
protocol within CCS is unproblematic \cite{Bou18}. Neither \cite{EPTCS255.2} nor \cite{Bou18}
mention the concurrency relation $\aconc$ at all, and use a coinductive definition of justness
instead. However, as shown in \cite{vG19}, the concept of justness common to \cite{EPTCS255.2} and \cite{Bou18} can
equivalently be obtained as in \Sec{completeness criteria} from an asymmetric concurrency relation $\aconc$.
Thus, by declaring certain CCS actions to be signals, Bouwman effectively changes the concurrency
relation between CCS transitions labelled with those actions.

In \cite{GH15b} it was stated that fair schedulers and mutual exclusion protocols cannot be rendered
correctly in CCS without imposing a fairness assumption. In making that statement, the (symmetric)
concurrency relation $\conc$ between CCS transitions defined in \Sec{LTSC}---or equivalently, the
resulting notion of justness---was seen as an integral part of the semantics of CCS\@.
In the early days of CCS, a lot of work has been done in formalising
notions of concurrency for CCS and related process algebras \cite{NPW81,GM84,BC87,Wi87a,GV87,DDM87,Ol87}.
All that work is consistent with the concurrency relation $\conc$ defined in \Sec{LTSC}.
Changing the concurrency relation, as implicitly proposed by Bouwman, alters the language CCS as
seen from the perspective of \cite{NPW81,GM84,BC87,Wi87a,GV87,DDM87,Ol87}. However, it is entirely
consistent with the interleaving semantics of CCS given by Milner \cite{Mi90ccs}.

\subsection{Modelling and Verification of Peterson's Algorithm with mCRL2}\label{sec:mCRL2}

ACP \cite{BW90,Fok00} and mCRL2 \cite{GM14} are CCS-like process algebras that fall under the scope of
the impossibility result of \cite{GH15b}. That is, when defining a concurrency relation on the
transitions of ACP or mCRL2 processes in the traditional way, consistent with the viewpoint of
\cite{NPW81,GM84,BC87,Wi87a,GV87,DDM87,Ol87}, and defining justness as in \Sec{completeness criteria}
in terms of this concurrency relation, Peterson's algorithm cannot be correctly rendered in ACP or
mCRL2 when merely assuming justness, i.e., without resorting to a fairness assumption.
Nevertheless, by the argument of \cite{Bou18}, Peterson's algorithm \emph{can} be correctly rendered
in these formalisms under the assumption of justness when using an alternative concurrency
relation, one obtained by treating certain actions as signals.

Using this insight, Bouwman, Luttik and Willemse \cite{BLW20} render Peterson's algorithm in an
instance of ACP or mCRL2 that uses much more informative actions. This could be done without
adapting those languages in any way, simply by choosing appropriate actions.  Each action labelling
a transition contains additional information on the components of the system from which this
transition stems. This information is preserved under synchronisation, when a transition of a
parallel composition is built from transitions of each of the components.  The latter requires the
general communication format of ACP or mCRL2, as in CCS synchronisation merely result in
$\tau$-transitions. A price paid for this approach is that the resulting LTSs have much fewer
$\tau$-transitions, so that there are fewer opportunities for state-space reduction by abstraction
from internal activity.

In \Sec{translating} I showed how $B$-justness, and thereby the correctness criteria for mutual
exclusion protocols, can be expressed in standard LTL enriched with a number of atomic propositions,
such as ${\it en}^t$ and $\sharp t$. Bouwman, Luttik and Willemse \cite{BLW20} show not only that
the same can be done in the modal $\mu$-calculus, but also that all the necessary atomic
propositions can be expressed in terms of the carefully chosen actions that are used in modelling
the protocol. This made it possible to model the protocol in such  way that all correctness
requirements can be checked with the existing mCRL2 toolset \cite{BGKLNVWWW19}.

\addtocontents{toc}{\protect\vspace{-4pt}}
\part{A Speed-dependent Rendering of Peterson's Protocol}\label{timeouts}

In Part~\ref{impossibility results} I showed that when assuming atomicity and speed independence,
Peterson's algorithm does not have the correctness property \hyperlink{ME}{\MEC[J]}---and
in \cite{vG18c} I argued that the same can be said for any other mutual exclusion protocol.
That is, it satisfies starvation-freedom only under a weak fairness assumption, which in my opinion
is not warranted --- if it was, even the encapsulated gatekeeper of \Sec{interface} would be an
acceptable mutual exclusion protocol.

Thus, to obtain a correct version of Peterson's algorithm (or mutual exclusion in general)
one has to either drop the assumption of atomicity, or speed-independence.
The former possibility has been elaborated in Part~\ref{impossibility results};
\Sec{Peterson} showed then when assuming non-blocking reading (option (2) from \Sec{atomicity},
resulting in the verdict $\ell_2 \aconc m_4$) instead of atomicity, Peterson's algorithm is entirely correct.
\Sec{Peterson CCSS} recalled how to model this with process algebra or Petri nets.

The latter possibility will be elaborated here. I present a speed-dependent incarnation of
Peterson's protocol that satisfies all correctness requirements, even under the assumption of atomicity.
Moreover, in the model of \Sec{Peterson timeouts}, requirement \hyperlink{ME}{\MEC[J]} can be strengthened into
\hyperlink{ME}{\MEC[Pr]}, thereby attaining the best quality criteria in the
hierarchy of Figure~\ref{hierarchy}. As pointed out in Part~\ref{impossibility results}, this is not
possible when keeping speed independence, even when dropping atomicity.

The idea is extremely simple. As explained in \Sec{Peterson}, all that is needed to obtain
starvation-freedom is the certainty that when Process~$\procA$ reaches the state $\ell_2$, it
will in fact execute the instruction $\ell_2$. The only thing that can stop
Process~$\procA$ in state $\ell_2$ from executing $\ell_2$, is the register $\rA$
being too busy being read by Process~$\procB$ to find time for being written\linebreak[3] by Process~$\procA$.
Now assume that there exists an amount $t_0$ of time, such that, if for a period of at least $t_0$,
Process~$\procA$ is in state $\ell_2$ and the register $\rA$ is available, in the sense that it
is not being read by Process~$\procB$, then in that period the write action $\ell_2$ will commence.
Now further assume that Process~$\procB$ will spend a period of at least $t_0$ in its critical or in its
noncritical section. Then Process~$\procA$ will have enough time to perform the action $\rA:=\tr$,
and starvation-freedom is ensured. This is the speed-dependent version of Peterson's protocol I propose.
In fact I cannot exclude that mutual exclusion protocols work in practice exactly because timing
constraints such as sketched above are always met.

The above solution is sufficiently clear not to need mathematical proof.
Nevertheless I proceed with an implementation of the above idea in process algebra.
The goal of this is mostly to see which process algebra we need to formalise time-dependent
reasoning such as performed above. Naturally a timed process algebra that associates real numbers to
various passages of time would be entirely equipped for this task. However, I will show that the
idea can already be formalised in a realm of untimed process algebra, in the sense that the
progress of time is not quantified.

\section{CCS with Time-outs}\label{sec:CCSt}

Following \cite{vG21,vG20b}, my process algebra will be CCS$_\rt$, which is CCS, as presented in \Sec{CCS}, but with
$\alpha$ ranging over $Act := \A \dcup \bar\A \dcup \{\tau,\rt\}$, with $\rt$ a fresh \emph{time-out action}.
Relabellings $f$ extend to this extended set of actions $Act$ by $f(\rt):=\rt$.
The interpretation of this language as an LTSC proceeds exactly as in \Sec{LTSC}.

All actions $\alpha \in Act$ are assumed to occur instantaneously.
The time-out action $\rt$ models the end of a time-consuming activity from which we abstract.
When a system arrives in a state $P$, and at that time $X$ is the set of actions allowed
(= not blocked) by the environment, there are two possibilities.
If $P$ has an outgoing transition $u$ with $\ell(u) \in X\cup \{\tau\}$,
the system immediately takes one of the outgoing transitions $u$ with $\ell(u) \in X\cup \{\tau\}$,
without spending any time in state $P$. The choice between these transitions is entirely nondeterministic.
The system cannot immediately take a transition $u$ with $\ell(u)\in A{\setminus}X$, because the
action $\ell(u)$ is blocked by the environment. Neither can it immediately take a transition $u$
with $\ell(u)=\rt$, because such transitions model the end of an activity with a finite but positive
duration that started when reaching state $P$.

In case $P$ has no outgoing transition  $u$ with $\ell(u) \in X\cup \{\tau\}$,
the system idles in state $P$ for a positive amount of time. This idling can end in two possible ways.
Either one of the time-out transitions  $P \goesto\rt Q$ occurs, or the environment spontaneously
changes the set of actions it allows into a different set $Y$ with the property that
$P \goesto a Q$ for some $a \in Y$. In the latter case a transition $u$ with $\source(u)=P$
and $\ell(u)=a \in Y$ occurs. The choice between the various ways to end a period of idling is
entirely nondeterministic.
It is possible to stay forever in state $P$ only if there are no outgoing time-out transitions.%
\footnote{The environment in which a CCS$_\rt$ process $P$ runs can be anthropomorphised as a user behind a
  switchboard who can toggle each visible action as ``blocked'' or ``allowed'' \cite[Section~5]{vG01}.
  The user can toggle switches either as an immediate response on a visible action performed by $P$,
  or at arbitrary points in time. Alternatively, such an environment can be seen as a CCS$_\rt$
  process $E$, that synchronously runs in parallel with $P$, yielding the environment/system
  composition $(E|P){\setminus}A$. As an example, take $P=\rt.c.{\bf 0}+b.{\bf 0}$ and
  $E=\rt.\bar{b}.{\bf 0}+ \bar{c}.{\bf 0}$. This environment first allows only $c$ to occur, but
  after some time allows $b$ (while blocking $c$).
  The composition $(E|P){\setminus}A$ starts by idling, and then performs a $\tau$-transition, that
  from the perspective of $P$ is either $c$ or $b$. Which one depends on which of the two
  time-outs, from $P$ or $E$, occurs first.}

A fundamental law describing the interaction between $\tau$- and $\rt$-transitions, motivated by the
above, is $\tau.P + \rt.Q = \tau.P$. It says that when faced with a choice between a $\tau$- and a
$\rt$-transition, a system will never take the $\rt$-transition. I could have devised an operational
semantics of CCS$_\rt$, featuring negative premises, that suppresses the generation of transitions
$R \goto[C]\rt Q$ when there is a transition $R \goto[D]\tau P$. However, following \cite{vG21,vG20b}, I take a different,
and simpler, approach. The operational semantics of CCS$_\rt$ is exactly like the one of CCS, and
generates such spurious transitions $R \goto[C]\rt Q$; instead, its semantics assures that these
transitions are never taken. In \cite{vG20b} a branching time semantics is proposed, and in
\cite{vG21} a linear-time semantics---the closest approximation of partial trace semantics \cite{vG01}
that yields a congruence for the operators of CCS$_\rt$. Both these semantics satisfy $\tau.P + \rt.Q = \tau.P$.

Here I achieve the same by calling a path \emph{potentially complete} when it features no
transitions $R \goto[C]\rt Q$ when there also exists a transition $R \goto[D]\tau P$.
A completeness criterion now should set apart a subset of the potentially complete paths as being complete.
So all paths containing spurious transitions $R \goto[C]\rt Q$ count as incomplete, and hence do not
contribute to the evaluation of judgements in reactive temporal logic.
In depictions of LTSs for fragments of CCS$_\rt$ I will display the spurious transitions dotted, to
emphasise that they cannot be taken.

A transition $R \goto[C]\rt Q$ also cannot be taken when there is an alternative $R \goto[D] a P$,
with $a$ an action that surely will not be blocked by the environment when the system is in state $R$.
Thus, whether or not a transition is spurious depends on the mood of the environment at the time this
transition is enabled. This dependency is encoded in the semantic equivalences of \cite{vG21} and
\cite{vG20b}. Given this, it was no extra effort to simultaneously inhibit the selection of
transitions that are spurious in any environment.

\section{Spurious Transitions and Completeness Criteria for LTSs with Time-outs}\label{sec:spurious}

In LTSs with $\rt$-transitions, it makes sense to allow judgements $P \models^{CC}_{B,E} \phi$
with $B \subseteq E \subseteq A$, where $A$ is the set of all actions except $\tau$ and $\rt$.
Here $B$ is the set of actions that can be permanently blocked by the environment, and $E$ the ones
that can be blocked for finite periods of time.
Thus, the annotations $B$ and $E$ rule out those environments in which an action from
$A{\setminus}B$ is blocked permanently, or an action from $A{\setminus}E$ is blocked temporarily.
My interest is in the cases $CC={\it Pr}$ and $CC= J$.

\begin{definition}{spurious}
A transition $u$ is \emph{$E$-spurious} if $\ell(u)=\rt$ and there exists a transition $v\in \Tr$
with $\source(v)=\source(u)$ and $\ell(v)\in (A{\setminus}E) \cup \{\tau\}$.  It is \emph{spurious}
iff it is $A$-spurious.
\end{definition}
\noindent
Note that $u$ is spurious iff it is $E$-spurious for all $E$. This is the case iff it is a
$\rt$-transition sharing its source state with a $\tau$-transition. As actions from $A{\setminus}E$ cannot be
blocked by the environment, not even temporarily, $E$-spurious transitions cannot be taken.

\begin{example}{spurious}
Consider the variant of the vending machine from \makebox[91pt]{}\linebreak\vspace{-13.2pt}
\makeatletter
\let\par\@@par
\par\parshape0
\begin{wrapfigure}[7]{r}{0.2\textwidth}
 \vspace{-5ex}
 \input{pretzelOff}
 \centerline{\raisebox{1ex}{\box\graph}}
\end{wrapfigure}
\noindent
\Sec{motivation} that turns itself off after some period of inactivity. This is modelled by the two
time-out transitions on the right. The machine also has an {\it on} transition, to be used by its
operator to start it up. The dashed\linebreak[4] $\rt$-transition is $\{c,{\it on}\}$-spurious: it cannot be taken in an
environment where a user of the machine never blocks the production of a pretzel.
\end{example}

\begin{definition}{BEprogresssing}
A path $\pi$ is \emph{potentially $E$-complete} if it contains no $E$-spurious transitions.
It is \emph{$(B,E)$-progressing} if it (a) is potentially $E$-complete, and
\hyperlink{Bprogressing}{(b)} is either infinite or ends in a state of which all outgoing
transitions have a label from $B$.
It is \emph{$(B,E)$-just} if (a) it is potentially $E$-complete,
and \hyperlink{Bjust}{(b)} for each $u\in\Tr$ with $\ell(u)\notin B$ and whose source state $s := \source(u)$
occurs in $\pi$, any suffix of $\pi$ starting at $s$ contains a transition $v$ with $u \naconc v$.

\hypertarget{quinary}{The judgement $s \models^{\it Pr}_{B,E} \phi$ (resp.\ $s \models^{\it J}_{B,E} \phi$) holds
if $\pi\models\phi$ holds for all $(B,E)$-progressing (resp.\ $(B,E)$-just) paths $\pi$ starting in $s$. }
\end{definition}
\noindent
For a finite path to be complete, its last state may have outgoing transitions with labels from $B$ only,
for a run comes to an end only when all subsequent activity is permanently blocked by the environment.
In the absence of $\rt$-transitions, judgements $s \models^{\it CC}_{B,E} \phi$ are independent of $E$,
and agree with the ones defined in Part~\ref{RTL} of this paper.

\begin{example}{BEprogresssing}
  In \ex{spurious} one has ${\aconc} = \emptyset$, so $(B,E)$-just is the same as $(B,E)$-progressing.
  Let $B=\{c,{\it on}\}$. Each $(B,B)$-just path contains as many $p$- as $c$-transitions (possibly $\infty$).
  However, a $(B,A)$-just path may contain strictly more $c$-transitions.
\end{example}


In the context of the present paper, when describing properties for a given or desired process $P$,
I see no reason to combine judgements $P \models^{CC}_{B,E} \phi$ with different values of $E$.
This suggests writing \plat{$P \models^{CC}_{B,E} \phi$ as $(P,E) \models^{CC}_{B} \phi$}.
This way, the quality criteria of Sections~\ref{sec:formalising FS} and~\ref{sec:formalising mutex}
can remain unchanged, and apply to systems $(P,E)$. Here $P$ is a hypothetical fair scheduler or
mutual exclusion protocol, and $E$ the set of its actions that can be temporarily blocked by the environment.

To gauge the influence of the environment on the visible actions $\en$, $\lni$, $\ec$ and $\lc$ of a
mutual exclusion protocol, one can see the processes $i$ that compete for the critical section as
clients that communicate with the protocol through synchronisation on these actions.
As explained in \Sec{formalising mutex}, the actions $\lni$ belong in $B$ (except when formulating
requirement \hyperlink{ME}{\MEF[]}) because $\lni$ is permanently blocked in case Client $i$
chooses not to leave its noncritical section again. The actions $\lc$ belong in $E$, but not in $B$,
because the client may need some time before leaving its critical section, but is assumed to do this eventually.
As mentioned in \Sec{formalising mutex}, the actions $\ec$ and $\en$ do not belong in $B$, for we
assume the client to eventually enter the (non)critical section when allowed by the protocol.
There is a choice between putting these actions in $E$ or not. Putting them in $E$ models that the
client may delay a while before entering the (non)critical section when allowed, whereas putting them in
$A{\setminus}E$ models that when the protocol for Client $i$ is in its \emph{entry} or \emph{exit} section,
the actual client will patiently wait until granted access to the (non)critical section, and take
advantage of this opportunity as soon as it arises. Taking $E=E_l := \{\lni,\lc \mid i= 1,\dots,N\}$
appears most natural, but taking $E=A=\{\lni,\ec,\lc,\en \mid i= 1,\dots,N\}$ is a reasonable alternative.
The latter leads to stronger judgements, in the sense that when a protocol $P$ is correct when
taking $E:=A$, it is surely correct when taking $E:=E_l$. To model both options at the same time,
in Figure~\ref{PetersonSD} I will draw $E_l$-spurious transitions dashed. Those transitions cannot
be taken when choosing $E:=E_l$, but they can be taken when choosing $E:=A$.

\begin{figure}[ht]
{\tiny
\input{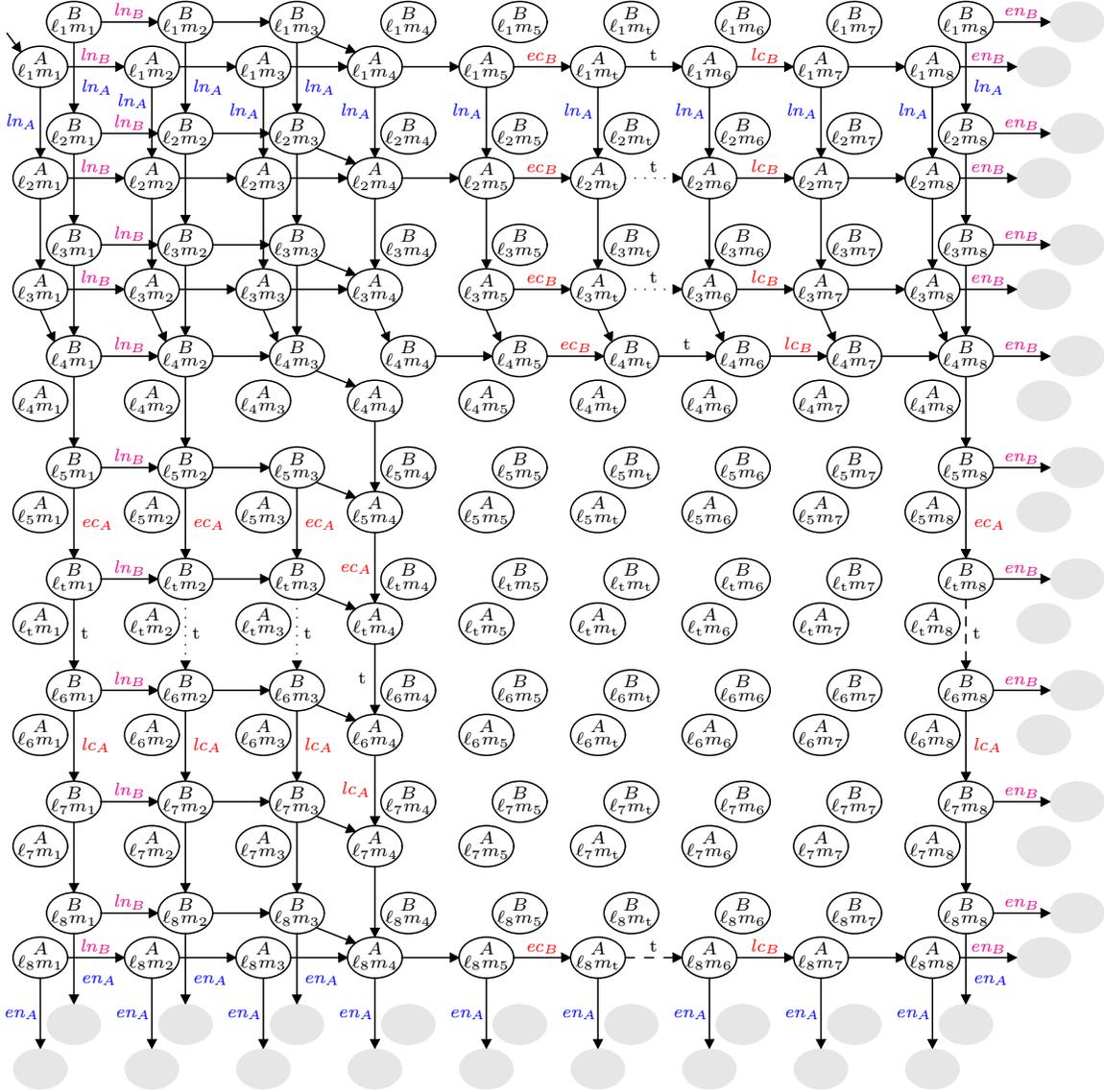}
\centerline{\box\graph}
}
\caption{\it Speed dependent LTS of Peterson's mutual exclusion algorithm}
\label{PetersonSD}
\end{figure}

\section{Modelling Peterson's Protocol in CCS with Timeouts}\label{sec:Peterson timeouts}

My model of Peterson's algorithm in CCS$_\rt$ differs from the one from \Sec{Peterson CCS} in only
one way: a $\rt$-action is inserted  between $\ec$ and $\lc$ for $i=\procA,\procB$, in the two
processes $\procA$ and $\procB$. Thus\hfill
\(
\procA \stackrel{{\it def}}{=} \lnA\mathbin.\overline{\ass{\rA}{\tr}}\mathbin.\overline{\ass{\tu}{B}}\mathbin.
  (\noti{\rB}{\fa}\ +\ \noti{\tu}{A})\mathbin.\ecA\mathbin.\rt\mathbin.
   \lcA\mathbin.\overline{\ass{\rA}{\fa}}\mathbin.\enA\mathbin.\procA.
\)\\
This models that a process spends a positive but finite amount of time in its critical section.
The LTS of the resulting CCS$_\rt$ rendering of Peterson's protocol is displayed in Figure~\ref{PetersonSD}.
Exactly as in \Sec{Peterson}/\ref{sec:Peterson CCS}, it follows that this model satisfies the
requirements
\hyperlink{ME1}{\MEA},
\hyperlink{ME}{\MEB},
\hyperlink{ME}{\MED[Pr]},
\hyperlink{ME}{\MEE[J]} and
\hyperlink{ME}{\MEF[J]}.
Additionally, it satisfies \hyperlink{ME}{\MEC[Pr]}, as follows immediately from the LTS\@.

The same result would be obtained by letting time pass in the noncritical section, instead of, or in
addition to, the critical section. It can be argued that it is not realistic to assume that assignments
like $\ell_2$ and $\ell_3$ occur instantaneously. However, this part of the modelling in CCS$_\rt$ is
merely an abstraction, and can be taken to mean that the time needed to execute such an assignment is
significantly smaller than the time a process spends in its critical and/or noncritical section.
Using CCS$_\rt$, one can also make a model in which time is spent between each two instructions.
In such a rendering one would obtain \hyperlink{ME}{\MEC[J]}, thus needing justness for starvation-freedom.

\addtocontents{toc}{\protect\vspace{10pt}}
\section{Conclusion}

This paper introduces temporal judgements of the form \hyperlink{quinary}{$\D \models^{CC}_{B,E} \phi$}, where
\begin{itemize}
\item $\D$ is a distributed system or its representation as pseudocode, a Petri net, a process in some
  process algebra, or a state in a labelled transition system or in a Kripke structure,
\item $\phi$ is a formula from a temporal logic, such as LTL or CTL,
\item ${\it CC}$ is a completeness criterion, telling which execution paths model actual system runs,
\item and $B$ and $E$ model the influence of the environment on reactive system behaviour, by stipulating
  which actions can be blocked permanently and temporary, respectively.
\end{itemize}
I call this \emph{reactive temporal logic}, or \emph{reactive LTL} when $\phi$ is an LTL formula.
Standard temporal logic has judgements $\D \models \phi$, obtained by default choices for ${\it CC}$, $B$ and $E$.

In the absence of \emph{time-out} transitions, the truth of judgements $\D \models^{CC}_{B,E} \phi$
is independent of $E$, so that $\models$ reduces to a quaternary relation. In this context
I present encodings of reactive LTL into standard LTL, at the expense of adding many
atomic propositions, and I present a fragment of LTL that only describes \emph{safety properties},
telling that nothing bad will ever happen. On this fragment, the values of ${\it CC}$ and $B$ do not matter,
so that reactive temporal judgements say no more than standard ones.

I formulate the correctness requirements of mutual exclusion protocols and fair schedules in
reactive LTL, so that it is unambiguously determined which processes present correct mutual exclusion
protocols, or fair schedules, and which do not. As some of the criteria are parametrised by the
choice of a completeness criterion, I obtain a hierarchy of correctness requirements,
where the choice of a stronger completeness criterion yields a lower quality mutual exclusion
protocol or fair scheduler. 

I formulate two assumptions that are commonly made when studying mutual exclusion, and call them
\emph{atomicity} and \emph{speed independence}. Both stem from Dijkstra's paper in which the mutual
exclusion problem was originally posed. I claim that under these assumptions correct mutual
exclusion protocols do not exist, unless one accepts the lowest quality criteria from the
above-mentioned hierarchy, namely the choice of \emph{fairness} as completeness criterion.
I substantiate this claim in detail for Peterson's mutual exclusion protocol.
I consider the choice of fairness as unwarranted, because the real world is not fair.
Moreover, when fairness would be an acceptable choice I propose a much simpler mutual exclusion
protocol---the \emph{gatekeeper}---that would also be correct, but which I expect would be rejected
by most experts, exactly because its blatant employ of fairness.

I render Peterson's protocol as a Petri net and as an expression in the process algebra CCS\@.
Since both atomicity and speed independence are build in in these formalisms, it is unavoidable
that the so formalised protocol is correct only under the assumption of fairness.

Good requirements for mutual exclusion protocols are obtained by using \emph{justness} as parameter
in the above hierarchy of quality criteria. Justness is a completeness criteria that is weaker than
fairness, and typically warranted in applications. Justness can be formalised in terms of a
concurrency relation between transitions in labelled transition systems or Petri nets.
I use Peterson's protocol to illustrate that any speed independent formalisation of mutual
exclusion (implicitly or explicitly) requires an asymmetric concurrency relation.

\emph{Progress} is a completeness criteria even weaker than justness, so that its use as parameter
in the correctness criteria specifies even higher quality mutual exclusion protocols. I claim that
such protocols do not exist when assuming speed independence, even when dropping the assumption of atomicity.
Also this claim is substantiated in detail for Peterson's protocol.

One alternative for atomicity is to allow read and write actions to overlap in time.
Assuming that two overlapping writes can write any legal value in a register, and a read overlapping
with a write may read any legal value, Lamport's bakery algorithm \cite{bakery} and
Aravind's mutual exclusion algorithm \cite{Ar11} are known to work correctly: they satisfy all my
requirements with justness as parameter. However, the algorithms of Peterson~\cite{Pet81} and
Szyma\'nski's~\cite{Szy88} do not \cite{Spr21,LS21}.

Another alternative to atomicity is to let a write action interrupt a read. This yields an entirely
correct model of Peterson's algorithm, satisfying all requirements with justness as parameter.
This can be modelled, for instance, in terms of Petri nets extended with read arcs \cite{Vogler02},
or CCS extended with signals \cite{CDV09,EPTCS255.2}. These approaches yield an asymmetric
concurrency relation.

Here I present a correct rendering of Peterson's algorithm that assumes atomicity and a symmetric
concurrency relation. It satisfies all requirements with justness as parameter, and the main
starvation-freedom requirement even with progress as parameter. Naturally, this goes at the expense
of speed independence. I formalise this in a variant of the process algebra CCS enriched with
\emph{time-out} transitions. These allow to model the passage of time in a qualitative way,
abstracting from exact durations.

\paragraph{Acknowledgements}
I am grateful to Wan Fokkink, Peter H\"ofner, Bas Luttik, Liam O'Connor, Myrthe Spronck,
Walter Vogler, Weiyou Wang and three reviewers for insightful feedback.

\bibliographystyle{eptcsalpha}
\bibliography{../../../../Biblio/abbreviations,../../../../Biblio/new,../../../../Biblio/dbase,glabbeek}
\end{document}

%% file: dependence.tex
\color{Green}
\expandafter\ifx\csname graph\endcsname\relax
   \csname newbox\expandafter\endcsname\csname graph\endcsname
\fi
\ifx\graphtemp\undefined
  \csname newdimen\endcsname\graphtemp
\fi
\expandafter\setbox\csname graph\endcsname
 =\vtop{\vskip 0pt\hbox{%
    \graphtemp=.5ex
    \advance\graphtemp by 0.355in
    \rlap{\kern 0.105in\lower\graphtemp\hbox to 0pt{\hss \raisebox{-1pt}{\ref{RTL}}\hss}}%
    \graphtemp=.5ex
    \advance\graphtemp by 0.355in
    \rlap{\kern 0.405in\lower\graphtemp\hbox to 0pt{\hss \ds{motivation}\hss}}%
\pdfliteral{
q [] 0 d 1 J 1 j
0.576 w
0.072 w
q 0 g
20.52 -23.76 m
24.12 -25.56 l
20.52 -27.36 l
20.52 -23.76 l
B Q
0.576 w
12.6 -25.56 m
20.52 -25.56 l
S
Q
}%
    \graphtemp=.5ex
    \advance\graphtemp by 1.355in
    \rlap{\kern 0.405in\lower\graphtemp\hbox to 0pt{\hss \ds{Kripke}\hss}}%
\pdfliteral{
q [] 0 d 1 J 1 j
0.576 w
0.072 w
q 0 g
30.96 -86.4 m
29.16 -90 l
27.36 -86.4 l
30.96 -86.4 l
B Q
0.576 w
q [0 3.552] 0 d
29.16 -33.12 m
29.16 -86.4 l
S Q
Q
}%
    \graphtemp=.5ex
    \advance\graphtemp by 1.355in
    \rlap{\kern 0.905in\lower\graphtemp\hbox to 0pt{\hss \ds{Models}\hss}}%
\pdfliteral{
q [] 0 d 1 J 1 j
0.576 w
0.072 w
q 0 g
54 -95.76 m
57.6 -97.56 l
54 -99.36 l
54 -95.76 l
B Q
0.576 w
36.72 -97.56 m
54 -97.56 l
S
Q
}%
    \graphtemp=.5ex
    \advance\graphtemp by 1.355in
    \rlap{\kern 1.405in\lower\graphtemp\hbox to 0pt{\hss \ds{LTS}\hss}}%
\pdfliteral{
q [] 0 d 1 J 1 j
0.576 w
0.072 w
q 0 g
88.056 -95.76 m
91.656 -97.56 l
88.056 -99.36 l
88.056 -95.76 l
B Q
0.576 w
72.72 -97.56 m
88.056 -97.56 l
S
Q
}%
    \graphtemp=.5ex
    \advance\graphtemp by 1.855in
    \rlap{\kern 1.405in\lower\graphtemp\hbox to 0pt{\hss \ds{nets}\hss}}%
\pdfliteral{
q [] 0 d 1 J 1 j
0.576 w
0.072 w
q 0 g
90.648 -123.336 m
91.656 -127.296 l
87.912 -125.712 l
90.648 -123.336 l
B Q
0.576 w
70.488 -102.888 m
89.28 -124.56 l
S
Q
}%
    \graphtemp=.5ex
    \advance\graphtemp by 0.855in
    \rlap{\kern 1.405in\lower\graphtemp\hbox to 0pt{\hss \ds{CCS}\hss}}%
\pdfliteral{
q [] 0 d 1 J 1 j
0.576 w
0.072 w
q 0 g
87.912 -69.408 m
91.656 -67.896 l
90.648 -71.784 l
87.912 -69.408 l
B Q
0.576 w
70.488 -92.232 m
89.28 -70.56 l
S
Q
}%
    \graphtemp=.5ex
    \advance\graphtemp by 1.355in
    \rlap{\kern 1.905in\lower\graphtemp\hbox to 0pt{\hss \ds{LTSC}\hss}}%
\pdfliteral{
q [] 0 d 1 J 1 j
0.576 w
0.072 w
q 0 g
124.056 -95.76 m
127.656 -97.56 l
124.056 -99.36 l
124.056 -95.76 l
B Q
0.576 w
110.736 -97.56 m
124.056 -97.56 l
S
0.072 w
q 0 g
124.056 -105.696 m
127.656 -103.896 l
126.936 -107.856 l
124.056 -105.696 l
B Q
0.576 w
q [3.6 3.625675] 0 d
110.736 -127.296 m
125.496 -106.776 l
S Q
0.072 w
q 0 g
126.936 -87.264 m
127.656 -91.296 l
124.056 -89.424 l
126.936 -87.264 l
B Q
0.576 w
q [3.6 3.606203] 0 d
110.736 -67.896 m
125.496 -88.344 l
S Q
Q
}%
    \graphtemp=.5ex
    \advance\graphtemp by 1.355in
    \rlap{\kern 2.405in\lower\graphtemp\hbox to 0pt{\hss \ds{completeness criteria}\hss}}%
\pdfliteral{
q [] 0 d 1 J 1 j
0.576 w
0.072 w
q 0 g
162 -95.76 m
165.6 -97.56 l
162 -99.36 l
162 -95.76 l
B Q
0.576 w
146.736 -97.56 m
162 -97.56 l
S
Q
}%
    \graphtemp=.5ex
    \advance\graphtemp by 1.355in
    \rlap{\kern 2.905in\lower\graphtemp\hbox to 0pt{\hss \ds{blocking}\hss}}%
\pdfliteral{
q [] 0 d 1 J 1 j
0.576 w
0.072 w
q 0 g
198 -95.76 m
201.6 -97.56 l
198 -99.36 l
198 -95.76 l
B Q
0.576 w
180.72 -97.56 m
198 -97.56 l
S
0.072 w
q 0 g
202.608 -88.344 m
203.832 -92.232 l
200.016 -90.864 l
202.608 -88.344 l
B Q
0.576 w
36.734062 -25.541465 m
98.220248 -22.467324 158.096626 -45.77568 201.322001 -89.611377 c
S
Q
}%
    \graphtemp=.5ex
    \advance\graphtemp by 0.105in
    \rlap{\kern 3.155in\lower\graphtemp\hbox to 0pt{\hss \ds{translating}\hss}}%
\pdfliteral{
q [] 0 d 1 J 1 j
0.576 w
0.072 w
q 0 g
220.392 -12.456 m
223.632 -10.008 l
223.632 -14.04 l
220.392 -12.456 l
B Q
0.576 w
209.195663 -90.008937 m
206.366217 -63.732746 210.822479 -37.184654 222.076361 -13.272437 c
S
Q
}%
    \graphtemp=.5ex
    \advance\graphtemp by 0.855in
    \rlap{\kern 3.155in\lower\graphtemp\hbox to 0pt{\hss \ds{safety}\hss}}%
\pdfliteral{
q [] 0 d 1 J 1 j
0.576 w
0.072 w
q 0 g
220.536 -70.704 m
223.776 -68.328 l
223.776 -72.36 l
220.536 -70.704 l
B Q
0.576 w
q [3.6 5.354605] 0 d
212.544 -90.792 m
222.192 -71.568 l
S Q
0.072 w
q 0 g
219.6 -52.848 m
221.832 -56.232 l
217.8 -56.016 l
219.6 -52.848 l
B Q
0.576 w
34.517835 -92.190273 m
78.081625 -35.716675 156.422014 -19.642032 218.704962 -54.396965 c
S
Q
}%
\color{red}
    \graphtemp=.5ex
    \advance\graphtemp by 1.855in
    \rlap{\kern 1.855in\lower\graphtemp\hbox to 0pt{\hss \raisebox{-1pt}{\ref{requirements}}\hss}}%
    \graphtemp=.5ex
    \advance\graphtemp by 1.855in
    \rlap{\kern 2.155in\lower\graphtemp\hbox to 0pt{\hss \ds{history}\hss}}%
\pdfliteral{
q [] 0 d 1 J 1 j
0.576 w
0.072 w
q 0 g
146.52 -131.76 m
150.12 -133.56 l
146.52 -135.36 l
146.52 -131.76 l
B Q
0.576 w
138.6 -133.56 m
146.52 -133.56 l
S
Q
}%
    \graphtemp=.5ex
    \advance\graphtemp by 1.855in
    \rlap{\kern 2.655in\lower\graphtemp\hbox to 0pt{\hss \ds{FS}\hss}}%
\pdfliteral{
q [] 0 d 1 J 1 j
0.576 w
0.072 w
q 0 g
180 -131.76 m
183.6 -133.56 l
180 -135.36 l
180 -131.76 l
B Q
0.576 w
162.72 -133.56 m
180 -133.56 l
S
Q
}%
    \graphtemp=.5ex
    \advance\graphtemp by 1.855in
    \rlap{\kern 3.155in\lower\graphtemp\hbox to 0pt{\hss \ds{formalising FS}\hss}}%
\pdfliteral{
q [] 0 d 1 J 1 j
0.576 w
0.072 w
q 0 g
216 -131.76 m
219.6 -133.56 l
216 -135.36 l
216 -131.76 l
B Q
0.576 w
198.72 -133.56 m
216 -133.56 l
S
0.072 w
q 0 g
223.776 -122.76 m
223.776 -126.792 l
220.536 -124.416 l
223.776 -122.76 l
B Q
0.576 w
212.544 -104.328 m
222.192 -123.552 l
S
0.072 w
q 0 g
228.96 -122.4 m
227.16 -126 l
225.36 -122.4 l
228.96 -122.4 l
B Q
0.576 w
q [3.6 3.497143] 0 d
227.16 -69.12 m
227.16 -122.4 l
S Q
Q
}%
    \graphtemp=.5ex
    \advance\graphtemp by 1.355in
    \rlap{\kern 3.405in\lower\graphtemp\hbox to 0pt{\hss \ds{formalising mutex}\hss}}%
\pdfliteral{
q [] 0 d 1 J 1 j
0.576 w
0.072 w
q 0 g
234 -95.76 m
237.6 -97.56 l
234 -99.36 l
234 -95.76 l
B Q
0.576 w
216.72 -97.56 m
234 -97.56 l
S
0.072 w
q 0 g
241.776 -86.76 m
241.776 -90.792 l
238.536 -88.416 l
241.776 -86.76 l
B Q
0.576 w
q [3.6 5.354605] 0 d
230.544 -68.328 m
240.192 -87.552 l
S Q
0.072 w
q 0 g
238.536 -106.704 m
241.776 -104.328 l
241.776 -108.36 l
238.536 -106.704 l
B Q
0.576 w
q [3.6 5.354605] 0 d
230.544 -126.792 m
240.192 -107.568 l
S Q
0.072 w
q 0 g
234.144 -100.008 m
238.176 -100.368 l
235.44 -103.392 l
234.144 -100.008 l
B Q
0.576 w
162.144 -130.752 m
234.792 -101.736 l
S
Q
}%
    \graphtemp=.5ex
    \advance\graphtemp by 0.855in
    \rlap{\kern 3.655in\lower\graphtemp\hbox to 0pt{\hss \ds{state-oriented}\hss}}%
\pdfliteral{
q [] 0 d 1 J 1 j
0.576 w
0.072 w
q 0 g
256.536 -70.704 m
259.776 -68.328 l
259.776 -72.36 l
256.536 -70.704 l
B Q
0.576 w
248.544 -90.792 m
258.192 -71.568 l
S
Q
}%
    \graphtemp=.5ex
    \advance\graphtemp by 1.855in
    \rlap{\kern 3.655in\lower\graphtemp\hbox to 0pt{\hss \ds{hierarchy}\hss}}%
\pdfliteral{
q [] 0 d 1 J 1 j
0.576 w
0.072 w
q 0 g
252 -131.76 m
255.6 -133.56 l
252 -135.36 l
252 -131.76 l
B Q
0.576 w
234.72 -133.56 m
252 -133.56 l
S
0.072 w
q 0 g
259.776 -122.76 m
259.776 -126.792 l
256.536 -124.416 l
259.776 -122.76 l
B Q
0.576 w
248.544 -104.328 m
258.192 -123.552 l
S
Q
}%
    \graphtemp=.5ex
    \advance\graphtemp by 1.855in
    \rlap{\kern 4.155in\lower\graphtemp\hbox to 0pt{\hss \ds{interface}\hss}}%
\pdfliteral{
q [] 0 d 1 J 1 j
0.576 w
0.072 w
q 0 g
288 -131.76 m
291.6 -133.56 l
288 -135.36 l
288 -131.76 l
B Q
0.576 w
270.72 -133.56 m
288 -133.56 l
S
Q
}%
\color{black}
    \graphtemp=.5ex
    \advance\graphtemp by 1.355in
    \rlap{\kern 3.905in\lower\graphtemp\hbox to 0pt{\hss \raisebox{-1pt}{\ref{impossibility results}}\hss}}%
\pdfliteral{
q [] 0 d 1 J 1 j
0.576 w
0.072 w
q 0 g
270 -95.76 m
273.6 -97.56 l
270 -99.36 l
270 -95.76 l
B Q
0.576 w
q [3.6 3.24] 0 d
252.72 -97.56 m
270 -97.56 l
S Q
0.072 w
q 0 g
274.536 -106.704 m
277.776 -104.328 l
277.776 -108.36 l
274.536 -106.704 l
B Q
0.576 w
q [3.6 5.354605] 0 d
266.544 -126.792 m
276.192 -107.568 l
S Q
Q
}%
    \graphtemp=.5ex
    \advance\graphtemp by 1.355in
    \rlap{\kern 4.405in\lower\graphtemp\hbox to 0pt{\hss \ds{Peterson}\hss}}%
\pdfliteral{
q [] 0 d 1 J 1 j
0.576 w
0.072 w
q 0 g
306 -95.76 m
309.6 -97.56 l
306 -99.36 l
306 -95.76 l
B Q
0.576 w
288.72 -97.56 m
306 -97.56 l
S
0.072 w
q 0 g
308.088 -90.216 m
310.752 -93.312 l
306.72 -93.528 l
308.088 -90.216 l
B Q
0.576 w
251.585726 -93.330266 m
269.250128 -85.372219 289.376306 -84.868956 307.416371 -91.9342 c
S
Q
}%
    \graphtemp=.5ex
    \advance\graphtemp by 1.855in
    \rlap{\kern 4.655in\lower\graphtemp\hbox to 0pt{\hss \ds{EG}\hss}}%
\pdfliteral{
q [] 0 d 1 J 1 j
0.576 w
0.072 w
q 0 g
331.776 -122.76 m
331.776 -126.792 l
328.536 -124.416 l
331.776 -122.76 l
B Q
0.576 w
320.544 -104.328 m
330.192 -123.552 l
S
0.072 w
q 0 g
324 -131.76 m
327.6 -133.56 l
324 -135.36 l
324 -131.76 l
B Q
0.576 w
306.72 -133.56 m
324 -133.56 l
S
Q
}%
    \graphtemp=.5ex
    \advance\graphtemp by 2.105in
    \rlap{\kern 5.155in\lower\graphtemp\hbox to 0pt{\hss \ds{Peterson Petri}\hss}}%
\pdfliteral{
q [] 0 d 1 J 1 j
0.576 w
0.072 w
q 0 g
359.928 -149.832 m
363.6 -151.56 l
360.072 -153.432 l
359.928 -149.832 l
B Q
0.576 w
359.972171 -151.716189 m
276.603309 -154.049819 193.189312 -149.136262 110.670986 -137.030937 c
S
0.072 w
q 0 g
364.536 -142.416 m
365.832 -146.232 l
362.016 -144.936 l
364.536 -142.416 l
B Q
0.576 w
322.488 -102.888 m
363.24 -143.64 l
S
Q
}%
    \graphtemp=.5ex
    \advance\graphtemp by 0.605in
    \rlap{\kern 5.405in\lower\graphtemp\hbox to 0pt{\hss \ds{Peterson CCS}\hss}}%
\pdfliteral{
q [] 0 d 1 J 1 j
0.576 w
0.072 w
q 0 g
378.72 -47.016 m
382.752 -47.088 l
380.232 -50.256 l
378.72 -47.016 l
B Q
0.576 w
322.501492 -92.192544 m
338.378177 -73.9949 357.754436 -59.17819 379.476479 -48.624696 c
S
0.072 w
q 0 g
378.216 -41.328 m
381.6 -43.56 l
377.784 -44.928 l
378.216 -41.328 l
B Q
0.576 w
110.710886 -61.556526 m
197.881051 -38.738234 288.559856 -32.485326 378.037005 -43.12261 c
S
Q
}%
    \graphtemp=.5ex
    \advance\graphtemp by 1.355in
    \rlap{\kern 4.905in\lower\graphtemp\hbox to 0pt{\hss \ds{atomicity}\hss}}%
\pdfliteral{
q [] 0 d 1 J 1 j
0.576 w
0.072 w
q 0 g
342 -95.76 m
345.6 -97.56 l
342 -99.36 l
342 -95.76 l
B Q
0.576 w
324.72 -97.56 m
342 -97.56 l
S
Q
}%
    \graphtemp=.5ex
    \advance\graphtemp by 1.855in
    \rlap{\kern 5.155in\lower\graphtemp\hbox to 0pt{\hss \ds{Peterson overlap}\hss}}%
\pdfliteral{
q [] 0 d 1 J 1 j
0.576 w
0.072 w
q 0 g
367.776 -122.76 m
367.776 -126.792 l
364.536 -124.416 l
367.776 -122.76 l
B Q
0.576 w
356.544 -104.328 m
366.192 -123.552 l
S
Q
}%
    \graphtemp=.5ex
    \advance\graphtemp by 0.855in
    \rlap{\kern 5.155in\lower\graphtemp\hbox to 0pt{\hss \ds{impossible}\hss}}%
\pdfliteral{
q [] 0 d 1 J 1 j
0.576 w
0.072 w
q 0 g
364.536 -70.704 m
367.776 -68.328 l
367.776 -72.36 l
364.536 -70.704 l
B Q
0.576 w
356.544 -90.792 m
366.192 -71.568 l
S
Q
}%
    \graphtemp=.5ex
    \advance\graphtemp by 1.355in
    \rlap{\kern 5.405in\lower\graphtemp\hbox to 0pt{\hss \ds{Peterson CCSS}\hss}}%
\pdfliteral{
q [] 0 d 1 J 1 j
0.576 w
0.072 w
q 0 g
378 -95.76 m
381.6 -97.56 l
378 -99.36 l
378 -95.76 l
B Q
0.576 w
360.72 -97.56 m
378 -97.56 l
S
0.072 w
q 0 g
385.776 -86.76 m
385.776 -90.792 l
382.536 -88.416 l
385.776 -86.76 l
B Q
0.576 w
q [3.6 5.354605] 0 d
374.544 -68.328 m
384.192 -87.552 l
S Q
Q
}%
    \graphtemp=.5ex
    \advance\graphtemp by 1.855in
    \rlap{\kern 5.905in\lower\graphtemp\hbox to 0pt{\hss \ds{read arcs}\hss}}%
\pdfliteral{
q [] 0 d 1 J 1 j
0.576 w
0.072 w
q 0 g
414.648 -123.336 m
415.656 -127.296 l
411.912 -125.712 l
414.648 -123.336 l
B Q
0.576 w
394.488 -102.888 m
413.28 -124.56 l
S
0.072 w
q 0 g
410.184 -136.8 m
414.216 -137.376 l
411.408 -140.256 l
410.184 -136.8 l
B Q
0.576 w
378.72 -149.76 m
410.76 -138.528 l
S
Q
}%
    \graphtemp=.5ex
    \advance\graphtemp by 1.355in
    \rlap{\kern 5.905in\lower\graphtemp\hbox to 0pt{\hss \ds{broadcast}\hss}}%
\pdfliteral{
q [] 0 d 1 J 1 j
0.576 w
0.072 w
q 0 g
412.056 -95.76 m
415.656 -97.56 l
412.056 -99.36 l
412.056 -95.76 l
B Q
0.576 w
396.72 -97.56 m
412.056 -97.56 l
S
0.072 w
q 0 g
415.656 -87.264 m
415.656 -91.296 l
412.416 -88.848 l
415.656 -87.264 l
B Q
0.576 w
394.488 -48.888 m
414 -88.056 l
S
Q
}%
    \graphtemp=.5ex
    \advance\graphtemp by 0.980in
    \rlap{\kern 5.905in\lower\graphtemp\hbox to 0pt{\hss \ds{signals}\hss}}%
\pdfliteral{
q [] 0 d 1 J 1 j
0.576 w
0.072 w
q 0 g
411.696 -75.96 m
415.656 -75.096 l
413.928 -78.696 l
411.696 -75.96 l
B Q
0.576 w
394.488 -92.232 m
412.848 -77.328 l
S
0.072 w
q 0 g
414.072 -60.552 m
415.656 -64.296 l
411.696 -63.216 l
414.072 -60.552 l
B Q
0.576 w
395.568 -46.728 m
412.92 -61.92 l
S
Q
}%
    \graphtemp=.5ex
    \advance\graphtemp by 0.605in
    \rlap{\kern 5.905in\lower\graphtemp\hbox to 0pt{\hss \ds{Bouwman}\hss}}%
\pdfliteral{
q [] 0 d 1 J 1 j
0.576 w
0.072 w
q 0 g
423.36 -53.424 m
425.16 -49.896 l
426.96 -53.424 l
423.36 -53.424 l
B Q
0.576 w
425.16 -64.296 m
425.16 -53.424 l
S
Q
}%
    \graphtemp=.5ex
    \advance\graphtemp by 0.105in
    \rlap{\kern 5.905in\lower\graphtemp\hbox to 0pt{\hss \ds{mCRL2}\hss}}%
\pdfliteral{
q [] 0 d 1 J 1 j
0.576 w
0.072 w
q 0 g
423.36 -17.496 m
425.16 -13.896 l
426.96 -17.496 l
423.36 -17.496 l
B Q
0.576 w
425.16 -37.296 m
425.16 -17.496 l
S
0.072 w
q 0 g
409.896 -5.76 m
413.496 -7.56 l
409.896 -9.36 l
409.896 -5.76 l
B Q
0.576 w
q [3.6 3.549] 0 d
234.72 -7.56 m
409.896 -7.56 l
S Q
Q
}%
\color{DarkBlue}
    \graphtemp=.5ex
    \advance\graphtemp by 0.855in
    \rlap{\kern 4.655in\lower\graphtemp\hbox to 0pt{\hss \raisebox{-1pt}{\ref{timeouts}}\hss}}%
\pdfliteral{
q [] 0 d 1 J 1 j
0.576 w
0.072 w
q 0 g
328.536 -70.704 m
331.776 -68.328 l
331.776 -72.36 l
328.536 -70.704 l
B Q
0.576 w
q [3.6 5.354605] 0 d
320.544 -90.792 m
330.192 -71.568 l
S Q
0.072 w
q 0 g
333.36 -72.72 m
335.16 -69.12 l
336.96 -72.72 l
333.36 -72.72 l
B Q
0.576 w
q [3.6 3.497143] 0 d
335.16 -126 m
335.16 -72.72 l
S Q
0.072 w
q 0 g
338.544 -72.36 m
338.544 -68.328 l
341.784 -70.704 l
338.544 -72.36 l
B Q
0.576 w
q [3.6 5.354605] 0 d
349.776 -90.792 m
340.128 -71.568 l
S Q
0.072 w
q 0 g
346.32 -63.36 m
342.72 -61.56 l
346.32 -59.76 l
346.32 -63.36 l
B Q
0.576 w
q [3.6 3.24] 0 d
363.6 -61.56 m
346.32 -61.56 l
S Q
Q
}%
    \graphtemp=.5ex
    \advance\graphtemp by 0.355in
    \rlap{\kern 4.655in\lower\graphtemp\hbox to 0pt{\hss \ds{CCSt}\hss}}%
\pdfliteral{
q [] 0 d 1 J 1 j
0.576 w
0.072 w
q 0 g
333.36 -34.92 m
335.16 -31.32 l
336.96 -34.92 l
333.36 -34.92 l
B Q
0.576 w
335.16 -55.296 m
335.16 -34.92 l
S
0.072 w
q 0 g
324 -23.76 m
327.6 -25.56 l
324 -27.36 l
324 -23.76 l
B Q
0.576 w
110.710642 -60.128611 m
179.668886 -37.953389 251.57865 -26.311169 324.011151 -25.595174 c
S
Q
}%
    \graphtemp=.5ex
    \advance\graphtemp by 0.355in
    \rlap{\kern 5.155in\lower\graphtemp\hbox to 0pt{\hss \ds{spurious}\hss}}%
\pdfliteral{
q [] 0 d 1 J 1 j
0.576 w
0.072 w
q 0 g
360 -23.76 m
363.6 -25.56 l
360 -27.36 l
360 -23.76 l
B Q
0.576 w
342.72 -25.56 m
360 -25.56 l
S
0.072 w
q 0 g
360.216 -28.44 m
364.248 -28.656 l
361.728 -31.752 l
360.216 -28.44 l
B Q
0.576 w
216.072 -94.464 m
360.936 -30.096 l
S
Q
}%
    \graphtemp=.5ex
    \advance\graphtemp by 0.355in
    \rlap{\kern 5.655in\lower\graphtemp\hbox to 0pt{\hss \ds{Peterson timeouts}\hss}}%
\pdfliteral{
q [] 0 d 1 J 1 j
0.576 w
0.072 w
q 0 g
396 -23.76 m
399.6 -25.56 l
396 -27.36 l
396 -23.76 l
B Q
0.576 w
378.72 -25.56 m
396 -25.56 l
S
0.072 w
q 0 g
398.016 -32.184 m
401.832 -30.888 l
400.536 -34.704 l
398.016 -32.184 l
B Q
0.576 w
394.488 -38.232 m
399.24 -33.48 l
S
Q
}%
    \hbox{\vrule depth2.210in width0pt height 0pt}%
    \kern 6.067in
  }%
}%

%% file: pretzel.tex
\expandafter\ifx\csname graph\endcsname\relax
   \csname newbox\expandafter\endcsname\csname graph\endcsname
\fi
\ifx\graphtemp\undefined
  \csname newdimen\endcsname\graphtemp
\fi
\expandafter\setbox\csname graph\endcsname
 =\vtop{\vskip 0pt\hbox{%
\pdfliteral{
q [] 0 d 1 J 1 j
0.576 w
0.576 w
31.968 -20.016 m
31.968 -24.867269 28.035269 -28.8 23.184 -28.8 c
18.332731 -28.8 14.4 -24.867269 14.4 -20.016 c
14.4 -15.164731 18.332731 -11.232 23.184 -11.232 c
28.035269 -11.232 31.968 -15.164731 31.968 -20.016 c
S
0.072 w
q 0 g
7.2 -18.216 m
14.4 -20.016 l
7.2 -21.816 l
7.2 -18.216 l
B Q
0.576 w
0 -20.016 m
7.2 -20.016 l
S
72 -20.016 m
72 -24.867269 68.067269 -28.8 63.216 -28.8 c
58.364731 -28.8 54.432 -24.867269 54.432 -20.016 c
54.432 -15.164731 58.364731 -11.232 63.216 -11.232 c
68.067269 -11.232 72 -15.164731 72 -20.016 c
S
0.072 w
q 0 g
50.256 -29.304 m
56.952 -26.208 l
52.488 -32.112 l
50.256 -29.304 l
B Q
0.576 w
51.378988 -30.650468 m
43.916205 -34.457143 34.813838 -32.612005 29.422545 -26.199684 c
S
0.072 w
q 0 g
36.144 -10.728 m
29.448 -13.752 l
33.912 -7.92 l
36.144 -10.728 l
B Q
0.576 w
35.021022 -9.309527 m
42.483808 -5.502856 51.586173 -7.348 56.977463 -13.760325 c
S
Q
}%
    \graphtemp=.5ex
    \advance\graphtemp by 0.389in
    \rlap{\kern 0.600in\lower\graphtemp\hbox to 0pt{\hss $c$\hss}}%
    \graphtemp=.5ex
    \advance\graphtemp by 0.000in
    \rlap{\kern 0.600in\lower\graphtemp\hbox to 0pt{\hss $p$\hss}}%
    \hbox{\vrule depth0.453in width0pt height 0pt}%
    \kern 1.000in
  }%
}%

%% file: Bart.tex
\expandafter\ifx\csname graph\endcsname\relax
   \csname newbox\expandafter\endcsname\csname graph\endcsname
\fi
\ifx\graphtemp\undefined
  \csname newdimen\endcsname\graphtemp
\fi
\expandafter\setbox\csname graph\endcsname
 =\vtop{\vskip 0pt\hbox{%
\pdfliteral{
q [] 0 d 1 J 1 j
0.576 w
0.576 w
24.48 -15.408 m
24.48 -20.259269 20.547269 -24.192 15.696 -24.192 c
10.844731 -24.192 6.912 -20.259269 6.912 -15.408 c
6.912 -10.556731 10.844731 -6.624 15.696 -6.624 c
20.547269 -6.624 24.48 -10.556731 24.48 -15.408 c
S
0.072 w
q 0 g
5.544 -2.952 m
9.432 -9.288 l
3.024 -5.544 l
5.544 -2.952 l
B Q
0.576 w
0 0 m
4.32 -4.248 l
S
34.488 -49.464 m
34.488 -54.315269 30.555269 -58.248 25.704 -58.248 c
20.852731 -58.248 16.92 -54.315269 16.92 -49.464 c
16.92 -44.612731 20.852731 -40.68 25.704 -40.68 c
30.555269 -40.68 34.488 -44.612731 34.488 -49.464 c
S
Q
}%
    \graphtemp=.5ex
    \advance\graphtemp by 0.687in
    \rlap{\kern 0.357in\lower\graphtemp\hbox to 0pt{\hss $A$\hss}}%
\pdfliteral{
q [] 0 d 1 J 1 j
0.576 w
74.52 -49.464 m
74.52 -54.315269 70.587269 -58.248 65.736 -58.248 c
60.884731 -58.248 56.952 -54.315269 56.952 -49.464 c
56.952 -44.612731 60.884731 -40.68 65.736 -40.68 c
70.587269 -40.68 74.52 -44.612731 74.52 -49.464 c
S
Q
}%
    \graphtemp=.5ex
    \advance\graphtemp by 0.687in
    \rlap{\kern 0.913in\lower\graphtemp\hbox to 0pt{\hss $B$\hss}}%
\pdfliteral{
q [] 0 d 1 J 1 j
0.576 w
0.072 w
q 0 g
41.688 -51.264 m
34.488 -49.464 l
41.688 -47.664 l
41.688 -51.264 l
B Q
q 0 g
49.752 -47.664 m
56.952 -49.464 l
49.752 -51.264 l
49.752 -47.664 l
B Q
0.576 w
41.688 -49.464 m
49.752 -49.464 l
S
62.496 -9.432 m
62.496 -14.283269 58.563269 -18.216 53.712 -18.216 c
48.860731 -18.216 44.928 -14.283269 44.928 -9.432 c
44.928 -4.580731 48.860731 -0.648 53.712 -0.648 c
58.563269 -0.648 62.496 -4.580731 62.496 -9.432 c
S
Q
}%
    \graphtemp=.5ex
    \advance\graphtemp by 0.131in
    \rlap{\kern 0.746in\lower\graphtemp\hbox to 0pt{\hss $C$\hss}}%
\pdfliteral{
q [] 0 d 1 J 1 j
0.576 w
0.072 w
q 0 g
62.856 -33.624 m
63.216 -41.04 l
59.4 -34.632 l
62.856 -33.624 l
B Q
q 0 g
56.592 -25.272 m
56.232 -17.856 l
60.048 -24.264 l
56.592 -25.272 l
B Q
0.576 w
61.128 -34.128 m
58.32 -24.768 l
S
0.072 w
q 0 g
43.056 -21.528 m
48.672 -16.632 l
46.008 -23.544 l
43.056 -21.528 l
B Q
q 0 g
36.36 -37.368 m
30.744 -42.192 l
33.408 -35.28 l
36.36 -37.368 l
B Q
0.576 w
44.568 -22.536 m
34.92 -36.36 l
S
0.072 w
q 0 g
22.968 -33.552 m
23.256 -40.968 l
19.512 -34.56 l
22.968 -33.552 l
B Q
0.576 w
18.216 -23.904 m
21.24 -34.056 l
S
0.072 w
q 0 g
53.496 -38.952 m
58.464 -44.496 l
51.48 -41.904 l
53.496 -38.952 l
B Q
0.576 w
22.968 -20.376 m
52.488 -40.464 l
S
0.072 w
q 0 g
37.656 -10.152 m
45 -10.8 l
38.232 -13.68 l
37.656 -10.152 l
B Q
0.576 w
24.408 -14.04 m
37.944 -11.952 l
S
Q
}%
    \hbox{\vrule depth0.809in width0pt height 0pt}%
    \kern 1.035in
  }%
}%

%% file: Bart2.tex
\expandafter\ifx\csname graph\endcsname\relax
   \csname newbox\expandafter\endcsname\csname graph\endcsname
\fi
\ifx\graphtemp\undefined
  \csname newdimen\endcsname\graphtemp
\fi
\expandafter\setbox\csname graph\endcsname
 =\vtop{\vskip 0pt\hbox{%
\pdfliteral{
q [] 0 d 1 J 1 j
0.576 w
0.576 w
31.968 -8.784 m
31.968 -13.635269 28.035269 -17.568 23.184 -17.568 c
18.332731 -17.568 14.4 -13.635269 14.4 -8.784 c
14.4 -3.932731 18.332731 0 23.184 0 c
28.035269 0 31.968 -3.932731 31.968 -8.784 c
S
0.072 w
q 0 g
7.2 -6.984 m
14.4 -8.784 l
7.2 -10.584 l
7.2 -6.984 l
B Q
0.576 w
0 -8.784 m
7.2 -8.784 l
S
72 -8.784 m
72 -13.635269 68.067269 -17.568 63.216 -17.568 c
58.364731 -17.568 54.432 -13.635269 54.432 -8.784 c
54.432 -3.932731 58.364731 0 63.216 0 c
68.067269 0 72 -3.932731 72 -8.784 c
S
Q
}%
    \graphtemp=.5ex
    \advance\graphtemp by 0.122in
    \rlap{\kern 0.878in\lower\graphtemp\hbox to 0pt{\hss $B$\hss}}%
\pdfliteral{
q [] 0 d 1 J 1 j
0.576 w
0.072 w
q 0 g
47.232 -6.984 m
54.432 -8.784 l
47.232 -10.584 l
47.232 -6.984 l
B Q
0.576 w
31.968 -8.784 m
47.232 -8.784 l
S
Q
}%
    \hbox{\vrule depth0.244in width0pt height 0pt}%
    \kern 1.000in
  }%
}%

%% file: Bart3.tex
\expandafter\ifx\csname graph\endcsname\relax
   \csname newbox\expandafter\endcsname\csname graph\endcsname
\fi
\ifx\graphtemp\undefined
  \csname newdimen\endcsname\graphtemp
\fi
\expandafter\setbox\csname graph\endcsname
 =\vtop{\vskip 0pt\hbox{%
\pdfliteral{
q [] 0 d 1 J 1 j
0.576 w
0.576 w
31.968 -8.784 m
31.968 -13.635269 28.035269 -17.568 23.184 -17.568 c
18.332731 -17.568 14.4 -13.635269 14.4 -8.784 c
14.4 -3.932731 18.332731 0 23.184 0 c
28.035269 0 31.968 -3.932731 31.968 -8.784 c
S
0.072 w
q 0 g
7.2 -6.984 m
14.4 -8.784 l
7.2 -10.584 l
7.2 -6.984 l
B Q
0.576 w
0 -8.784 m
7.2 -8.784 l
S
0.072 w
q 0 g
33.264 -21.384 m
29.448 -15.048 l
35.784 -18.864 l
33.264 -21.384 l
B Q
0.576 w
16.992 -15.048 m
14.112 -17.928 l
12.1536 -19.8864 11.232 -22.73184 11.232 -26.82 c
11.232 -30.90816 15.06816 -32.832 23.22 -32.832 c
31.37184 -32.832 35.208 -30.90816 35.208 -26.82 c
35.208 -22.73184 34.36704 -19.96704 32.58 -18.18 c
29.952 -15.552 l
S
72 -8.784 m
72 -13.635269 68.067269 -17.568 63.216 -17.568 c
58.364731 -17.568 54.432 -13.635269 54.432 -8.784 c
54.432 -3.932731 58.364731 0 63.216 0 c
68.067269 0 72 -3.932731 72 -8.784 c
S
Q
}%
    \graphtemp=.5ex
    \advance\graphtemp by 0.122in
    \rlap{\kern 0.878in\lower\graphtemp\hbox to 0pt{\hss $B$\hss}}%
\pdfliteral{
q [] 0 d 1 J 1 j
0.576 w
0.072 w
q 0 g
47.232 -6.984 m
54.432 -8.784 l
47.232 -10.584 l
47.232 -6.984 l
B Q
0.576 w
31.968 -8.784 m
47.232 -8.784 l
S
0.072 w
q 0 g
73.224 -21.384 m
69.408 -15.048 l
75.816 -18.864 l
73.224 -21.384 l
B Q
0.576 w
56.952 -15.048 m
54.072 -17.928 l
52.1136 -19.8864 51.192 -22.73184 51.192 -26.82 c
51.192 -30.90816 55.02816 -32.832 63.18 -32.832 c
71.33184 -32.832 75.168 -30.90816 75.168 -26.82 c
75.168 -22.73184 74.32704 -19.96704 72.54 -18.18 c
69.912 -15.552 l
S
Q
}%
    \hbox{\vrule depth0.435in width0pt height 0pt}%
    \kern 1.034in
  }%
}%

%% file: models.tex
\expandafter\ifx\csname graph\endcsname\relax
   \csname newbox\expandafter\endcsname\csname graph\endcsname
\fi
\ifx\graphtemp\undefined
  \csname newdimen\endcsname\graphtemp
\fi
\expandafter\setbox\csname graph\endcsname
 =\vtop{\vskip 0pt\hbox{%
\pdfliteral{
q [] 0 d 1 J 1 j
0.576 w
0.576 w
54 -162 m
54 -176.911688 41.911688 -189 27 -189 c
12.088312 -189 0 -176.911688 0 -162 c
0 -147.088312 12.088312 -135 27 -135 c
41.911688 -135 54 -147.088312 54 -162 c
S
Q
}%
    \graphtemp=\baselineskip
    \multiply\graphtemp by -1
    \divide\graphtemp by 2
    \advance\graphtemp by .5ex
    \advance\graphtemp by 2.250in
    \rlap{\kern 0.375in\lower\graphtemp\hbox to 0pt{\hss Pseudo\hss}}%
    \graphtemp=\baselineskip
    \multiply\graphtemp by 1
    \divide\graphtemp by 2
    \advance\graphtemp by .5ex
    \advance\graphtemp by 2.250in
    \rlap{\kern 0.375in\lower\graphtemp\hbox to 0pt{\hss code\hss}}%
\pdfliteral{
q [] 0 d 1 J 1 j
0.576 w
54 -72 m
54 -86.911688 41.911688 -99 27 -99 c
12.088312 -99 0 -86.911688 0 -72 c
0 -57.088312 12.088312 -45 27 -45 c
41.911688 -45 54 -57.088312 54 -72 c
S
Q
}%
    \graphtemp=\baselineskip
    \multiply\graphtemp by -1
    \divide\graphtemp by 2
    \advance\graphtemp by .5ex
    \advance\graphtemp by 1.000in
    \rlap{\kern 0.375in\lower\graphtemp\hbox to 0pt{\hss Process\hss}}%
    \graphtemp=\baselineskip
    \multiply\graphtemp by 1
    \divide\graphtemp by 2
    \advance\graphtemp by .5ex
    \advance\graphtemp by 1.000in
    \rlap{\kern 0.375in\lower\graphtemp\hbox to 0pt{\hss algebra\hss}}%
\pdfliteral{
q [] 0 d 1 J 1 j
0.576 w
0.072 w
q 0 g
25.2 -106.2 m
27 -99 l
28.8 -106.2 l
25.2 -106.2 l
B Q
0.576 w
27 -135 m
27 -106.2 l
S
144 -27 m
144 -41.911688 131.911688 -54 117 -54 c
102.088312 -54 90 -41.911688 90 -27 c
90 -12.088312 102.088312 0 117 0 c
131.911688 0 144 -12.088312 144 -27 c
S
Q
}%
    \graphtemp=\baselineskip
    \multiply\graphtemp by -1
    \divide\graphtemp by 2
    \advance\graphtemp by .5ex
    \advance\graphtemp by 0.375in
    \rlap{\kern 1.625in\lower\graphtemp\hbox to 0pt{\hss Petri\hss}}%
    \graphtemp=\baselineskip
    \multiply\graphtemp by 1
    \divide\graphtemp by 2
    \advance\graphtemp by .5ex
    \advance\graphtemp by 0.375in
    \rlap{\kern 1.625in\lower\graphtemp\hbox to 0pt{\hss nets\hss}}%
\pdfliteral{
q [] 0 d 1 J 1 j
0.576 w
0.072 w
q 0 g
96.552 -54.432 m
102.024 -49.464 l
99.504 -56.448 l
96.552 -54.432 l
B Q
0.576 w
41.976 -139.536 m
98.064 -55.44 l
S
0.072 w
q 0 g
85.608 -40.68 m
92.88 -39.096 l
87.192 -43.92 l
85.608 -40.68 l
B Q
0.576 w
51.12 -59.904 m
86.4 -42.264 l
S
234 -162 m
234 -176.911688 221.911688 -189 207 -189 c
192.088312 -189 180 -176.911688 180 -162 c
180 -147.088312 192.088312 -135 207 -135 c
221.911688 -135 234 -147.088312 234 -162 c
S
Q
}%
    \graphtemp=.5ex
    \advance\graphtemp by 2.250in
    \rlap{\kern 2.875in\lower\graphtemp\hbox to 0pt{\hss LTSs\hss}}%
\pdfliteral{
q [] 0 d 1 J 1 j
0.576 w
0.072 w
q 0 g
172.8 -160.2 m
180 -162 l
172.8 -163.8 l
172.8 -160.2 l
B Q
0.576 w
54 -162 m
172.8 -162 l
S
0.072 w
q 0 g
177.192 -145.08 m
182.88 -149.904 l
175.608 -148.32 l
177.192 -145.08 l
B Q
0.576 w
51.12 -84.096 m
176.4 -146.736 l
S
0.072 w
q 0 g
189.504 -132.552 m
192.024 -139.536 l
186.552 -134.568 l
189.504 -132.552 l
B Q
0.576 w
131.976 -49.464 m
188.064 -133.56 l
S
234 -72 m
234 -86.911688 221.911688 -99 207 -99 c
192.088312 -99 180 -86.911688 180 -72 c
180 -57.088312 192.088312 -45 207 -45 c
221.911688 -45 234 -57.088312 234 -72 c
S
Q
}%
    \graphtemp=\baselineskip
    \multiply\graphtemp by -1
    \divide\graphtemp by 2
    \advance\graphtemp by .5ex
    \advance\graphtemp by 1.000in
    \rlap{\kern 2.875in\lower\graphtemp\hbox to 0pt{\hss LTSs\hss}}%
    \graphtemp=\baselineskip
    \multiply\graphtemp by 1
    \divide\graphtemp by 2
    \advance\graphtemp by .5ex
    \advance\graphtemp by 1.000in
    \rlap{\kern 2.875in\lower\graphtemp\hbox to 0pt{\hss $\begin{array}{c}+\\[-5pt] \aconc\end{array}$\hss}}%
\pdfliteral{
q [] 0 d 1 J 1 j
0.576 w
0.072 w
q 0 g
175.608 -85.68 m
182.88 -84.096 l
177.192 -88.92 l
175.608 -85.68 l
B Q
0.576 w
51.12 -149.904 m
176.4 -87.264 l
S
0.072 w
q 0 g
172.8 -70.2 m
180 -72 l
172.8 -73.8 l
172.8 -70.2 l
B Q
0.576 w
54 -72 m
172.8 -72 l
S
0.072 w
q 0 g
177.192 -55.08 m
182.88 -59.904 l
175.608 -58.32 l
177.192 -55.08 l
B Q
0.576 w
141.12 -39.096 m
176.4 -56.736 l
S
0.072 w
q 0 g
208.8 -127.8 m
207 -135 l
205.2 -127.8 l
208.8 -127.8 l
B Q
0.576 w
207 -127.8 m
207 -99 l
S
324 -162 m
324 -176.911688 311.911688 -189 297 -189 c
282.088312 -189 270 -176.911688 270 -162 c
270 -147.088312 282.088312 -135 297 -135 c
311.911688 -135 324 -147.088312 324 -162 c
S
Q
}%
    \graphtemp=\baselineskip
    \multiply\graphtemp by -1
    \divide\graphtemp by 2
    \advance\graphtemp by .5ex
    \advance\graphtemp by 2.250in
    \rlap{\kern 4.125in\lower\graphtemp\hbox to 0pt{\hss Kripke\hss}}%
    \graphtemp=\baselineskip
    \multiply\graphtemp by 1
    \divide\graphtemp by 2
    \advance\graphtemp by .5ex
    \advance\graphtemp by 2.250in
    \rlap{\kern 4.125in\lower\graphtemp\hbox to 0pt{\hss structures\hss}}%
\pdfliteral{
q [] 0 d 1 J 1 j
0.576 w
0.072 w
q 0 g
262.8 -160.2 m
270 -162 l
262.8 -163.8 l
262.8 -160.2 l
B Q
0.576 w
234 -162 m
262.8 -162 l
S
Q
}%
    \graphtemp=.5ex
    \advance\graphtemp by 2.250in
    \rlap{\kern 4.875in\lower\graphtemp\hbox to 0pt{\hss $\models$ LTL\hss}}%
    \hbox{\vrule depth2.625in width0pt height 0pt}%
    \kern 4.875in
  }%
}%

%% file: Gatekeeper1.tex
\expandafter\ifx\csname graph\endcsname\relax
   \csname newbox\expandafter\endcsname\csname graph\endcsname
\fi
\ifx\graphtemp\undefined
  \csname newdimen\endcsname\graphtemp
\fi
\expandafter\setbox\csname graph\endcsname
 =\vtop{\vskip 0pt\hbox{%
\pdfliteral{
q [] 0 d 1 J 1 j
0.576 w
0.576 w
15.984 -10.584 m
15.984 -13.566338 13.566338 -15.984 10.584 -15.984 c
7.601662 -15.984 5.184 -13.566338 5.184 -10.584 c
5.184 -7.601662 7.601662 -5.184 10.584 -5.184 c
13.566338 -5.184 15.984 -7.601662 15.984 -10.584 c
S
Q
}%
    \graphtemp=.5ex
    \advance\graphtemp by 0.147in
    \rlap{\kern 0.147in\lower\graphtemp\hbox to 0pt{\hss {$\scriptscriptstyle X$}\hss}}%
\pdfliteral{
q [] 0 d 1 J 1 j
0.576 w
0.072 w
q 0 g
2.952 -0.432 m
6.768 -6.768 l
0.432 -2.952 l
2.952 -0.432 l
B Q
0.576 w
0 0 m
1.656 -1.656 l
S
51.984 -10.584 m
51.984 -13.566338 49.566338 -15.984 46.584 -15.984 c
43.601662 -15.984 41.184 -13.566338 41.184 -10.584 c
41.184 -7.601662 43.601662 -5.184 46.584 -5.184 c
49.566338 -5.184 51.984 -7.601662 51.984 -10.584 c
S
Q
}%
    \graphtemp=.5ex
    \advance\graphtemp by 0.147in
    \rlap{\kern 0.647in\lower\graphtemp\hbox to 0pt{\hss {$\scriptscriptstyle Y$}\hss}}%
\pdfliteral{
q [] 0 d 1 J 1 j
0.576 w
0.072 w
q 0 g
33.984 -8.784 m
41.184 -10.584 l
33.984 -12.384 l
33.984 -8.784 l
B Q
0.576 w
15.984 -10.584 m
33.984 -10.584 l
S
Q
}%
    \graphtemp=\baselineskip
    \multiply\graphtemp by -1
    \divide\graphtemp by 2
    \advance\graphtemp by .5ex
    \advance\graphtemp by 0.147in
    \rlap{\kern 0.397in\lower\graphtemp\hbox to 0pt{\hss $r_1$\hss}}%
\pdfliteral{
q [] 0 d 1 J 1 j
0.576 w
87.984 -10.584 m
87.984 -13.566338 85.566338 -15.984 82.584 -15.984 c
79.601662 -15.984 77.184 -13.566338 77.184 -10.584 c
77.184 -7.601662 79.601662 -5.184 82.584 -5.184 c
85.566338 -5.184 87.984 -7.601662 87.984 -10.584 c
S
0.072 w
q 0 g
69.984 -8.784 m
77.184 -10.584 l
69.984 -12.384 l
69.984 -8.784 l
B Q
0.576 w
51.984 -10.584 m
69.984 -10.584 l
S
Q
}%
    \graphtemp=\baselineskip
    \multiply\graphtemp by -1
    \divide\graphtemp by 2
    \advance\graphtemp by .5ex
    \advance\graphtemp by 0.147in
    \rlap{\kern 0.897in\lower\graphtemp\hbox to 0pt{\hss $t_1$\hss}}%
\pdfliteral{
q [] 0 d 1 J 1 j
0.576 w
116.784 -10.584 m
116.784 -13.566338 114.366338 -15.984 111.384 -15.984 c
108.401662 -15.984 105.984 -13.566338 105.984 -10.584 c
105.984 -7.601662 108.401662 -5.184 111.384 -5.184 c
114.366338 -5.184 116.784 -7.601662 116.784 -10.584 c
h q 0.9 g
f Q
0.072 w
q 0 g
98.784 -8.784 m
105.984 -10.584 l
98.784 -12.384 l
98.784 -8.784 l
B Q
0.576 w
87.984 -10.584 m
98.784 -10.584 l
S
Q
}%
    \graphtemp=\baselineskip
    \multiply\graphtemp by -1
    \divide\graphtemp by 2
    \advance\graphtemp by .5ex
    \advance\graphtemp by 0.147in
    \rlap{\kern 1.347in\lower\graphtemp\hbox to 0pt{\hss $e$\hss}}%
\pdfliteral{
q [] 0 d 1 J 1 j
0.576 w
15.984 -46.584 m
15.984 -49.566338 13.566338 -51.984 10.584 -51.984 c
7.601662 -51.984 5.184 -49.566338 5.184 -46.584 c
5.184 -43.601662 7.601662 -41.184 10.584 -41.184 c
13.566338 -41.184 15.984 -43.601662 15.984 -46.584 c
S
Q
}%
    \graphtemp=.5ex
    \advance\graphtemp by 0.647in
    \rlap{\kern 0.147in\lower\graphtemp\hbox to 0pt{\hss {$\scriptscriptstyle Z$}\hss}}%
    \graphtemp=.5ex
    \advance\graphtemp by 0.547in
    \rlap{\kern 0.397in\lower\graphtemp\hbox to 0pt{\hss $r_1$\hss}}%
    \graphtemp=.5ex
    \advance\graphtemp by 0.397in
    \rlap{\kern 0.532in\lower\graphtemp\hbox to 0pt{\hss $r_2$\hss}}%
    \graphtemp=.5ex
    \advance\graphtemp by 0.747in
    \rlap{\kern 0.897in\lower\graphtemp\hbox to 0pt{\hss $t_1$\hss}}%
    \graphtemp=.5ex
    \advance\graphtemp by 0.897in
    \rlap{\kern 0.747in\lower\graphtemp\hbox to 0pt{\hss $t_2$\hss}}%
\pdfliteral{
q [] 0 d 1 J 1 j
0.576 w
0.072 w
q 0 g
12.384 -33.984 m
10.584 -41.184 l
8.784 -33.984 l
12.384 -33.984 l
B Q
0.576 w
10.584 -15.984 m
10.584 -33.984 l
S
Q
}%
    \graphtemp=.5ex
    \advance\graphtemp by 0.397in
    \rlap{\kern 0.147in\lower\graphtemp\hbox to 0pt{\hss $r_2\hspace{14pt}$\hss}}%
\pdfliteral{
q [] 0 d 1 J 1 j
0.576 w
48.024 -53.784 m
48.024 -57.283276 44.15574 -60.12 39.384 -60.12 c
34.61226 -60.12 30.744 -57.283276 30.744 -53.784 c
30.744 -50.284724 34.61226 -47.448 39.384 -47.448 c
44.15574 -47.448 48.024 -50.284724 48.024 -53.784 c
S
0.072 w
q 0 g
26.712 -45.792 m
33.264 -49.32 l
25.848 -49.32 l
26.712 -45.792 l
B Q
0.576 w
15.977797 -46.584689 m
19.438514 -46.60162 22.890664 -46.930409 26.292358 -47.567066 c
S
62.424 -39.384 m
62.424 -42.883276 58.55574 -45.72 53.784 -45.72 c
49.01226 -45.72 45.144 -42.883276 45.144 -39.384 c
45.144 -35.884724 49.01226 -33.048 53.784 -33.048 c
58.55574 -33.048 62.424 -35.884724 62.424 -39.384 c
S
0.072 w
q 0 g
48.312 -27.504 m
47.664 -34.92 l
44.712 -28.08 l
48.312 -27.504 l
B Q
0.576 w
46.500282 -27.810832 m
46.113084 -23.871828 46.132918 -19.903405 46.559469 -15.968469 c
S
91.224 -46.584 m
91.224 -50.083276 87.35574 -52.92 82.584 -52.92 c
77.81226 -52.92 73.944 -50.083276 73.944 -46.584 c
73.944 -43.084724 77.81226 -40.248 82.584 -40.248 c
87.35574 -40.248 91.224 -43.084724 91.224 -46.584 c
S
0.072 w
q 0 g
67.032 -43.848 m
73.944 -46.584 l
66.6 -47.448 l
67.032 -43.848 l
B Q
0.576 w
66.835981 -45.726622 m
64.495215 -45.268295 62.185178 -44.665015 59.919213 -43.920265 c
S
0.072 w
q 0 g
84.384 -33.048 m
82.584 -40.248 l
80.784 -33.048 l
84.384 -33.048 l
B Q
0.576 w
82.584 -15.984 m
82.584 -33.048 l
S
Q
}%
    \graphtemp=.5ex
    \advance\graphtemp by 0.391in
    \rlap{\kern 1.147in\lower\graphtemp\hbox to 0pt{\hss $r_2\hspace{14pt}$\hss}}%
\pdfliteral{
q [] 0 d 1 J 1 j
0.576 w
116.784 -46.584 m
116.784 -49.566338 114.366338 -51.984 111.384 -51.984 c
108.401662 -51.984 105.984 -49.566338 105.984 -46.584 c
105.984 -43.601662 108.401662 -41.184 111.384 -41.184 c
114.366338 -41.184 116.784 -43.601662 116.784 -46.584 c
h q 0.9 g
f Q
0.072 w
q 0 g
98.784 -44.784 m
105.984 -46.584 l
98.784 -48.384 l
98.784 -44.784 l
B Q
0.576 w
91.224 -46.584 m
98.784 -46.584 l
S
Q
}%
    \graphtemp=\baselineskip
    \multiply\graphtemp by 1
    \divide\graphtemp by 2
    \advance\graphtemp by .5ex
    \advance\graphtemp by 0.647in
    \rlap{\kern 1.369in\lower\graphtemp\hbox to 0pt{\hss $e$\hss}}%
\pdfliteral{
q [] 0 d 1 J 1 j
0.576 w
15.984 -82.584 m
15.984 -85.566338 13.566338 -87.984 10.584 -87.984 c
7.601662 -87.984 5.184 -85.566338 5.184 -82.584 c
5.184 -79.601662 7.601662 -77.184 10.584 -77.184 c
13.566338 -77.184 15.984 -79.601662 15.984 -82.584 c
S
0.072 w
q 0 g
12.384 -69.984 m
10.584 -77.184 l
8.784 -69.984 l
12.384 -69.984 l
B Q
0.576 w
10.584 -51.984 m
10.584 -69.984 l
S
Q
}%
    \graphtemp=.5ex
    \advance\graphtemp by 0.897in
    \rlap{\kern 0.147in\lower\graphtemp\hbox to 0pt{\hss $t_2\hspace{14pt}$\hss}}%
\pdfliteral{
q [] 0 d 1 J 1 j
0.576 w
55.224 -82.584 m
55.224 -86.083276 51.35574 -88.92 46.584 -88.92 c
41.81226 -88.92 37.944 -86.083276 37.944 -82.584 c
37.944 -79.084724 41.81226 -76.248 46.584 -76.248 c
51.35574 -76.248 55.224 -79.084724 55.224 -82.584 c
S
0.072 w
q 0 g
30.744 -80.784 m
37.944 -82.584 l
30.744 -84.384 l
30.744 -80.784 l
B Q
0.576 w
15.984 -82.584 m
30.744 -82.584 l
S
Q
}%
    \graphtemp=\baselineskip
    \multiply\graphtemp by -1
    \divide\graphtemp by 2
    \advance\graphtemp by .5ex
    \advance\graphtemp by 1.147in
    \rlap{\kern 0.374in\lower\graphtemp\hbox to 0pt{\hss $r_1$\hss}}%
\pdfliteral{
q [] 0 d 1 J 1 j
0.576 w
0.072 w
q 0 g
48.6 -69.12 m
46.584 -76.248 l
45 -69.048 l
48.6 -69.12 l
B Q
0.576 w
45.581431 -58.249628 m
46.371463 -61.817 46.820265 -65.451452 46.921769 -69.103847 c
S
15.984 -107.784 m
15.984 -110.766338 13.566338 -113.184 10.584 -113.184 c
7.601662 -113.184 5.184 -110.766338 5.184 -107.784 c
5.184 -104.801662 7.601662 -102.384 10.584 -102.384 c
13.566338 -102.384 15.984 -104.801662 15.984 -107.784 c
h q 0.9 g
f Q
0.072 w
q 0 g
12.384 -95.184 m
10.584 -102.384 l
8.784 -95.184 l
12.384 -95.184 l
B Q
0.576 w
10.584 -87.984 m
10.584 -95.184 l
S
Q
}%
    \graphtemp=.5ex
    \advance\graphtemp by 1.322in
    \rlap{\kern 0.147in\lower\graphtemp\hbox to 0pt{\hss $e\hspace{14pt}$\hss}}%
\pdfliteral{
q [] 0 d 1 J 1 j
0.576 w
51.984 -107.784 m
51.984 -110.766338 49.566338 -113.184 46.584 -113.184 c
43.601662 -113.184 41.184 -110.766338 41.184 -107.784 c
41.184 -104.801662 43.601662 -102.384 46.584 -102.384 c
49.566338 -102.384 51.984 -104.801662 51.984 -107.784 c
h q 0.9 g
f Q
0.072 w
q 0 g
48.384 -95.184 m
46.584 -102.384 l
44.784 -95.184 l
48.384 -95.184 l
B Q
0.576 w
46.584 -88.848 m
46.584 -95.184 l
S
Q
}%
    \graphtemp=.5ex
    \advance\graphtemp by 1.328in
    \rlap{\kern 0.647in\lower\graphtemp\hbox to 0pt{\hss $\hspace{14pt}e$\hss}}%
    \hbox{\vrule depth1.572in width0pt height 0pt}%
    \kern 1.622in
  }%
}%

%% file: gatekeeper2.tex
\expandafter\ifx\csname graph\endcsname\relax
   \csname newbox\expandafter\endcsname\csname graph\endcsname
\fi
\ifx\graphtemp\undefined
  \csname newdimen\endcsname\graphtemp
\fi
\expandafter\setbox\csname graph\endcsname
 =\vtop{\vskip 0pt\hbox{%
\pdfliteral{
q [] 0 d 1 J 1 j
0.576 w
0.576 w
15.984 -10.584 m
15.984 -13.566338 13.566338 -15.984 10.584 -15.984 c
7.601662 -15.984 5.184 -13.566338 5.184 -10.584 c
5.184 -7.601662 7.601662 -5.184 10.584 -5.184 c
13.566338 -5.184 15.984 -7.601662 15.984 -10.584 c
S
Q
}%
    \graphtemp=.5ex
    \advance\graphtemp by 0.147in
    \rlap{\kern 0.147in\lower\graphtemp\hbox to 0pt{\hss {$\scriptscriptstyle X$}\hss}}%
\pdfliteral{
q [] 0 d 1 J 1 j
0.576 w
0.072 w
q 0 g
2.952 -0.432 m
6.768 -6.768 l
0.432 -2.952 l
2.952 -0.432 l
B Q
0.576 w
0 0 m
1.656 -1.656 l
S
51.984 -10.584 m
51.984 -13.566338 49.566338 -15.984 46.584 -15.984 c
43.601662 -15.984 41.184 -13.566338 41.184 -10.584 c
41.184 -7.601662 43.601662 -5.184 46.584 -5.184 c
49.566338 -5.184 51.984 -7.601662 51.984 -10.584 c
S
Q
}%
    \graphtemp=.5ex
    \advance\graphtemp by 0.147in
    \rlap{\kern 0.647in\lower\graphtemp\hbox to 0pt{\hss {$\scriptscriptstyle Y$}\hss}}%
\pdfliteral{
q [] 0 d 1 J 1 j
0.576 w
0.072 w
q 0 g
33.984 -8.784 m
41.184 -10.584 l
33.984 -12.384 l
33.984 -8.784 l
B Q
0.576 w
15.984 -10.584 m
33.984 -10.584 l
S
Q
}%
    \graphtemp=\baselineskip
    \multiply\graphtemp by -1
    \divide\graphtemp by 2
    \advance\graphtemp by .5ex
    \advance\graphtemp by 0.147in
    \rlap{\kern 0.397in\lower\graphtemp\hbox to 0pt{\hss $\lni[1]$\hss}}%
\pdfliteral{
q [] 0 d 1 J 1 j
0.576 w
87.984 -10.584 m
87.984 -13.566338 85.566338 -15.984 82.584 -15.984 c
79.601662 -15.984 77.184 -13.566338 77.184 -10.584 c
77.184 -7.601662 79.601662 -5.184 82.584 -5.184 c
85.566338 -5.184 87.984 -7.601662 87.984 -10.584 c
S
0.072 w
q 0 g
69.984 -8.784 m
77.184 -10.584 l
69.984 -12.384 l
69.984 -8.784 l
B Q
0.576 w
51.984 -10.584 m
69.984 -10.584 l
S
Q
}%
    \graphtemp=\baselineskip
    \multiply\graphtemp by -1
    \divide\graphtemp by 2
    \advance\graphtemp by .5ex
    \advance\graphtemp by 0.147in
    \rlap{\kern 0.897in\lower\graphtemp\hbox to 0pt{\hss $\ec[1]$\hss}}%
\pdfliteral{
q [] 0 d 1 J 1 j
0.576 w
123.984 -10.584 m
123.984 -13.566338 121.566338 -15.984 118.584 -15.984 c
115.601662 -15.984 113.184 -13.566338 113.184 -10.584 c
113.184 -7.601662 115.601662 -5.184 118.584 -5.184 c
121.566338 -5.184 123.984 -7.601662 123.984 -10.584 c
S
0.072 w
q 0 g
105.984 -8.784 m
113.184 -10.584 l
105.984 -12.384 l
105.984 -8.784 l
B Q
0.576 w
87.984 -10.584 m
105.984 -10.584 l
S
Q
}%
    \graphtemp=\baselineskip
    \multiply\graphtemp by -1
    \divide\graphtemp by 2
    \advance\graphtemp by .5ex
    \advance\graphtemp by 0.147in
    \rlap{\kern 1.397in\lower\graphtemp\hbox to 0pt{\hss $\lc[1]$\hss}}%
\pdfliteral{
q [] 0 d 1 J 1 j
0.576 w
152.784 -10.584 m
152.784 -13.566338 150.366338 -15.984 147.384 -15.984 c
144.401662 -15.984 141.984 -13.566338 141.984 -10.584 c
141.984 -7.601662 144.401662 -5.184 147.384 -5.184 c
150.366338 -5.184 152.784 -7.601662 152.784 -10.584 c
h q 0.9 g
f Q
0.072 w
q 0 g
134.784 -8.784 m
141.984 -10.584 l
134.784 -12.384 l
134.784 -8.784 l
B Q
0.576 w
123.984 -10.584 m
134.784 -10.584 l
S
Q
}%
    \graphtemp=\baselineskip
    \multiply\graphtemp by -1
    \divide\graphtemp by 2
    \advance\graphtemp by .5ex
    \advance\graphtemp by 0.147in
    \rlap{\kern 1.847in\lower\graphtemp\hbox to 0pt{\hss $\en[1]$\hss}}%
\pdfliteral{
q [] 0 d 1 J 1 j
0.576 w
15.984 -46.584 m
15.984 -49.566338 13.566338 -51.984 10.584 -51.984 c
7.601662 -51.984 5.184 -49.566338 5.184 -46.584 c
5.184 -43.601662 7.601662 -41.184 10.584 -41.184 c
13.566338 -41.184 15.984 -43.601662 15.984 -46.584 c
S
Q
}%
    \graphtemp=.5ex
    \advance\graphtemp by 0.647in
    \rlap{\kern 0.147in\lower\graphtemp\hbox to 0pt{\hss {$\scriptscriptstyle Z$}\hss}}%
    \graphtemp=.5ex
    \advance\graphtemp by 0.547in
    \rlap{\kern 0.397in\lower\graphtemp\hbox to 0pt{\hss $\lni[1]$\hss}}%
    \graphtemp=.5ex
    \advance\graphtemp by 0.397in
    \rlap{\kern 0.797in\lower\graphtemp\hbox to 0pt{\hss $\lni[2]$\hss}}%
    \graphtemp=.5ex
    \advance\graphtemp by 0.747in
    \rlap{\kern 0.897in\lower\graphtemp\hbox to 0pt{\hss $\ec[1]$\hss}}%
    \graphtemp=.5ex
    \advance\graphtemp by 0.897in
    \rlap{\kern 0.797in\lower\graphtemp\hbox to 0pt{\hss $\ec[2]$\hss}}%
\pdfliteral{
q [] 0 d 1 J 1 j
0.576 w
0.072 w
q 0 g
12.384 -33.984 m
10.584 -41.184 l
8.784 -33.984 l
12.384 -33.984 l
B Q
0.576 w
10.584 -15.984 m
10.584 -33.984 l
S
Q
}%
    \graphtemp=.5ex
    \advance\graphtemp by 0.397in
    \rlap{\kern 0.147in\lower\graphtemp\hbox to 0pt{\hss $\hspace{18pt}\lni[2]$\hss}}%
\pdfliteral{
q [] 0 d 1 J 1 j
0.576 w
46.224 -53.784 m
46.224 -57.283276 42.35574 -60.12 37.584 -60.12 c
32.81226 -60.12 28.944 -57.283276 28.944 -53.784 c
28.944 -50.284724 32.81226 -47.448 37.584 -47.448 c
42.35574 -47.448 46.224 -50.284724 46.224 -53.784 c
S
0.072 w
q 0 g
24.912 -45.792 m
31.464 -49.32 l
24.048 -49.32 l
24.912 -45.792 l
B Q
0.576 w
15.963718 -46.554478 m
18.823366 -46.664364 21.670849 -46.987211 24.48251 -47.520339 c
S
62.424 -39.384 m
62.424 -42.883276 58.55574 -45.72 53.784 -45.72 c
49.01226 -45.72 45.144 -42.883276 45.144 -39.384 c
45.144 -35.884724 49.01226 -33.048 53.784 -33.048 c
58.55574 -33.048 62.424 -35.884724 62.424 -39.384 c
S
0.072 w
q 0 g
48.312 -27.504 m
47.664 -34.92 l
44.712 -28.08 l
48.312 -27.504 l
B Q
0.576 w
46.500282 -27.810832 m
46.113084 -23.871828 46.132918 -19.903405 46.559469 -15.968469 c
S
91.224 -46.584 m
91.224 -50.083276 87.35574 -52.92 82.584 -52.92 c
77.81226 -52.92 73.944 -50.083276 73.944 -46.584 c
73.944 -43.084724 77.81226 -40.248 82.584 -40.248 c
87.35574 -40.248 91.224 -43.084724 91.224 -46.584 c
S
0.072 w
q 0 g
67.032 -43.848 m
73.944 -46.584 l
66.6 -47.448 l
67.032 -43.848 l
B Q
0.576 w
66.835981 -45.726622 m
64.495215 -45.268295 62.185178 -44.665015 59.919213 -43.920265 c
S
0.072 w
q 0 g
84.384 -33.048 m
82.584 -40.248 l
80.784 -33.048 l
84.384 -33.048 l
B Q
0.576 w
82.584 -15.984 m
82.584 -33.048 l
S
Q
}%
    \graphtemp=.5ex
    \advance\graphtemp by 0.391in
    \rlap{\kern 1.147in\lower\graphtemp\hbox to 0pt{\hss $\hspace{18pt}\lni[2]$\hss}}%
\pdfliteral{
q [] 0 d 1 J 1 j
0.576 w
123.984 -46.584 m
123.984 -49.566338 121.566338 -51.984 118.584 -51.984 c
115.601662 -51.984 113.184 -49.566338 113.184 -46.584 c
113.184 -43.601662 115.601662 -41.184 118.584 -41.184 c
121.566338 -41.184 123.984 -43.601662 123.984 -46.584 c
S
0.072 w
q 0 g
105.984 -44.784 m
113.184 -46.584 l
105.984 -48.384 l
105.984 -44.784 l
B Q
0.576 w
91.224 -46.584 m
105.984 -46.584 l
S
Q
}%
    \graphtemp=\baselineskip
    \multiply\graphtemp by 1
    \divide\graphtemp by 2
    \advance\graphtemp by .5ex
    \advance\graphtemp by 0.647in
    \rlap{\kern 1.419in\lower\graphtemp\hbox to 0pt{\hss $\lc[1]$\hss}}%
\pdfliteral{
q [] 0 d 1 J 1 j
0.576 w
0.072 w
q 0 g
120.384 -33.984 m
118.584 -41.184 l
116.784 -33.984 l
120.384 -33.984 l
B Q
0.576 w
118.584 -15.984 m
118.584 -33.984 l
S
Q
}%
    \graphtemp=.5ex
    \advance\graphtemp by 0.397in
    \rlap{\kern 1.647in\lower\graphtemp\hbox to 0pt{\hss $\hspace{18pt}\lni[2]$\hss}}%
\pdfliteral{
q [] 0 d 1 J 1 j
0.576 w
152.784 -46.584 m
152.784 -49.566338 150.366338 -51.984 147.384 -51.984 c
144.401662 -51.984 141.984 -49.566338 141.984 -46.584 c
141.984 -43.601662 144.401662 -41.184 147.384 -41.184 c
150.366338 -41.184 152.784 -43.601662 152.784 -46.584 c
h q 0.9 g
f Q
0.072 w
q 0 g
134.784 -44.784 m
141.984 -46.584 l
134.784 -48.384 l
134.784 -44.784 l
B Q
0.576 w
123.984 -46.584 m
134.784 -46.584 l
S
Q
}%
    \graphtemp=\baselineskip
    \multiply\graphtemp by 1
    \divide\graphtemp by 2
    \advance\graphtemp by .5ex
    \advance\graphtemp by 0.647in
    \rlap{\kern 1.847in\lower\graphtemp\hbox to 0pt{\hss $\en[1]$\hss}}%
\pdfliteral{
q [] 0 d 1 J 1 j
0.576 w
15.984 -82.584 m
15.984 -85.566338 13.566338 -87.984 10.584 -87.984 c
7.601662 -87.984 5.184 -85.566338 5.184 -82.584 c
5.184 -79.601662 7.601662 -77.184 10.584 -77.184 c
13.566338 -77.184 15.984 -79.601662 15.984 -82.584 c
S
0.072 w
q 0 g
12.384 -69.984 m
10.584 -77.184 l
8.784 -69.984 l
12.384 -69.984 l
B Q
0.576 w
10.584 -51.984 m
10.584 -69.984 l
S
Q
}%
    \graphtemp=.5ex
    \advance\graphtemp by 0.897in
    \rlap{\kern 0.147in\lower\graphtemp\hbox to 0pt{\hss $\hspace{18pt}\ec[2]$\hss}}%
\pdfliteral{
q [] 0 d 1 J 1 j
0.576 w
55.224 -82.584 m
55.224 -86.083276 51.35574 -88.92 46.584 -88.92 c
41.81226 -88.92 37.944 -86.083276 37.944 -82.584 c
37.944 -79.084724 41.81226 -76.248 46.584 -76.248 c
51.35574 -76.248 55.224 -79.084724 55.224 -82.584 c
S
0.072 w
q 0 g
30.744 -80.784 m
37.944 -82.584 l
30.744 -84.384 l
30.744 -80.784 l
B Q
0.576 w
15.984 -82.584 m
30.744 -82.584 l
S
Q
}%
    \graphtemp=\baselineskip
    \multiply\graphtemp by 1
    \divide\graphtemp by 2
    \advance\graphtemp by .5ex
    \advance\graphtemp by 1.147in
    \rlap{\kern 0.374in\lower\graphtemp\hbox to 0pt{\hss $\lni[1]$\hss}}%
\pdfliteral{
q [] 0 d 1 J 1 j
0.576 w
0.072 w
q 0 g
47.952 -68.976 m
46.584 -76.248 l
44.352 -69.192 l
47.952 -68.976 l
B Q
0.576 w
43.765405 -58.230202 m
44.929213 -61.763999 45.747901 -65.402242 46.209997 -69.09394 c
S
15.984 -118.584 m
15.984 -121.566338 13.566338 -123.984 10.584 -123.984 c
7.601662 -123.984 5.184 -121.566338 5.184 -118.584 c
5.184 -115.601662 7.601662 -113.184 10.584 -113.184 c
13.566338 -113.184 15.984 -115.601662 15.984 -118.584 c
S
0.072 w
q 0 g
12.384 -105.984 m
10.584 -113.184 l
8.784 -105.984 l
12.384 -105.984 l
B Q
0.576 w
10.584 -87.984 m
10.584 -105.984 l
S
Q
}%
    \graphtemp=.5ex
    \advance\graphtemp by 1.397in
    \rlap{\kern 0.147in\lower\graphtemp\hbox to 0pt{\hss $\hspace{18pt}\lc[2]$\hss}}%
\pdfliteral{
q [] 0 d 1 J 1 j
0.576 w
51.984 -118.584 m
51.984 -121.566338 49.566338 -123.984 46.584 -123.984 c
43.601662 -123.984 41.184 -121.566338 41.184 -118.584 c
41.184 -115.601662 43.601662 -113.184 46.584 -113.184 c
49.566338 -113.184 51.984 -115.601662 51.984 -118.584 c
S
0.072 w
q 0 g
33.984 -116.784 m
41.184 -118.584 l
33.984 -120.384 l
33.984 -116.784 l
B Q
0.576 w
15.984 -118.584 m
33.984 -118.584 l
S
Q
}%
    \graphtemp=\baselineskip
    \multiply\graphtemp by 1
    \divide\graphtemp by 2
    \advance\graphtemp by .5ex
    \advance\graphtemp by 1.647in
    \rlap{\kern 0.397in\lower\graphtemp\hbox to 0pt{\hss $\lni[1]$\hss}}%
\pdfliteral{
q [] 0 d 1 J 1 j
0.576 w
0.072 w
q 0 g
48.384 -105.984 m
46.584 -113.184 l
44.784 -105.984 l
48.384 -105.984 l
B Q
0.576 w
46.584 -88.848 m
46.584 -105.984 l
S
Q
}%
    \graphtemp=.5ex
    \advance\graphtemp by 1.403in
    \rlap{\kern 0.647in\lower\graphtemp\hbox to 0pt{\hss $\hspace{18pt}\lc[2]$\hss}}%
\pdfliteral{
q [] 0 d 1 J 1 j
0.576 w
15.984 -149.184 m
15.984 -152.166338 13.566338 -154.584 10.584 -154.584 c
7.601662 -154.584 5.184 -152.166338 5.184 -149.184 c
5.184 -146.201662 7.601662 -143.784 10.584 -143.784 c
13.566338 -143.784 15.984 -146.201662 15.984 -149.184 c
h q 0.9 g
f Q
0.072 w
q 0 g
12.384 -136.584 m
10.584 -143.784 l
8.784 -136.584 l
12.384 -136.584 l
B Q
0.576 w
10.584 -123.984 m
10.584 -136.584 l
S
Q
}%
    \graphtemp=.5ex
    \advance\graphtemp by 1.859in
    \rlap{\kern 0.147in\lower\graphtemp\hbox to 0pt{\hss $\hspace{18pt}\en[2]$\hss}}%
\pdfliteral{
q [] 0 d 1 J 1 j
0.576 w
51.984 -149.184 m
51.984 -152.166338 49.566338 -154.584 46.584 -154.584 c
43.601662 -154.584 41.184 -152.166338 41.184 -149.184 c
41.184 -146.201662 43.601662 -143.784 46.584 -143.784 c
49.566338 -143.784 51.984 -146.201662 51.984 -149.184 c
h q 0.9 g
f Q
0.072 w
q 0 g
48.384 -136.584 m
46.584 -143.784 l
44.784 -136.584 l
48.384 -136.584 l
B Q
0.576 w
46.584 -123.984 m
46.584 -136.584 l
S
Q
}%
    \graphtemp=.5ex
    \advance\graphtemp by 1.859in
    \rlap{\kern 0.647in\lower\graphtemp\hbox to 0pt{\hss $\hspace{18pt}\en[2]$\hss}}%
    \hbox{\vrule depth2.147in width0pt height 0pt}%
    \kern 2.122in
  }%
}%

%% file: gatekeeper.tex
\expandafter\ifx\csname graph\endcsname\relax
   \csname newbox\expandafter\endcsname\csname graph\endcsname
\fi
\ifx\graphtemp\undefined
  \csname newdimen\endcsname\graphtemp
\fi
\expandafter\setbox\csname graph\endcsname
 =\vtop{\vskip 0pt\hbox{%
\pdfliteral{
q [] 0 d 1 J 1 j
0.576 w
0.576 w
18.792 -12.456 m
18.792 -15.955276 15.955276 -18.792 12.456 -18.792 c
8.956724 -18.792 6.12 -15.955276 6.12 -12.456 c
6.12 -8.956724 8.956724 -6.12 12.456 -6.12 c
15.955276 -6.12 18.792 -8.956724 18.792 -12.456 c
S
0.072 w
q 0 g
4.104 -1.584 m
7.992 -7.992 l
1.584 -4.104 l
4.104 -1.584 l
B Q
0.576 w
0 0 m
2.88 -2.88 l
S
61.128 -12.456 m
61.128 -15.955276 58.291276 -18.792 54.792 -18.792 c
51.292724 -18.792 48.456 -15.955276 48.456 -12.456 c
48.456 -8.956724 51.292724 -6.12 54.792 -6.12 c
58.291276 -6.12 61.128 -8.956724 61.128 -12.456 c
S
Q
}%
    \graphtemp=.5ex
    \advance\graphtemp by 0.173in
    \rlap{\kern 0.761in\lower\graphtemp\hbox to 0pt{\hss $I_2$\hss}}%
\pdfliteral{
q [] 0 d 1 J 1 j
0.576 w
0.072 w
q 0 g
41.256 -10.656 m
48.456 -12.456 l
41.256 -14.256 l
41.256 -10.656 l
B Q
0.576 w
18.792 -12.456 m
41.256 -12.456 l
S
103.464 -12.456 m
103.464 -15.955276 100.627276 -18.792 97.128 -18.792 c
93.628724 -18.792 90.792 -15.955276 90.792 -12.456 c
90.792 -8.956724 93.628724 -6.12 97.128 -6.12 c
100.627276 -6.12 103.464 -8.956724 103.464 -12.456 c
S
Q
}%
    \graphtemp=.5ex
    \advance\graphtemp by 0.173in
    \rlap{\kern 1.349in\lower\graphtemp\hbox to 0pt{\hss $C_2$\hss}}%
\pdfliteral{
q [] 0 d 1 J 1 j
0.576 w
0.072 w
q 0 g
83.592 -10.656 m
90.792 -12.456 l
83.592 -14.256 l
83.592 -10.656 l
B Q
0.576 w
61.128 -12.456 m
83.592 -12.456 l
S
137.376 -12.456 m
137.376 -15.955276 134.539276 -18.792 131.04 -18.792 c
127.540724 -18.792 124.704 -15.955276 124.704 -12.456 c
124.704 -8.956724 127.540724 -6.12 131.04 -6.12 c
134.539276 -6.12 137.376 -8.956724 137.376 -12.456 c
h q 0.9 g
f Q
0.072 w
q 0 g
117.504 -10.656 m
124.704 -12.456 l
117.504 -14.256 l
117.504 -10.656 l
B Q
0.576 w
103.536 -12.456 m
117.504 -12.456 l
S
18.792 -54.792 m
18.792 -58.291276 15.955276 -61.128 12.456 -61.128 c
8.956724 -61.128 6.12 -58.291276 6.12 -54.792 c
6.12 -51.292724 8.956724 -48.456 12.456 -48.456 c
15.955276 -48.456 18.792 -51.292724 18.792 -54.792 c
S
Q
}%
    \graphtemp=.5ex
    \advance\graphtemp by 0.761in
    \rlap{\kern 0.173in\lower\graphtemp\hbox to 0pt{\hss $I_1$\hss}}%
\pdfliteral{
q [] 0 d 1 J 1 j
0.576 w
0.072 w
q 0 g
14.256 -41.256 m
12.456 -48.456 l
10.656 -41.256 l
14.256 -41.256 l
B Q
0.576 w
12.456 -18.792 m
12.456 -41.256 l
S
56.448 -63.288 m
56.448 -67.383744 51.902795 -70.704 46.296 -70.704 c
40.689205 -70.704 36.144 -67.383744 36.144 -63.288 c
36.144 -59.192256 40.689205 -55.872 46.296 -55.872 c
51.902795 -55.872 56.448 -59.192256 56.448 -63.288 c
S
Q
}%
    \graphtemp=.5ex
    \advance\graphtemp by 0.879in
    \rlap{\kern 0.643in\lower\graphtemp\hbox to 0pt{\hss $I_1  I_2$\hss}}%
\pdfliteral{
q [] 0 d 1 J 1 j
0.576 w
0.072 w
q 0 g
32.616 -54.432 m
39.168 -58.032 l
31.752 -57.96 l
32.616 -54.432 l
B Q
0.576 w
18.835464 -54.792811 m
23.33053 -54.814802 27.81225 -55.284086 32.214371 -56.193723 c
S
73.44 -46.296 m
73.44 -50.391744 68.894795 -53.712 63.288 -53.712 c
57.681205 -53.712 53.136 -50.391744 53.136 -46.296 c
53.136 -42.200256 57.681205 -38.88 63.288 -38.88 c
68.894795 -38.88 73.44 -42.200256 73.44 -46.296 c
S
Q
}%
    \graphtemp=.5ex
    \advance\graphtemp by 0.643in
    \rlap{\kern 0.879in\lower\graphtemp\hbox to 0pt{\hss $I_1  I_2$\hss}}%
\pdfliteral{
q [] 0 d 1 J 1 j
0.576 w
0.072 w
q 0 g
56.664 -33.696 m
56.088 -41.112 l
53.064 -34.272 l
56.664 -33.696 l
B Q
0.576 w
54.893096 -34.026989 m
54.301404 -28.981299 54.279097 -23.885067 54.826595 -18.834392 c
S
107.28 -54.792 m
107.28 -58.887744 102.734795 -62.208 97.128 -62.208 c
91.521205 -62.208 86.976 -58.887744 86.976 -54.792 c
86.976 -50.696256 91.521205 -47.376 97.128 -47.376 c
102.734795 -47.376 107.28 -50.696256 107.28 -54.792 c
S
Q
}%
    \graphtemp=.5ex
    \advance\graphtemp by 0.761in
    \rlap{\kern 1.349in\lower\graphtemp\hbox to 0pt{\hss $I_1  C_2$\hss}}%
\pdfliteral{
q [] 0 d 1 J 1 j
0.576 w
0.072 w
q 0 g
80.064 -52.128 m
86.976 -54.792 l
79.632 -55.728 l
80.064 -52.128 l
B Q
0.576 w
79.814135 -53.977839 m
76.633708 -53.416805 73.497465 -52.629088 70.429395 -51.620712 c
S
0.072 w
q 0 g
98.928 -40.176 m
97.128 -47.376 l
95.328 -40.176 l
98.928 -40.176 l
B Q
0.576 w
97.128 -18.792 m
97.128 -40.176 l
S
137.376 -54.792 m
137.376 -58.291276 134.539276 -61.128 131.04 -61.128 c
127.540724 -61.128 124.704 -58.291276 124.704 -54.792 c
124.704 -51.292724 127.540724 -48.456 131.04 -48.456 c
134.539276 -48.456 137.376 -51.292724 137.376 -54.792 c
h q 0.9 g
f Q
0.072 w
q 0 g
117.504 -52.992 m
124.704 -54.792 l
117.504 -56.592 l
117.504 -52.992 l
B Q
0.576 w
107.352 -54.792 m
117.504 -54.792 l
S
18.792 -97.128 m
18.792 -100.627276 15.955276 -103.464 12.456 -103.464 c
8.956724 -103.464 6.12 -100.627276 6.12 -97.128 c
6.12 -93.628724 8.956724 -90.792 12.456 -90.792 c
15.955276 -90.792 18.792 -93.628724 18.792 -97.128 c
S
Q
}%
    \graphtemp=.5ex
    \advance\graphtemp by 1.349in
    \rlap{\kern 0.173in\lower\graphtemp\hbox to 0pt{\hss $C_1$\hss}}%
\pdfliteral{
q [] 0 d 1 J 1 j
0.576 w
0.072 w
q 0 g
14.256 -83.592 m
12.456 -90.792 l
10.656 -83.592 l
14.256 -83.592 l
B Q
0.576 w
12.456 -61.128 m
12.456 -83.592 l
S
64.944 -97.128 m
64.944 -101.223744 60.398795 -104.544 54.792 -104.544 c
49.185205 -104.544 44.64 -101.223744 44.64 -97.128 c
44.64 -93.032256 49.185205 -89.712 54.792 -89.712 c
60.398795 -89.712 64.944 -93.032256 64.944 -97.128 c
S
Q
}%
    \graphtemp=.5ex
    \advance\graphtemp by 1.349in
    \rlap{\kern 0.761in\lower\graphtemp\hbox to 0pt{\hss $C_1  I_2$\hss}}%
\pdfliteral{
q [] 0 d 1 J 1 j
0.576 w
0.072 w
q 0 g
37.44 -95.328 m
44.64 -97.128 l
37.44 -98.928 l
37.44 -95.328 l
B Q
0.576 w
18.792 -97.128 m
37.44 -97.128 l
S
0.072 w
q 0 g
56.952 -82.656 m
54.792 -89.712 l
53.352 -82.44 l
56.952 -82.656 l
B Q
0.576 w
53.549278 -68.510545 m
54.570341 -73.121131 55.106575 -77.825834 55.149235 -82.547937 c
S
18.792 -126.792 m
18.792 -130.291276 15.955276 -133.128 12.456 -133.128 c
8.956724 -133.128 6.12 -130.291276 6.12 -126.792 c
6.12 -123.292724 8.956724 -120.456 12.456 -120.456 c
15.955276 -120.456 18.792 -123.292724 18.792 -126.792 c
h q 0.9 g
f Q
0.072 w
q 0 g
14.256 -113.256 m
12.456 -120.456 l
10.656 -113.256 l
14.256 -113.256 l
B Q
0.576 w
12.456 -103.536 m
12.456 -113.256 l
S
61.128 -126.792 m
61.128 -130.291276 58.291276 -133.128 54.792 -133.128 c
51.292724 -133.128 48.456 -130.291276 48.456 -126.792 c
48.456 -123.292724 51.292724 -120.456 54.792 -120.456 c
58.291276 -120.456 61.128 -123.292724 61.128 -126.792 c
h q 0.9 g
f Q
0.072 w
q 0 g
56.592 -113.256 m
54.792 -120.456 l
52.992 -113.256 l
56.592 -113.256 l
B Q
0.576 w
54.792 -104.544 m
54.792 -113.256 l
S
Q
}%
    \hbox{\vrule depth1.849in width0pt height 0pt}%
    \kern 1.908in
  }%
}%

%% file: hierarchy.tex
\expandafter\ifx\csname graph\endcsname\relax
   \csname newbox\expandafter\endcsname\csname graph\endcsname
\fi
\ifx\graphtemp\undefined
  \csname newdimen\endcsname\graphtemp
\fi
\expandafter\setbox\csname graph\endcsname
 =\vtop{\vskip 0pt\hbox{%
\pdfliteral{
q [] 0 d 1 J 1 j
0.576 w
0.576 w
108.648 -146.376 m
108.648 -149.040222 106.488222 -151.2 103.824 -151.2 c
101.159778 -151.2 99 -149.040222 99 -146.376 c
99 -143.711778 101.159778 -141.552 103.824 -141.552 c
106.488222 -141.552 108.648 -143.711778 108.648 -146.376 c
h q 0.5 g
B Q
Q
}%
    \graphtemp=.5ex
    \advance\graphtemp by 2.433in
    \rlap{\kern 1.042in\lower\graphtemp\hbox to 0pt{\hss Request: {\it J}\hspace{13pt}\hss}}%
\pdfliteral{
q [] 0 d 1 J 1 j
0.576 w
0.072 w
q 0 g
90.648 -157.032 m
96.984 -153.216 l
93.168 -159.552 l
90.648 -157.032 l
B Q
0.576 w
81.792 -168.408 m
91.944 -158.256 l
S
Q
}%
    \graphtemp=.5ex
    \advance\graphtemp by 2.433in
    \rlap{\kern 1.842in\lower\graphtemp\hbox to 0pt{\hss \hspace{13pt}Granting: {\it Pr}\hss}}%
\pdfliteral{
q [] 0 d 1 J 1 j
0.576 w
0.072 w
q 0 g
114.408 -159.552 m
110.592 -153.216 l
116.928 -157.032 l
114.408 -159.552 l
B Q
0.576 w
125.784 -168.408 m
115.704 -158.256 l
S
Q
}%
    \graphtemp=.5ex
    \advance\graphtemp by 2.033in
    \rlap{\kern 2.475in\lower\graphtemp\hbox to 0pt{\hss {\color{blue}See Part IV of this paper}\hss}}%
\pdfliteral{
q [] 0 d 1 J 1 j
0.576 w
156.6 -98.424 m
156.6 -101.088222 154.440222 -103.248 151.776 -103.248 c
149.111778 -103.248 146.952 -101.088222 146.952 -98.424 c
146.952 -95.759778 149.111778 -93.6 151.776 -93.6 c
154.440222 -93.6 156.6 -95.759778 156.6 -98.424 c
h q 0.5 g
B Q
Q
}%
    \graphtemp=.5ex
    \advance\graphtemp by 1.767in
    \rlap{\kern 2.508in\lower\graphtemp\hbox to 0pt{\hss \hspace{13pt}Granting: {\it J}\hss}}%
\pdfliteral{
q [] 0 d 1 J 1 j
0.576 w
0.072 w
q 0 g
162.432 -111.528 m
158.616 -105.192 l
164.952 -109.008 l
162.432 -111.528 l
B Q
0.576 w
173.808 -120.384 m
163.656 -110.304 l
S
107.208 -142.992 m
148.392 -101.808 l
S
Q
}%
    \graphtemp=.5ex
    \advance\graphtemp by 1.367in
    \rlap{\kern 3.242in\lower\graphtemp\hbox to 0pt{\hss {\color{red}Peterson in CCS with signals}\hss}}%
\pdfliteral{
q [] 0 d 1 J 1 j
0.576 w
204.624 -50.4 m
204.624 -53.064222 202.464222 -55.224 199.8 -55.224 c
197.135778 -55.224 194.976 -53.064222 194.976 -50.4 c
194.976 -47.735778 197.135778 -45.576 199.8 -45.576 c
202.464222 -45.576 204.624 -47.735778 204.624 -50.4 c
h q 0.5 g
B Q
Q
}%
    \graphtemp=.5ex
    \advance\graphtemp by 1.100in
    \rlap{\kern 3.175in\lower\graphtemp\hbox to 0pt{\hss \hspace{13pt}Granting: {\it WF}\hss}}%
\pdfliteral{
q [] 0 d 1 J 1 j
0.576 w
0.072 w
q 0 g
210.384 -63.576 m
206.568 -57.168 l
212.976 -60.984 l
210.384 -63.576 l
B Q
0.576 w
221.832 -72.432 m
211.68 -62.28 l
S
155.16 -95.04 m
196.416 -53.784 l
S
Q
}%
    \graphtemp=.5ex
    \advance\graphtemp by 0.533in
    \rlap{\kern 3.775in\lower\graphtemp\hbox to 0pt{\hss {\color{red}Encapsulated gatekeeper}\hss}}%
    \graphtemp=.5ex
    \advance\graphtemp by 0.700in
    \rlap{\kern 3.508in\lower\graphtemp\hbox to 0pt{\hss {\color{red}Peterson in CCS}\hss}}%
\pdfliteral{
q [] 0 d 1 J 1 j
0.576 w
60.624 -98.424 m
60.624 -101.088222 58.464222 -103.248 55.8 -103.248 c
53.135778 -103.248 50.976 -101.088222 50.976 -98.424 c
50.976 -95.759778 53.135778 -93.6 55.8 -93.6 c
58.464222 -93.6 60.624 -95.759778 60.624 -98.424 c
h q 0.5 g
B Q
Q
}%
    \graphtemp=.5ex
    \advance\graphtemp by 1.767in
    \rlap{\kern 0.375in\lower\graphtemp\hbox to 0pt{\hss Request: {\it WF}\hss}}%
\pdfliteral{
q [] 0 d 1 J 1 j
0.576 w
0.072 w
q 0 g
42.624 -109.008 m
49.032 -105.192 l
45.216 -111.528 l
42.624 -109.008 l
B Q
0.576 w
33.768 -120.384 m
43.92 -110.304 l
S
Q
}%
    \graphtemp=.5ex
    \advance\graphtemp by 1.233in
    \rlap{\kern 0.308in\lower\graphtemp\hbox to 0pt{\hss {\color{red}Gatekeeper}\hss}}%
\pdfliteral{
q [] 0 d 1 J 1 j
0.576 w
100.44 -142.992 m
59.184 -101.808 l
S
106.2 -50.4 m
106.2 -51.712229 105.136229 -52.776 103.824 -52.776 c
102.511771 -52.776 101.448 -51.712229 101.448 -50.4 c
101.448 -49.087771 102.511771 -48.024 103.824 -48.024 c
105.136229 -48.024 106.2 -49.087771 106.2 -50.4 c
S
q [0 3.569808] 0 d
148.392 -95.04 m
105.48 -52.128 l
S Q
q [0 3.569808] 0 d
59.184 -95.04 m
102.096 -52.128 l
S Q
154.152 -2.376 m
154.152 -3.688229 153.088229 -4.752 151.776 -4.752 c
150.463771 -4.752 149.4 -3.688229 149.4 -2.376 c
149.4 -1.063771 150.463771 0 151.776 0 c
153.088229 0 154.152 -1.063771 154.152 -2.376 c
S
q [0 3.569808] 0 d
196.416 -47.016 m
153.504 -4.104 l
S Q
q [0 3.504422] 0 d
105.48 -48.672 m
150.12 -4.104 l
S Q
Q
}%
    \hbox{\vrule depth2.683in width0pt height 0pt}%
    \kern 3.775in
  }%
}%

%% file: Encapsulated.tex
\expandafter\ifx\csname graph\endcsname\relax
   \csname newbox\expandafter\endcsname\csname graph\endcsname
\fi
\ifx\graphtemp\undefined
  \csname newdimen\endcsname\graphtemp
\fi
\expandafter\setbox\csname graph\endcsname
 =\vtop{\vskip 0pt\hbox{%
\pdfliteral{
q [] 0 d 1 J 1 j
0.576 w
0.576 w
15.984 -10.584 m
15.984 -13.566338 13.566338 -15.984 10.584 -15.984 c
7.601662 -15.984 5.184 -13.566338 5.184 -10.584 c
5.184 -7.601662 7.601662 -5.184 10.584 -5.184 c
13.566338 -5.184 15.984 -7.601662 15.984 -10.584 c
S
Q
}%
    \graphtemp=.5ex
    \advance\graphtemp by 0.147in
    \rlap{\kern 0.147in\lower\graphtemp\hbox to 0pt{\hss {$\scriptscriptstyle X$}\hss}}%
\pdfliteral{
q [] 0 d 1 J 1 j
0.576 w
0.072 w
q 0 g
2.952 -0.432 m
6.768 -6.768 l
0.432 -2.952 l
2.952 -0.432 l
B Q
0.576 w
0 0 m
1.656 -1.656 l
S
51.984 -10.584 m
51.984 -13.566338 49.566338 -15.984 46.584 -15.984 c
43.601662 -15.984 41.184 -13.566338 41.184 -10.584 c
41.184 -7.601662 43.601662 -5.184 46.584 -5.184 c
49.566338 -5.184 51.984 -7.601662 51.984 -10.584 c
S
Q
}%
    \graphtemp=.5ex
    \advance\graphtemp by 0.147in
    \rlap{\kern 0.647in\lower\graphtemp\hbox to 0pt{\hss {$\scriptscriptstyle X$}\hss}}%
\pdfliteral{
q [] 0 d 1 J 1 j
0.576 w
0.072 w
q 0 g
33.984 -8.784 m
41.184 -10.584 l
33.984 -12.384 l
33.984 -8.784 l
B Q
0.576 w
15.984 -10.584 m
33.984 -10.584 l
S
Q
}%
    \graphtemp=\baselineskip
    \multiply\graphtemp by -1
    \divide\graphtemp by 2
    \advance\graphtemp by .5ex
    \advance\graphtemp by 0.147in
    \rlap{\kern 0.397in\lower\graphtemp\hbox to 0pt{\hss $\lni[1]$\hss}}%
\pdfliteral{
q [] 0 d 1 J 1 j
0.576 w
87.984 -10.584 m
87.984 -13.566338 85.566338 -15.984 82.584 -15.984 c
79.601662 -15.984 77.184 -13.566338 77.184 -10.584 c
77.184 -7.601662 79.601662 -5.184 82.584 -5.184 c
85.566338 -5.184 87.984 -7.601662 87.984 -10.584 c
S
Q
}%
    \graphtemp=.5ex
    \advance\graphtemp by 0.147in
    \rlap{\kern 1.147in\lower\graphtemp\hbox to 0pt{\hss {$\scriptscriptstyle Y$}\hss}}%
\pdfliteral{
q [] 0 d 1 J 1 j
0.576 w
0.072 w
q 0 g
69.984 -8.784 m
77.184 -10.584 l
69.984 -12.384 l
69.984 -8.784 l
B Q
0.576 w
51.984 -10.584 m
69.984 -10.584 l
S
Q
}%
    \graphtemp=\baselineskip
    \multiply\graphtemp by -1
    \divide\graphtemp by 2
    \advance\graphtemp by .5ex
    \advance\graphtemp by 0.147in
    \rlap{\kern 0.897in\lower\graphtemp\hbox to 0pt{\hss $\tau_{c1}$\hss}}%
\pdfliteral{
q [] 0 d 1 J 1 j
0.576 w
123.984 -10.584 m
123.984 -13.566338 121.566338 -15.984 118.584 -15.984 c
115.601662 -15.984 113.184 -13.566338 113.184 -10.584 c
113.184 -7.601662 115.601662 -5.184 118.584 -5.184 c
121.566338 -5.184 123.984 -7.601662 123.984 -10.584 c
S
0.072 w
q 0 g
105.984 -8.784 m
113.184 -10.584 l
105.984 -12.384 l
105.984 -8.784 l
B Q
0.576 w
87.984 -10.584 m
105.984 -10.584 l
S
Q
}%
    \graphtemp=\baselineskip
    \multiply\graphtemp by -1
    \divide\graphtemp by 2
    \advance\graphtemp by .5ex
    \advance\graphtemp by 0.147in
    \rlap{\kern 1.397in\lower\graphtemp\hbox to 0pt{\hss $\ec[1]$\hss}}%
\pdfliteral{
q [] 0 d 1 J 1 j
0.576 w
159.984 -10.584 m
159.984 -13.566338 157.566338 -15.984 154.584 -15.984 c
151.601662 -15.984 149.184 -13.566338 149.184 -10.584 c
149.184 -7.601662 151.601662 -5.184 154.584 -5.184 c
157.566338 -5.184 159.984 -7.601662 159.984 -10.584 c
S
0.072 w
q 0 g
141.984 -8.784 m
149.184 -10.584 l
141.984 -12.384 l
141.984 -8.784 l
B Q
0.576 w
123.984 -10.584 m
141.984 -10.584 l
S
Q
}%
    \graphtemp=\baselineskip
    \multiply\graphtemp by -1
    \divide\graphtemp by 2
    \advance\graphtemp by .5ex
    \advance\graphtemp by 0.147in
    \rlap{\kern 1.897in\lower\graphtemp\hbox to 0pt{\hss $\lc[1]$\hss}}%
\pdfliteral{
q [] 0 d 1 J 1 j
0.576 w
195.984 -10.584 m
195.984 -13.566338 193.566338 -15.984 190.584 -15.984 c
187.601662 -15.984 185.184 -13.566338 185.184 -10.584 c
185.184 -7.601662 187.601662 -5.184 190.584 -5.184 c
193.566338 -5.184 195.984 -7.601662 195.984 -10.584 c
S
0.072 w
q 0 g
177.984 -8.784 m
185.184 -10.584 l
177.984 -12.384 l
177.984 -8.784 l
B Q
0.576 w
159.984 -10.584 m
177.984 -10.584 l
S
Q
}%
    \graphtemp=\baselineskip
    \multiply\graphtemp by -1
    \divide\graphtemp by 2
    \advance\graphtemp by .5ex
    \advance\graphtemp by 0.147in
    \rlap{\kern 2.397in\lower\graphtemp\hbox to 0pt{\hss $\tau_{d1}$\hss}}%
\pdfliteral{
q [] 0 d 1 J 1 j
0.576 w
224.784 -10.584 m
224.784 -13.566338 222.366338 -15.984 219.384 -15.984 c
216.401662 -15.984 213.984 -13.566338 213.984 -10.584 c
213.984 -7.601662 216.401662 -5.184 219.384 -5.184 c
222.366338 -5.184 224.784 -7.601662 224.784 -10.584 c
h q 0.9 g
f Q
0.072 w
q 0 g
206.784 -8.784 m
213.984 -10.584 l
206.784 -12.384 l
206.784 -8.784 l
B Q
0.576 w
195.984 -10.584 m
206.784 -10.584 l
S
Q
}%
    \graphtemp=\baselineskip
    \multiply\graphtemp by -1
    \divide\graphtemp by 2
    \advance\graphtemp by .5ex
    \advance\graphtemp by 0.147in
    \rlap{\kern 2.847in\lower\graphtemp\hbox to 0pt{\hss $\en[1]$\hss}}%
\pdfliteral{
q [] 0 d 1 J 1 j
0.576 w
15.984 -46.584 m
15.984 -49.566338 13.566338 -51.984 10.584 -51.984 c
7.601662 -51.984 5.184 -49.566338 5.184 -46.584 c
5.184 -43.601662 7.601662 -41.184 10.584 -41.184 c
13.566338 -41.184 15.984 -43.601662 15.984 -46.584 c
S
Q
}%
    \graphtemp=.5ex
    \advance\graphtemp by 0.647in
    \rlap{\kern 0.147in\lower\graphtemp\hbox to 0pt{\hss {$\scriptscriptstyle X$}\hss}}%
\pdfliteral{
q [] 0 d 1 J 1 j
0.576 w
0.072 w
q 0 g
12.384 -33.984 m
10.584 -41.184 l
8.784 -33.984 l
12.384 -33.984 l
B Q
0.576 w
10.584 -15.984 m
10.584 -33.984 l
S
Q
}%
    \graphtemp=.5ex
    \advance\graphtemp by 0.397in
    \rlap{\kern 0.147in\lower\graphtemp\hbox to 0pt{\hss $\hspace{18pt}\lni[2]$\hss}}%
\pdfliteral{
q [] 0 d 1 J 1 j
0.576 w
15.984 -82.584 m
15.984 -85.566338 13.566338 -87.984 10.584 -87.984 c
7.601662 -87.984 5.184 -85.566338 5.184 -82.584 c
5.184 -79.601662 7.601662 -77.184 10.584 -77.184 c
13.566338 -77.184 15.984 -79.601662 15.984 -82.584 c
S
Q
}%
    \graphtemp=.5ex
    \advance\graphtemp by 1.147in
    \rlap{\kern 0.147in\lower\graphtemp\hbox to 0pt{\hss {$\scriptscriptstyle Z$}\hss}}%
\pdfliteral{
q [] 0 d 1 J 1 j
0.576 w
0.072 w
q 0 g
12.384 -69.984 m
10.584 -77.184 l
8.784 -69.984 l
12.384 -69.984 l
B Q
0.576 w
10.584 -51.984 m
10.584 -69.984 l
S
Q
}%
    \graphtemp=.5ex
    \advance\graphtemp by 0.897in
    \rlap{\kern 0.147in\lower\graphtemp\hbox to 0pt{\hss $\hspace{18pt}\tau_{c2}$\hss}}%
\pdfliteral{
q [] 0 d 1 J 1 j
0.576 w
15.984 -118.584 m
15.984 -121.566338 13.566338 -123.984 10.584 -123.984 c
7.601662 -123.984 5.184 -121.566338 5.184 -118.584 c
5.184 -115.601662 7.601662 -113.184 10.584 -113.184 c
13.566338 -113.184 15.984 -115.601662 15.984 -118.584 c
S
0.072 w
q 0 g
12.384 -105.984 m
10.584 -113.184 l
8.784 -105.984 l
12.384 -105.984 l
B Q
0.576 w
10.584 -87.984 m
10.584 -105.984 l
S
Q
}%
    \graphtemp=.5ex
    \advance\graphtemp by 1.397in
    \rlap{\kern 0.147in\lower\graphtemp\hbox to 0pt{\hss $\hspace{18pt}\ec[2]$\hss}}%
\pdfliteral{
q [] 0 d 1 J 1 j
0.576 w
15.984 -154.584 m
15.984 -157.566338 13.566338 -159.984 10.584 -159.984 c
7.601662 -159.984 5.184 -157.566338 5.184 -154.584 c
5.184 -151.601662 7.601662 -149.184 10.584 -149.184 c
13.566338 -149.184 15.984 -151.601662 15.984 -154.584 c
S
0.072 w
q 0 g
12.384 -141.984 m
10.584 -149.184 l
8.784 -141.984 l
12.384 -141.984 l
B Q
0.576 w
10.584 -123.984 m
10.584 -141.984 l
S
Q
}%
    \graphtemp=.5ex
    \advance\graphtemp by 1.897in
    \rlap{\kern 0.147in\lower\graphtemp\hbox to 0pt{\hss $\hspace{18pt}\lc[2]$\hss}}%
\pdfliteral{
q [] 0 d 1 J 1 j
0.576 w
15.984 -190.584 m
15.984 -193.566338 13.566338 -195.984 10.584 -195.984 c
7.601662 -195.984 5.184 -193.566338 5.184 -190.584 c
5.184 -187.601662 7.601662 -185.184 10.584 -185.184 c
13.566338 -185.184 15.984 -187.601662 15.984 -190.584 c
S
0.072 w
q 0 g
12.384 -177.984 m
10.584 -185.184 l
8.784 -177.984 l
12.384 -177.984 l
B Q
0.576 w
10.584 -159.984 m
10.584 -177.984 l
S
Q
}%
    \graphtemp=.5ex
    \advance\graphtemp by 2.397in
    \rlap{\kern 0.147in\lower\graphtemp\hbox to 0pt{\hss $\hspace{18pt}\tau_{d2}$\hss}}%
\pdfliteral{
q [] 0 d 1 J 1 j
0.576 w
15.984 -221.184 m
15.984 -224.166338 13.566338 -226.584 10.584 -226.584 c
7.601662 -226.584 5.184 -224.166338 5.184 -221.184 c
5.184 -218.201662 7.601662 -215.784 10.584 -215.784 c
13.566338 -215.784 15.984 -218.201662 15.984 -221.184 c
h q 0.9 g
f Q
0.072 w
q 0 g
12.384 -208.584 m
10.584 -215.784 l
8.784 -208.584 l
12.384 -208.584 l
B Q
0.576 w
10.584 -195.984 m
10.584 -208.584 l
S
Q
}%
    \graphtemp=.5ex
    \advance\graphtemp by 2.859in
    \rlap{\kern 0.147in\lower\graphtemp\hbox to 0pt{\hss $\hspace{18pt}\en[2]$\hss}}%
\pdfliteral{
q [] 0 d 1 J 1 j
0.576 w
51.984 -46.584 m
51.984 -49.566338 49.566338 -51.984 46.584 -51.984 c
43.601662 -51.984 41.184 -49.566338 41.184 -46.584 c
41.184 -43.601662 43.601662 -41.184 46.584 -41.184 c
49.566338 -41.184 51.984 -43.601662 51.984 -46.584 c
S
Q
}%
    \graphtemp=.5ex
    \advance\graphtemp by 0.647in
    \rlap{\kern 0.647in\lower\graphtemp\hbox to 0pt{\hss {$\scriptscriptstyle X$}\hss}}%
\pdfliteral{
q [] 0 d 1 J 1 j
0.576 w
0.072 w
q 0 g
48.384 -33.984 m
46.584 -41.184 l
44.784 -33.984 l
48.384 -33.984 l
B Q
0.576 w
46.584 -15.984 m
46.584 -33.984 l
S
Q
}%
    \graphtemp=.5ex
    \advance\graphtemp by 0.397in
    \rlap{\kern 0.647in\lower\graphtemp\hbox to 0pt{\hss $\hspace{18pt}\lni[2]$\hss}}%
\pdfliteral{
q [] 0 d 1 J 1 j
0.576 w
0.072 w
q 0 g
33.984 -44.784 m
41.184 -46.584 l
33.984 -48.384 l
33.984 -44.784 l
B Q
0.576 w
15.984 -46.584 m
33.984 -46.584 l
S
Q
}%
    \graphtemp=\baselineskip
    \multiply\graphtemp by -1
    \divide\graphtemp by 2
    \advance\graphtemp by .5ex
    \advance\graphtemp by 0.647in
    \rlap{\kern 0.397in\lower\graphtemp\hbox to 0pt{\hss $\lni[1]$\hss}}%
\pdfliteral{
q [] 0 d 1 J 1 j
0.576 w
87.984 -46.584 m
87.984 -49.566338 85.566338 -51.984 82.584 -51.984 c
79.601662 -51.984 77.184 -49.566338 77.184 -46.584 c
77.184 -43.601662 79.601662 -41.184 82.584 -41.184 c
85.566338 -41.184 87.984 -43.601662 87.984 -46.584 c
S
Q
}%
    \graphtemp=.5ex
    \advance\graphtemp by 0.647in
    \rlap{\kern 1.147in\lower\graphtemp\hbox to 0pt{\hss {$\scriptscriptstyle Y$}\hss}}%
\pdfliteral{
q [] 0 d 1 J 1 j
0.576 w
0.072 w
q 0 g
84.384 -33.984 m
82.584 -41.184 l
80.784 -33.984 l
84.384 -33.984 l
B Q
0.576 w
82.584 -15.984 m
82.584 -33.984 l
S
Q
}%
    \graphtemp=.5ex
    \advance\graphtemp by 0.397in
    \rlap{\kern 1.147in\lower\graphtemp\hbox to 0pt{\hss $\hspace{18pt}\lni[2]$\hss}}%
\pdfliteral{
q [] 0 d 1 J 1 j
0.576 w
0.072 w
q 0 g
69.984 -44.784 m
77.184 -46.584 l
69.984 -48.384 l
69.984 -44.784 l
B Q
0.576 w
51.984 -46.584 m
69.984 -46.584 l
S
Q
}%
    \graphtemp=\baselineskip
    \multiply\graphtemp by -1
    \divide\graphtemp by 2
    \advance\graphtemp by .5ex
    \advance\graphtemp by 0.647in
    \rlap{\kern 0.897in\lower\graphtemp\hbox to 0pt{\hss $\tau_{c1}$\hss}}%
\pdfliteral{
q [] 0 d 1 J 1 j
0.576 w
123.984 -46.584 m
123.984 -49.566338 121.566338 -51.984 118.584 -51.984 c
115.601662 -51.984 113.184 -49.566338 113.184 -46.584 c
113.184 -43.601662 115.601662 -41.184 118.584 -41.184 c
121.566338 -41.184 123.984 -43.601662 123.984 -46.584 c
S
0.072 w
q 0 g
120.384 -33.984 m
118.584 -41.184 l
116.784 -33.984 l
120.384 -33.984 l
B Q
0.576 w
118.584 -15.984 m
118.584 -33.984 l
S
Q
}%
    \graphtemp=.5ex
    \advance\graphtemp by 0.397in
    \rlap{\kern 1.647in\lower\graphtemp\hbox to 0pt{\hss $\hspace{18pt}\lni[2]$\hss}}%
\pdfliteral{
q [] 0 d 1 J 1 j
0.576 w
0.072 w
q 0 g
105.984 -44.784 m
113.184 -46.584 l
105.984 -48.384 l
105.984 -44.784 l
B Q
0.576 w
87.984 -46.584 m
105.984 -46.584 l
S
Q
}%
    \graphtemp=\baselineskip
    \multiply\graphtemp by -1
    \divide\graphtemp by 2
    \advance\graphtemp by .5ex
    \advance\graphtemp by 0.647in
    \rlap{\kern 1.397in\lower\graphtemp\hbox to 0pt{\hss $\ec[1]$\hss}}%
\pdfliteral{
q [] 0 d 1 J 1 j
0.576 w
159.984 -46.584 m
159.984 -49.566338 157.566338 -51.984 154.584 -51.984 c
151.601662 -51.984 149.184 -49.566338 149.184 -46.584 c
149.184 -43.601662 151.601662 -41.184 154.584 -41.184 c
157.566338 -41.184 159.984 -43.601662 159.984 -46.584 c
S
0.072 w
q 0 g
156.384 -33.984 m
154.584 -41.184 l
152.784 -33.984 l
156.384 -33.984 l
B Q
0.576 w
154.584 -15.984 m
154.584 -33.984 l
S
Q
}%
    \graphtemp=.5ex
    \advance\graphtemp by 0.397in
    \rlap{\kern 2.147in\lower\graphtemp\hbox to 0pt{\hss $\hspace{18pt}\lni[2]$\hss}}%
\pdfliteral{
q [] 0 d 1 J 1 j
0.576 w
0.072 w
q 0 g
141.984 -44.784 m
149.184 -46.584 l
141.984 -48.384 l
141.984 -44.784 l
B Q
0.576 w
123.984 -46.584 m
141.984 -46.584 l
S
Q
}%
    \graphtemp=\baselineskip
    \multiply\graphtemp by -1
    \divide\graphtemp by 2
    \advance\graphtemp by .5ex
    \advance\graphtemp by 0.647in
    \rlap{\kern 1.897in\lower\graphtemp\hbox to 0pt{\hss $\lc[1]$\hss}}%
\pdfliteral{
q [] 0 d 1 J 1 j
0.576 w
195.984 -46.584 m
195.984 -49.566338 193.566338 -51.984 190.584 -51.984 c
187.601662 -51.984 185.184 -49.566338 185.184 -46.584 c
185.184 -43.601662 187.601662 -41.184 190.584 -41.184 c
193.566338 -41.184 195.984 -43.601662 195.984 -46.584 c
S
0.072 w
q 0 g
192.384 -33.984 m
190.584 -41.184 l
188.784 -33.984 l
192.384 -33.984 l
B Q
0.576 w
190.584 -15.984 m
190.584 -33.984 l
S
Q
}%
    \graphtemp=.5ex
    \advance\graphtemp by 0.397in
    \rlap{\kern 2.647in\lower\graphtemp\hbox to 0pt{\hss $\hspace{18pt}\lni[2]$\hss}}%
\pdfliteral{
q [] 0 d 1 J 1 j
0.576 w
0.072 w
q 0 g
177.984 -44.784 m
185.184 -46.584 l
177.984 -48.384 l
177.984 -44.784 l
B Q
0.576 w
159.984 -46.584 m
177.984 -46.584 l
S
Q
}%
    \graphtemp=\baselineskip
    \multiply\graphtemp by -1
    \divide\graphtemp by 2
    \advance\graphtemp by .5ex
    \advance\graphtemp by 0.647in
    \rlap{\kern 2.397in\lower\graphtemp\hbox to 0pt{\hss $\tau_{d1}$\hss}}%
\pdfliteral{
q [] 0 d 1 J 1 j
0.576 w
224.784 -46.584 m
224.784 -49.566338 222.366338 -51.984 219.384 -51.984 c
216.401662 -51.984 213.984 -49.566338 213.984 -46.584 c
213.984 -43.601662 216.401662 -41.184 219.384 -41.184 c
222.366338 -41.184 224.784 -43.601662 224.784 -46.584 c
h q 0.9 g
f Q
0.072 w
q 0 g
206.784 -44.784 m
213.984 -46.584 l
206.784 -48.384 l
206.784 -44.784 l
B Q
0.576 w
195.984 -46.584 m
206.784 -46.584 l
S
Q
}%
    \graphtemp=\baselineskip
    \multiply\graphtemp by -1
    \divide\graphtemp by 2
    \advance\graphtemp by .5ex
    \advance\graphtemp by 0.647in
    \rlap{\kern 2.847in\lower\graphtemp\hbox to 0pt{\hss $\en[1]$\hss}}%
\pdfliteral{
q [] 0 d 1 J 1 j
0.576 w
51.984 -82.584 m
51.984 -85.566338 49.566338 -87.984 46.584 -87.984 c
43.601662 -87.984 41.184 -85.566338 41.184 -82.584 c
41.184 -79.601662 43.601662 -77.184 46.584 -77.184 c
49.566338 -77.184 51.984 -79.601662 51.984 -82.584 c
S
Q
}%
    \graphtemp=.5ex
    \advance\graphtemp by 1.147in
    \rlap{\kern 0.647in\lower\graphtemp\hbox to 0pt{\hss {$\scriptscriptstyle Z$}\hss}}%
\pdfliteral{
q [] 0 d 1 J 1 j
0.576 w
0.072 w
q 0 g
33.984 -80.784 m
41.184 -82.584 l
33.984 -84.384 l
33.984 -80.784 l
B Q
0.576 w
15.984 -82.584 m
33.984 -82.584 l
S
Q
}%
    \graphtemp=\baselineskip
    \multiply\graphtemp by -1
    \divide\graphtemp by 2
    \advance\graphtemp by .5ex
    \advance\graphtemp by 1.147in
    \rlap{\kern 0.397in\lower\graphtemp\hbox to 0pt{\hss $\lni[1]$\hss}}%
    \graphtemp=.5ex
    \advance\graphtemp by 1.072in
    \rlap{\kern 0.897in\lower\graphtemp\hbox to 0pt{\hss $\tau_{c1}$\hss}}%
    \graphtemp=.5ex
    \advance\graphtemp by 0.897in
    \rlap{\kern 1.272in\lower\graphtemp\hbox to 0pt{\hss $\tau_{c1}$\hss}}%
    \graphtemp=.5ex
    \advance\graphtemp by 1.247in
    \rlap{\kern 1.397in\lower\graphtemp\hbox to 0pt{\hss $\ec[1]$\hss}}%
    \graphtemp=.5ex
    \advance\graphtemp by 1.397in
    \rlap{\kern 1.297in\lower\graphtemp\hbox to 0pt{\hss $\ec[2]$\hss}}%
\pdfliteral{
q [] 0 d 1 J 1 j
0.576 w
0.072 w
q 0 g
48.384 -69.984 m
46.584 -77.184 l
44.784 -69.984 l
48.384 -69.984 l
B Q
0.576 w
46.584 -51.984 m
46.584 -69.984 l
S
Q
}%
    \graphtemp=.5ex
    \advance\graphtemp by 0.897in
    \rlap{\kern 0.647in\lower\graphtemp\hbox to 0pt{\hss $\hspace{18pt}\tau_{c2}$\hss}}%
\pdfliteral{
q [] 0 d 1 J 1 j
0.576 w
82.224 -89.784 m
82.224 -93.283276 78.35574 -96.12 73.584 -96.12 c
68.81226 -96.12 64.944 -93.283276 64.944 -89.784 c
64.944 -86.284724 68.81226 -83.448 73.584 -83.448 c
78.35574 -83.448 82.224 -86.284724 82.224 -89.784 c
S
0.072 w
q 0 g
60.912 -81.792 m
67.464 -85.32 l
60.048 -85.32 l
60.912 -81.792 l
B Q
0.576 w
51.963718 -82.554478 m
54.823366 -82.664364 57.670849 -82.987211 60.48251 -83.520339 c
S
98.424 -75.384 m
98.424 -78.883276 94.55574 -81.72 89.784 -81.72 c
85.01226 -81.72 81.144 -78.883276 81.144 -75.384 c
81.144 -71.884724 85.01226 -69.048 89.784 -69.048 c
94.55574 -69.048 98.424 -71.884724 98.424 -75.384 c
S
0.072 w
q 0 g
84.312 -63.504 m
83.664 -70.92 l
80.712 -64.08 l
84.312 -63.504 l
B Q
0.576 w
82.500282 -63.810832 m
82.113084 -59.871828 82.132918 -55.903405 82.559469 -51.968469 c
S
127.224 -82.584 m
127.224 -86.083276 123.35574 -88.92 118.584 -88.92 c
113.81226 -88.92 109.944 -86.083276 109.944 -82.584 c
109.944 -79.084724 113.81226 -76.248 118.584 -76.248 c
123.35574 -76.248 127.224 -79.084724 127.224 -82.584 c
S
0.072 w
q 0 g
103.032 -79.848 m
109.944 -82.584 l
102.6 -83.448 l
103.032 -79.848 l
B Q
0.576 w
102.835981 -81.654622 m
100.495215 -81.196295 98.185178 -80.593015 95.919213 -79.848265 c
S
0.072 w
q 0 g
120.384 -69.048 m
118.584 -76.248 l
116.784 -69.048 l
120.384 -69.048 l
B Q
0.576 w
118.584 -51.984 m
118.584 -69.048 l
S
Q
}%
    \graphtemp=.5ex
    \advance\graphtemp by 0.891in
    \rlap{\kern 1.647in\lower\graphtemp\hbox to 0pt{\hss $\hspace{18pt}\tau_{c2}$\hss}}%
\pdfliteral{
q [] 0 d 1 J 1 j
0.576 w
159.984 -82.584 m
159.984 -85.566338 157.566338 -87.984 154.584 -87.984 c
151.601662 -87.984 149.184 -85.566338 149.184 -82.584 c
149.184 -79.601662 151.601662 -77.184 154.584 -77.184 c
157.566338 -77.184 159.984 -79.601662 159.984 -82.584 c
S
0.072 w
q 0 g
141.984 -80.784 m
149.184 -82.584 l
141.984 -84.384 l
141.984 -80.784 l
B Q
0.576 w
127.224 -82.584 m
141.984 -82.584 l
S
Q
}%
    \graphtemp=\baselineskip
    \multiply\graphtemp by 1
    \divide\graphtemp by 2
    \advance\graphtemp by .5ex
    \advance\graphtemp by 1.147in
    \rlap{\kern 1.919in\lower\graphtemp\hbox to 0pt{\hss $\lc[1]$\hss}}%
\pdfliteral{
q [] 0 d 1 J 1 j
0.576 w
0.072 w
q 0 g
156.384 -69.984 m
154.584 -77.184 l
152.784 -69.984 l
156.384 -69.984 l
B Q
0.576 w
154.584 -51.984 m
154.584 -69.984 l
S
Q
}%
    \graphtemp=.5ex
    \advance\graphtemp by 0.897in
    \rlap{\kern 2.147in\lower\graphtemp\hbox to 0pt{\hss $\hspace{18pt}\tau_{c2}$\hss}}%
\pdfliteral{
q [] 0 d 1 J 1 j
0.576 w
195.984 -82.584 m
195.984 -85.566338 193.566338 -87.984 190.584 -87.984 c
187.601662 -87.984 185.184 -85.566338 185.184 -82.584 c
185.184 -79.601662 187.601662 -77.184 190.584 -77.184 c
193.566338 -77.184 195.984 -79.601662 195.984 -82.584 c
S
0.072 w
q 0 g
177.984 -80.784 m
185.184 -82.584 l
177.984 -84.384 l
177.984 -80.784 l
B Q
0.576 w
159.984 -82.584 m
177.984 -82.584 l
S
Q
}%
    \graphtemp=\baselineskip
    \multiply\graphtemp by 1
    \divide\graphtemp by 2
    \advance\graphtemp by .5ex
    \advance\graphtemp by 1.147in
    \rlap{\kern 2.397in\lower\graphtemp\hbox to 0pt{\hss $\tau_{d1}$\hss}}%
\pdfliteral{
q [] 0 d 1 J 1 j
0.576 w
0.072 w
q 0 g
192.384 -69.984 m
190.584 -77.184 l
188.784 -69.984 l
192.384 -69.984 l
B Q
0.576 w
190.584 -51.984 m
190.584 -69.984 l
S
Q
}%
    \graphtemp=.5ex
    \advance\graphtemp by 0.897in
    \rlap{\kern 2.647in\lower\graphtemp\hbox to 0pt{\hss $\hspace{18pt}\tau_{c2}$\hss}}%
\pdfliteral{
q [] 0 d 1 J 1 j
0.576 w
224.784 -82.584 m
224.784 -85.566338 222.366338 -87.984 219.384 -87.984 c
216.401662 -87.984 213.984 -85.566338 213.984 -82.584 c
213.984 -79.601662 216.401662 -77.184 219.384 -77.184 c
222.366338 -77.184 224.784 -79.601662 224.784 -82.584 c
h q 0.9 g
f Q
0.072 w
q 0 g
206.784 -80.784 m
213.984 -82.584 l
206.784 -84.384 l
206.784 -80.784 l
B Q
0.576 w
195.984 -82.584 m
206.784 -82.584 l
S
Q
}%
    \graphtemp=\baselineskip
    \multiply\graphtemp by 1
    \divide\graphtemp by 2
    \advance\graphtemp by .5ex
    \advance\graphtemp by 1.147in
    \rlap{\kern 2.847in\lower\graphtemp\hbox to 0pt{\hss $\en[1]$\hss}}%
\pdfliteral{
q [] 0 d 1 J 1 j
0.576 w
51.984 -118.584 m
51.984 -121.566338 49.566338 -123.984 46.584 -123.984 c
43.601662 -123.984 41.184 -121.566338 41.184 -118.584 c
41.184 -115.601662 43.601662 -113.184 46.584 -113.184 c
49.566338 -113.184 51.984 -115.601662 51.984 -118.584 c
S
0.072 w
q 0 g
33.984 -116.784 m
41.184 -118.584 l
33.984 -120.384 l
33.984 -116.784 l
B Q
0.576 w
15.984 -118.584 m
33.984 -118.584 l
S
Q
}%
    \graphtemp=\baselineskip
    \multiply\graphtemp by -1
    \divide\graphtemp by 2
    \advance\graphtemp by .5ex
    \advance\graphtemp by 1.647in
    \rlap{\kern 0.397in\lower\graphtemp\hbox to 0pt{\hss $\lni[1]$\hss}}%
\pdfliteral{
q [] 0 d 1 J 1 j
0.576 w
0.072 w
q 0 g
48.384 -105.984 m
46.584 -113.184 l
44.784 -105.984 l
48.384 -105.984 l
B Q
0.576 w
46.584 -87.984 m
46.584 -105.984 l
S
Q
}%
    \graphtemp=.5ex
    \advance\graphtemp by 1.397in
    \rlap{\kern 0.647in\lower\graphtemp\hbox to 0pt{\hss $\hspace{18pt}\ec[2]$\hss}}%
\pdfliteral{
q [] 0 d 1 J 1 j
0.576 w
91.224 -118.584 m
91.224 -122.083276 87.35574 -124.92 82.584 -124.92 c
77.81226 -124.92 73.944 -122.083276 73.944 -118.584 c
73.944 -115.084724 77.81226 -112.248 82.584 -112.248 c
87.35574 -112.248 91.224 -115.084724 91.224 -118.584 c
S
0.072 w
q 0 g
66.744 -116.784 m
73.944 -118.584 l
66.744 -120.384 l
66.744 -116.784 l
B Q
0.576 w
51.984 -118.584 m
66.744 -118.584 l
S
Q
}%
    \graphtemp=\baselineskip
    \multiply\graphtemp by -1
    \divide\graphtemp by 2
    \advance\graphtemp by .5ex
    \advance\graphtemp by 1.647in
    \rlap{\kern 0.874in\lower\graphtemp\hbox to 0pt{\hss $\tau_{c1}$\hss}}%
\pdfliteral{
q [] 0 d 1 J 1 j
0.576 w
0.072 w
q 0 g
83.952 -104.976 m
82.584 -112.248 l
80.352 -105.192 l
83.952 -104.976 l
B Q
0.576 w
79.693405 -94.230202 m
80.857213 -97.763999 81.675901 -101.402242 82.137997 -105.09394 c
S
51.984 -154.584 m
51.984 -157.566338 49.566338 -159.984 46.584 -159.984 c
43.601662 -159.984 41.184 -157.566338 41.184 -154.584 c
41.184 -151.601662 43.601662 -149.184 46.584 -149.184 c
49.566338 -149.184 51.984 -151.601662 51.984 -154.584 c
S
0.072 w
q 0 g
33.984 -152.784 m
41.184 -154.584 l
33.984 -156.384 l
33.984 -152.784 l
B Q
0.576 w
15.984 -154.584 m
33.984 -154.584 l
S
Q
}%
    \graphtemp=\baselineskip
    \multiply\graphtemp by -1
    \divide\graphtemp by 2
    \advance\graphtemp by .5ex
    \advance\graphtemp by 2.147in
    \rlap{\kern 0.397in\lower\graphtemp\hbox to 0pt{\hss $\lni[1]$\hss}}%
\pdfliteral{
q [] 0 d 1 J 1 j
0.576 w
0.072 w
q 0 g
48.384 -141.984 m
46.584 -149.184 l
44.784 -141.984 l
48.384 -141.984 l
B Q
0.576 w
46.584 -123.984 m
46.584 -141.984 l
S
Q
}%
    \graphtemp=.5ex
    \advance\graphtemp by 1.897in
    \rlap{\kern 0.647in\lower\graphtemp\hbox to 0pt{\hss $\hspace{18pt}\lc[2]$\hss}}%
\pdfliteral{
q [] 0 d 1 J 1 j
0.576 w
87.984 -154.584 m
87.984 -157.566338 85.566338 -159.984 82.584 -159.984 c
79.601662 -159.984 77.184 -157.566338 77.184 -154.584 c
77.184 -151.601662 79.601662 -149.184 82.584 -149.184 c
85.566338 -149.184 87.984 -151.601662 87.984 -154.584 c
S
0.072 w
q 0 g
69.984 -152.784 m
77.184 -154.584 l
69.984 -156.384 l
69.984 -152.784 l
B Q
0.576 w
51.984 -154.584 m
69.984 -154.584 l
S
Q
}%
    \graphtemp=\baselineskip
    \multiply\graphtemp by -1
    \divide\graphtemp by 2
    \advance\graphtemp by .5ex
    \advance\graphtemp by 2.147in
    \rlap{\kern 0.897in\lower\graphtemp\hbox to 0pt{\hss $\tau_{c1}$\hss}}%
\pdfliteral{
q [] 0 d 1 J 1 j
0.576 w
0.072 w
q 0 g
84.384 -141.984 m
82.584 -149.184 l
80.784 -141.984 l
84.384 -141.984 l
B Q
0.576 w
82.584 -124.848 m
82.584 -141.984 l
S
Q
}%
    \graphtemp=.5ex
    \advance\graphtemp by 1.903in
    \rlap{\kern 1.147in\lower\graphtemp\hbox to 0pt{\hss $\hspace{18pt}\lc[2]$\hss}}%
\pdfliteral{
q [] 0 d 1 J 1 j
0.576 w
51.984 -190.584 m
51.984 -193.566338 49.566338 -195.984 46.584 -195.984 c
43.601662 -195.984 41.184 -193.566338 41.184 -190.584 c
41.184 -187.601662 43.601662 -185.184 46.584 -185.184 c
49.566338 -185.184 51.984 -187.601662 51.984 -190.584 c
S
0.072 w
q 0 g
33.984 -188.784 m
41.184 -190.584 l
33.984 -192.384 l
33.984 -188.784 l
B Q
0.576 w
15.984 -190.584 m
33.984 -190.584 l
S
Q
}%
    \graphtemp=\baselineskip
    \multiply\graphtemp by -1
    \divide\graphtemp by 2
    \advance\graphtemp by .5ex
    \advance\graphtemp by 2.647in
    \rlap{\kern 0.397in\lower\graphtemp\hbox to 0pt{\hss $\lni[1]$\hss}}%
\pdfliteral{
q [] 0 d 1 J 1 j
0.576 w
0.072 w
q 0 g
48.384 -177.984 m
46.584 -185.184 l
44.784 -177.984 l
48.384 -177.984 l
B Q
0.576 w
46.584 -159.984 m
46.584 -177.984 l
S
Q
}%
    \graphtemp=.5ex
    \advance\graphtemp by 2.397in
    \rlap{\kern 0.647in\lower\graphtemp\hbox to 0pt{\hss $\hspace{18pt}\tau_{d2}$\hss}}%
\pdfliteral{
q [] 0 d 1 J 1 j
0.576 w
87.984 -190.584 m
87.984 -193.566338 85.566338 -195.984 82.584 -195.984 c
79.601662 -195.984 77.184 -193.566338 77.184 -190.584 c
77.184 -187.601662 79.601662 -185.184 82.584 -185.184 c
85.566338 -185.184 87.984 -187.601662 87.984 -190.584 c
S
0.072 w
q 0 g
69.984 -188.784 m
77.184 -190.584 l
69.984 -192.384 l
69.984 -188.784 l
B Q
0.576 w
51.984 -190.584 m
69.984 -190.584 l
S
Q
}%
    \graphtemp=\baselineskip
    \multiply\graphtemp by -1
    \divide\graphtemp by 2
    \advance\graphtemp by .5ex
    \advance\graphtemp by 2.647in
    \rlap{\kern 0.897in\lower\graphtemp\hbox to 0pt{\hss $\tau_{c1}$\hss}}%
\pdfliteral{
q [] 0 d 1 J 1 j
0.576 w
0.072 w
q 0 g
84.384 -177.984 m
82.584 -185.184 l
80.784 -177.984 l
84.384 -177.984 l
B Q
0.576 w
82.584 -159.984 m
82.584 -177.984 l
S
Q
}%
    \graphtemp=.5ex
    \advance\graphtemp by 2.397in
    \rlap{\kern 1.147in\lower\graphtemp\hbox to 0pt{\hss $\hspace{18pt}\tau_{d2}$\hss}}%
\pdfliteral{
q [] 0 d 1 J 1 j
0.576 w
51.984 -221.184 m
51.984 -224.166338 49.566338 -226.584 46.584 -226.584 c
43.601662 -226.584 41.184 -224.166338 41.184 -221.184 c
41.184 -218.201662 43.601662 -215.784 46.584 -215.784 c
49.566338 -215.784 51.984 -218.201662 51.984 -221.184 c
h q 0.9 g
f Q
0.072 w
q 0 g
48.384 -208.584 m
46.584 -215.784 l
44.784 -208.584 l
48.384 -208.584 l
B Q
0.576 w
46.584 -195.984 m
46.584 -208.584 l
S
Q
}%
    \graphtemp=.5ex
    \advance\graphtemp by 2.859in
    \rlap{\kern 0.647in\lower\graphtemp\hbox to 0pt{\hss $\hspace{18pt}\en[2]$\hss}}%
\pdfliteral{
q [] 0 d 1 J 1 j
0.576 w
87.984 -221.184 m
87.984 -224.166338 85.566338 -226.584 82.584 -226.584 c
79.601662 -226.584 77.184 -224.166338 77.184 -221.184 c
77.184 -218.201662 79.601662 -215.784 82.584 -215.784 c
85.566338 -215.784 87.984 -218.201662 87.984 -221.184 c
h q 0.9 g
f Q
0.072 w
q 0 g
84.384 -208.584 m
82.584 -215.784 l
80.784 -208.584 l
84.384 -208.584 l
B Q
0.576 w
82.584 -195.984 m
82.584 -208.584 l
S
Q
}%
    \graphtemp=.5ex
    \advance\graphtemp by 2.859in
    \rlap{\kern 1.147in\lower\graphtemp\hbox to 0pt{\hss $\hspace{18pt}\en[2]$\hss}}%
    \hbox{\vrule depth3.147in width0pt height 0pt}%
    \kern 3.122in
  }%
}%

%% file: PetersonPN.tex
\expandafter\ifx\csname graph\endcsname\relax
   \csname newbox\expandafter\endcsname\csname graph\endcsname
\fi
\ifx\graphtemp\undefined
  \csname newdimen\endcsname\graphtemp
\fi
\expandafter\setbox\csname graph\endcsname
 =\vtop{\vskip 0pt\hbox{%
\pdfliteral{
q [] 0 d 1 J 1 j
0.576 w
0.576 w
20.16 -10.08 m
20.16 -15.64703 15.64703 -20.16 10.08 -20.16 c
4.51297 -20.16 0 -15.64703 0 -10.08 c
0 -4.51297 4.51297 0 10.08 0 c
15.64703 0 20.16 -4.51297 20.16 -10.08 c
S
Q
}%
    \graphtemp=.5ex
    \advance\graphtemp by 0.140in
    \rlap{\kern 0.140in\lower\graphtemp\hbox to 0pt{\hss $\bullet$\hss}}%
\pdfliteral{
q [] 0 d 1 J 1 j
0.576 w
57.6 -20.16 m
77.76 -20.16 l
77.76 0 l
57.6 0 l
57.6 -20.16 l
S
Q
}%
    \graphtemp=.5ex
    \advance\graphtemp by 0.140in
    \rlap{\kern 0.940in\lower\graphtemp\hbox to 0pt{\hss $\lnA$\hss}}%
\pdfliteral{
q [] 0 d 1 J 1 j
0.576 w
0.072 w
q 0 g
50.4 -8.28 m
57.6 -10.08 l
50.4 -11.88 l
50.4 -8.28 l
B Q
0.576 w
20.16 -10.08 m
50.4 -10.08 l
S
135.36 -10.08 m
135.36 -15.64703 130.84703 -20.16 125.28 -20.16 c
119.71297 -20.16 115.2 -15.64703 115.2 -10.08 c
115.2 -4.51297 119.71297 0 125.28 0 c
130.84703 0 135.36 -4.51297 135.36 -10.08 c
S
Q
}%
    \graphtemp=.5ex
    \advance\graphtemp by 0.140in
    \rlap{\kern 1.740in\lower\graphtemp\hbox to 0pt{\hss $\ell_2$\hss}}%
\pdfliteral{
q [] 0 d 1 J 1 j
0.576 w
0.072 w
q 0 g
108 -8.28 m
115.2 -10.08 l
108 -11.88 l
108 -8.28 l
B Q
0.576 w
77.76 -10.08 m
108 -10.08 l
S
115.2 -56.16 m
135.36 -56.16 l
135.36 -36 l
115.2 -36 l
115.2 -56.16 l
S
Q
}%
    \graphtemp=.5ex
    \advance\graphtemp by 0.640in
    \rlap{\kern 1.740in\lower\graphtemp\hbox to 0pt{\hss $\ell_2$\hss}}%
\pdfliteral{
q [] 0 d 1 J 1 j
0.576 w
0.072 w
q 0 g
127.08 -28.8 m
125.28 -36 l
123.48 -28.8 l
127.08 -28.8 l
B Q
0.576 w
125.28 -20.16 m
125.28 -28.8 l
S
135.36 -82.08 m
135.36 -87.64703 130.84703 -92.16 125.28 -92.16 c
119.71297 -92.16 115.2 -87.64703 115.2 -82.08 c
115.2 -76.51297 119.71297 -72 125.28 -72 c
130.84703 -72 135.36 -76.51297 135.36 -82.08 c
S
Q
}%
    \graphtemp=.5ex
    \advance\graphtemp by 1.140in
    \rlap{\kern 1.740in\lower\graphtemp\hbox to 0pt{\hss $\ell_3$\hss}}%
\pdfliteral{
q [] 0 d 1 J 1 j
0.576 w
0.072 w
q 0 g
127.08 -64.8 m
125.28 -72 l
123.48 -64.8 l
127.08 -64.8 l
B Q
0.576 w
125.28 -56.16 m
125.28 -64.8 l
S
135.36 -154.08 m
135.36 -159.64703 130.84703 -164.16 125.28 -164.16 c
119.71297 -164.16 115.2 -159.64703 115.2 -154.08 c
115.2 -148.51297 119.71297 -144 125.28 -144 c
130.84703 -144 135.36 -148.51297 135.36 -154.08 c
S
Q
}%
    \graphtemp=.5ex
    \advance\graphtemp by 2.140in
    \rlap{\kern 1.740in\lower\graphtemp\hbox to 0pt{\hss $\ell_4$\hss}}%
\pdfliteral{
q [] 0 d 1 J 1 j
0.576 w
135.36 -226.08 m
135.36 -231.64703 130.84703 -236.16 125.28 -236.16 c
119.71297 -236.16 115.2 -231.64703 115.2 -226.08 c
115.2 -220.51297 119.71297 -216 125.28 -216 c
130.84703 -216 135.36 -220.51297 135.36 -226.08 c
S
Q
}%
    \graphtemp=.5ex
    \advance\graphtemp by 3.140in
    \rlap{\kern 1.740in\lower\graphtemp\hbox to 0pt{\hss $\ell_5$\hss}}%
\pdfliteral{
q [] 0 d 1 J 1 j
0.576 w
57.6 -236.16 m
77.76 -236.16 l
77.76 -216 l
57.6 -216 l
57.6 -236.16 l
S
Q
}%
    \graphtemp=.5ex
    \advance\graphtemp by 3.140in
    \rlap{\kern 0.940in\lower\graphtemp\hbox to 0pt{\hss $\ecA$\hss}}%
\pdfliteral{
q [] 0 d 1 J 1 j
0.576 w
0.072 w
q 0 g
84.96 -227.88 m
77.76 -226.08 l
84.96 -224.28 l
84.96 -227.88 l
B Q
0.576 w
115.2 -226.08 m
84.96 -226.08 l
S
20.16 -226.08 m
20.16 -231.64703 15.64703 -236.16 10.08 -236.16 c
4.51297 -236.16 0 -231.64703 0 -226.08 c
0 -220.51297 4.51297 -216 10.08 -216 c
15.64703 -216 20.16 -220.51297 20.16 -226.08 c
S
Q
}%
    \graphtemp=.5ex
    \advance\graphtemp by 3.140in
    \rlap{\kern 0.140in\lower\graphtemp\hbox to 0pt{\hss $\ell_6$\hss}}%
\pdfliteral{
q [] 0 d 1 J 1 j
0.576 w
0.072 w
q 0 g
27.36 -227.88 m
20.16 -226.08 l
27.36 -224.28 l
27.36 -227.88 l
B Q
0.576 w
57.6 -226.08 m
27.36 -226.08 l
S
0 -200.16 m
20.16 -200.16 l
20.16 -180 l
0 -180 l
0 -200.16 l
S
Q
}%
    \graphtemp=.5ex
    \advance\graphtemp by 2.640in
    \rlap{\kern 0.140in\lower\graphtemp\hbox to 0pt{\hss $\lcA$\hss}}%
\pdfliteral{
q [] 0 d 1 J 1 j
0.576 w
0.072 w
q 0 g
8.28 -207.36 m
10.08 -200.16 l
11.88 -207.36 l
8.28 -207.36 l
B Q
0.576 w
10.08 -216 m
10.08 -207.36 l
S
20.16 -154.08 m
20.16 -159.64703 15.64703 -164.16 10.08 -164.16 c
4.51297 -164.16 0 -159.64703 0 -154.08 c
0 -148.51297 4.51297 -144 10.08 -144 c
15.64703 -144 20.16 -148.51297 20.16 -154.08 c
S
Q
}%
    \graphtemp=.5ex
    \advance\graphtemp by 2.140in
    \rlap{\kern 0.140in\lower\graphtemp\hbox to 0pt{\hss $\ell_7$\hss}}%
\pdfliteral{
q [] 0 d 1 J 1 j
0.576 w
0.072 w
q 0 g
8.28 -171.36 m
10.08 -164.16 l
11.88 -171.36 l
8.28 -171.36 l
B Q
0.576 w
10.08 -180 m
10.08 -171.36 l
S
0 -128.16 m
20.16 -128.16 l
20.16 -108 l
0 -108 l
0 -128.16 l
S
Q
}%
    \graphtemp=.5ex
    \advance\graphtemp by 1.640in
    \rlap{\kern 0.140in\lower\graphtemp\hbox to 0pt{\hss $\ell_7$\hss}}%
\pdfliteral{
q [] 0 d 1 J 1 j
0.576 w
0.072 w
q 0 g
8.28 -135.36 m
10.08 -128.16 l
11.88 -135.36 l
8.28 -135.36 l
B Q
0.576 w
10.08 -144 m
10.08 -135.36 l
S
20.16 -82.08 m
20.16 -87.64703 15.64703 -92.16 10.08 -92.16 c
4.51297 -92.16 0 -87.64703 0 -82.08 c
0 -76.51297 4.51297 -72 10.08 -72 c
15.64703 -72 20.16 -76.51297 20.16 -82.08 c
S
Q
}%
    \graphtemp=.5ex
    \advance\graphtemp by 1.140in
    \rlap{\kern 0.140in\lower\graphtemp\hbox to 0pt{\hss $\ell_8$\hss}}%
\pdfliteral{
q [] 0 d 1 J 1 j
0.576 w
0.072 w
q 0 g
8.28 -99.36 m
10.08 -92.16 l
11.88 -99.36 l
8.28 -99.36 l
B Q
0.576 w
10.08 -108 m
10.08 -99.36 l
S
0 -56.16 m
20.16 -56.16 l
20.16 -36 l
0 -36 l
0 -56.16 l
S
Q
}%
    \graphtemp=.5ex
    \advance\graphtemp by 0.640in
    \rlap{\kern 0.140in\lower\graphtemp\hbox to 0pt{\hss $\enA$\hss}}%
\pdfliteral{
q [] 0 d 1 J 1 j
0.576 w
0.072 w
q 0 g
8.28 -63.36 m
10.08 -56.16 l
11.88 -63.36 l
8.28 -63.36 l
B Q
0.576 w
10.08 -72 m
10.08 -63.36 l
S
0.072 w
q 0 g
8.28 -27.36 m
10.08 -20.16 l
11.88 -27.36 l
8.28 -27.36 l
B Q
0.576 w
10.08 -36 m
10.08 -27.36 l
S
91.08 -46.08 m
91.08 -56.021125 80.603463 -64.08 67.68 -64.08 c
54.756537 -64.08 44.28 -56.021125 44.28 -46.08 c
44.28 -36.138875 54.756537 -28.08 67.68 -28.08 c
80.603463 -28.08 91.08 -36.138875 91.08 -46.08 c
S
Q
}%
    \graphtemp=\baselineskip
    \multiply\graphtemp by -1
    \divide\graphtemp by 2
    \advance\graphtemp by .5ex
    \advance\graphtemp by 0.640in
    \rlap{\kern 0.940in\lower\graphtemp\hbox to 0pt{\hss $\rA$\hss}}%
    \graphtemp=\baselineskip
    \multiply\graphtemp by 1
    \divide\graphtemp by 2
    \advance\graphtemp by .5ex
    \advance\graphtemp by 0.640in
    \rlap{\kern 0.940in\lower\graphtemp\hbox to 0pt{\hss $\tr$\hss}}%
\pdfliteral{
q [] 0 d 1 J 1 j
0.576 w
0.072 w
q 0 g
98.28 -47.88 m
91.08 -46.08 l
98.28 -44.28 l
98.28 -47.88 l
B Q
0.576 w
98.28 -46.08 m
115.2 -46.08 l
S
0.072 w
q 0 g
25.488 -102.888 m
20.16 -108 l
22.464 -100.944 l
25.488 -102.888 l
B Q
0.576 w
23.976 -101.88 m
51.12 -58.824 l
S
91.08 -118.08 m
91.08 -128.021125 80.603463 -136.08 67.68 -136.08 c
54.756537 -136.08 44.28 -128.021125 44.28 -118.08 c
44.28 -108.138875 54.756537 -100.08 67.68 -100.08 c
80.603463 -100.08 91.08 -108.138875 91.08 -118.08 c
S
Q
}%
    \graphtemp=\baselineskip
    \multiply\graphtemp by -1
    \divide\graphtemp by 2
    \advance\graphtemp by .5ex
    \advance\graphtemp by 1.640in
    \rlap{\kern 0.940in\lower\graphtemp\hbox to 0pt{\hss $\rA$\hss}}%
    \graphtemp=\baselineskip
    \multiply\graphtemp by 1
    \divide\graphtemp by 2
    \advance\graphtemp by .5ex
    \advance\graphtemp by 1.640in
    \rlap{\kern 0.940in\lower\graphtemp\hbox to 0pt{\hss $\hspace{6pt}\fa$\hss}}%
    \graphtemp=.5ex
    \advance\graphtemp by 1.640in
    \rlap{\kern 0.940in\lower\graphtemp\hbox to 0pt{\hss $\bullet$\hss}}%
\pdfliteral{
q [] 0 d 1 J 1 j
0.576 w
0.072 w
q 0 g
37.08 -116.28 m
44.28 -118.08 l
37.08 -119.88 l
37.08 -116.28 l
B Q
0.576 w
37.08 -118.08 m
20.16 -118.08 l
S
0.072 w
q 0 g
109.872 -61.272 m
115.2 -56.16 l
112.896 -63.216 l
109.872 -61.272 l
B Q
0.576 w
111.384 -62.28 m
84.24 -105.336 l
S
430.56 -10.08 m
430.56 -15.64703 426.04703 -20.16 420.48 -20.16 c
414.91297 -20.16 410.4 -15.64703 410.4 -10.08 c
410.4 -4.51297 414.91297 0 420.48 0 c
426.04703 0 430.56 -4.51297 430.56 -10.08 c
S
Q
}%
    \graphtemp=.5ex
    \advance\graphtemp by 0.140in
    \rlap{\kern 5.840in\lower\graphtemp\hbox to 0pt{\hss $\bullet$\hss}}%
\pdfliteral{
q [] 0 d 1 J 1 j
0.576 w
352.8 -20.16 m
372.96 -20.16 l
372.96 0 l
352.8 0 l
352.8 -20.16 l
S
Q
}%
    \graphtemp=.5ex
    \advance\graphtemp by 0.140in
    \rlap{\kern 5.040in\lower\graphtemp\hbox to 0pt{\hss $\lnB$\hss}}%
\pdfliteral{
q [] 0 d 1 J 1 j
0.576 w
0.072 w
q 0 g
380.16 -11.88 m
372.96 -10.08 l
380.16 -8.28 l
380.16 -11.88 l
B Q
0.576 w
410.4 -10.08 m
380.16 -10.08 l
S
315.36 -10.08 m
315.36 -15.64703 310.84703 -20.16 305.28 -20.16 c
299.71297 -20.16 295.2 -15.64703 295.2 -10.08 c
295.2 -4.51297 299.71297 0 305.28 0 c
310.84703 0 315.36 -4.51297 315.36 -10.08 c
S
Q
}%
    \graphtemp=.5ex
    \advance\graphtemp by 0.140in
    \rlap{\kern 4.240in\lower\graphtemp\hbox to 0pt{\hss $m_2$\hss}}%
\pdfliteral{
q [] 0 d 1 J 1 j
0.576 w
0.072 w
q 0 g
322.56 -11.88 m
315.36 -10.08 l
322.56 -8.28 l
322.56 -11.88 l
B Q
0.576 w
352.8 -10.08 m
322.56 -10.08 l
S
295.2 -56.16 m
315.36 -56.16 l
315.36 -36 l
295.2 -36 l
295.2 -56.16 l
S
Q
}%
    \graphtemp=.5ex
    \advance\graphtemp by 0.640in
    \rlap{\kern 4.240in\lower\graphtemp\hbox to 0pt{\hss $m_2$\hss}}%
\pdfliteral{
q [] 0 d 1 J 1 j
0.576 w
0.072 w
q 0 g
307.08 -28.8 m
305.28 -36 l
303.48 -28.8 l
307.08 -28.8 l
B Q
0.576 w
305.28 -20.16 m
305.28 -28.8 l
S
315.36 -82.08 m
315.36 -87.64703 310.84703 -92.16 305.28 -92.16 c
299.71297 -92.16 295.2 -87.64703 295.2 -82.08 c
295.2 -76.51297 299.71297 -72 305.28 -72 c
310.84703 -72 315.36 -76.51297 315.36 -82.08 c
S
Q
}%
    \graphtemp=.5ex
    \advance\graphtemp by 1.140in
    \rlap{\kern 4.240in\lower\graphtemp\hbox to 0pt{\hss $m_3$\hss}}%
\pdfliteral{
q [] 0 d 1 J 1 j
0.576 w
0.072 w
q 0 g
307.08 -64.8 m
305.28 -72 l
303.48 -64.8 l
307.08 -64.8 l
B Q
0.576 w
305.28 -56.16 m
305.28 -64.8 l
S
315.36 -154.08 m
315.36 -159.64703 310.84703 -164.16 305.28 -164.16 c
299.71297 -164.16 295.2 -159.64703 295.2 -154.08 c
295.2 -148.51297 299.71297 -144 305.28 -144 c
310.84703 -144 315.36 -148.51297 315.36 -154.08 c
S
Q
}%
    \graphtemp=.5ex
    \advance\graphtemp by 2.140in
    \rlap{\kern 4.240in\lower\graphtemp\hbox to 0pt{\hss $m_4$\hss}}%
\pdfliteral{
q [] 0 d 1 J 1 j
0.576 w
315.36 -226.08 m
315.36 -231.64703 310.84703 -236.16 305.28 -236.16 c
299.71297 -236.16 295.2 -231.64703 295.2 -226.08 c
295.2 -220.51297 299.71297 -216 305.28 -216 c
310.84703 -216 315.36 -220.51297 315.36 -226.08 c
S
Q
}%
    \graphtemp=.5ex
    \advance\graphtemp by 3.140in
    \rlap{\kern 4.240in\lower\graphtemp\hbox to 0pt{\hss $m_5$\hss}}%
\pdfliteral{
q [] 0 d 1 J 1 j
0.576 w
352.8 -236.16 m
372.96 -236.16 l
372.96 -216 l
352.8 -216 l
352.8 -236.16 l
S
Q
}%
    \graphtemp=.5ex
    \advance\graphtemp by 3.140in
    \rlap{\kern 5.040in\lower\graphtemp\hbox to 0pt{\hss $\ecB$\hss}}%
\pdfliteral{
q [] 0 d 1 J 1 j
0.576 w
0.072 w
q 0 g
345.6 -224.28 m
352.8 -226.08 l
345.6 -227.88 l
345.6 -224.28 l
B Q
0.576 w
315.36 -226.08 m
345.6 -226.08 l
S
430.56 -226.08 m
430.56 -231.64703 426.04703 -236.16 420.48 -236.16 c
414.91297 -236.16 410.4 -231.64703 410.4 -226.08 c
410.4 -220.51297 414.91297 -216 420.48 -216 c
426.04703 -216 430.56 -220.51297 430.56 -226.08 c
S
Q
}%
    \graphtemp=.5ex
    \advance\graphtemp by 3.140in
    \rlap{\kern 5.840in\lower\graphtemp\hbox to 0pt{\hss $m_6$\hss}}%
\pdfliteral{
q [] 0 d 1 J 1 j
0.576 w
0.072 w
q 0 g
403.2 -224.28 m
410.4 -226.08 l
403.2 -227.88 l
403.2 -224.28 l
B Q
0.576 w
372.96 -226.08 m
403.2 -226.08 l
S
410.4 -200.16 m
430.56 -200.16 l
430.56 -180 l
410.4 -180 l
410.4 -200.16 l
S
Q
}%
    \graphtemp=.5ex
    \advance\graphtemp by 2.640in
    \rlap{\kern 5.840in\lower\graphtemp\hbox to 0pt{\hss $\lcB$\hss}}%
\pdfliteral{
q [] 0 d 1 J 1 j
0.576 w
0.072 w
q 0 g
418.68 -207.36 m
420.48 -200.16 l
422.28 -207.36 l
418.68 -207.36 l
B Q
0.576 w
420.48 -216 m
420.48 -207.36 l
S
430.56 -154.08 m
430.56 -159.64703 426.04703 -164.16 420.48 -164.16 c
414.91297 -164.16 410.4 -159.64703 410.4 -154.08 c
410.4 -148.51297 414.91297 -144 420.48 -144 c
426.04703 -144 430.56 -148.51297 430.56 -154.08 c
S
Q
}%
    \graphtemp=.5ex
    \advance\graphtemp by 2.140in
    \rlap{\kern 5.840in\lower\graphtemp\hbox to 0pt{\hss $m_7$\hss}}%
\pdfliteral{
q [] 0 d 1 J 1 j
0.576 w
0.072 w
q 0 g
418.68 -171.36 m
420.48 -164.16 l
422.28 -171.36 l
418.68 -171.36 l
B Q
0.576 w
420.48 -180 m
420.48 -171.36 l
S
410.4 -128.16 m
430.56 -128.16 l
430.56 -108 l
410.4 -108 l
410.4 -128.16 l
S
Q
}%
    \graphtemp=.5ex
    \advance\graphtemp by 1.640in
    \rlap{\kern 5.840in\lower\graphtemp\hbox to 0pt{\hss $m_7$\hss}}%
\pdfliteral{
q [] 0 d 1 J 1 j
0.576 w
0.072 w
q 0 g
418.68 -135.36 m
420.48 -128.16 l
422.28 -135.36 l
418.68 -135.36 l
B Q
0.576 w
420.48 -144 m
420.48 -135.36 l
S
430.56 -82.08 m
430.56 -87.64703 426.04703 -92.16 420.48 -92.16 c
414.91297 -92.16 410.4 -87.64703 410.4 -82.08 c
410.4 -76.51297 414.91297 -72 420.48 -72 c
426.04703 -72 430.56 -76.51297 430.56 -82.08 c
S
Q
}%
    \graphtemp=.5ex
    \advance\graphtemp by 1.140in
    \rlap{\kern 5.840in\lower\graphtemp\hbox to 0pt{\hss $m_8$\hss}}%
\pdfliteral{
q [] 0 d 1 J 1 j
0.576 w
0.072 w
q 0 g
418.68 -99.36 m
420.48 -92.16 l
422.28 -99.36 l
418.68 -99.36 l
B Q
0.576 w
420.48 -108 m
420.48 -99.36 l
S
410.4 -56.16 m
430.56 -56.16 l
430.56 -36 l
410.4 -36 l
410.4 -56.16 l
S
Q
}%
    \graphtemp=.5ex
    \advance\graphtemp by 0.640in
    \rlap{\kern 5.840in\lower\graphtemp\hbox to 0pt{\hss $\enB$\hss}}%
\pdfliteral{
q [] 0 d 1 J 1 j
0.576 w
0.072 w
q 0 g
418.68 -63.36 m
420.48 -56.16 l
422.28 -63.36 l
418.68 -63.36 l
B Q
0.576 w
420.48 -72 m
420.48 -63.36 l
S
0.072 w
q 0 g
418.68 -27.36 m
420.48 -20.16 l
422.28 -27.36 l
418.68 -27.36 l
B Q
0.576 w
420.48 -36 m
420.48 -27.36 l
S
386.28 -46.08 m
386.28 -56.021125 375.803463 -64.08 362.88 -64.08 c
349.956537 -64.08 339.48 -56.021125 339.48 -46.08 c
339.48 -36.138875 349.956537 -28.08 362.88 -28.08 c
375.803463 -28.08 386.28 -36.138875 386.28 -46.08 c
S
Q
}%
    \graphtemp=\baselineskip
    \multiply\graphtemp by -1
    \divide\graphtemp by 2
    \advance\graphtemp by .5ex
    \advance\graphtemp by 0.640in
    \rlap{\kern 5.040in\lower\graphtemp\hbox to 0pt{\hss $\rB$\hss}}%
    \graphtemp=\baselineskip
    \multiply\graphtemp by 1
    \divide\graphtemp by 2
    \advance\graphtemp by .5ex
    \advance\graphtemp by 0.640in
    \rlap{\kern 5.040in\lower\graphtemp\hbox to 0pt{\hss $\tr$\hss}}%
\pdfliteral{
q [] 0 d 1 J 1 j
0.576 w
0.072 w
q 0 g
332.28 -44.28 m
339.48 -46.08 l
332.28 -47.88 l
332.28 -44.28 l
B Q
0.576 w
332.28 -46.08 m
315.36 -46.08 l
S
0.072 w
q 0 g
408.096 -100.944 m
410.4 -108 l
405.072 -102.888 l
408.096 -100.944 l
B Q
0.576 w
406.584 -101.88 m
379.44 -58.824 l
S
386.28 -118.08 m
386.28 -128.021125 375.803463 -136.08 362.88 -136.08 c
349.956537 -136.08 339.48 -128.021125 339.48 -118.08 c
339.48 -108.138875 349.956537 -100.08 362.88 -100.08 c
375.803463 -100.08 386.28 -108.138875 386.28 -118.08 c
S
Q
}%
    \graphtemp=\baselineskip
    \multiply\graphtemp by -1
    \divide\graphtemp by 2
    \advance\graphtemp by .5ex
    \advance\graphtemp by 1.640in
    \rlap{\kern 5.040in\lower\graphtemp\hbox to 0pt{\hss $\rB$\hss}}%
    \graphtemp=\baselineskip
    \multiply\graphtemp by 1
    \divide\graphtemp by 2
    \advance\graphtemp by .5ex
    \advance\graphtemp by 1.640in
    \rlap{\kern 5.040in\lower\graphtemp\hbox to 0pt{\hss $\hspace{6pt}\fa$\hss}}%
    \graphtemp=.5ex
    \advance\graphtemp by 1.640in
    \rlap{\kern 5.040in\lower\graphtemp\hbox to 0pt{\hss $\bullet$\hss}}%
\pdfliteral{
q [] 0 d 1 J 1 j
0.576 w
0.072 w
q 0 g
393.48 -119.88 m
386.28 -118.08 l
393.48 -116.28 l
393.48 -119.88 l
B Q
0.576 w
393.48 -118.08 m
410.4 -118.08 l
S
0.072 w
q 0 g
317.664 -63.216 m
315.36 -56.16 l
320.688 -61.272 l
317.664 -63.216 l
B Q
0.576 w
319.176 -62.28 m
346.32 -105.336 l
S
151.2 -106.56 m
171.36 -106.56 l
171.36 -86.4 l
151.2 -86.4 l
151.2 -106.56 l
S
Q
}%
    \graphtemp=.5ex
    \advance\graphtemp by 1.340in
    \rlap{\kern 2.240in\lower\graphtemp\hbox to 0pt{\hss $\ell_3$\hss}}%
\pdfliteral{
q [] 0 d 1 J 1 j
0.576 w
0.072 w
q 0 g
144.72 -85.608 m
151.2 -89.28 l
143.784 -89.064 l
144.72 -85.608 l
B Q
0.576 w
135 -84.744 m
144.288 -87.336 l
S
0.072 w
q 0 g
135.144 -139.752 m
130.104 -145.224 l
131.976 -138.024 l
135.144 -139.752 l
B Q
0.576 w
151.2 -106.56 m
133.56 -138.888 l
S
233.28 -96.48 m
233.28 -105.427013 225.221125 -112.68 215.28 -112.68 c
205.338875 -112.68 197.28 -105.427013 197.28 -96.48 c
197.28 -87.532987 205.338875 -80.28 215.28 -80.28 c
225.221125 -80.28 233.28 -87.532987 233.28 -96.48 c
S
Q
}%
    \graphtemp=\baselineskip
    \multiply\graphtemp by -1
    \divide\graphtemp by 2
    \advance\graphtemp by .5ex
    \advance\graphtemp by 1.340in
    \rlap{\kern 2.990in\lower\graphtemp\hbox to 0pt{\hss $\tu$\hss}}%
    \graphtemp=\baselineskip
    \multiply\graphtemp by 1
    \divide\graphtemp by 2
    \advance\graphtemp by .5ex
    \advance\graphtemp by 1.340in
    \rlap{\kern 2.990in\lower\graphtemp\hbox to 0pt{\hss $\hspace{10pt}A$\hss}}%
    \graphtemp=.5ex
    \advance\graphtemp by 1.340in
    \rlap{\kern 2.990in\lower\graphtemp\hbox to 0pt{\hss $\bullet$\hss}}%
\pdfliteral{
q [] 0 d 1 J 1 j
0.576 w
0.072 w
q 0 g
178.56 -98.28 m
171.36 -96.48 l
178.56 -94.68 l
178.56 -98.28 l
B Q
0.576 w
197.28 -96.48 m
178.56 -96.48 l
S
259.2 -106.56 m
279.36 -106.56 l
279.36 -86.4 l
259.2 -86.4 l
259.2 -106.56 l
S
Q
}%
    \graphtemp=.5ex
    \advance\graphtemp by 1.340in
    \rlap{\kern 3.740in\lower\graphtemp\hbox to 0pt{\hss $m_3$\hss}}%
\pdfliteral{
q [] 0 d 1 J 1 j
0.576 w
0.072 w
q 0 g
252 -94.68 m
259.2 -96.48 l
252 -98.28 l
252 -94.68 l
B Q
q 0 g
240.48 -98.28 m
233.28 -96.48 l
240.48 -94.68 l
240.48 -98.28 l
B Q
0.576 w
252 -96.48 m
240.48 -96.48 l
S
0.072 w
q 0 g
286.776 -89.064 m
279.36 -89.28 l
285.84 -85.608 l
286.776 -89.064 l
B Q
0.576 w
295.56 -84.744 m
286.272 -87.336 l
S
0.072 w
q 0 g
298.584 -138.024 m
300.456 -145.224 l
295.416 -139.752 l
298.584 -138.024 l
B Q
0.576 w
279.36 -106.56 m
297 -138.888 l
S
151.2 -149.76 m
171.36 -149.76 l
171.36 -129.6 l
151.2 -129.6 l
151.2 -149.76 l
S
Q
}%
    \graphtemp=.5ex
    \advance\graphtemp by 1.940in
    \rlap{\kern 2.240in\lower\graphtemp\hbox to 0pt{\hss $\ell_3$\hss}}%
\pdfliteral{
q [] 0 d 1 J 1 j
0.576 w
0.072 w
q 0 g
149.328 -122.4 m
151.2 -129.6 l
146.16 -124.128 l
149.328 -122.4 l
B Q
0.576 w
130.104 -90.936 m
147.744 -123.264 l
S
0.072 w
q 0 g
142.416 -151.2 m
135 -151.416 l
141.48 -147.744 l
142.416 -151.2 l
B Q
0.576 w
151.2 -146.88 m
141.912 -149.472 l
S
233.28 -139.68 m
233.28 -148.627013 225.221125 -155.88 215.28 -155.88 c
205.338875 -155.88 197.28 -148.627013 197.28 -139.68 c
197.28 -130.732987 205.338875 -123.48 215.28 -123.48 c
225.221125 -123.48 233.28 -130.732987 233.28 -139.68 c
S
Q
}%
    \graphtemp=\baselineskip
    \multiply\graphtemp by -1
    \divide\graphtemp by 2
    \advance\graphtemp by .5ex
    \advance\graphtemp by 1.940in
    \rlap{\kern 2.990in\lower\graphtemp\hbox to 0pt{\hss $\tu$\hss}}%
    \graphtemp=\baselineskip
    \multiply\graphtemp by 1
    \divide\graphtemp by 2
    \advance\graphtemp by .5ex
    \advance\graphtemp by 1.940in
    \rlap{\kern 2.990in\lower\graphtemp\hbox to 0pt{\hss $B$\hss}}%
\pdfliteral{
q [] 0 d 1 J 1 j
0.576 w
0.072 w
q 0 g
197.64 -122.616 m
202.536 -128.232 l
195.624 -125.568 l
197.64 -122.616 l
B Q
0.576 w
171.36 -106.56 m
196.632 -124.128 l
S
0.072 w
q 0 g
178.56 -141.48 m
171.36 -139.68 l
178.56 -137.88 l
178.56 -141.48 l
B Q
q 0 g
190.08 -137.88 m
197.28 -139.68 l
190.08 -141.48 l
190.08 -137.88 l
B Q
0.576 w
178.56 -139.68 m
190.08 -139.68 l
S
259.2 -149.76 m
279.36 -149.76 l
279.36 -129.6 l
259.2 -129.6 l
259.2 -149.76 l
S
Q
}%
    \graphtemp=.5ex
    \advance\graphtemp by 1.940in
    \rlap{\kern 3.740in\lower\graphtemp\hbox to 0pt{\hss $m_3$\hss}}%
\pdfliteral{
q [] 0 d 1 J 1 j
0.576 w
0.072 w
q 0 g
252 -137.88 m
259.2 -139.68 l
252 -141.48 l
252 -137.88 l
B Q
0.576 w
233.28 -139.68 m
252 -139.68 l
S
0.072 w
q 0 g
232.92 -113.544 m
228.024 -107.928 l
234.936 -110.592 l
232.92 -113.544 l
B Q
0.576 w
259.2 -129.6 m
233.928 -112.032 l
S
0.072 w
q 0 g
284.4 -124.128 m
279.36 -129.6 l
281.232 -122.4 l
284.4 -124.128 l
B Q
0.576 w
300.456 -90.936 m
282.816 -123.264 l
S
0.072 w
q 0 g
289.08 -147.744 m
295.56 -151.416 l
288.144 -151.2 l
289.08 -147.744 l
B Q
0.576 w
279.36 -146.88 m
288.648 -149.472 l
S
151.2 -192.96 m
171.36 -192.96 l
171.36 -172.8 l
151.2 -172.8 l
151.2 -192.96 l
S
Q
}%
    \graphtemp=.5ex
    \advance\graphtemp by 2.540in
    \rlap{\kern 2.240in\lower\graphtemp\hbox to 0pt{\hss $\ell_4^{\,t\!}$\hss}}%
\pdfliteral{
q [] 0 d 1 J 1 j
0.576 w
0.072 w
q 0 g
172.728 -167.328 m
167.76 -172.8 l
169.56 -165.6 l
172.728 -167.328 l
B Q
q 0 g
197.568 -113.4 m
202.536 -107.928 l
200.736 -115.128 l
197.568 -113.4 l
B Q
0.576 w
171.144 -166.464 m
199.152 -114.264 l
S
0.072 w
q 0 g
146.808 -169.704 m
151.2 -175.68 l
144.504 -172.44 l
146.808 -169.704 l
B Q
0.576 w
133.056 -160.56 m
145.656 -171.072 l
S
0.072 w
q 0 g
137.376 -213.552 m
131.472 -218.16 l
134.496 -211.392 l
137.376 -213.552 l
B Q
0.576 w
151.2 -192.96 m
135.936 -212.472 l
S
259.2 -192.96 m
279.36 -192.96 l
279.36 -172.8 l
259.2 -172.8 l
259.2 -192.96 l
S
Q
}%
    \graphtemp=.5ex
    \advance\graphtemp by 2.540in
    \rlap{\kern 3.740in\lower\graphtemp\hbox to 0pt{\hss $m_4^{\,t\!}$\hss}}%
\pdfliteral{
q [] 0 d 1 J 1 j
0.576 w
0.072 w
q 0 g
254.304 -167.184 m
259.2 -172.8 l
252.288 -170.136 l
254.304 -167.184 l
B Q
q 0 g
232.92 -156.744 m
228.024 -151.128 l
234.936 -153.792 l
232.92 -156.744 l
B Q
0.576 w
253.296 -168.696 m
233.928 -155.232 l
S
0.072 w
q 0 g
286.056 -172.44 m
279.36 -175.68 l
283.752 -169.704 l
286.056 -172.44 l
B Q
0.576 w
297.504 -160.56 m
284.904 -171.072 l
S
0.072 w
q 0 g
296.064 -211.392 m
299.088 -218.16 l
293.184 -213.552 l
296.064 -211.392 l
B Q
0.576 w
279.36 -192.96 m
294.624 -212.472 l
S
151.2 -236.16 m
171.36 -236.16 l
171.36 -216 l
151.2 -216 l
151.2 -236.16 l
S
Q
}%
    \graphtemp=.5ex
    \advance\graphtemp by 3.140in
    \rlap{\kern 2.240in\lower\graphtemp\hbox to 0pt{\hss $\ell_4^{B\!}$\hss}}%
\pdfliteral{
q [] 0 d 1 J 1 j
0.576 w
0.072 w
q 0 g
150.048 -208.656 m
151.2 -216 l
146.736 -210.024 l
150.048 -208.656 l
B Q
0.576 w
129.168 -163.368 m
148.392 -209.376 l
S
0.072 w
q 0 g
142.56 -227.88 m
135.36 -226.08 l
142.56 -224.28 l
142.56 -227.88 l
B Q
0.576 w
151.2 -226.08 m
142.56 -226.08 l
S
0.072 w
q 0 g
176.112 -241.848 m
171.36 -236.16 l
178.272 -238.968 l
176.112 -241.848 l
B Q
q 0 g
361.224 -143.352 m
362.88 -136.08 l
364.824 -143.208 l
361.224 -143.352 l
B Q
0.576 w
363.020347 -143.27562 m
362.676093 -187.173517 338.151816 -227.303692 299.244455 -247.634976 c
260.337093 -267.96626 213.388902 -265.184468 177.154704 -240.400869 c
S
259.2 -236.16 m
279.36 -236.16 l
279.36 -216 l
259.2 -216 l
259.2 -236.16 l
S
Q
}%
    \graphtemp=.5ex
    \advance\graphtemp by 3.140in
    \rlap{\kern 3.740in\lower\graphtemp\hbox to 0pt{\hss $m_4^{A\!}$\hss}}%
\pdfliteral{
q [] 0 d 1 J 1 j
0.576 w
0.072 w
q 0 g
283.824 -210.024 m
279.36 -216 l
280.512 -208.656 l
283.824 -210.024 l
B Q
0.576 w
301.392 -163.368 m
282.168 -209.376 l
S
0.072 w
q 0 g
288 -224.28 m
295.2 -226.08 l
288 -227.88 l
288 -224.28 l
B Q
0.576 w
279.36 -226.08 m
288 -226.08 l
S
0.072 w
q 0 g
65.736 -143.208 m
67.68 -136.08 l
69.336 -143.352 l
65.736 -143.208 l
B Q
q 0 g
252.288 -238.968 m
259.2 -236.16 l
254.448 -241.848 l
252.288 -238.968 l
B Q
0.576 w
253.405262 -240.400892 m
217.171056 -265.184479 170.222863 -267.966254 131.315509 -247.634957 c
92.408154 -227.303659 67.883892 -187.173476 67.539653 -143.275579 c
S
Q
}%
    \hbox{\vrule depth3.627in width0pt height 0pt}%
    \kern 5.980in
  }%
}%

%% file: smallPN.tex
\expandafter\ifx\csname graph\endcsname\relax
   \csname newbox\expandafter\endcsname\csname graph\endcsname
\fi
\ifx\graphtemp\undefined
  \csname newdimen\endcsname\graphtemp
\fi
\expandafter\setbox\csname graph\endcsname
 =\vtop{\vskip 0pt\hbox{%
\pdfliteral{
q [] 0 d 1 J 1 j
0.576 w
0.576 w
57.6 -6.84 m
57.6 -10.617628 44.705801 -13.68 28.8 -13.68 c
12.894199 -13.68 0 -10.617628 0 -6.84 c
0 -3.062372 12.894199 0 28.8 0 c
44.705801 0 57.6 -3.062372 57.6 -6.84 c
S
Q
}%
    \graphtemp=.5ex
    \advance\graphtemp by 0.095in
    \rlap{\kern 0.400in\lower\graphtemp\hbox to 0pt{\hss \tiny $\rA=\fa$\hss}}%
\pdfliteral{
q [] 0 d 1 J 1 j
0.576 w
90.72 -13.68 m
110.88 -13.68 l
110.88 0 l
90.72 0 l
90.72 -13.68 l
S
Q
}%
    \graphtemp=.5ex
    \advance\graphtemp by 0.095in
    \rlap{\kern 1.400in\lower\graphtemp\hbox to 0pt{\hss $m_4^{\scriptscriptstyle A\!}$\hss}}%
\pdfliteral{
q [] 0 d 1 J 1 j
0.576 w
0.072 w
q 0 g
64.8 -8.64 m
57.6 -6.84 l
64.8 -5.04 l
64.8 -8.64 l
B Q
q 0 g
83.52 -5.04 m
90.72 -6.84 l
83.52 -8.64 l
83.52 -5.04 l
B Q
0.576 w
64.8 -6.84 m
83.52 -6.84 l
S
Q
}%
    \hbox{\vrule depth0.190in width0pt height 0pt}%
    \kern 1.540in
  }%
}%

%% file: smallPN2.tex
\expandafter\ifx\csname graph\endcsname\relax
   \csname newbox\expandafter\endcsname\csname graph\endcsname
\fi
\ifx\graphtemp\undefined
  \csname newdimen\endcsname\graphtemp
\fi
\expandafter\setbox\csname graph\endcsname
 =\vtop{\vskip 0pt\hbox{%
\pdfliteral{
q [] 0 d 1 J 1 j
0.576 w
0.576 w
110.88 -7.2 m
110.88 -10.977628 97.985801 -14.04 82.08 -14.04 c
66.174199 -14.04 53.28 -10.977628 53.28 -7.2 c
53.28 -3.422372 66.174199 -0.36 82.08 -0.36 c
97.985801 -0.36 110.88 -3.422372 110.88 -7.2 c
S
Q
}%
    \graphtemp=.5ex
    \advance\graphtemp by 0.100in
    \rlap{\kern 1.140in\lower\graphtemp\hbox to 0pt{\hss \tiny $\rB=\fa$\hss}}%
\pdfliteral{
q [] 0 d 1 J 1 j
0.576 w
0 -14.4 m
20.16 -14.4 l
20.16 0 l
0 0 l
0 -14.4 l
S
Q
}%
    \graphtemp=.5ex
    \advance\graphtemp by 0.100in
    \rlap{\kern 0.140in\lower\graphtemp\hbox to 0pt{\hss $\ell_4^{\scriptscriptstyle B\!}$\hss}}%
\pdfliteral{
q [] 0 d 1 J 1 j
0.576 w
0.072 w
q 0 g
46.08 -5.4 m
53.28 -7.2 l
46.08 -9 l
46.08 -5.4 l
B Q
q 0 g
27.36 -9 m
20.16 -7.2 l
27.36 -5.4 l
27.36 -9 l
B Q
0.576 w
46.08 -7.2 m
27.36 -7.2 l
S
Q
}%
    \hbox{\vrule depth0.200in width0pt height 0pt}%
    \kern 1.540in
  }%
}%

%% file: Boolean.tex
\expandafter\ifx\csname graph\endcsname\relax
   \csname newbox\expandafter\endcsname\csname graph\endcsname
\fi
\ifx\graphtemp\undefined
  \csname newdimen\endcsname\graphtemp
\fi
\expandafter\setbox\csname graph\endcsname
 =\vtop{\vskip 0pt\hbox{%
\pdfliteral{
q [] 0 d 1 J 1 j
0.576 w
0.576 w
51.12 -30.24 m
51.12 -37.39761 45.31761 -43.2 38.16 -43.2 c
31.00239 -43.2 25.2 -37.39761 25.2 -30.24 c
25.2 -23.08239 31.00239 -17.28 38.16 -17.28 c
45.31761 -17.28 51.12 -23.08239 51.12 -30.24 c
S
Q
}%
    \graphtemp=.5ex
    \advance\graphtemp by 0.420in
    \rlap{\kern 0.530in\lower\graphtemp\hbox to 0pt{\hss $\tr$\hss}}%
\pdfliteral{
q [] 0 d 1 J 1 j
0.576 w
0.072 w
q 0 g
39.96 -10.08 m
38.16 -17.28 l
36.36 -10.08 l
39.96 -10.08 l
B Q
0.576 w
38.16 0 m
38.16 -10.08 l
S
0.072 w
q 0 g
24.192 -15.408 m
29.016 -21.096 l
22.104 -18.36 l
24.192 -15.408 l
B Q
0.576 w
25.2 -30.24 m
13.68 -30.24 l
5.8464 -30.24 2.16 -27.36 2.16 -21.24 c
2.16 -15.12 4.464 -12.24 9.36 -12.24 c
14.256 -12.24 18.4608 -13.58784 22.5 -16.452 c
28.44 -20.664 l
S
Q
}%
    \graphtemp=.5ex
    \advance\graphtemp by 0.120in
    \rlap{\kern 0.030in\lower\graphtemp\hbox to 0pt{\hss $\overline{\noti{x}{\tr}}$\hss}}%
\pdfliteral{
q [] 0 d 1 J 1 j
0.576 w
0.072 w
q 0 g
23.328 -44.208 m
29.016 -39.384 l
26.28 -46.296 l
23.328 -44.208 l
B Q
0.576 w
38.16 -43.2 m
38.16 -54.72 l
38.16 -62.5536 35.28 -66.24 29.16 -66.24 c
23.04 -66.24 20.16 -63.936 20.16 -59.04 c
20.16 -54.144 21.50784 -49.9392 24.372 -45.9 c
28.584 -39.96 l
S
Q
}%
    \graphtemp=.5ex
    \advance\graphtemp by 0.720in
    \rlap{\kern 0.000in\lower\graphtemp\hbox to 0pt{\hss $\ass{x}{\tr}$\hss}}%
\pdfliteral{
q [] 0 d 1 J 1 j
0.576 w
144.72 -30.24 m
144.72 -37.39761 138.91761 -43.2 131.76 -43.2 c
124.60239 -43.2 118.8 -37.39761 118.8 -30.24 c
118.8 -23.08239 124.60239 -17.28 131.76 -17.28 c
138.91761 -17.28 144.72 -23.08239 144.72 -30.24 c
S
Q
}%
    \graphtemp=.5ex
    \advance\graphtemp by 0.420in
    \rlap{\kern 1.830in\lower\graphtemp\hbox to 0pt{\hss $\fa$\hss}}%
\pdfliteral{
q [] 0 d 1 J 1 j
0.576 w
0.072 w
q 0 g
116.064 -17.64 m
122.616 -21.096 l
115.2 -21.096 l
116.064 -17.64 l
B Q
0.576 w
47.324054 -21.061268 m
69.617134 -15.024916 93.035108 -14.43934 115.602021 -19.353951 c
S
Q
}%
    \graphtemp=.5ex
    \advance\graphtemp by 0.120in
    \rlap{\kern 1.200in\lower\graphtemp\hbox to 0pt{\hss $\ass{x}{\fa}$\hss}}%
\pdfliteral{
q [] 0 d 1 J 1 j
0.576 w
0.072 w
q 0 g
53.856 -42.84 m
47.304 -39.384 l
54.72 -39.384 l
53.856 -42.84 l
B Q
0.576 w
122.595898 -39.490745 m
100.302816 -45.52709 76.884842 -46.112657 54.31793 -41.198038 c
S
Q
}%
    \graphtemp=.5ex
    \advance\graphtemp by 0.740in
    \rlap{\kern 1.200in\lower\graphtemp\hbox to 0pt{\hss $\ass{x}{\tr}$\hss}}%
\pdfliteral{
q [] 0 d 1 J 1 j
0.576 w
0.072 w
q 0 g
143.64 -46.296 m
140.904 -39.384 l
146.592 -44.208 l
143.64 -46.296 l
B Q
0.576 w
131.76 -43.2 m
131.76 -54.72 l
131.76 -62.5536 134.64 -66.24 140.76 -66.24 c
146.88 -66.24 149.76 -63.936 149.76 -59.04 c
149.76 -54.144 148.41216 -49.9392 145.548 -45.9 c
141.336 -39.96 l
S
Q
}%
    \graphtemp=.5ex
    \advance\graphtemp by 0.720in
    \rlap{\kern 2.380in\lower\graphtemp\hbox to 0pt{\hss $\ass{x}{\fa}$\hss}}%
\pdfliteral{
q [] 0 d 1 J 1 j
0.576 w
0.072 w
q 0 g
147.816 -18.36 m
140.904 -21.096 l
145.728 -15.408 l
147.816 -18.36 l
B Q
0.576 w
144.72 -30.24 m
156.24 -30.24 l
164.0736 -30.24 167.76 -27.36 167.76 -21.24 c
167.76 -15.12 165.456 -12.24 160.56 -12.24 c
155.664 -12.24 151.4592 -13.58784 147.42 -16.452 c
141.48 -20.664 l
S
Q
}%
    \graphtemp=.5ex
    \advance\graphtemp by 0.120in
    \rlap{\kern 2.130in\lower\graphtemp\hbox to 0pt{\hss $\overline{\noti{x}{\fa}}$\hss}}%
    \hbox{\vrule depth0.895in width0pt height 0pt}%
    \kern 2.380in
  }%
}%

%% file: pretzelOff.tex
\expandafter\ifx\csname graph\endcsname\relax
   \csname newbox\expandafter\endcsname\csname graph\endcsname
\fi
\ifx\graphtemp\undefined
  \csname newdimen\endcsname\graphtemp
\fi
\expandafter\setbox\csname graph\endcsname
 =\vtop{\vskip 0pt\hbox{%
\pdfliteral{
q [] 0 d 1 J 1 j
0.576 w
0.576 w
31.968 -20.016 m
31.968 -24.867269 28.035269 -28.8 23.184 -28.8 c
18.332731 -28.8 14.4 -24.867269 14.4 -20.016 c
14.4 -15.164731 18.332731 -11.232 23.184 -11.232 c
28.035269 -11.232 31.968 -15.164731 31.968 -20.016 c
S
72 -20.016 m
72 -24.867269 68.067269 -28.8 63.216 -28.8 c
58.364731 -28.8 54.432 -24.867269 54.432 -20.016 c
54.432 -15.164731 58.364731 -11.232 63.216 -11.232 c
68.067269 -11.232 72 -15.164731 72 -20.016 c
S
0.072 w
q 0 g
50.256 -29.304 m
56.952 -26.208 l
52.488 -32.112 l
50.256 -29.304 l
B Q
0.576 w
51.378988 -30.650468 m
43.916205 -34.457143 34.813838 -32.612005 29.422545 -26.199684 c
S
0.072 w
q 0 g
36.144 -10.728 m
29.448 -13.752 l
33.912 -7.92 l
36.144 -10.728 l
B Q
0.576 w
35.021022 -9.309527 m
42.483808 -5.502856 51.586173 -7.348 56.977463 -13.760325 c
S
Q
}%
    \graphtemp=.5ex
    \advance\graphtemp by 0.389in
    \rlap{\kern 0.600in\lower\graphtemp\hbox to 0pt{\hss $c$\hss}}%
    \graphtemp=.5ex
    \advance\graphtemp by 0.000in
    \rlap{\kern 0.600in\lower\graphtemp\hbox to 0pt{\hss $p$\hss}}%
\pdfliteral{
q [] 0 d 1 J 1 j
0.576 w
31.968 -59.976 m
31.968 -64.827269 28.035269 -68.76 23.184 -68.76 c
18.332731 -68.76 14.4 -64.827269 14.4 -59.976 c
14.4 -55.124731 18.332731 -51.192 23.184 -51.192 c
28.035269 -51.192 31.968 -55.124731 31.968 -59.976 c
S
0.072 w
q 0 g
39.384 -59.688 m
31.968 -59.976 l
38.376 -56.232 l
39.384 -59.688 l
B Q
0.576 w
42.148237 -56.477747 m
41.077964 -57.023006 39.983692 -57.519818 38.868744 -57.966676 c
S
48.230549 -52.655729 m
47.274956 -53.383489 46.287093 -54.06787 45.269957 -54.706797 c
S
53.531867 -47.808217 m
52.721724 -48.695043 51.872077 -49.544957 50.985506 -50.355379 c
S
57.881488 -42.091301 m
57.242882 -43.108637 56.55881 -44.096716 55.831352 -45.052538 c
S
61.13937 -35.689026 m
60.692864 -36.804115 60.196397 -37.898543 59.651475 -38.968988 c
S
63.200573 -28.807617 m
62.960544 -29.984553 62.667666 -31.15009 62.322827 -32.300689 c
S
0.072 w
q 0 g
35.28 -49.248 m
29.448 -53.784 l
32.472 -47.016 l
35.28 -49.248 l
B Q
0.576 w
29.439679 -26.254541 m
35.852002 -31.645832 37.697143 -40.748198 33.89047 -48.210983 c
S
0.072 w
q 0 g
11.088 -30.744 m
16.992 -26.208 l
13.896 -32.976 l
11.088 -30.744 l
B Q
0.576 w
17.000316 -53.809455 m
10.587995 -48.418161 8.742857 -39.315795 12.549532 -31.853012 c
S
Q
}%
    \graphtemp=.5ex
    \advance\graphtemp by 0.556in
    \rlap{\kern 0.433in\lower\graphtemp\hbox to 0pt{\hss $\rt$\hss}}%
    \graphtemp=.5ex
    \advance\graphtemp by 0.750in
    \rlap{\kern 0.794in\lower\graphtemp\hbox to 0pt{\hss $\rt$\hss}}%
    \graphtemp=.5ex
    \advance\graphtemp by 0.556in
    \rlap{\kern 0.028in\lower\graphtemp\hbox to 0pt{\hss {\it on}\hss}}%
\pdfliteral{
q [] 0 d 1 J 1 j
0.576 w
0.072 w
q 0 g
7.2 -58.176 m
14.4 -59.976 l
7.2 -61.776 l
7.2 -58.176 l
B Q
0.576 w
0 -59.976 m
7.2 -59.976 l
S
Q
}%
    \hbox{\vrule depth0.956in width0pt height 0pt}%
    \kern 1.000in
  }%
}%